\documentclass[12pt,twoside,leftblank,vi]{mitthesis}
\usepackage{lgrind}
\usepackage{color}
\usepackage{graphicx}
\usepackage{amsmath}
\usepackage{latexsym}
\usepackage{epstopdf}
\usepackage{amsmath,amssymb}

\usepackage{lipsum}
\usepackage{cite}
\usepackage{bm}
\usepackage{amsmath,amssymb,amsfonts,amsthm}
\usepackage{environ}
\usepackage{graphicx}
\usepackage{textcomp}
\usepackage{xcolor}
\usepackage{tabularx}
\usepackage{mathtools}
\usepackage{multirow}
\usepackage[makeroom]{cancel}
\def\BibTeX{{\rm B\kern-.05em{\sc i\kern-.025em b}\kern-.08em
    T\kern-.1667em\lower.7ex\hbox{E}\kern-.125emX}}
\usepackage{etoolbox}

\usepackage{wrapfig}
\usepackage{float}

\makeatletter
\patchcmd{\@makecaption}
  {\scshape}
  {}
  {}
  {}
\makeatletter
\patchcmd{\@makecaption}
  {\\}
  {.\ }
  {}
  {}
\makeatother

\DeclareMathOperator*{\argmax}{arg\,max}

\newcommand{\x}{\mathbf{x}}
\newcommand{\Ep}{\mathbb{E}_{p(\mathbf{x})}}
\newcommand{\A}{\mathbf{A}}
\newcommand{\I}{\mathbf{I}}
\newcommand{\W}{\mathbf{W}}

\newcommand{\D}{\mathbf{D}}
\newcommand{\E}{\mathbf{E}}
\newcommand{\F}{\mathbf{F}}
\newcommand{\diag}{{\rm diag}}
\newcommand{\s}{\mathbf{s}}
\newcommand{\bmu}{\bm{\mu}}
\newcommand{\bkappa}{\bm{\kappa}}
\newcommand{\btheta}{\bm{\theta}}
\newcommand{\bphi}{\bm{\phi}}

\newcommand{\bSigma}{\mathbf{\Sigma}}
\newcommand{\z}{\mathbf{z}}
\newcommand{\bb}{\mathbf{b}}
\newcommand{\y}{\mathbf{y}}
\newcommand{\N}{\mathcal{N}}
\newcommand{\C}{\mathbf{C}}
\newcommand{\U}{\mathbf{U}}
\newcommand{\V}{\mathbf{V}}
\newcommand{\B}{\mathbf{B}}
\newcommand{\J}{\mathbf{J}}

\usepackage{amsmath}
\usepackage{amssymb}
\usepackage{times}
\usepackage{graphicx}
\usepackage{color}
\usepackage{multirow}
\usepackage{epsfig}
\usepackage{url}
\usepackage{dcolumn}% Align table columns on decimal point
\usepackage{subcaption}
\usepackage[numbers]{natbib}
\usepackage{rotating}
\usepackage{bbm}
\usepackage{latexsym}

\usepackage{amsmath,verbatim,latexsym,amssymb,indentfirst,mathrsfs,mathtools,amsthm,bbm,bm,hyperref,url,cancel}%subcaption
\usepackage[font=small,labelfont=bf,format=plain,justification=raggedright,singlelinecheck=false]{caption}
\usepackage{verbatim,indentfirst}
\usepackage[title,titletoc]{appendix}

\usepackage{epstopdf}
\newcommand{\caphead}[1]{{\bf #1}}

\renewcommand{\thesection}{\Roman{section}}

\makeatletter
\def\p@subsection{}
\makeatother
\makeatletter
\def\p@subsubsection{}
\makeatother

\newtheorem{corollary}{Corollary}

\newtheorem{proposition}{Proposition}

\makeatletter
\newcommand\footnoteref[1]{\protected@xdef\@thefnmark{\ref{#1}}\@footnotemark}
\makeatother

\usepackage{soul}

\usepackage{color}
\newcommand{\KL}{{\rm KL}}

\newcommand{\Tr}{{\rm Tr}}   % Trace
\def\id{\mathbbm{1}}   % Identity
\newcommand{\kB}{k_\mathrm{B}}  % k_B
\newcommand{\Sites}{N}  % System size
\newcommand{\LParen}{ \bm{(} }
\newcommand{\RParen}{ \bm{)} }

\newcommand*{\Set}[1]{\left\{  #1  \right\}}

\renewcommand\th{ {\rm th} }

\newcommand*{\bra}[1]{\langle #1\rvert}
\newcommand*{\ket}[1]{\lvert #1 \rangle}

\usepackage{color}

\usepackage{graphicx}
\usepackage{amsmath}
\usepackage{latexsym}
\usepackage{epstopdf}
\usepackage{amsmath,amssymb}
\usepackage{algorithm}
\usepackage{algpseudocode}
\usepackage{bm}
\renewcommand{\alglinenumber}[1]{}

\usepackage{multirow}
\usepackage{mathtools}
\usepackage{physics}
\usepackage[normalem]{ulem} % to strikethrough text

\newcommand{\beq}{\begin{equation}}
\newcommand{\eeq}{\end{equation}}

\def\bra#1{\langle #1|}
\def\ket#1{|#1\rangle}
\newcommand{\h}{\mathbf{h}}
\newcommand{\vis}{\mathbf{v}}
\newcommand{\bTheta}{\mathbf{\Theta}}

\usepackage{tikz}
\usetikzlibrary{patterns}
\allowdisplaybreaks

\usepackage{comment}

\newcommand\HSTeq{\stackrel{\mathclap{\normalfont\mbox{HST}}}{=}}

\usepackage{cmap}
\usepackage[T1]{fontenc}
\def\all{all}
\ifx\files\all  \else %\typeout{Including only \files.} \includeonly{\files} \fi

\begin{document}

\title{NON-EQUILIBRIUM PHYSICS: \\ FROM SPIN GLASSES TO MACHINE AND NEURAL LEARNING}

\author{Weishun Zhong}
% If you wish to list your previous degrees on the cover page, use the 
% previous degrees command:
%       \prevdegrees{A.A., Harvard University (1985)}
% You can use the \\ command to list multiple previous degrees
%       \prevdegrees{B.S., University of California (1978) \\
%                    S.M., Massachusetts Institute of Technology (1981)}
\department{Department of Physics}

% If the thesis is for two degrees simultaneously, list them both
% separated by \and like this:
% \degree{Doctor of Philosophy \and Master of Science}
\degree{Doctor of Philosophy in Physics}

% As of the 2007-08 academic year, valid degree months are September, 
% February, or June.  The default is June.
\degreemonth{June}
\degreeyear{2023}
\thesisdate{May 19, 2023}
%% By default, the thesis will be copyrighted to MIT.  If you need to copyright
%% the thesis to yourself, just specify the `vi' documentclass option.  If for
%% some reason you want to exactly specify the copyright notice text, you can
%% use the \copyrightnoticetext command.  
%\copyrightnoticetext{\copyright IBM, 1990.  Do not open till Xmas.}

% If there is more than one supervisor, use the \supervisor command
% once for each.
\supervisor{Haim Sompolinsky}{Professor of Molecular and Cellular Biology and of Physics (in Residence), \\Harvard University}
\supervisor{Mehran Kardar}{Francis Friedman Professor of Physics}

% This is the department committee chairman, not the thesis committee
% chairman.  You should replace this with your Department's Committee
% Chairman.
\chairman{Lindley Winslow}{Associate Department Head of Physics}

% Make the titlepage based on the above information.  If you need
% something special and can't use the standard form, you can specify
% the exact text of the titlepage yourself.  Put it in a titlepage
% environment and leave blank lines where you want vertical space.
% The spaces will be adjusted to fill the entire page.  The dotted
% lines for the signatures are made with the \signature command.
\maketitle

% The abstractpage environment sets up everything on the page except
% the text itself.  The title and other header material are put at the
% top of the page, and the supervisors are listed at the bottom.  A
% new page is begun both before and after.  Of course, an abstract may
% be more than one page itself.  If you need more control over the
% format of the page, you can use the abstract environment, which puts
% the word "Abstract" at the beginning and single spaces its text.

%% You can either \input (*not* \include) your abstract file, or you can put
%% the text of the abstract directly between the \begin{abstractpage} and
%% \end{abstractpage} commands.

% First copy: start a new page, and save the page number.
\cleardoublepage
% Uncomment the next line if you do NOT want a page number on your
% abstract and acknowledgments pages.
% \pagestyle{empty}
\setcounter{savepage}{\thepage}
\begin{abstractpage}

Disordered many-body systems exhibit a wide range of emergent phenomena across different scales. These complex behaviors can be utilized for various information processing tasks such as error correction, learning, and optimization. Despite the empirical success of utilizing these systems for intelligent tasks, the underlying principles that govern their emergent intelligent behaviors remain largely unknown. In this thesis, we aim to characterize such emergent intelligence in disordered systems through statistical physics. We chart a roadmap for our efforts in this thesis based on two axes: learning mechanisms (long-term memory vs. working memory) and learning dynamics (artificial vs. natural). We begin our exploration from the long-term memory and artificial dynamics continent of this atlas, where we examine the structure-function relationships in feedforward neural networks, the prototypical example of neural learning. Using replica theory, information theory, and optimal transport, we study the computational consequences of imposing connectivity constraints on the network, such as distribution constraints, sign constraints, and disentangling constraints. We evaluate the performances based on metrics such as capacity, generalization, and generative ability. Next, we explore the working memory and artificial dynamics corner of the atlas and investigate the non-equilibrium driven dynamics of recurrent neural networks under external inputs. Then, we move to the working memory and natural dynamics island and study the ability of driven spin-glasses to perform discriminative tasks such as novelty detection and classification. Finally, we conclude our exploration at the long-term memory and natural dynamics kingdom and investigate the generative modeling ability in many-body localized systems. Throughout our journey, we uncover relationships between learning mechanisms and physical dynamics that could serve as guiding principles for designing intelligent systems. We hope that our investigation into the emergent intelligence of seemingly disparate learning systems can expand our current understanding of intelligence beyond neural systems and uncover a wider range of computational substrates suitable for AI applications.

\end{abstractpage}

% Additional copy: start a new page, and reset the page number.  This way,
% the second copy of the abstract is not counted as separate pages.
% Uncomment the next 6 lines if you need two copies of the abstract
% page.
% \setcounter{page}{\thesavepage}
% \begin{abstractpage}
% \input{abstract}
% \end{abstractpage}

\cleardoublepage

\clearpage
\begin{center}
    \thispagestyle{empty}
    \vspace*{\fill}
    This thesis is dedicated to my beloved mother, Qianyu Hu. 
    \vspace*{\fill}
\end{center}
\clearpage

\section*{Acknowledgments}

Six years ago, I wanted to pursue a PhD because I thought there is no better way to get a front seat in the exciting adventure of scientific discovery. 
I remember immediately clicking the acceptance button upon reading the sentence of "together we are pushing back the frontiers of human understanding of space and time and of matter and energy in 
all its forms, from the subatomic to the cosmological and from the elementary to the complex", in the opening paragraph of MIT physics's admission offer letter. Only much later, I would find out that the second paragraphs is a warning that this would normally take five to six years. 

% It has been incredibly exciting and rewarding six years, and I have gained more than what I expected - not just a passive audience, but also a participant myself in this epic journey of mankind. I want to express my deepest gratitude to those who made it possible.

These past six years have been extraordinarily exciting and fulfilling, exceeding my expectations. I have not only been a passive spectator but also an active participant in this epic journey of mankind. I would like to extend my profound gratitude to all those who made this incredible experience possible.

I would like to express my deepest gratitude to my thesis advisor, Prof. Haim Sompolinsky, whose invaluable guidance and expertise have molded me into the scientist I am today. I am especially grateful to Haim for demonstrating the importance of adhering to the highest standards of scientific rigor and for teaching me to relentlessly confront challenging problems with unwavering determination. Furthermore, I want to thank him for being an inspiring mentor who consistently encourages me to strive for excellence, while never ceasing to pursue it himself. I also want to thank him for being a true role model of a leader in his field, and for taking me under his wings years ago when I felt lost. 

I am immensely grateful to my thesis coadvisor, Prof. Mehran Kardar, for his unwavering support and invaluable guidance throughout my graduate school journey. I particularly appreciate Mehran's encouragement to explore my own research ideas and forge my unique path. Additionally, I want to thank him for sharing his vast knowledge of physics with me, covering almost every aspects of statistical mechanics, and for serving as an exemplary role model of a great statistical mechanician. 

I would like to express my gratitude to Prof. Jeremy England for welcoming me into his group upon my arrival at MIT. I want to thank him for sharing his vision of life-like physics with me, which have profoundly influenced both my scientific and philosophical perspectives on the world. Furthermore, I want to thank him for demonstrating how statistical physics can shed light on the living world, which played a significant role in inspiring me to pursue graduate studies.

I am greatly indebted to Prof. Susanne Yelin and Prof. Nicole Yunger Halpern for their generosity with their time and invaluable help during my postdoc applications. I also want to thank them for being outstanding collaborators and mentors.

I am also grateful to Prof. Frank Wolfs for the unwavering support, guidance and kindness throughout my academic journey. His example of bravery and living life with a full heart has been an inspiration to me.

I would like to extend a special thanks to my thesis committee members, Profs. Leonid Mirny and Marin Soljačić. Additionally, I am grateful for the interactions I had with faculty members from MIT Physics, including Profs. Nikta Fakhri, Jeff Gore, and Daniel Harlow.

I also want to thank all my collaborators, who I learned much of my knowledge from: Daniel D. Lee, Jacob M. Gold, Xun Gao, Sarah Marzen, Arvind Murugan, Khadijeh Najafi, Cengiz Pehlevan, Zhiyue Lu, Ramis Movassagh, Harshvardhan Sikka, Ben Sorscher, David J. Schwab, Oles Shtanko, Yoav Soen.

Next, I want to thank my dearest friends, without whom this journey wouldn't have been as enjoyable. I want to thank members from the Sompolinsky group for the camaraderie: Madhu Advani, Sueyeon Chung, Ouns el Harzli, Naoki Hiratani, Qianyi Li, Haozhe Shan, Nimrod Shaham, Shane Shang, Julia Steinberg, Alexander van Meegen, Zechen Zhang. Our lunchtime conversations were always the highlight of my day, and I cherish those memories. 

I want to thank my friends from the England group and the Physics of Living systems, for making me feel at home even during  trying times: Gili Bisker, Pavel Chvykov, Todd Gingrich, Jacob Gold, Jordan Horowitz, Hridesh Kedia, Jinghui Liu, Jeremy Owen. 

I want to thank my friends from Cambridge, for all the fun memories and making here my second hometown: Anqi Chen, Simon Grosse-Holz, Emil Khabiboulline, Chengfeng Mao, Daniya Seitova, Yue Wang, Ming Zheng, Zhenghao Fu, Justin Hou, Shang Liu, Tongtong Liu, Ruihao Zhu. 

I would also like to express my heartfelt gratitude to my wife, Jun Yin, for all the laughs and tears we shared, for always being there for me through the ups and downs, and for the journey we've shared as we've grown into better individuals. Lastly, I want to thank my parents Yongping Zhong and Qianyu Hu for the unconditional love and support, for indulging me to chase my dreams. and for teaching me how to embrace and nurture curiosity about the world and to love wholeheartedly.

%%%%%%%%%%%%%%%%%%%%%%%%%%%%%%%%%%%%%%%%%%%%%%%%%%%%%%%%%%%%%%%%%%%%%%
% -*-latex-*-

% Some departments (e.g. 5) require an additional signature page.  See
% signature.tex for more information and uncomment the following line if
% applicable.
% \include{signature}
\pagestyle{plain}
%auto-ignore
  % -*- Mode:TeX -*-
%% This file simply contains the commands that actually generate the table of
%% contents and lists of figures and tables.  You can omit any or all of
%% these files by simply taking out the appropriate command.  For more
%% information on these files, see appendix C.3.3 of the LaTeX manual. 
\tableofcontents
%\newpage
%\listoffigures
%\newpage
%\listoftables

%auto-ignore
%% This is an example first chapter.  You should put chapter/appendix that you
%% write into a separate file, and add a line \include{yourfilename} to
%% main.tex, where `yourfilename.tex' is the name of the chapter/appendix file.
%% You can process specific files by typing their names in at the 
%% \files=
%% prompt when you run the file main.tex through LaTeX.
\chapter{The roadmap}

\section{Motivation}
Can a collection of atoms and molecules exhibit intelligence? Our brain serves as an example; however, not just any collection can think like the human brain. Is it possible to configure natural and engineered disordered many-body systems for intelligent tasks typically associated with nervous systems, such as learning, memory, and optimization? In this thesis, we attempt to answer these questions by initiating a statistical mechanics program called \textbf{many-body intelligence} – the study of emergent intelligence from the collective dynamics of many-body systems.

By closely examining different neural and physical systems that demonstrate intelligent behaviors, we aim to achieve the following objectives: (1) develop statistical mechanical theories for systems that exhibit distinctively intelligent functions such as learning and memory; (2) harness the power of non-equilibrium many-body systems for intelligent tasks, and create novel learning systems with near-term applications.

%%%%%%add atlas here and describe it%%%%%%%%%%%%%%%%%%
\begin{figure}[H]
\centering{}\includegraphics[scale=0.8]{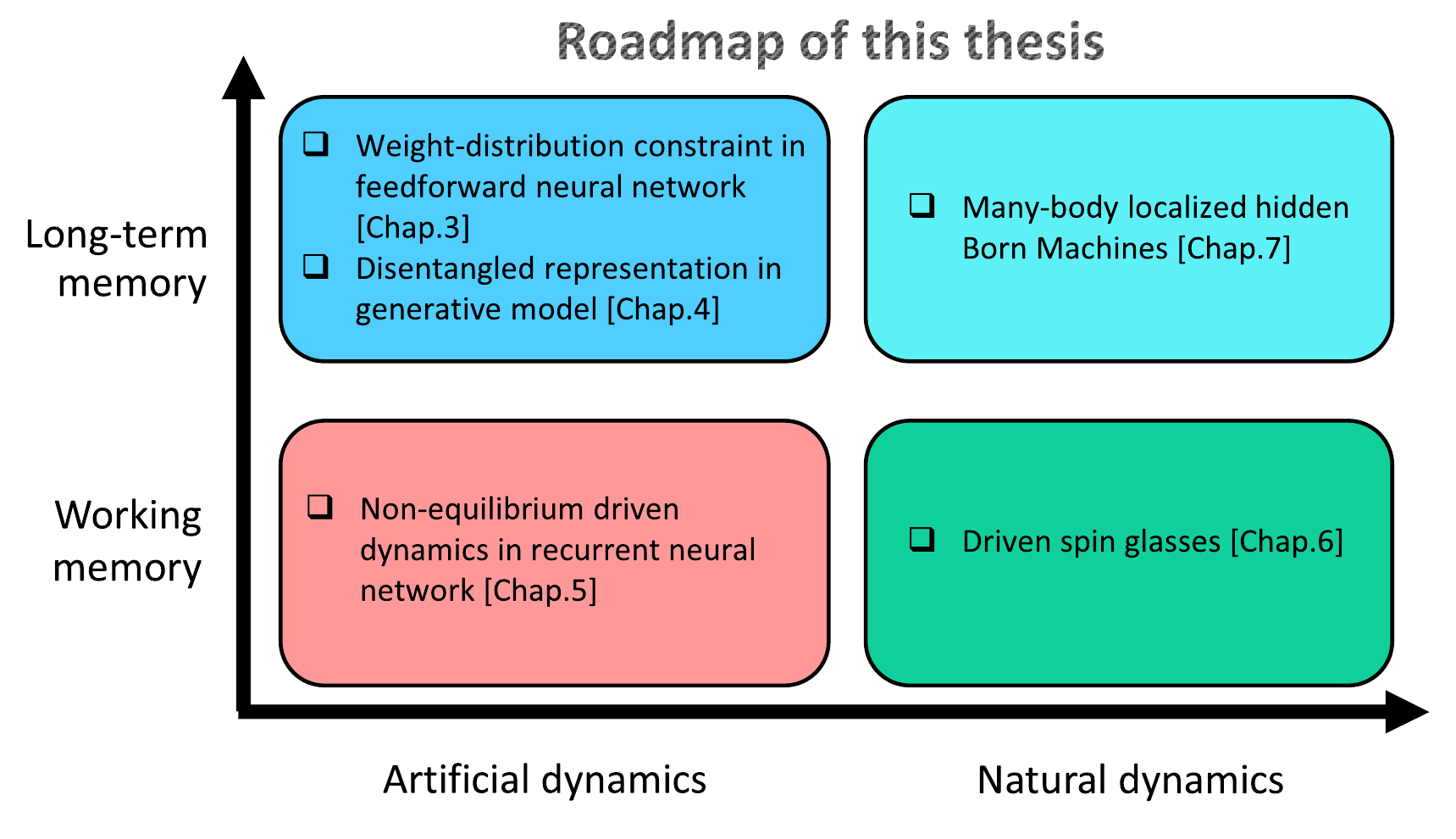}\caption{\label{fig:atlas}We categorize different intelligent many-body systems studied in this thesis based on their learning dynamics (horizontal axis) and memory mechanisms (vertical axis).}
\end{figure}
%%%%%%add atlas here and describe it%%%%%%%%%%%%%%%%%%

\section{Summary}
Intelligence is a multidimensional concept with diverse definitions. This thesis focuses on two aspects of intelligence: memory mechanisms and learning dynamics. The latter can be categorized as either natural or artificial, based on whether they follow natural or artificial dynamics (depicted on the horizontal axis of Fig.\ref{fig:atlas}). For instance, stochastic gradient descent in neural network training exemplifies artificial dynamics \cite{hertz2018introduction}, while Hamiltonian time-evolution in physical systems represents natural dynamics \cite{keim2019memory,stern2023learning}. On the other hand, memory mechanisms can be classified as long-term memory versus working memory (represented on the vertical axis of Fig.\ref{fig:atlas}), depending on whether the system's couplings change during the learning process \cite{cowan2008differences}. In systems that form long-term memory, the internal couplings are modified by external inputs, as in the case of standard neural networks and kernel machines \cite{engel2001statistical}. Conversely, in systems that uses working memory, only the internal state is altered by external inputs, while the couplings remain constant, as seen in reservoir computing \cite{lukovsevivcius2009reservoir,tanaka2019recent}. \\

The organization of the thesis is as follows: Chapter 1 is a roadmap similar to the current extended summary. In Chapter 2, we review the essential theoretical tools needed to investigate these topics, including the statistical mechanics of spin glasses and the replica method for feedforward neural networks.

Ergodicity breaking is essential for learning and memory in non-equilibrium many-body systems. Classical examples include spin glasses and neural networks, which is where we will begin. In Chapter 3 (see Fig.\ref{fig:chap3} for a snippet), we enter the artificial realm of the roadmap, where we use a combination of replica theory, information geometry, and optimal transport to study feedforward neural networks subject to connectivity constraints \cite{zhong2022theory}. Typically, incorporating such structural constraints into network regularization has posed challenges for the development of learning theories. We constructed an analytical theory that quantified the effect of imposing arbitrary network weight-distribution constraints. Our theory predicted that the network memory capacity was proportional to the geodesic distance between the imposed and original distributions on the Wasserstein statistical manifold, and further predicted optimal prior distributions for achieving the best generalization performance. Our theory and the accompanying algorithm unified three distinct elements: learning capability, information geometry, and optimal transport, providing a principled approach to reconstructing ground-truth biological neural circuits from connectomics data.
\begin{wrapfigure}{r}{0.5\textwidth}
%\vspace{-1em}
\centering
    \includegraphics[width=7.5cm]{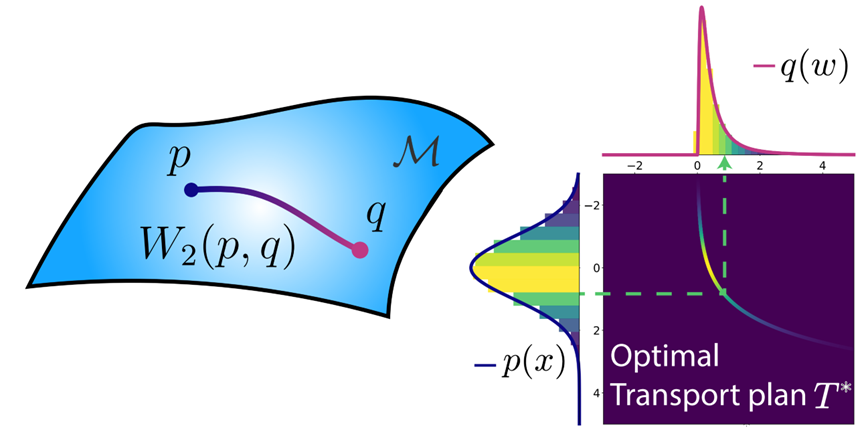}
    %\vspace{-2em}
    \caption{\textbf{Chapter 3} Learning in feedforward neural networks under constrained weight-distribution.}
    \label{fig:chap3}
    %\vspace{-0.5em}
\end{wrapfigure}
In the unsupervised learning setting, Chapter 4 focuses on the criteria for learning good representations in generative models \cite{sikka2019closer}. We established a trade-off between inference fidelity and disentangling ability in variational autoencoders, and proposed a solvable model in which optimal performance can be predicted analytically.

In recurrent architectures, Chapter 5 examines non-equilibrium driven dynamics in models of hippocampus spatial navigation systems \cite{zhong2020nonequilibrium}. We established a fundamental bound on how quickly recurrent dynamics can track sensory inputs and developed an analytical theory that predicts how memory retrieval depends on external inputs.

\begin{wrapfigure}{r}{0.5\textwidth}
%\vspace{-1em}
\centering
    \includegraphics[width=7.5cm]{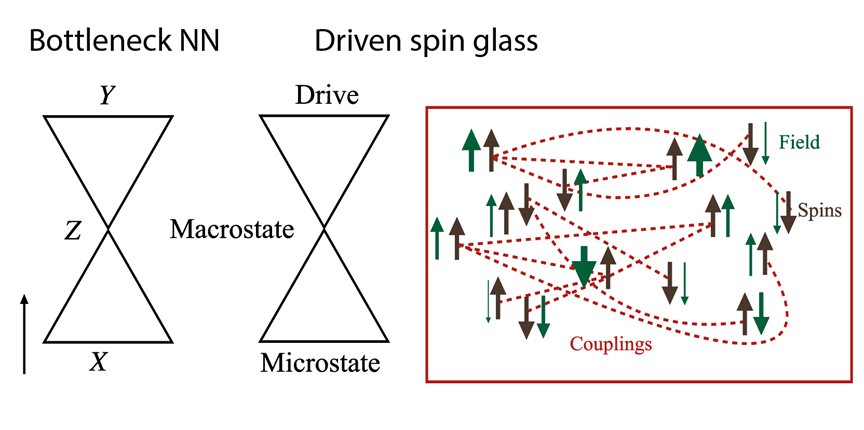}
    %\vspace{-2em}
    \caption{\textbf{Chapter 6} A driven spin glass system capable of performing discriminative learning.}
    \label{fig:chap6}
    %\vspace{-0.5em}
\end{wrapfigure}

In the natural domain, Chapter 6 (see Fig.\ref{fig:chap6} for a snippet) demonstrates that driven spin-glasses can perform a wide range of learning tasks typically seen only in artificial systems, such as classification, memory, and novelty detection \cite{zhong2021machine}. We further showed that traditional thermodynamic variables were no longer effective in characterizing these novel many-body learning phenomena, and that machine learning could make better predictions by using macroscopic variables constructed from nonlinear combinations of traditional ones.
\begin{wrapfigure}{r}{0.5\textwidth}
%\vspace{-1em}
\centering
    \includegraphics[width=7.5cm]{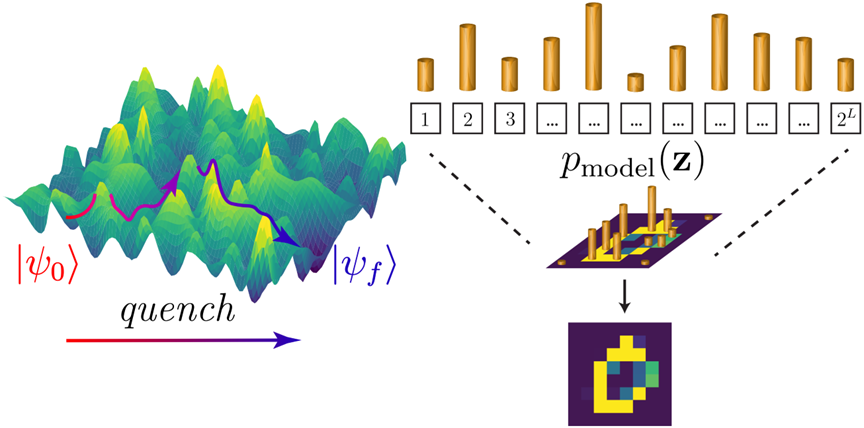}
    %\vspace{-2em}
    \caption{\textbf{Chapter 7} A system of many-body localized quantum spins capable of performing generative modeling.}
    \label{fig:chap7}
    %\vspace{-0.5em}
\end{wrapfigure}

Lastly, another prime example of ergodicity breaking arising in the quantum domain is Many-Body Localization (MBL). In Chapter 7 (see Fig.\ref{fig:chap7} for a snippet), to explore the potential for using MBL for learning in quantum many-body systems, we proposed a class of quantum generative models that we named "MBL hidden Born machines" \cite{zhong2022many}. We demonstrated that the trainability of basic Born machines could be significantly improved by including hidden units and that the MBL Born machine possessed greater expressive power than classical models. Our work revealed important relationships between learning and physical dynamics that could serve as guiding principles for designing quantum generative models.

Finally, in Chapter 8, we outline potential future directions.

\section{Outlook}

We envision configuring many-body systems for machine intelligence using mathematical tools of statistical mechanics and numerical tools from machine learning, and conversely, improving the understanding and practice of state-of-the-art machine learning using many-body physics. Our goal is to advance our understanding of intelligence from both directions, Ultimately, we aspire to develop many-body intelligence into a discipline that can contribute to answering scientifically meaningful, and societally impactful questions in physical sciences, AI, and beyond.
%auto-ignore
\chapter{The replica method for spin glasses and neural learning}
\label{review}

In this chapter, we develop the basic notions and techniques used throughout this thesis. First, we review the statistical mechanics of spin glasses, including the Edwards-Anderson model and the Sherrington-Kirkpatrick model. The materials presented in Section \ref{spin} closely follow the introduction of spin glasses in \cite{binder1986spin}.

In Section \ref{gardner}, we review the application of spin glass methods to feedforward neural networks, pioneered by Elizabeth Gardner in \cite{gardner1987maximum} \cite{gardner1988optimal} to study the capacity of perceptrons, and later developed by \cite{seung1992statistical} to study generalization performance. For pedagogical reviews, see \cite{engel2001statistical} and \cite{advani2013statistical}. 

\section{Spin Glasses}
\label{spin}

\subsection{Annealed vs quenched average}
Consider a system that can be characterized by statistical variables denoted by $S_i$ (where $i$ ranges from $1$ to $N$). This system exhibits randomness that can be captured by a random variable ${x}$, which fluctuates over time. An example of ${x}$ could be the location of a spin within a lattice, where the time evolution of the system corresponds to the spin diffusing through the lattice. Here, we define $\tau_{\text{dis}}$ as the typical fluctuation time and $\tau_{\exp}$ as the observation time. If $\tau_{\exp} \gg \tau_{\text{dis}}$, the random variables will eventually relax to thermal equilibrium and will be averaged over in a manner similar to statistical averages. For instance, the free energy of the system will become
\begin{equation}
\begin{split}
F &= -k_B T \text{ln}[Z\{x\}]_{\text{avg}} \\
Z\{x\} &= \Tr\exp[\mathcal{H}\{x,S_i\}/k_BT],
\end{split}
\end{equation}
where the trace is over all the spins ${S_i}$, the average referred to as an "annealed average" is not applicable in experiments that study the spin glass phase because atomic diffusion times are quite long at low temperatures. Instead, we must consider the regime where $\tau_{\text{dis}} \gg \tau_{\exp}$, which is known as a "quenched average". In this case, each random variable takes on a unique value while the statistical variables fluctuate. Therefore, we need to calculate the partition function for a specific random variable configuration, denoted as $Z{x}$. However, it will become clear later on that averaging over $Z{x}$ is inappropriate in this regime. Instead, we need to average over all replicas of the system.

In general, extensive variables can be averaged. Brout (1959) presents an intuitive argument to support this idea. He suggests considering a single, very large system that is divided into multiple macroscopic subsystems, each with a distinct set of random variables. Assuming that the coupling between subsystems is small, the value of any extensive variable for the entire system should be equivalent to the average of that quantity over all the subsystems. For large original systems, we can average over a large number of subsystems and expect that the result will only differ slightly from the complete average over all possible choices of ${x}$. For instance, the magnetization per spin $m$ should behave like
\begin{equation}
m\{x\} - [m]_{\text{avg}} \to 0 \;\; (\text{for} \; N \to \infty),
\end{equation}
for any set ${x}$ that occurs with a reasonable probability. Systems that satisfy this property are referred to as "self-averaging." This means that a single large system produces the same result for extensive quantities as a configurational average. However, for finite $N$, the Brout argument predicts that the probability distribution for the density of extensive quantities, such as the free-energy density $f$, will be Gaussian with a width of $N^{-1/2}$,
\begin{equation}
p(f) \propto \exp\left\{-\frac{N(f-[f]_{\text{avg}})^2}{2(\Delta f)^2}\right\}.
\end{equation}
Now averaging over the partition function we will get
\begin{equation}
\begin{aligned}
f_{\text{ann}} &= -\frac{k_BT}{N}\ln [Z]_{\text{avg}} \\
&= [f]_{\text{avg}} + (\Delta f)^2/k_BT.
\end{aligned}
\end{equation}
Now it is clear that $f_{\text{ann}} \geq f_{\text{avg}}$. Therefore, it is $f$, or $\ln Z$, instead of $Z$, that one should average. A correct way of performing this averaging is by using the replica method, described in the following section \ref{replica}.

When an experimental measurement is conducted over a specific period of time, the fluctuations in the system occur at a rate determined by the spectrum of relaxation times ${\tau}$. If the observation time $\tau_{\exp} \gg \tau_{\text{max}}$, which is the maximum relaxation time, then the system explores all regions of its phase space with the probability described by the Boltzmann distribution. In this scenario, the system satisfies the ergodic hypothesis of thermodynamics, and the time average calculated by the experiment corresponds to an average of all the system's states in its phase space. However, the ergodic hypothesis does not hold for spin glasses, where $\tau_{\text{max}} \gg \tau_{\exp}$, and ergodicity is violated. In this case, alternative averaging methods are necessary, and the replica method needs to be employed.
%
% % % % %
\subsection{The Replica Method}
\label{replica}
As discussed earlier, we should average over free energy, instead of the partition function $Z$, 
\begin{equation}
f = [f\{x\}]_{\text{avg}} = -\frac{k_BT}{N}\left[\ln Z\{x\}\right]_{\text{avg}}.
\end{equation}
However, directly computing the average is not feasible since the random variables are present within a log. In situations where the disorder is weak, it is possible to separate $\mathcal{H}{\{x\}}$ into a nonrandom component $\mathcal{H}_0$ and a random perturbation $\delta \mathcal{H}{\{x\}}$, and then perform the average term-by-term. Unfortunately, spin glasses are highly disordered systems, and the nonrandom part is much smaller than the random part, making it impossible to follow this procedure. However, utilizing the exact relationship
\begin{equation}
\label{logidentity}
\begin{split}
[\text{ln}Z\{x\}]_{\text{avg}} &= \lim\limits_{n \to 0} \frac{1}{n} ([Z^n\{x\}]_{\text{avg}}-1) 
\\
&= \lim\limits_{n \to 0} \frac{\partial}{\partial n} [Z\{x\}]_{\text{avg}}.
\end{split}
\end{equation}
for positive interger $n$, one can express $Z^n\{x\}$ in terms of $n$ identical replicas of the system,
\begin{equation}
\begin{split}
Z^n\{x\} &= \prod_{\alpha = 1}^{n} Z_\alpha \{x\}
\\
& = \prod_{\alpha = 1}^{n} \exp\left[-\frac{\mathcal{H}{\{x,S_i^{\alpha}\}}}{k_BT}\right] 
\\
&= \exp\left[-\frac{1}{k_BT}\sum_{\alpha=1}^n \mathcal{H}{\{x,S_i^{\alpha}\}}\right] ,
\end{split}
\end{equation}
where $Z_{\alpha}$ is the partition function of the $\alpha$-th replica. For positive integer $n$, it is easy to carry out the average $[\;]_{\text{avg}}$. Then we can express the above in terms of an effective Hamiltonian $H_{\text{eff}}$ that does not contain any disorder. 
\begin{equation}
Z_n \equiv [Z^n\{x\}]_{\text{avg}} \equiv \Tr \exp\left[-\frac{\mathcal{H}_{\text{eff}}(n)}{k_BT}\right]  
\end{equation} 
where the trace is over all variables ${S_i^{\alpha}}$ of all spins of all replicas. Note that before we perform the averaging, different replicas do not interact with each other. However, after averaging, we effectively introduce interactions among different replicas. 
We can take the following Hamiltonian as an example. Consider
\begin{equation}
\label{Heff}
\mathcal{H} = -\frac{1}{2}\sum_{i,j}^{N}J_{ij}S_i S_j - h\sum_{i,j}^{N}S_i^z,
\end{equation}
where the interaction term $J_{ij}$ are random variables with distribution $P(J_{ij})$. Then at $h=0$, \eqref{Heff} becomes
\begin{equation}
Z_n = \Tr \prod_{i,j}^{N} dJ_{ij}P(J_{ij})\exp\bigg(\frac{J_{ij}}{k_BT}\sum_{\alpha=1}^{n}S_i^{\alpha}S_j^{\alpha}\bigg).
\end{equation}
Taylor expanding the above equation we have
\begin{equation}
\mathcal{H}_{\text{eff}}(n)/k_BT = -\frac{1}{2}\sum_{i,j}^{N}\sum_{k=1}^{\infty}\frac{1}{k!}\frac{J_{ij}^{\text{cum}}(k)}{(k_BT)^k}\bigg(\sum_{\alpha=1}^{n}S_i^{\alpha}S_j^{\alpha}\bigg)^k,
\end{equation}
where again the trace is over all variables ${S_i^{\alpha}}$ of all spins of all replicas, , and $J_{ij}^{\text{cum}}(k)$ is the $k$th cumulant of $J_{ij}$,
\begin{equation}
\begin{aligned}
J_{ij}^{\text{cum}}(1) &= [J_{ij}]_{\text{avg}} = \bar{J}
\\
J_{ij}^{\text{cum}}(2) &= [J_{ij}^2]_{\text{avg}} - [J_{ij}]^2_{\text{avg}} \equiv (\Delta J_{ij})^2
\end{aligned}
\end{equation}
Hence, it is evident that cumulants higher than the first indicate interactions between different replicas of the disordered system. Furthermore, since $\mathcal{H}_{\text{eff}}$ is now a theory that lacks disorder and is translationally-invariant, we can employ the conventional method to solve it. For instance, we can use the mean-field approximation and substitute the $S_i^{\alpha}$'s with their respective expectation values, leading to a set of self-consistency equations for these expectation values $\langle S_i^{\alpha}\rangle $.

Note that our discussion has thus far been limited to positive integer values of $n$. For the replica method to be effective, we must be able to take the limit as $n\to 0$. Consequently, we need to analytically continue $n$ to arbitrary real numbers. It is apparent that $\mathcal{H}_{\text{eff}}$ is unaffected by relabeling of the replicas when $n$ is a positive integer as defined. However, this symmetry is not preserved when we analytically continue $n$ to arbitrary real numbers, leading to the concept of "replica symmetry breaking."

It is not sufficient to only be able to calculate the free energy, in the following we provide an example of how to use the replica method to calculate the correlation function. 

Let's consider the magnetization per spin,
\begin{equation}
\begin{split}
M &= [\langle S_i \rangle_T]_{\text{avg}}
\\
&=\left[\frac{\Tr S_i\exp (-\mathcal{H}\{x\}/k_B T)}{Z\{x\}}\right]_{\text{avg}},
\end{split}
\end{equation}
where $\langle \cdot \rangle_T$ denotes thermal average with respect to Boltzmann distribution. Multiplying both the numerator and denominator by a factor of $(Z\{x\})^{n-1}$, we have
\begin{equation}
M = \left[\frac{Z^{n-1}\Tr S_i\exp-\mathcal{H}\{x\}/k_BT}{Z^n}\right]_{\text{avg}}.
\end{equation}
Now in the limit $n\to 0$, the denominator becomes essentially unity and does not need to be averaged, so the averaging is only for the numerator. Note that the trace is over all spins of all replicas, and we can identify the averaging as just the expectation value of $S_i^{\alpha}$,
\begin{equation}
M = \langle S_i^{\alpha}\rangle ,
\end{equation}  
where the bracket denotes averaging over $\mathcal{H}_{\text{eff}}$ and $\alpha$ is any of the replicas. 
Next we consider the fluctuations,
\begin{equation}
\begin{split}
q &= [\langle S_i\rangle ^2_T]_{\text{avg}} 
\\
&= \left[\frac{[\Tr S_i\exp(-\mathcal{H}\{x\}/k_BT)][\Tr S_i\exp(-\mathcal{H}\{x\}/k_BT)]}{Z^2}\right]_{\text{avg}},
\end{split}
\end{equation}
where $q$ is the overlap, often important in spin glasses and serves as an order parameter for low-temperature phase transitions. Performing the same trick as above, we can identify $q$ as
\begin{equation}
q = \langle S_i^{\alpha}S_i^{\beta}\rangle\qquad (\alpha \neq \beta),
\end{equation}
for all replicas $\alpha$ and $\beta$. It is easy to generalize the above relation to $k$-point correlations:
\begin{equation}
[\langle S_i\rangle^k_T]_{\text{avg}} = \langle S_i^{\alpha_1}S_i^{\alpha_2}... S_i^{\alpha_k}\rangle,
\end{equation}
where all the replicas $\alpha_i$'s are distinct. The lesson from the above equality is that for every thermal average on the original theory, there is a distinct replica, and the choice of replica does not affect the above equality. 

% % %
\subsection{The Edwards-Anderson model}
\label{EA model}

In 1975, Edwards and Anderson consider a Hamiltonian of the type
\begin{equation}
\label{EA}
\mathcal{H} = -\sum_{\langle i,j\rangle}^{N}J_{ij}S_i S_j - H\sum_{i,j}^{N}S_i,
\end{equation}
where spins are on sites of a regular lattice with nearest neighbour interaction, and $J_{ij}$ is random with distribution $P(J_{ij})$. The standard choice for the distribution is Gaussian, and it is called the Gaussian Edwards-Anderson model,
\begin{equation}
P(J_{ij}) = \frac{1}{\sqrt{2\pi (\Delta J_ij)^2}} \exp\left[-\frac{(J_{ij}-\bar{J}_{ij})^2}{2(\Delta J_{ij})^2}\right].
\end{equation} 
Since cumulants higher than second order vanishes for Gaussian distribution, we have
\begin{equation}
J_{ij}^{\text{cum}}(k\geq3)\equiv 0.
\end{equation}
Therefore, the effective Hamiltonian takes a rather simple form,
\begin{equation}
\begin{split}
\mathcal{H}_{\text{eff}}(n)/k_BT = &-\frac{1}{2}\sum_{i,j}^{N} \frac{\bar{J_{ij}}}{k_BT}\sum_{\alpha=1}^{n}S_i^{\alpha}S_j^{\alpha}
\\
&-\frac{1}{4}\sum_{i,j}^{N}\bigg(\frac{\Delta J_{ij}}{k_BT}\bigg)^2\sum_{\alpha,\beta}^{n}S_i^{\alpha}S_j^{\alpha}S_i^{\beta}S_j^{\beta}.
\end{split}
\end{equation}
%

% % % % % %
\subsection{The Sherrington-Kirkpatrick model}
\label{SK model}

The Sherrington-Kirkpatrick model's Hamiltonian takes the same form as the E-A model \eqref{EA}, but instead of the finite-range interaction, the S-K model considers infinite range interaction among the spins. Also, instead of a globally constant magnetic field $H$, a local field $H_i$ is applied to every site of the lattice,
\begin{equation}
\label{SK}
\mathcal{H} = -\frac{1}{2}\sum_{i,j}^{N}J_{ij}S_i S_j - \sum_{i,j}^{N}H_iS_i,
\end{equation}
where couplings between $i,j$ does not depend on range.The distribution of $J_{ij}$ is given by 
\begin{equation}
P(J_{ij}) = \frac{1}{J}\bigg(\frac{N}{2\pi}\bigg)\exp\left[\frac{-N(J_{ij}-J_0/N)^2}{2J^2}\right],
\end{equation}
and therefore,
\begin{equation}
\label{moments}
\begin{split}
[J_{ij}]_{\text{avg}} &= \frac{J_0}{N}
\\
[J_{ij}^2]_{\text{avg}} &- [J_{ij}]_{\text{avg}}^2 = \frac{J^2}{N}.
\end{split}
\end{equation}
The  $1/N$ factor above is to ensure that there's a sensible and nontrivial thermodynamic limit $N\to \infty$. Note that here the distribution $P(J_{ij})$ need not be Gaussian, as long as its first two moments are given by \eqref{moments}, and the higher order moments are bounded. 
We consider first the quantity (in the following we use $\langle i,j \rangle$ to denote summing over distinct $i,j$ pairs only once)
\begin{equation}
[Z^n]_{\text{avg}} = \sum_{S_i^{\alpha}}\int_{-\infty}^{\infty}\bigg(\prod_{\langle i,j\rangle}P(J_{ij})dJ_{ij}\bigg)\exp \bigg\{ \beta\sum_{\langle i,j\rangle}J_{ij}\sum_{\alpha=1}^{n}S_i^{\alpha}S_j^{\alpha}+\beta\sum_{i}H_i\sum_{\alpha=1}^{n}S_i^{\alpha}\bigg\},
\end{equation}
where $\alpha$ is the replica index. The integral can be easily evaluated for Gaussian distribution, and we have
\begin{equation}
[Z^n]_{\text{avg}} = \sum_{S_i^{\alpha}}\exp\bigg[\frac{1}{N}\sum_{\langle i,j\rangle}\bigg(\frac{1}{2}(\beta J)^2 \sum_{\alpha,\beta}S_i^{\alpha}S_j^{\alpha}S_i^{\beta}S_j^{\beta}+\beta J_0\sum_{\alpha}S_i^{\alpha}S_j^{\alpha}\bigg)+\beta \sum_i H_i \sum_{\alpha} S_i^{\alpha} \bigg].
\end{equation}
Dropping $1/N$ corrections in the exponent, and note that $(S_i^{\alpha})^2 = 1$, we have
\begin{equation}
[Z^n]_{\text{avg}} = \exp\bigg[\frac{1}{4}(\beta J)^2 n N\bigg] \sum_{S_i^{\alpha}} \exp\bigg[\frac{(\beta J)^2}{2N}\sum_{\alpha<\beta}\bigg(\sum_i S_i^{\alpha}S_i^{\beta}\bigg)^2+\frac{\beta J_0}{2N}\sum_{\alpha}\bigg(\sum_i S_i^{\alpha}\bigg)^2 + \beta \sum_i H_i \sum_{\alpha} S_i^{\alpha}\bigg].
\end{equation}
We linearize the square terms in the above expression with the Hubbard-Stratonovitch identity
\begin{equation}
\label{HS}
\exp\bigg(\frac{\lambda a^2}{2}\bigg) = \bigg(\frac{\lambda}{2 \pi}\bigg)^{1/2} \int_{-\infty}^{\infty} dx \exp\bigg(-\frac{\lambda x^2}{2}+a\lambda x\bigg),
\end{equation}
by introducing auxiliary variables $q_{\alpha \beta}$ and $m_{\alpha}$. Then we have
\begin{equation}
\begin{split}
[Z^n]_{\text{avg}} &= \exp \bigg(\frac{1}{4}(\beta J)^2n N\bigg)
\\
& \times \int_{-\infty}^{\infty}\bigg[\prod_{\alpha<\beta}\bigg(\frac{N}{2 \pi}\bigg)^{1/2}\beta J dq_{\alpha \beta}\bigg]\bigg[\prod_{\alpha}\bigg(\frac{N\beta J_0}{2\pi}\bigg)^{1/2}dm_{\alpha}\bigg]
\\
& \times \exp\bigg(-\frac{N(\beta J)^2}{2}\sum_{\alpha<\beta}q_{\alpha \beta}^2 - \frac{N\beta J_0}{2} \sum_{\alpha}m_{\alpha}^2 + N\log \Tr \exp L[q_{\alpha \beta},m_{\alpha}]\bigg),
\end{split}
\end{equation}
where 
\begin{equation}
L[q_{\alpha \beta},m_{\alpha}] = (\beta J)^2\sum_{\alpha < \beta}q_{\alpha \beta}S^{\alpha} S^{\beta} + \beta \sum_{\alpha} (J_0 m_{\alpha}+ H)S^{\alpha},
\end{equation}
and the trace is over all spins of replica $S^{\alpha}$. Note that $q_{\alpha \beta}$ with $\alpha < \beta$ has $n(n-1)/2$ independent entries, and we can define $q_{\alpha \beta}$ to be symmetric, i.e., $q_{\alpha \beta} = q_{\beta\alpha}$.
Now make use of \eqref{logidentity} we then arrive at
\begin{equation}
\label{FE}
-\beta f = \lim\limits_{n\to 0 }\bigg[\frac{(\beta J)^2}{4}\bigg(1-\frac{1}{n}\sum_{\alpha,\beta}q_{\alpha \beta}^2\bigg)+\frac{\beta J_0}{2}\frac{1}{n}\sum_{\alpha}m_{\alpha}^2+\frac{1}{n}\log\Tr\exp L\bigg],
\end{equation}
where the summation is over all distinct replica pairs. We then need to evaluate the self-consistency conditions for $q_{\alpha \beta}$ and $m_{\alpha}$,
\begin{equation}
\frac{\partial f}{\partial q_{\alpha \beta}} = \frac{\partial f}{\partial m_{\alpha}} = 0,
\end{equation}
which reads
\begin{equation}
\label{SC}
\begin{split}
q_{\alpha \beta} &= \langle S^{\alpha}S^{\beta}\rangle = \lim\limits_{n\to 0 }\frac{\Tr S^{\alpha}S^{\beta}\exp L[q_{\alpha \beta},m_{\alpha}]}{\Tr\exp L[q_{\alpha \beta},m_{\alpha}]}
\\
m_{\alpha} &= \langle S^{\alpha}\rangle = \lim\limits_{n\to 0 }\frac{\Tr S^{\alpha}\exp L[q_{\alpha \beta},m_{\alpha}]}{\Tr \exp L[q_{\alpha \beta},m_{alpha}]}
\end{split}
\end{equation}
Since relabeling the replica indices $\alpha, \beta$ is a symmetry of the solution, we can assume that for all replicas,
\begin{equation}
\begin{split}
q_{\alpha \beta} &= q
\\
m_{\alpha} &= M.
\end{split}
\end{equation}
Let's define 
\begin{equation}
\tilde{H}(z) = Jq^{1/2}z + J_0 M + H,
\end{equation}
then the free energy \eqref{FE} can be simplified into
\begin{equation}
-\beta f = \frac{(\beta J)^2}{4}(1-q)^2 - \frac{\beta J_0}{2}M^2 + \frac{1}{\sqrt{2\pi}}\int_{\infty}^{\infty}dz e^{-z^2/2} \text{log}[2\text{cosh}\beta \tilde{H}(z)]dz,
\end{equation}
and the self consistency conditions \eqref{SC} becomes
\begin{equation}
\begin{split}
q &= \frac{1}{\sqrt{2\pi}}\int_{-\infty}^{\infty} e^{-z^2/2}\text{tanh}^2[\beta \tilde{H}(z)]dz
\\
M &= \frac{1}{\sqrt{2\pi}}\int_{-\infty}^{\infty} e^{-z^2/2}\text{tanh}[\beta \tilde{H}(z)]dz.
\end{split}
\end{equation}
As in the case for mean-field theory in Ising model, solving these two equations analytically are difficult, but we can solve them numerically. It turns out that $q$ plays the role of order parameter for the spin-glass phase/ferromagnetic phase transition. For $H=0$, we plot the phase diagram of the Sherrington-Kirkpatrick model in Fig. \ref{phase} \cite{binder1986spin}. Note that we have spin glass phase when $M=0, q \neq 0$, paramagnetic phase when $M=q=0$, and ferromagnetic phase when $M\neq0, q\neq0$.
%
%%%%
\begin{figure}[t]
	\begin{center}
		\includegraphics[scale=0.4]{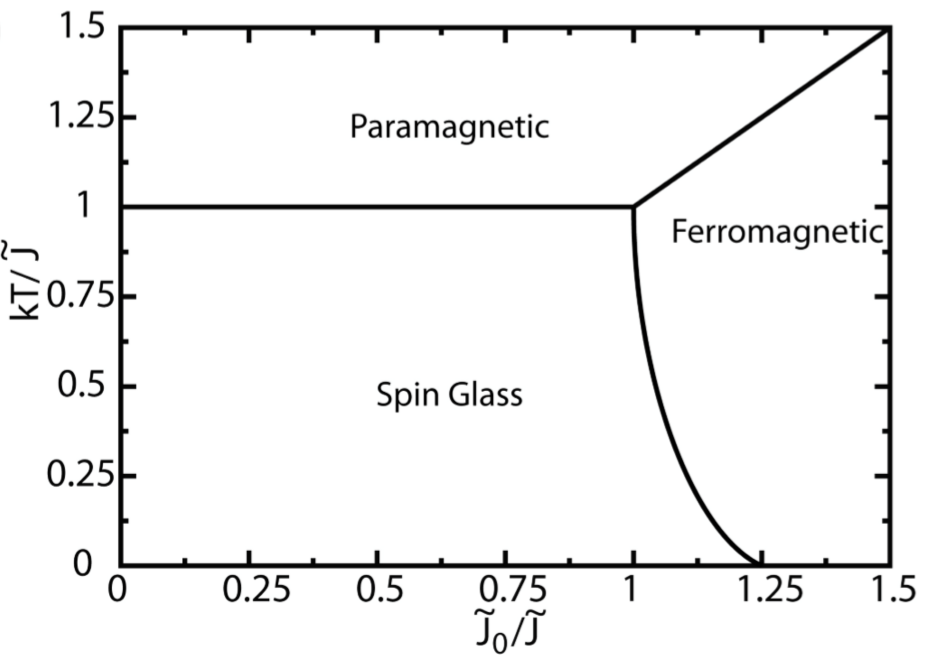}
	\end{center}
	\caption{\label{phase} Phase diagram of the Sherrington-Kirkpatrick model.(Note: $\tilde{J_0}=J_0$, and $\tilde{J} = J$ in our above derivation.)}
\end{figure}
%%%%
%

\section{Gardner capacity}

In this section, we review the statistical mechanical formulation
for perceptron capacity, also known as the Gardner calculation \cite{gardner1988optimal,gardner1988space}.
Calculations in this section is not exactly the same as the original
papers but follow roughly the same idea. 

\subsubsection*{Preliminaries}

Throughout this chapter, we make frequent use of Gaussian integrals.
We introduce short-hand notations $\int Dt\equiv\int\frac{dt}{\sqrt{2\pi}}e^{-t^{2}/2}$
and $H(x)\equiv\int_{x}^{\infty}Dt$. Also, when we do not specify
the integration range it is understood that we are integrating from
$-\infty$ to $\infty$. 

We start with a perceptron with weight vector $\boldsymbol{w}\in\mathbb{R}^{N}$,
normalized to $||\boldsymbol{w}||^{2}=N.$ Our data consists of pairs$\left\{ \boldsymbol{\xi}^{\mu},\zeta^{\mu}\right\} _{\mu=1}^{P}$,
where $\boldsymbol{\xi}^{\mu}$ is an $N$-dimentional random vector
drawn i.i.d. from a standard normal distribution, $p(\xi_{i}^{\mu})=\mathcal{N}(0,1)$,
and $\zeta^{\mu}$ are random binary class labels with $p(\zeta^{\mu})=\frac{1}{2}\delta(\zeta^{\mu}+1)+\frac{1}{2}\delta(\zeta^{\mu}-1)$.
The goal is to find a hyperplane through the origin, perpendicular
to $\boldsymbol{w}$, such that it separates the two classes of examples
correctly (Fig.\ref{fig:capacity_schemetic}). In the following, we
stop distinguishing $\boldsymbol{w}$ and the hyperplane that it defines. 

We call $\boldsymbol{w}$ a separating hyperplane when it correctly
classifies all the examples with margin $\kappa>0$:
\begin{equation}
\zeta^{\mu}\frac{\boldsymbol{w}\cdot\boldsymbol{\xi}^{\mu}}{||\boldsymbol{w}||}\geq\kappa.\label{c2eq:margin}
\end{equation}
\begin{figure}
\centering{}\includegraphics[scale=0.6]{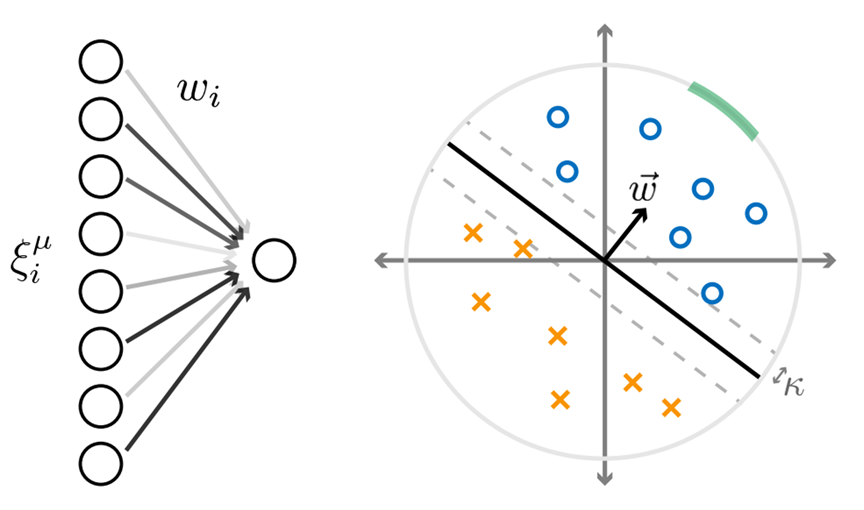}\caption{\label{fig:capacity_schemetic}Schematics of the perceptron classification
problem. Left: A perceptron with inputs $\xi_{i}^{\mu}$ and synaptic
weights $w_{i}$. Right: visualization of the perceptron binary classification
problem in 2-dimension. Solid line is the decision surface, which
is perpendicular to the perceptron weight vector $\vec{w}$. Dash
line corresponds to the geometric margin $\kappa$, which is defined
as the minimal distance to the examples $\vec{\xi^{\mu}}$ (shown
in blue and orange, different colors represent the two classes). }
\end{figure}
Note that since $w_{i}\sim\mathcal{O}(1),$$\sum w_{i}\xi_{i}^{\mu}\sim\mathcal{O}(\sqrt{N}),$and
$||\boldsymbol{w}||=\sqrt{N},$the LHS of Eq.\ref{c2eq:margin} is
$\mathcal{O}(1).$

We want to calculate the volume fraction $V$ of the viable weights
to all possible weights
\begin{equation}
V=\frac{\int d\boldsymbol{w}\left[\prod_{\mu=1}^{P}\Theta\left(\zeta^{\mu}\frac{\boldsymbol{w}\cdot\boldsymbol{\xi}^{\mu}}{||\boldsymbol{w}||}-\kappa\right)\right]\delta(||\boldsymbol{w}||^{2}-N)}{\int d\boldsymbol{w}\delta(||\boldsymbol{w}||^{2}-N)}.\label{c2eq:volume}
\end{equation}
We would like to perform a quenched average over random patterns $\boldsymbol{\xi}^{\mu}$
and labels $\zeta^{\mu}$. This amounts to calculating $\left\langle \log V\right\rangle $,
which can be done using the replica trick $\left\langle \log V\right\rangle =\lim_{n\to0}(\left\langle V^{n}\right\rangle -1)/n$.
We consider first integer $n$, and at the end perform analytic continuation
of $n\to0$. The replicated Gardner volume is:
\begin{equation}
V^{n}=\frac{\prod_{\alpha=1}^{n}\int d\boldsymbol{w}^{\alpha}\left[\prod_{\mu=1}^{P}\Theta\left(\zeta^{\mu}\frac{\boldsymbol{w}^{\alpha}\cdot\boldsymbol{\xi}^{\mu}}{||\boldsymbol{w}^{\alpha}||}-\kappa\right)\right]\delta(||\boldsymbol{w}^{\alpha}||^{2}-N)}{\prod_{\alpha=1}^{n}\int d\boldsymbol{w}^{\alpha}\delta(||\boldsymbol{w}^{\alpha}||^{2}-N)}\label{c2eq:volume_cap}
\end{equation}
We start by rewriting the Heaviside step function using Fourier representation
of the $\delta$-function $\delta(x)=\int_{-\infty}^{\infty}\frac{dk}{2\pi}e^{ikx}$
as (defining $z_{\alpha}^{\mu}=\zeta^{\mu}\frac{\boldsymbol{w}^{\alpha}\cdot\boldsymbol{\xi}^{\mu}}{||\boldsymbol{w}^{\alpha}||}$)
\begin{equation}
\Theta\left(z_{\alpha}^{\mu}-\kappa\right)=\int_{\kappa}^{\infty}d\rho_{\alpha}^{\mu}\delta(\rho_{\alpha}^{\mu}-z_{\alpha}^{\mu})=\int_{\kappa}^{\infty}d\rho_{\alpha}^{\mu}\int\frac{dx_{\alpha}^{\mu}}{2\pi}e^{ix_{\alpha}^{\mu}(\rho_{\alpha}^{\mu}-z_{\alpha}^{\mu})}.
\end{equation}
Note that now all the $\boldsymbol{\xi}^{\mu},\zeta^{\mu}$ dependence
is in $e^{-ix_{\alpha}^{\mu}z_{\alpha}^{\mu}}$. We perform the average
with respect to $\xi_{i}^{\mu}\sim p(\xi_{i}^{\mu})=\mathcal{N}(0,1)$
and $p(\zeta^{\mu})=\frac{1}{2}\delta(\zeta^{\mu}+1)+\frac{1}{2}\delta(\zeta^{\mu}-1)$
(also note that $||\boldsymbol{w}^{\alpha}||=\sqrt{N}$):
\begin{equation}
\begin{split}\left\langle \prod_{\mu\alpha}e^{-ix_{\alpha}^{\mu}z_{\alpha}^{\mu}}\right\rangle _{\xi\eta} & =\prod_{\mu j}\left\langle \exp\left\{ -\frac{i}{\sqrt{N}}\zeta^{\mu}\xi_{j}^{\mu}\sum_{\alpha}x_{\alpha}^{\mu}w_{j}^{\alpha}\right\} \right\rangle _{\xi\zeta}\\
 & =\prod_{\mu i}\left\langle \exp\left\{ -\frac{(\zeta^{\mu})^{2}}{2N}\sum_{\alpha\beta}x_{\alpha}^{\mu}x_{\beta}^{\mu}w_{i}^{\alpha}w_{i}^{\beta}\right\} \right\rangle _{\zeta}\\
 & =\prod_{\mu}\exp\left\{ -\frac{1}{2N}\sum_{\alpha\beta}x_{\alpha}^{\mu}x_{\beta}^{\mu}\sum_{i}w_{i}^{\alpha}w_{i}^{\beta}\right\} .
\end{split}
\label{c2eq:cap_dist_energy_x_z}
\end{equation}
Introducing the replica overlap parameter $q_{\alpha\beta}=\frac{1}{N}\sum_{i}w_{i}^{\alpha}w_{i}^{\beta}$,
and notice that the $\mu$ index gives $P$ identical copies of the
same integral. We can suppress the $\mu$ indices and write
\begin{equation}
\begin{split}\left\langle \prod_{\mu\alpha}\Theta\left(z_{\alpha}^{\mu}-\kappa\right)\right\rangle _{\xi\zeta} & =\left[\int_{\kappa}^{\infty}\left(\prod_{\alpha}\frac{d\rho_{\alpha}dx_{\alpha}}{2\pi}\right)e^{K}\right]^{P}\end{split}
,
\end{equation}
where
\begin{equation}
K=i\sum_{\alpha}x_{\alpha}\rho_{\alpha}-\frac{1}{2}\sum_{\alpha\beta}q_{\alpha\beta}x_{\alpha}x_{\beta}\label{c2eq:cap_dist_K_replica}
\end{equation}
captures all the data dependence in the quenched free energy landscape,
and therefore it is called the `energetic' part of the free energy.
In contrast, the $\delta$-functions in Eqn.\ref{c2eq:volume_cap}
are called `entropic' part because they regulate what kind of weights
are considered in the version space (space of viable weights).

\subsubsection*{The entropic part}

The delta-function we have is from the introduction of $q_{\alpha\beta}$
(note that $\delta(||\boldsymbol{w}^{\alpha}||^{2}-N)$ amounts to
$q_{\alpha=\beta}=1$),
\begin{equation}
\begin{split}\delta(Nq_{\alpha\beta}-\sum_{i}w_{i}^{\alpha}w_{i}^{\beta}) & =\int\frac{d\hat{q}_{\alpha\beta}}{2\pi}\exp\left\{ iN\hat{q}_{\alpha\beta}q_{\alpha\beta}-i\hat{q}_{\alpha\beta}\sum_{i}w_{i}^{\alpha}w_{i}^{\beta}\right\} \end{split}
.
\end{equation}
Note that the normalization constraint $\delta(||\boldsymbol{w}^{\alpha}||^{2}-N)$
is automatically satisfied by requiring $q_{\alpha\alpha}=1$. Using
replica-symmetric ansatz: 
\begin{equation}
\hat{q}_{\alpha\beta}=-\frac{i}{2}(\Delta\hat{q}\delta_{\alpha\beta}+\hat{q}_{1}),\;\;q_{\alpha\beta}=(1-q)\delta_{\alpha\beta}+q
\end{equation}
We have
\begin{equation}
iN\sum_{\alpha\beta}\hat{q}_{\alpha\beta}q_{\alpha\beta}=\frac{nN}{2}\left[\Delta\hat{q}+\hat{q}_{1}(1-q)\right]+\mathcal{O}(n^{2}).
\end{equation}
and
\begin{equation}
\begin{split}-i\sum_{\alpha\beta}\hat{q}_{\alpha\beta}\sum_{i}w_{i}^{\alpha}w_{i}^{\beta} & =-\frac{1}{2}(\Delta\hat{q}+\hat{q}_{1})\sum_{\alpha}\sum_{i}(w_{i}^{\alpha})^{2}-\frac{1}{2}\hat{q}_{1}\sum_{(\alpha\beta)}\sum_{i}w_{i}^{\alpha}w_{i}^{\beta}\\
 & =-\frac{1}{2}\Delta\hat{q}\sum_{\alpha}\sum_{i}(w_{i}^{\alpha})^{2}-\frac{1}{2}\hat{q}_{1}\sum_{i}\left(\sum_{\alpha}w_{i}^{\alpha}\right)^{2}\\
 & \HSTeq-\frac{1}{2}\Delta\hat{q}\sum_{\alpha}\sum_{i}(w_{i}^{\alpha})^{2}+\sqrt{-\hat{q}_{1}}\sum_{i}t_{i}\left(\sum_{\alpha}w_{i}^{\alpha}\right),
\end{split}
\end{equation}
where in the last step HST denotes Hubbard-Stratonovich transformation
$\int\frac{dt}{\sqrt{2\pi}}e^{-t^{2}/2}e^{bt}=e^{b^{2}/2}$ that we
use to linearize the quadratic term at the cost of introducing an
auxiliary Gaussian variable $t$ to be averaged over later.

We can now express the complete free energy while disregarding overall
constant coefficients such as $2\pi$'s and $i$'s in the integration
measure, which become inconsequential when employing the saddle-point
approximation. Additionally, we will omit the denominator of $V$,
as it is independent of data and serves as an overall constant. It
is important to note that under the replica-symmetric assumption,
the replica index $\alpha$ generates $n$ identical copies of the
same integral, allowing for the suppression of the replica index $\alpha$
(also applicable to the synaptic index $i$):
\begin{equation}
\left\langle V^{n}\right\rangle =\int dqd\hat{\lambda}(k)d\Delta\hat{q}d\hat{q}_{1}e^{nN(G_{0}+G_{1})},\label{c2eq:cap_dist_free_energy}
\end{equation}
where
\begin{equation}
\begin{split}G_{0} & =\frac{1}{2}\Delta\hat{q}+\frac{1}{2}\hat{q}_{1}(1-q)+\left\langle \ln X(t)\right\rangle _{t},\\
X(t) & =\int_{-\infty}^{\infty}dw\exp\left\{ -\frac{1}{2}\Delta\hat{q}w^{2}+\sqrt{-\hat{q}_{1}}tw\right\} .
\end{split}
\end{equation}
Note that integrals in Eqn.\ref{c2eq:cap_dist_free_energy} can be
evaluated using saddle-point approximation in the thermodynamic limit
$N\to\infty$.

\subsubsection*{Limit $q\to1$}

We are interested in the critical load $\alpha_{c}$ where the version
space (space of viable weights) shrinks to a single point, i.e., there
exists only one viable solution. Since $q$ measures the typical overlap
between weight vectors in the version space, the uniqueness of the
solution implies $q\to1$ at $\alpha_{c}$. In this limit, the order
parameters $\left\{ \hat{q}_{1},\Delta\hat{q}\right\} $ diverges
and we need to express them in terms of undiverged order parameters
$\left\{ u,v\right\} $:
\begin{equation}
\hat{q}_{1}=\frac{-u^{2}}{(1-q)^{2}};\qquad\Delta\hat{q}=\frac{v}{1-q}
\end{equation}
Then $X(t)$ becomes
\begin{eqnarray}
X(t) & = & \int_{-\infty}^{\infty}dw\exp\frac{1}{1-q}\left\{ -\frac{1}{2}vw^{2}+utw\right\} 
\end{eqnarray}
We can perform yet another saddle point approximation to the $w$
integral. To $\mathcal{O}(\frac{1}{1-q})$ we have
\begin{eqnarray}
\left\langle \ln X(t)\right\rangle _{t} & = & \frac{1}{2(1-q)}\left[-v\left\langle w^{2}\right\rangle _{t}+2u\left\langle tw\right\rangle _{t}\right],
\end{eqnarray}
where the saddle value $w$ satisfies
\begin{equation}
w=\frac{ut}{v}
\end{equation}
Assuming $u/v>0$, the integration range of $t$ is unaffected by
the saddle point approximation, $\left\langle \cdot\right\rangle _{t}=\int_{-\infty}^{\infty}(\cdot)Dt$.

Note that 
\begin{equation}
\left\langle tw\right\rangle _{t}=\frac{v}{u}\left\langle w^{2}\right\rangle _{t}.
\end{equation}
So we have 
\begin{equation}
G_{0}=\frac{1}{2(1-q)}\left(v-u^{2}+v\left\langle w^{2}\right\rangle _{t}\right)
\end{equation}
Now we can perform the $t$ integral in $\left\langle w^{2}\right\rangle _{t}$
and obtain
\begin{equation}
G_{0}=\frac{1}{2(1-q)}\left(v-u^{2}+\frac{u^{2}}{v}\right)
\end{equation}
We seek saddle-point self-consistency equations with respect to order
parameters $u$ and $v$:
\begin{equation}
0=\frac{\partial G_{0}}{\partial u}\Rightarrow0=-2u+\frac{2u}{v}
\end{equation}
\begin{equation}
0=\frac{\partial G_{0}}{\partial v}\Rightarrow0=1-\frac{v^{2}}{u^{2}}
\end{equation}
Solving gives $u=v=1$. So $G_{0}$ becomes
\begin{equation}
G_{0}=\frac{1}{2(1-q)}.\label{c2eq:G0_final}
\end{equation}

\subsubsection*{The energetic part}

We would like to perform a similar procedure as shown above, to Eqn.\ref{c2eq:cap_dist_K_replica}
using the replica-symmetric ansatz. 

Under the replica-symmetric ansatz $q_{\alpha\beta}=(1-q)\delta_{\alpha\beta}+q$,
Eqn.\ref{c2eq:cap_dist_K_replica} becomes
\begin{equation}
\begin{split}K & =i\sum_{\alpha}x_{\alpha}\rho_{\alpha}-\frac{1-q}{2}\sum_{\alpha}x_{\alpha}^{2}-\frac{q}{2}\left(\sum_{\alpha}x_{\alpha}\right)^{2}\\
 & \HSTeq i\sum_{\alpha}x_{\alpha}\rho_{\alpha}-\frac{1-q}{2}\sum_{\alpha}x_{\alpha}^{2}-it\sqrt{q}\sum_{\alpha}x_{\alpha}.
\end{split}
\end{equation}
where we have again used the Hubbard-Stratonovich transformation to
linearize the quadratic piece. Performing the Gaussian integrals in
$x_{\alpha}$ (define $\alpha=\frac{P}{N}$),
\begin{equation}
nG_{1}=\alpha\log\left[\left\langle \int_{\kappa}^{\infty}\frac{d\rho}{\sqrt{2\pi(1-q)}}\exp\left\{ -\frac{(\rho+t\sqrt{q})^{2}}{2(1-q)}\right\} \right\rangle _{t}^{n}\right].
\end{equation}
At the limit $n\to0$,
\begin{equation}
nG_{1}=\alpha n\left\langle \log\left[\int_{\kappa}^{\infty}\frac{d\rho}{\sqrt{2\pi(1-q)}}\exp\left\{ -\frac{(\rho+t\sqrt{q})^{2}}{2(1-q)}\right\} \right]\right\rangle _{t}.
\end{equation}
Perform the Gaussian integral in $\rho$ and define $\tilde{\kappa}=\frac{\kappa+t\sqrt{q}}{\sqrt{1-q}}$,
we have
\begin{equation}
G_{1}=\alpha\int Dt\log H(\tilde{\kappa}).
\end{equation}
At the limit $q\to1,\alpha\to\alpha_{c}$, $\int_{-\infty}^{\infty}Dt$
is dominated by $\int_{-\kappa}^{\infty}Dt$, and $H(\tilde{\kappa})\to\frac{1}{\sqrt{2\pi}\tilde{\kappa}}e^{-\tilde{\kappa}^{2}/2}$.
The $\mathcal{O}\left(\frac{1}{1-q}\right)$ (leading order) contribution
gives
\begin{equation}
G_{1}=-\frac{1}{2(1-q)}\alpha_{c}\int_{-\kappa}^{\infty}Dt(\kappa+t)^{2}.\label{c2eq:cap_dist_G1_1pop}
\end{equation}
Let $G=G_{0}+G_{1}$. As $n\to0$, $\left\langle V^{n}\right\rangle =e^{n\left(NG\right)}\to1+n\left(NG\right)$,
and $\left\langle \log V\right\rangle =\lim_{n\to0}\frac{\left\langle V^{n}\right\rangle -1}{n}=NG$. 
Combining with Eqn.\ref{c2eq:G0_final}, we have 
\begin{equation}
\left\langle \log V\right\rangle =\frac{N}{2(1-q)}\left[1-\alpha_{c}\int_{-\kappa}^{\infty}Dt(\kappa+t)^{2}\right]\label{c2eq:cap_dist_logV}
\end{equation}
Capacity $\alpha_{c}$ is reached when Eqn.\ref{c2eq:cap_dist_logV}
goes to zero. We arrive at the famous Gardner capacity
\begin{equation}
\alpha_{c}(\kappa)=\left[\int_{-\kappa}^{\infty}Dt(\kappa+t)^{2}\right]^{-1},
\end{equation}
which reduces to $\alpha_{c}=2$ when $\kappa=0$. 

\section{Teacher-student setup for generalization performance}

In this section, we introduce the basics of teacher-student setup
for studying perceptron's generalization performance, first developed
in \cite{seung1992statistical}. We follow the pedagogical review
in \cite{engel2001statistical}.

Let's consider a noisy teacher perceptron, $\boldsymbol{w}_{t}\in\mathbb{R}^{N}$,
given random inputs $\boldsymbol{\xi}^{\mu}$ with $p(\xi_{i}^{\mu})=\mathcal{N}(0,1)$,
it generate labels by $\zeta^{\mu}=\text{sgn}(\boldsymbol{w}_{t}\cdot\boldsymbol{\xi}^{\mu}/||\boldsymbol{w}_{t}||+\eta^{\mu})$,
where $\eta^{\mu}$ is input noise and $\eta^{\mu}\sim\mathcal{N}(0,\sigma^{2})$.
The student perceptron $\boldsymbol{w}_{s}$ (noiseless) tries to
predict the labels of $\boldsymbol{\xi}^{\mu}$ by computing $\hat{\zeta}^{\mu}=\text{sgn}(\boldsymbol{w}_{s}\cdot\boldsymbol{\xi}^{\mu}/||\boldsymbol{w}_{s}||)$.
We are interested in finding the max-margin student for the dataset
$\{\boldsymbol{\xi}^{\mu},\zeta^{\mu}\}_{\mu=1}^{p}$ the noisy teacher
generates: $\max\kappa:\zeta^{\mu}\boldsymbol{w}_{s}\cdot\boldsymbol{\xi}^{\mu}\geq\kappa||\boldsymbol{w}_{s}||$. 

The generalization error in this problem is defined to be the averaged
number of errors the student perceptron makes
\begin{equation}
\varepsilon_{g}=\left\langle \Theta\left(-\hat{\zeta}\zeta\right)\right\rangle _{\boldsymbol{\xi}\zeta}\label{c2eq:gen_err}
\end{equation}
In the following, we normalize both the teacher and the student's
weight vectors to have $\left\Vert \boldsymbol{w}_{s}\right\Vert =\left\Vert \boldsymbol{w}_{t}\right\Vert =\sqrt{N}.$
Eq.\ref{c2eq:gen_err} can be rewritten as 
\begin{equation}
\varepsilon_{g}=\left\langle \Theta\left(\left(\frac{\boldsymbol{w}_{t}\cdot\boldsymbol{\xi}^{\mu}}{\sqrt{N}}+\eta^{\mu}\right)\left(\frac{\boldsymbol{w}_{s}\cdot\boldsymbol{\xi}^{\mu}}{\sqrt{N}}\right)\right)\right\rangle _{\boldsymbol{\xi}\zeta}
\end{equation}
We can carry out the average explicitly by introducing variables $h_{0}=\boldsymbol{w}_{t}\cdot\boldsymbol{\xi}^{\mu}/\sqrt{N}+\eta^{\mu}$
and $h=\boldsymbol{w}_{s}\cdot\boldsymbol{\xi}^{\mu}/\sqrt{N}$ and
the corresponding delta-functions to enforce these relations. It is
also convenient to introduce the teacher-student overlap
\begin{equation}
R=\frac{\boldsymbol{w}_{s}\cdot\boldsymbol{w}_{t}}{\left\Vert \boldsymbol{w}_{s}\right\Vert \left\Vert \boldsymbol{w}_{t}\right\Vert },\label{c2eq:overlap}
\end{equation}
which is a measure of how close the student weight vector is to that
of the teacher's. After performing the integrals, we obtain
\begin{equation}
\varepsilon_{g}=\frac{1}{\pi}\arccos\left(\frac{R}{\sqrt{1+\sigma^{2}}}\right)
\end{equation}
\begin{figure}
\centering{}\includegraphics[scale=0.3]{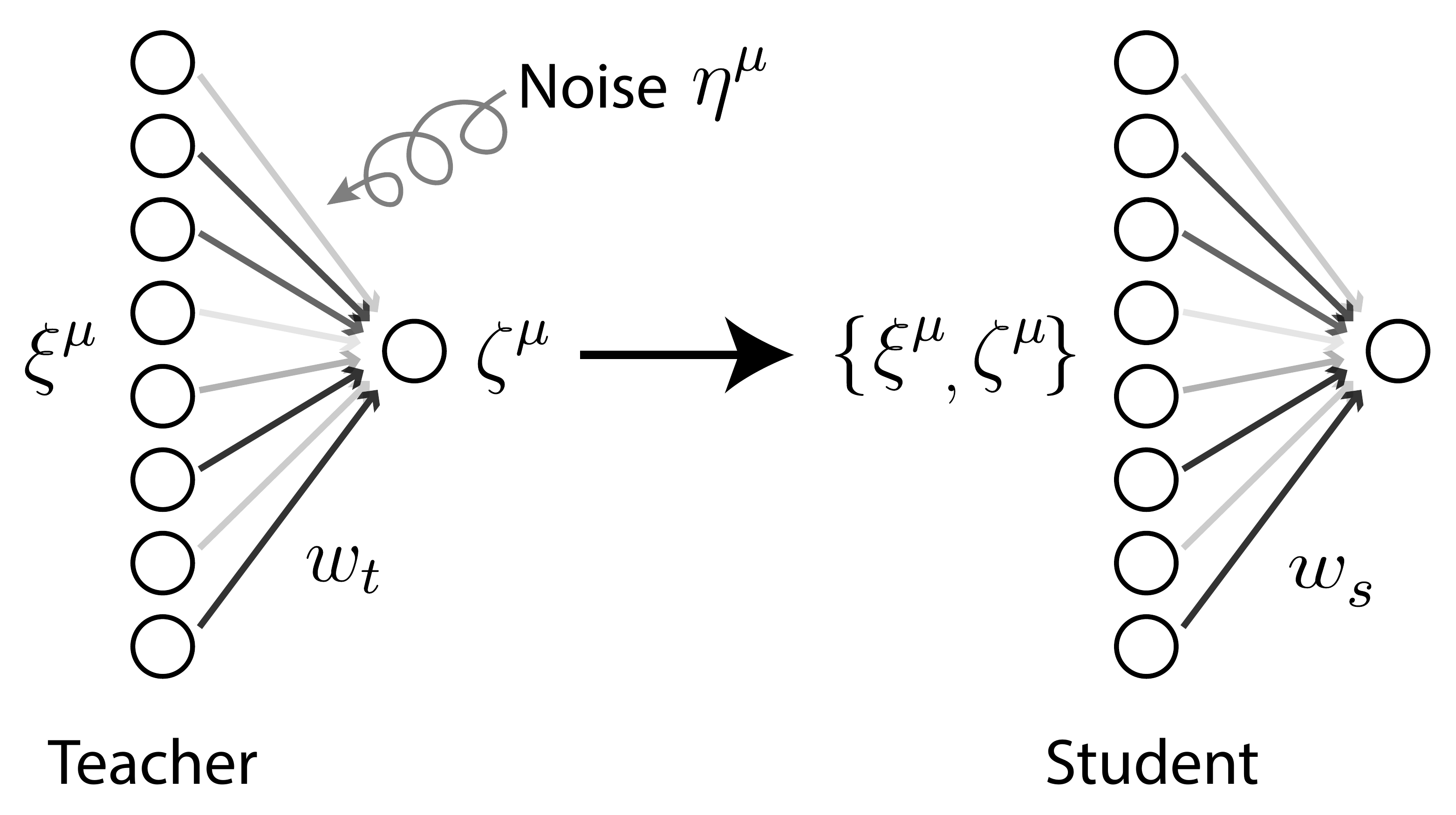}\caption{\label{fig:generalization_schemetic}Schematics of theteacher-student
setup. Left: A teacher perceptron $\boldsymbol{w}_{t}$ takes in inputs
$\boldsymbol{\xi}^{\mu}$, corrupted with noise $\eta^{\mu}$, then
generates output label $\zeta^{\mu}$. A student perceptron $\boldsymbol{w}_{s}$
tries to learn the input-output association $\{\boldsymbol{\xi}^{\mu},\zeta^{\mu}\}$
generated by the teacher.}
\end{figure}
In the following, to ease notation, we denote the teacher perceptron
$\boldsymbol{w_{t}}\equiv\boldsymbol{w}^{0}$ and the (replicated)
student perceptron $\boldsymbol{w}_{s}^{a}\equiv\boldsymbol{w}^{a}.$
Given random inputs $\boldsymbol{\xi}^{\mu}$ with $p(\xi_{i}^{\mu})=\mathcal{N}(0,1)$,
we generate labels by $\zeta^{\mu}=\text{sgn}(\boldsymbol{w}^{0}\cdot\boldsymbol{\xi}^{\mu}/||\boldsymbol{w}^{0}||+\eta^{\mu})$,
where $\eta^{\mu}$ is input noise and $\eta^{\mu}\sim\mathcal{N}(0,\sigma^{2})$.
The Gardner volume for this problem is: 
\begin{equation}
\left\langle V^{n}\right\rangle _{\xi\eta w^{0}}=\prod_{\alpha=1}^{n}\left\langle \int_{-\infty}^{\infty}\frac{d\boldsymbol{w}^{a}}{\sqrt{2\pi}}\prod_{\mu=1}^{p}\Theta\left(\text{sgn}\left(\frac{\boldsymbol{w^{0}}\cdot\boldsymbol{\xi}^{\mu}}{||\boldsymbol{w^{0}}||}+\eta^{\mu}\right)\frac{\boldsymbol{w^{a}}\cdot\boldsymbol{\xi}^{\mu}}{||\boldsymbol{w^{a}}||}-\kappa\right)\right\rangle _{\xi\eta w^{0}}.\label{sc_original}
\end{equation}
Let's define the local fields as
\begin{equation}
h_{\mu}^{a}=\frac{\boldsymbol{w^{a}}\cdot\boldsymbol{\xi}^{\mu}}{\sqrt{N}};\qquad h_{\mu}^{0}=\frac{\boldsymbol{w^{0}}\cdot\boldsymbol{\xi}^{\mu}}{\sqrt{N}}+\eta^{\mu}
\end{equation}
We leave the average over teacher $w^{0}$ to the very end.
\begin{equation}
\begin{split}\left\langle V^{n}\right\rangle _{\xi\eta} & =\prod_{\mu a}\int\frac{d\boldsymbol{w}^{a}}{\sqrt{2\pi}}\int dh_{\mu}^{a}\Theta\bigg(\text{sgn}(h_{\mu}^{0})h_{\mu}^{a}-\kappa\bigg)\left\langle \delta\left(h_{\mu}^{a}-\frac{\boldsymbol{w^{a}}\cdot\boldsymbol{\xi}^{\mu}}{\sqrt{N}}\right)\right\rangle _{\xi\eta}\left\langle \delta\left(h_{\mu}^{0}-\frac{\boldsymbol{w^{0}}\cdot\boldsymbol{\xi}^{\mu}}{\sqrt{N}}-\eta^{\mu}\right)\right\rangle _{\xi\eta}\\
 & =\int(\prod_{a=1}^{n}\frac{d\boldsymbol{w}^{a}}{\sqrt{2\pi}})\int\prod_{\mu a}\frac{dh_{\mu}^{a}d\hat{h}_{\mu}^{a}}{2\pi}\int\prod_{\mu}\frac{dh_{\mu}^{0}d\hat{h}_{\mu}^{0}}{2\pi}\prod_{\mu a}\Theta\bigg(\text{sgn}(h_{\mu}^{0})h_{\mu}^{a}-\kappa\bigg)\\
 & \times\bigg\langle\exp\bigg\{\sum_{\mu a}\bigg(i\hat{h}_{\mu}^{a}h_{\mu}^{a}-i\hat{h}_{\mu}^{a}\frac{\boldsymbol{w^{a}}\cdot\boldsymbol{\xi}^{\mu}}{\sqrt{N}}\bigg)+\sum_{\mu}\bigg(i\hat{h}_{\mu}^{0}h_{\mu}^{0}-i\hat{h}_{\mu}^{0}\frac{\boldsymbol{w^{0}}\cdot\boldsymbol{\xi}^{\mu}}{\sqrt{N}}-i\hat{h}_{\mu}^{0}\eta^{\mu}\bigg)\bigg\}\bigg\rangle_{\xi\eta}\\
 & =\int(\prod_{a=1}^{n}\frac{d\boldsymbol{w}^{a}}{\sqrt{2\pi}})\int\prod_{\mu a}\frac{dh_{\mu}^{a}d\hat{h}_{\mu}^{a}}{2\pi}\int\prod_{\mu}\frac{dh_{\mu}^{0}d\hat{h}_{\mu}^{0}}{2\pi}\prod_{\mu a}\Theta\bigg(\text{sgn}(h_{\mu}^{0})h_{\mu}^{a}-\kappa\bigg)\\
 & \times\exp\left\{ \sum_{\mu a}i\hat{h_{\mu}^{a}}h_{\mu}^{a}+\sum_{\mu}i\hat{h}_{\mu}^{0}h_{\mu}^{0}-\sum_{\mu}\frac{\sigma^{2}}{2}\hat{h}_{\mu}^{0}{}^{2}\right\} \\
 & \times\prod_{\mu}\exp\left\{ -\frac{1}{2N}\left[\sum_{a,b}\hat{h}_{\mu}^{a}\hat{h}_{\mu}^{b}\sum_{i}w_{i}^{a}w_{i}^{b}+N\left(\hat{h}_{\mu}^{0}\right)^{2}+2\sum_{a}\hat{h}_{\mu}^{a}\hat{h}_{\mu}^{0}\sum_{i}w_{i}^{a}w_{i}^{0}\right]\right\} ,
\end{split}
\end{equation}
where in the last step we perform the average over noise $\eta^{\mu}\sim\mathcal{N}(0,\sigma^{2})$
and patterns $p(\xi_{i}^{\mu})=\mathcal{N}(0,1)$, and make use of
the normalization conditions $\sum_{i}(w_{i}^{0})^{2}=N$ and $\sum_{i}(w_{i}^{a})^{2}=N$.

Now let's define
\begin{equation}
q^{ab}=\frac{1}{N}\sum_{i}w_{i}^{a}w_{i}^{b},\qquad R^{a}=\frac{1}{N}\sum_{i}w_{i}^{a}w_{i}^{0},\qquad\gamma=\frac{1}{\sqrt{1+\sigma^{2}}}
\end{equation}
Then, 
\begin{equation}
\begin{split}\langle V^{n}\rangle_{\xi\eta} & =\int(\prod_{a=1}^{n}\frac{d\boldsymbol{w}^{a}}{\sqrt{2\pi}})\int\prod_{\mu a}\frac{dh_{\mu}^{a}d\hat{h}_{\mu}^{a}}{2\pi}\int\prod_{\mu}\frac{dh_{\mu}^{0}d\hat{h}_{\mu}^{0}}{2\pi}\prod_{\mu a}\theta\bigg(\text{sgn}(h_{\mu}^{0})h_{\mu}^{a}-\kappa\bigg)\\
 & \times\int\prod_{a}Ndq^{ab}\int\prod_{a}NdR^{a}\prod_{ab}\delta(Nq^{ab}-\sum_{i}w_{i}^{a}w_{i}^{b})\prod_{a}\delta(NR^{a}-\sum_{i}w_{i}^{a}w_{i}^{0})\\
 & \times\exp\bigg\{ i\sum_{\mu a}\hat{h}_{\mu}^{a}h_{\mu}^{a}-\frac{1}{2}\sum_{\mu ab}\hat{h}_{\mu}^{a}\hat{h}_{\mu}^{b}q^{ab}-\frac{1}{2}\sum_{\mu}\gamma^{-2}(\hat{h}_{\mu}^{0})^{2}+i\sum_{\mu}\hat{h}_{\mu}^{0}h_{\mu}^{0}\\
 & -\sum_{\mu a}\hat{h}_{\mu}^{0}\hat{h}_{\mu}^{a}R^{a}\bigg\}
\end{split}
\end{equation}
We can do the $\hat{h}_{\mu}^{0}$ integral, and redefine $\bar{h}_{\mu}^{0}=\gamma h_{\mu}^{0}$
\begin{equation}
\begin{split}\langle V^{n}\rangle_{\xi\eta} & =\int(\prod_{a=1}^{n}\frac{d\boldsymbol{w}^{a}}{\sqrt{2\pi}})\int\prod_{ab}Ndq^{ab}\int\prod_{a}NdR^{a}\prod_{ab}\delta(Nq^{ab}-\sum_{i}w_{i}^{a}w_{i}^{b})\prod_{a}\delta(NR^{a}-\sum_{i}w_{i}^{a}w_{i}^{0})\\
 & \times\prod_{\mu}\bigg[\int\prod_{a}\frac{dh_{\mu}^{a}d\hat{h}_{\mu}^{a}}{2\pi}\int\frac{d\bar{h}_{\mu}^{0}}{\sqrt{2\pi}}\prod_{a}\theta\bigg(\text{sgn}(\bar{h}_{\mu}^{0})h_{\mu}^{a}-\kappa\bigg)\exp\bigg\{ i\sum_{a}\hat{h}_{\mu}^{a}h_{\mu}^{a}-\frac{1}{2}\sum_{ab}\hat{h}_{\mu}^{a}\hat{h}_{\mu}^{b}q^{ab}\\
 & +\frac{1}{2}\gamma^{2}\sum_{ab}\hat{h}_{\mu}^{a}\hat{h}_{\mu}^{b}R^{a}R^{b}-\frac{1}{2}(\bar{h}_{\mu}^{0})^{2}-i\gamma\sum_{a}h_{\mu}^{a}R^{a}\bar{h}_{\mu}^{0}\bigg\}\bigg]
\end{split}
\end{equation}
Next, we Fourier decompose the remaining $\delta$-functions by using
the identity 
\begin{equation}
\delta(x)=\int_{-i\infty}^{i\infty}\frac{d\hat{r}}{2\pi i}e^{-\hat{r}x}
\end{equation}
Note that for for the ease of notation later on, we choose to integrate
over the imaginary axis instead of the usual real axis. 
We also introduce conjugate variables $\hat{q}^{ab}$ and $\hat{R}^{a}$
to write the $\delta$-functions into its Fourier representations. 
After rescaling $\hat{q}^{ab}\to\hat{q}^{ab}/2$, and note that the
$\mu$-indexed terms factor into $p$-identical integrals, and the
$i$-indexed terms factor into $N$-identical integrals, we can bring
the Gardner volume into the following form ($\alpha\equiv p/N$):
\begin{equation}
\begin{split}\langle V^{n}\rangle_{\xi\eta} & =\int(\prod_{ab}dq^{ab}d\hat{q}^{ab})(\prod_{a}dR^{a}d\hat{R}^{a})e^{N\left(G_{0}+\alpha G_{E}\right)}\end{split}
,
\end{equation}
where ($\bar{h}_{\mu}^{0}=\gamma h_{\mu}^{0};\quad\gamma=1/\sqrt{1+\sigma^{2}}$)
the entropic part is
\begin{equation}
\begin{split}G_{0}= & -\frac{1}{2}\sum_{ab}\hat{q}^{ab}q^{ab}-\sum_{a}\hat{R}^{a}R^{a}+n\left\langle \ln Z\right\rangle _{w^{0}},\\
Z= & \int\left(\prod_{a}dw_{i}^{a}\right)\exp\bigg\{\frac{1}{2}\sum_{ab}\hat{q}^{ab}w_{i}^{a}w_{i}^{b}+\sum_{a}\hat{R}^{a}w_{i}^{a}w_{i}^{0}\bigg\},
\end{split}
\label{c2eq:entropic part}
\end{equation}
and the energetic part is
\begin{equation}
\begin{split}G_{1}= & \ln\int\prod_{a}\frac{d\hat{h}^{a}dh^{a}}{2\pi}\int D\bar{h}^{0}\prod_{a}\Theta\bigg(\text{sgn}(\frac{\bar{h}^{0}}{\gamma})h^{a}-\kappa\bigg)\\
 & \times\exp\bigg\{ i\sum_{a}\hat{h}^{a}h^{a}-i\gamma\bar{h}^{0}\sum_{a}h^{a}R^{a}-\frac{1}{2}\sum_{ab}\hat{h}^{a}\hat{h}^{b}(q^{ab}-\gamma^{2}R^{a}R^{b})\bigg\}.
\end{split}
\label{c2eq:energetic part}
\end{equation}

\subsection{Energetic part}

In this subsection, we try to perform the integrations in $G_{1}$.
In the following, we assume replica symmetric solutions 
\begin{equation}
q^{ab}=q+(1-q)\delta_{ab};\qquad R^{a}=R
\end{equation}
First note that $\gamma=\frac{1}{\sqrt{1+\sigma^{2}}}>0$, so $\text{sgn}(\bar{h}^{0}/\gamma)=\text{sgn}(\bar{h}^{0})$.
Now the effect of $\theta(\text{sgn}(\bar{h}^{0})h^{a}-\kappa)$ can
be understood as the following: 
\[
\begin{cases}
\bar{h}^{0}>0;\;h^{a}>\kappa\Rightarrow\int_{\kappa}^{\infty}dh^{a}\int_{0}^{\infty}d\bar{h}^{0}\\
\bar{h}^{0}<0;\;h^{a}<-\kappa\Rightarrow\int_{-\infty}^{-\kappa}dh^{a}\int_{-\infty}^{0}d\bar{h}^{0}=\int_{\kappa}^{\infty}dh^{a}\int_{0}^{\infty}d\bar{h}^{0}
\end{cases}
\]
Therefore, the net effect of the step-function is to modify the integration
range and an overall factor of 2 in the energetic part (Eqn. (\ref{c2eq:energetic part})),
\begin{equation}
\begin{split}G_{1} & =2\int\prod_{a}\frac{d\hat{h}^{a}}{2\pi}\int_{\kappa}^{\infty}dh^{a}\int_{0}^{\infty}D\bar{h}^{0}\exp\bigg\{ i\sum_{a}\hat{h}^{a}h^{a}-i\bar{h}^{0}\gamma R\sum_{a}h^{a}\\
 & -\frac{1}{2}(1-\gamma^{2}R^{2})\sum_{a}(\hat{h}^{a})^{2}\underbrace{-\frac{1}{2}(q-\gamma^{2}R^{2})\sum_{a\neq b}\hat{h}^{a}\hat{h}^{b}}_{=-\frac{1}{2}(q-\gamma^{2}R^{2})(\sum_{a}\hat{h}^{a})^{2}+\frac{1}{2}(q-\gamma^{2}R^{2})\sum_{a}(\hat{h}^{a})^{2}}\bigg\}
\end{split}
\end{equation}
We can linearize the $(\sum_{a}\hat{h}^{a})^{2}$ term using the Hubbard-Stratonovich
transformation 
\begin{equation}
-\frac{1}{2}(q-\gamma^{2}R^{2})(\sum_{a}\hat{h}^{a})^{2}=\int Dt\exp\bigg\{-i\sqrt{q-\gamma^{2}R^{2}}t\sum_{a}\hat{h}^{a}\bigg\}
\end{equation}
Then, 
\begin{equation}
\begin{split}G_{1} & =2\int\prod_{a}\frac{d\hat{h}^{a}}{2\pi}\int_{\kappa}^{\infty}dh^{a}\int_{0}^{\infty}D\bar{h}^{0}\int Dt\\
 & \times\exp\bigg\{ i\sum_{a}\hat{h}^{a}\bigg(h^{a}-\bar{h}^{0}\gamma R-\sqrt{q-\gamma^{2}R^{2}}t\bigg)-\frac{1}{2}(1-q)\sum_{a}(\hat{h}^{a})^{2}\bigg\}\\
 & =\ln2\int\prod_{a}\frac{d\hat{h}^{a}}{2\pi}\int_{\kappa}^{\infty}dh^{a}\int_{0}^{\infty}D\bar{h}^{0}\int Dt\\
 & \times\prod_{a}\exp\bigg\{-\frac{1-q}{2}\bigg[\hat{h}^{a}-i\frac{h^{a}-\bar{h}^{0}\gamma R-\sqrt{q-\gamma^{2}R^{2}}t}{1-q}\bigg]^{2}-\frac{(h^{a}-\bar{h}^{0}\gamma R-\sqrt{q-\gamma^{2}R^{2}}t)^{2}}{2(1-q)}\bigg\}
\end{split}
\end{equation}
We can do the $\hat{h}^{a}$-integral,
\begin{equation}
\begin{split}G_{1}=2\int_{0}^{\infty}D\bar{h}^{0}\int Dt\bigg[\int_{\kappa}^{\infty}\frac{dh^{a}}{\sqrt{2\pi(1-q)}}\exp\bigg\{-\frac{1}{2}\frac{(h^{a}-\bar{h}^{0}\gamma R-\sqrt{q-\gamma^{2}R^{2}}t)^{2}}{1-q}\bigg\}\bigg]^{n}\end{split}
\end{equation}
Then we can do the $h^{a}$ and $\bar{h}^{0}$-integrals by expressing
the result in terms of $H$-functions ($H(x)=\int_{x}^{\infty}Dt$)
and simplifying,
\begin{equation}
\begin{split}G_{1} & =\end{split}
2\int DtH\bigg(-\frac{\gamma Rt}{\sqrt{q-\gamma^{2}R^{2}}}\bigg)H^{n}\bigg(\frac{\kappa-\sqrt{q}t}{\sqrt{1-q}}\bigg)
\end{equation}
Note that as $n\to0$, at leading order 
\begin{equation}
\ln\int DtH(x)H^{n}(y)=n\int DtH(x)\ln H(y)
\end{equation}
So we have 
\begin{equation}
G_{1}/n=2\int DtH\bigg(-\frac{\gamma Rt}{\sqrt{q-\gamma^{2}R^{2}}}\bigg)\ln H\bigg(\frac{\kappa-\sqrt{q}t}{\sqrt{1-q}}\bigg)
\end{equation}
As $q\to1$, only the max-margin solution exists, $\kappa\to\kappa_{max}$,
and 
\begin{equation}
G_{1}/n=-\frac{\alpha}{1-q}\int_{-\infty}^{\kappa_{max}}DtH\bigg(-\frac{\gamma Rt}{\sqrt{1-\gamma^{2}R^{2}}}\bigg)(\kappa_{max}-t)^{2}.
\end{equation}

\subsection{Entropic part}

In this subsection, we perform the integrals in the entropic part.
We start by assuming a replica-symmetric solution for the auxiliary
variables introduced in the Fourier decomposition of the $\delta$-functions,
\begin{equation}
\hat{R}^{a}=\hat{R};\qquad\hat{q}^{ab}=\hat{q}+(\hat{q}_{1}-\hat{q})\delta_{ab}
\end{equation}
Then the log-term in the entropic part becomes, 
\begin{equation}
\begin{split}Z & =\int\left(\prod_{a}\frac{dw_{i}^{a}}{\sqrt{2\pi}}\right)\exp\bigg\{\frac{1}{2}(\hat{q}_{1}-\hat{q})\sum_{a}(w_{i}^{a})^{2}+\hat{R}w_{i}^{0}\sum_{a}w_{i}^{a}+\frac{1}{2}\hat{q}(\sum_{a}w_{i}^{a})^{2}\bigg\}\\
 & \HSTeq\int Dt\int(\prod_{a}\frac{dw_{i}^{a}}{\sqrt{2\pi}})\exp\bigg\{\frac{1}{2}(\hat{q}_{1}-\hat{q})\sum_{a}(w_{i}^{a})^{2}+(\hat{R}w_{i}^{0}+t\sqrt{\hat{q}})\sum_{a}w_{i}^{a}\bigg\},
\end{split}
\end{equation}
where we have introduced Gaussian variable $t$ to linearize quadratic
term as usual. Now the integral becomes $n$ identical copies and
we can drop the replica index $a$.

Note to $\mathcal{O}(n)$ we have
\begin{equation}
-\frac{1}{2}\sum_{ab}\hat{q}^{ab}q^{ab}=-\frac{n}{2}\left(\hat{q}_{1}-\hat{q}q\right)
\end{equation}
Therefore,
\begin{equation}
G_{0}/n=-\frac{1}{2}\hat{q}_{1}+\frac{1}{2}\hat{q}q-\hat{R}R+\left\langle \ln Z\right\rangle _{t,w^{0}}.
\end{equation}
We can bring the log term into the form of an induced distribution
$f(w)$,
\begin{equation}
\begin{split}Z= & \int_{-\infty}^{\infty}\frac{dw}{\sqrt{2\pi}}\exp\left[-f(w)\right]\\
f(w)= & \frac{1}{2}(\hat{q}-\hat{q}_{1})w^{2}-(\hat{R}w^{0}+t\sqrt{\hat{q}})w
\end{split}
.
\end{equation}
Under saddle-point approximation, we obtain a set of mean field self-consistency
equations for the order parameters: 
\begin{equation}
\begin{split}0=\frac{\partial G_{0}}{\partial\hat{q}_{1}} & \Rightarrow1=\left\langle \left\langle w^{2}\right\rangle _{f}\right\rangle _{t,w^{0}}\\
0=\frac{\partial G_{0}}{\partial\hat{R}} & \Rightarrow R=\left\langle w^{0}\left\langle w\right\rangle _{f}\right\rangle _{t,w^{0}}\\
0=\frac{\partial G_{0}}{\partial\hat{q}} & \Rightarrow q=\left\langle \left\langle w\right\rangle _{f}^{2}\right\rangle _{t,w^{0}}
\end{split}
,\label{c2eq:sc_saddle_1-3_qneq1}
\end{equation}
\begin{equation}
\begin{split}0=\frac{\partial G_{1}}{\partial q} & \Rightarrow\hat{q}=-2\alpha\partial_{q}G_{1}\\
0=\frac{\partial G_{1}}{\partial R} & \Rightarrow\hat{R}=\alpha\partial_{R}G_{1}
\end{split}
.\label{c2eq:sc_saddle_4-5_qneq1}
\end{equation}
Note that Eq.(\ref{c2eq:sc_saddle_1-3_qneq1}) has nice interpretations:
(1) the weights must be normalized WRT the induced distribution $f(w)$;
(2) the student-teacher overlap $R$ is the overlap between teacher
$w^{0}$ and the average student from the family $f(w)$; (3) the
replica-overlap $q$ is the overlap between students drawn from $f(w)$. 

\subsubsection*{$q\to1$ limit}

In this limit the order parameter diverges, and we define the new
set of undiverged order parameters as
\begin{equation}
\hat{R}=\frac{\tilde{R}}{1-q};\qquad\hat{q}=\frac{\tilde{q}^{2}}{(1-q)^{2}};\qquad\hat{q}-\hat{q}_{1}=\frac{\Delta}{1-q}.
\end{equation}
Then 
\begin{equation}
\begin{split}f(w)= & \frac{1}{1-q}\left[\frac{1}{2}\Delta w^{2}-(\tilde{R}w^{0}+t\tilde{q})w\right]\end{split}
\end{equation}
We can perform again saddle point approximation for $Z$, 
\begin{equation}
\ln Z=\frac{1}{1-q}\left\{ -\frac{1}{2}\left\langle \Delta w^{2}\right\rangle _{t,w^{0}}-\left\langle (\tilde{R}w^{0}+t\tilde{q})^{2}w^{2}\right\rangle _{t,w^{0}}\right\} 
\end{equation}
where at the saddle 
\begin{equation}
w=-\frac{\tilde{R}w^{0}+t\tilde{q}}{\Delta}
\end{equation}
Now $G_{0}$ can be simplified to
\begin{equation}
G_{0}/n=\frac{1}{2(1-q)}\left\{ \Delta-\tilde{q}^{2}-2\tilde{R}R+\frac{\tilde{R}^{2}\tilde{q}^{2}}{2\Delta}\right\} 
\end{equation}
We can write down the saddle point equations:
\begin{equation}
0=\frac{\partial G_{0}}{\partial\Delta}\Rightarrow0=1-\frac{\tilde{R}^{2}\tilde{q}^{2}}{2\Delta^{2}}
\end{equation}
\begin{equation}
0=\frac{\partial G_{0}}{\partial\tilde{q}}\Rightarrow0=-2\tilde{q}+\frac{\tilde{q}}{\Delta}
\end{equation}
\begin{equation}
0=\frac{\partial G_{0}}{\partial\tilde{R}}\Rightarrow0=-2R+\frac{\tilde{R}}{\Delta}
\end{equation}
Solving these equations gives $\tilde{R}=R,$$\Delta=1/2$, and $\tilde{q}^{2}=1/2-R^{2}.$
We can further simplify $G_{0}$ into
\begin{equation}
G_{0}/n=\frac{1}{2(1-q)}\left(1-R^{2}\right)
\end{equation}

\subsection{Summary}

Putting everything together, we have 
\begin{equation}
G/n=\frac{1}{(1-q)}\left\{ \frac{1-R^{2}}{2}-\alpha\int_{-\infty}^{\kappa_{max}}DtH\bigg(-\frac{\gamma Rt}{\sqrt{1-\gamma^{2}R^{2}}}\bigg)(\kappa_{max}-t)^{2}\right\} 
\end{equation}
The two remaining saddle-point equations, $\log V=0$ $(G=0)$ and
$0=\frac{\partial G_{0}}{\partial R}$ self-consistently determine
the two order parameters $\left\{ R,\kappa_{max}\right\} $,
\begin{equation}
1-R^{2}=2\alpha\int_{-\infty}^{\kappa_{max}}DtH\bigg(-\frac{\gamma Rt}{\sqrt{1-\gamma^{2}R^{2}}}\bigg)(\kappa_{max}-t)^{2}\label{c2eq:tsmf_1}
\end{equation}
\begin{equation}
R=\frac{\alpha\gamma}{\sqrt{2\pi}}\sqrt{1-\gamma^{2}R^{2}}\int_{-\tilde{\kappa}}^{\infty}Dt\bigg(\tilde{\kappa}+t\bigg)\label{c2eq:tsmf_2}
\end{equation}
where $\tilde{\kappa}=\kappa/\sqrt{1-\gamma^{2}R^{2}}$. 
We numerically solve Eq.(\ref{c2eq:tsmf_1})-(\ref{c2eq:tsmf_2})
in Fig.\ref{fig:generalization_numerics}. 
\begin{figure}
\centering{}\includegraphics[scale=0.3]{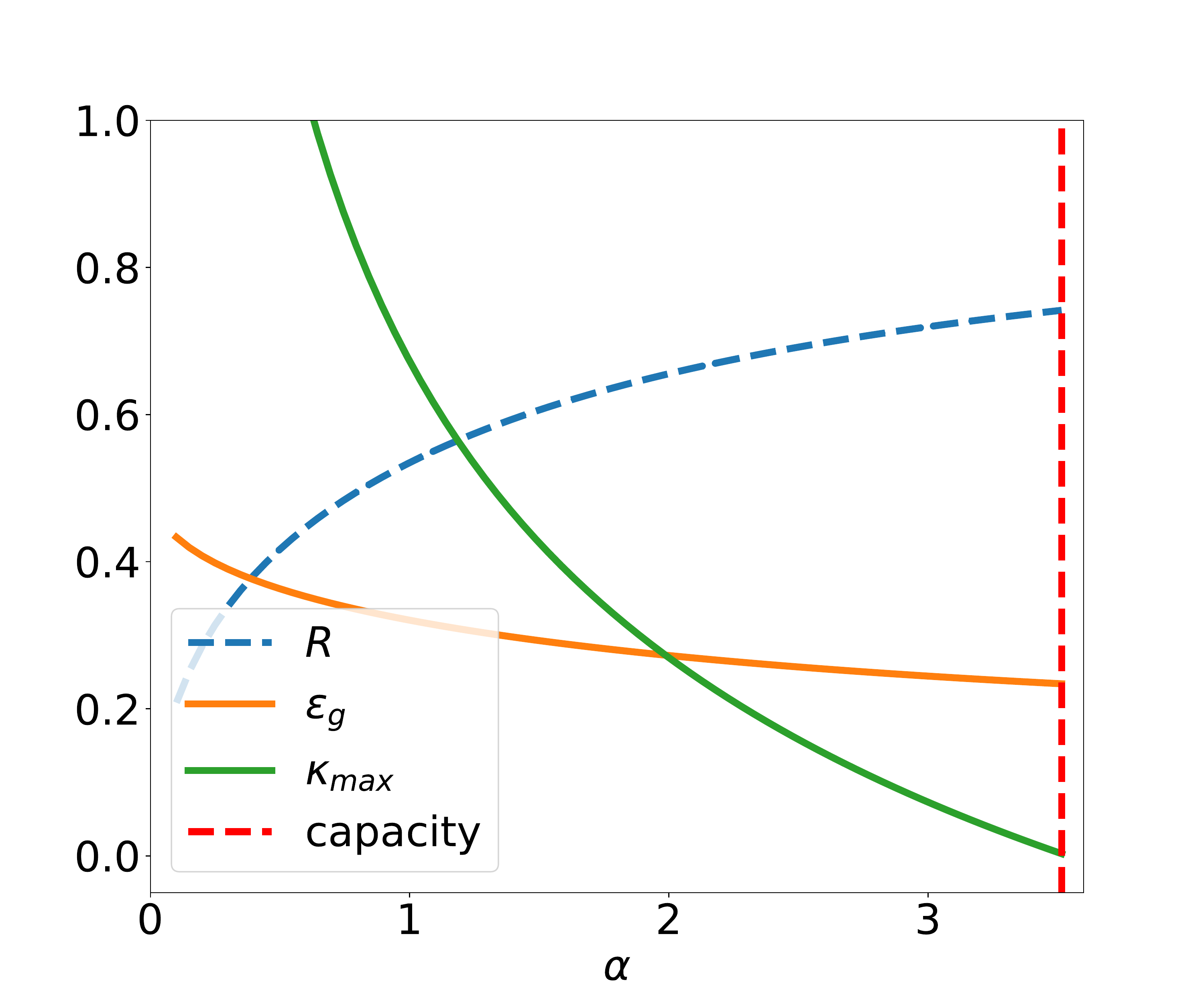}\caption{\label{fig:generalization_numerics}Numerical solutions of overlap
$R$, max-margin $\kappa_{max}$, and generalization error $\varepsilon_{g}$
as a function of load $\alpha$. The red dashed line represents capacity
due to nonzero teacher noise, beyond which the student can no longer
learn the data perfectly ($\kappa_{max}<0$). }
\end{figure}
\label{gardner}

%auto-ignore
%% This is an example first chapter.  You should put chapter/appendix that you
%% write into a separate file, and add a line \include{yourfilename} to
%% main.tex, where `yourfilename.tex' is the name of the chapter/appendix file.
%% You can process specific files by typing their names in at the 
%% \files=
%% prompt when you run the file main.tex through LaTeX.

\chapter{Feedforward neural networks under structural constraints}
\label{disco}
\section{Introduction}

Learning and memory are thought to take place at the microscopic level
by modifications of synaptic connections. Unlike learning in artificial
neural networks, synaptic plasticity in the brain operates under structural
biological constraints. Theoretical efforts to incorporate some of
these constraints have focused largely on the degree of connectivity
\cite{buzsaki2014log,koulakov2009correlated} and the constraints
on the sign of the synapses (Excitatory vs. Inhibitory) \cite{amit1989perceptron,brunel2004optimal},
but few include additional features of synaptic weight distributions
observed in the brain \cite{barbour2007can}. More generally, recent
large-scale connectomic studies \cite{kunst2019cellular,scheffer2020connectome,shapson2021connectomic}
are beginning to provide a wealth of structural information of neuronal
circuits at an unprecedented scope and level of precision, which presents
a remarkable opportunity for a more refined theoretical study of learning
and memory that takes into account these hitherto unavailable structural
information.

Perceptron \cite{rosenblatt1958perceptron} is arguably the simplest
model of computation by single neuron and is the fundamental building
block for many modern neural networks. Despite the drastic oversimplification,
studying the computational properties of (binary and analog) perceptron
has been used extensively in computational neuroscience since its
dawn, particularly in the cerebellum (as a model of sensory-motor
association) but also in cerebral cortex (for generic associative
memory functions) \cite{marr1969theory,albus1971theory,brunel2004optimal,chapeton2012efficient,brunel2016cortical,bouvier2018cerebellar}.
Forming associations is considered an ‘atomic’ building block for
generic cortical functions, and perceptron memory capacity sets a
tight bound on the memory capacity in recurrently connected neuronal
circuits with application to cortex and hippocampus\textcolor{red}{{}
}\cite{gardner1987maximum,rolls1998neural,rubin2017balanced}. Statistical
mechanical analysis predicts that near capacity, an unconstrained
perceptron classifying random input-output associations has normally
distributed weights \cite{gardner1988optimal,gardner1988space,4038449},
see Fig.\ref{fig:motivation}(a). In contrast, physiological experiments
suggest that biological synapses do not change their excitatory/inhibitory
identity during learning (but see recent \cite{kim2022co}). In order
to take perceptron a step closer to biological realism, prior work
has imposed sign constraints during learning \cite{amit1989perceptron,brunel2004optimal}.
In this case, the predicted weight distribution is a delta-function
centered at zero plus a half-normal distribution, see Fig.\ref{fig:motivation}(b).
However, a wide range of connectomic studies ranging from cortical
circuits in animals \cite{levy2012spatial,holmgren2003pyramidal,molnar2008complex,yang2013development,shapson2021connectomic,loewenstein2011multiplicative,avermann2012microcircuits},
to human cerebral cortex \cite{molnar2008complex,shapson2021connectomic}
have shown evidence of lognormally distributed synaptic connections.
As an example, Fig.\ref{fig:motivation}(c) shows the weight connection
distribution in mouse primary auditory cortex (data adapted from \cite{levy2012spatial}).
 Possible reasons for the ubiquitous lognormal distributions range
from biological structural/developmental constraints to computational
benefits \cite{teramae2014computational}. Various potential mechanisms
for lognormal distributions has been proposed, from multiplicative
gradient updates in feedforward networks\cite{kivinen1997exponentiated,loewenstein2011multiplicative},
to mixture of additive and multiplicative plasticity rules in spiking
networks\cite{gilson2011stability}, but the majority of these proposals
lead not just to lognormal distributions but also to sparsification
in the weights. Instead of adding yet another explanation to the
computational origin of lognormal distribution, here we take the observed
weight distribution as a prior on the network structure, and ask for
its computational consequences. The goal of the paper is to present
for the first time a quantitative and qualitative theory of neural
network learning performance under non-Gaussian and general weight
distributions (not limited to lognormal distributions). 

In this paper, we combine two powerful tools: statistical mechanics
and optimal transport theory, and present a theory of perceptron learning
that incorporates the knowledge of both distribution and sign information
as constraints, and gives accurate predictions for capacity and generalization
error. Interestingly, the theory predicts that the reduction in capacity
due to the constrained weight-distribution is related to the Wasserstein
distance between the cumulative distribution function of the constrained
weights and that of the standard normal distribution. Along with the
theoretical framework, we also present a learning algorithm derived
from information geometry that is capable of efficiently finding viable
perceptron weights that satisfy desired distribution and sign constraints.
This paper is organized as follows: in Section \ref{sec:capacity_one_population}
we derive the perceptron capacity for classifying random input-output
associations using statistical mechanics, and illustrate our theory
with a simple example. In Section \ref{sec:disco_algorithm}, we derive
our learning algorithm using optimal transport theory, and show that
distribution of weights found by the learning algorithm coincide with
geodesic distributions on a Wasserstein statistical manifold, and
therefore training can be interpreted as a geodesic flow. In Section
\ref{sec:compare_experiment} we analyze a parameterized family of
biologically realistic weight distributions, and use our theory to
predict the shape of the distribution with optimal parameters. We
map out the experimental parameter landscape for the estimated distribution
of synaptic weights in mammalian cortex and show that our theory's
prediction for optimal distribution is close to the experimentally
measured value. In Section \ref{sec:Generalization} we further develop
a statistical mechanical theory for teacher-student perceptron rule
learning and ask for the best way for the student to incorporate prior
knowledge about the weight distribution of the rule (i.e., the teacher).
Our theory shows that it is beneficial for the learner to adopt different
prior weight distributions during learning. 

\begin{figure}
\centering{}\includegraphics[scale=0.18]{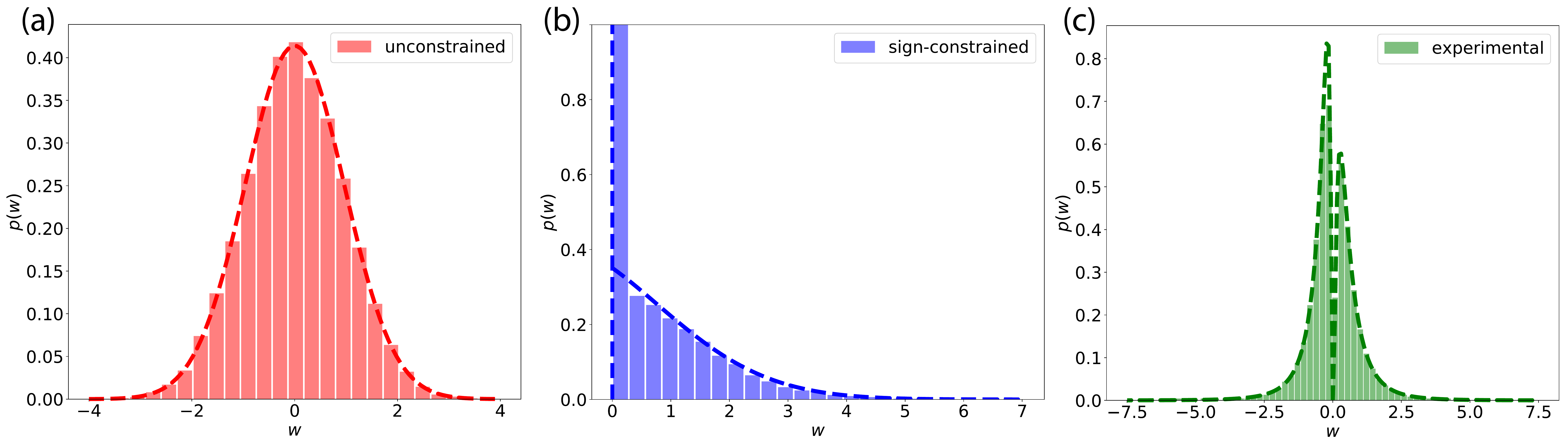}\caption{\label{fig:motivation}Theoretical and empirical synaptic weight distributions.
(a)-(b) predicted distribution following perceptron learning at capacity.
(a) Normal distribution when learning is unconstrained. (b) A delta-function
plus a half-normal distribution when learning is sign-constrained.
(c) Experimentally measured synaptic weight distribution (mouse primary
auditory cortex \cite{levy2012spatial}).}
\end{figure}

\section{Capacity}

\subsection{Learning under weight distribution constraints}

\label{sec:capacity_one_population}

We begin by considering a canonical learning problem: classifying
random input-output associations by a perceptron. In biological memory
systems, the heavily correlated sensory data is undergoing heavy preprocessing
including massive decorrelations, and previous work on brain related
perceptron modeling \cite{gardner1987maximum,brunel2004optimal,rubin2017balanced}
assumes similarly unstructured data. The data consists of pairs$\left\{ \boldsymbol{\xi}^{\mu},\zeta^{\mu}\right\} _{\mu=1}^{P}$,
where $\boldsymbol{\xi}^{\mu}$ is an $N$-dimentional random vector
drawn i.i.d. from a standard normal distribution, $p(\xi_{i}^{\mu})=\mathcal{N}(0,1)$,
and $\zeta^{\mu}$ are random binary class labels with $p(\zeta^{\mu})=\frac{1}{2}\delta(\zeta^{\mu}+1)+\frac{1}{2}\delta(\zeta^{\mu}-1)$.
The goal is to find a hyperplane through the origin, described by
a perceptron weight vector $\boldsymbol{w}\in\mathbb{R}^{N}$, normalized
to $||\boldsymbol{w}||^{2}=N.$

We call $\boldsymbol{w}$ a separating hyperplane when it correctly
classifies all the examples with margin $\kappa>0$:

\begin{equation}
\zeta^{\mu}\frac{\boldsymbol{w}\cdot\boldsymbol{\xi}^{\mu}}{||\boldsymbol{w}||}\geq\kappa.\label{eq:margin}
\end{equation}

We are interested in solutions $\boldsymbol{w}$ to Eqn.\ref{eq:margin}
that obey a prescribed distribution constraint, $w_{i}\sim q(w)$,
where $q$ is an arbitrary probability density function. We further
demand that $\langle w^{2}\rangle_{q(w)}=1$ to fix the overall scale
of the distribution (since the task is invariant to the overall scale
of $w$). Thus, the goal of learning is to find weights that satisfy
\ref{eq:margin} with the additional constraint that the empirical
density function $\hat{q}(w)=\frac{1}{N}\sum_{i}^{N}\delta(w-w_{i})$,
formed by the learned weights is similar to $q(w)$, and more precisely
that it converges to $q(w)$ as $N\rightarrow\infty$ (see Section
\ref{subsec:capacity_gardner} below). 

Extension of this setup that includes an arbitrary number of populations
each satisfying its own prescribed distribution constraints is discussed
in Section \ref{sec:compare_experiment} and in Appendix \ref{app:capacity_M_pop}.
Note that the sign constraint is a special case of this scenario with
two synaptic populations: one excitatory and one inhibitory. We further
discuss the generalization of this setup to include biased inputs
and sparse labels in Appendix \ref{app:capacity_biased_sparse}.

\subsection{Statistical mechanical theory of capacity}

\label{subsec:capacity_gardner}

We are interested in the thermodynamic limit where $P,N\to\infty$
, but the load $\alpha=\frac{P}{N}$ stays $\mathcal{O}(1)$. This
limit is amenable to mean-field analysis using statistical mechanics.

Following Gardner's seminal work \cite{gardner1988optimal,gardner1988space},
we consider the fraction $V$ of viable weights that satisfies both
Eqn.\ref{eq:margin} and the distribution constraint $\hat{q}=q$$,$
to all possible weights:

\begin{equation}
V=\frac{\int d\boldsymbol{w}\left[\prod_{\mu=1}^{P}\Theta\left(\zeta^{\mu}\frac{\boldsymbol{w}\cdot\boldsymbol{\xi}^{\mu}}{||\boldsymbol{w}||}-\kappa\right)\right]\delta(||\boldsymbol{w}||^{2}-N)\delta\bigg(\int dk\left(\hat{q}(k)-q(k)\right)\bigg)}{\int d\boldsymbol{w}\delta(||\boldsymbol{w}||^{2}-N)}.\label{eq:volume}
\end{equation}
In Eqn.\ref{eq:volume}, we impose the distribution constraint $\hat{q}=q$
by demanding that in the thermodynamic limit, all Fourier modes of
$q$ and $\hat{q}$ are the same , i.e., that $q(k)=\int dwe^{ikw}q(w)$
= $\hat{q}(k)=\frac{1}{N}\sum_{i}^{N}e^{ikw_{i}},$where in the last
equality we have used the definition of empirical distribution. We
perform a quenched average over random patterns $\boldsymbol{\xi}^{\mu}$
and labels $\zeta^{\mu}$. This amounts to calculating $\left\langle \log V\right\rangle $,
which can be done using the replica trick \cite{gardner1988optimal,gardner1988space}.

We focus on solutions with maximum margin $\kappa$ at a given load
$\alpha$, or equivalently, the maximum load capacity $\alpha_{c}(\kappa)$
of separable patterns given margin $\kappa$. We proceed by assuming
replica symmetry in our mean field analysis, which in general might
not hold because the constraint $\hat{q}=q$ is non-convex. For all
the results presented in the main text, replica symmetry solution
is supported by numerical simulations. In Appendix \ref{app:Replica-symmetry-breaking}
we explore the validity of replica symmetric solutions in the case
of strongly bimodal distributions and show that they fail only very
close to the binary (Ising) limit.

Detailed calculations of the mean-field theory are presented in Appendix
\ref{app:capacity_dist_const}. Our mean-field theory predicts that
the reduction in capacity due to the distribution constraint is proportional
to the Jacobian of the transformation from $w\sim q(w)$ to a normally
distributed variable $x(w)\sim\mathcal{N}(0,1)$,

\begin{equation}
\alpha_{c}(\kappa)=\alpha_{0}(\kappa)\left\langle \frac{dw}{dx}\right\rangle _{x}^{2},\label{eq:single_pop}
\end{equation}

where $\alpha_{0}(\kappa)=\left[\int_{-\kappa}^{\infty}Dt(\kappa+t)^{2}\right]^{-1}$
is the capacity of an unconstrained perceptron, from Gardner theory
\cite{gardner1988optimal,gardner1988space}, and $\kappa=0$ reduces
to the classical result of $\alpha_{0}(0)=2.$ The Jacobian factor,
$\langle dw/dx\rangle_{x}$, can be written in terms of the constrained
distribution's cumulative distribution function (CDF), $Q(w)$, and
the standard normal CDF $P(x)=\frac{1}{2}\left[1+\text{Erf}(\frac{x}{\sqrt{2}})\right]$,
namely,

\begin{equation}
\left\langle \frac{dw}{dx}\right\rangle _{x}=\int_{0}^{1}duQ^{-1}(u)P^{-1}(u).\label{eq:jacobian}
\end{equation}

Note that since the second moments are fixed to unity, $0\leq\left\langle \frac{dw}{dx}\right\rangle _{x}\leq1$
and it equals $1$ iff $p=q.$ 

\begin{figure}
\centering{}\includegraphics[scale=0.16]{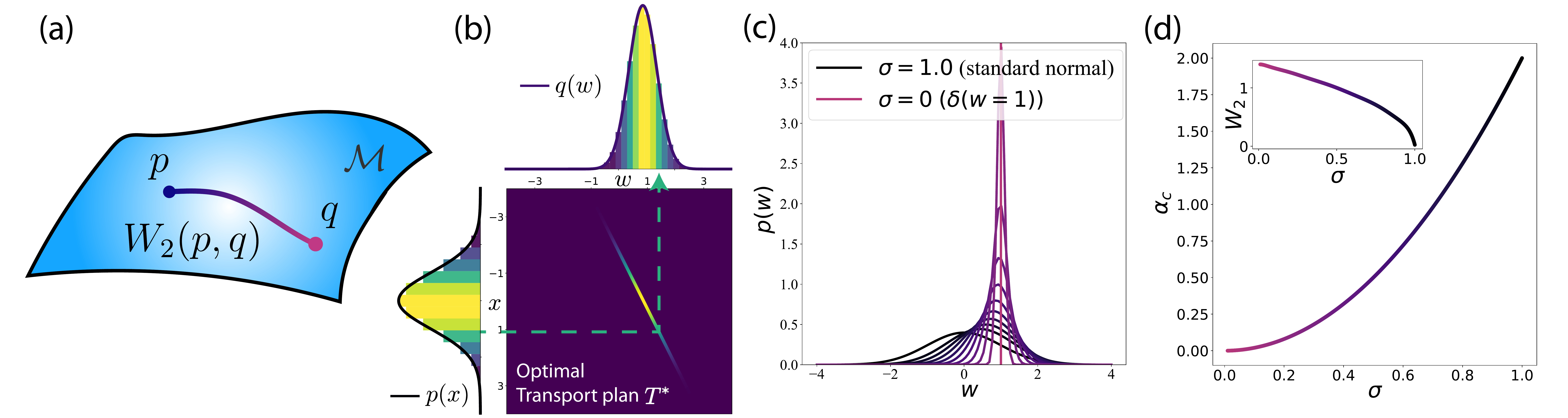}\caption{\label{fig:theory_illustration}An illustration of optimal transport
from a standard normal distribution $\mathcal{N}(0,1)$ to normal
distributions with nonzero mean $\mathcal{N}(\sqrt{1-\sigma^{2}},\sigma^{2})$.
(a) A schematic of the space $(\mathcal{M},W_{2})$ of probability
distributions. (b) An example optimal transport plan from standard
normal, $p(x)$, to a normal with $\sigma=0.5$, $q(w)$. The optimal
transport plan $T^{*}$ is plotted in between the distributions. $T^{*}$
moves $p(x)$ units of probability mass $x$ to location $w$, as
indicated by the dashed line, and the colors are chosen to reflect
the amount of probability mass to be transported. (c) $\mathcal{N}(\sqrt{1-\sigma^{2}},\sigma^{2})$
interpolates between standard normal ($\sigma=0$) to a $\delta$-function
at 1 ($\sigma=1$). (d) Capacity $\alpha_{c}(\kappa=0)$ as a function
of $\sigma$. Inset shows the $W_{2}$ distance as a function of $\sigma$.}
\end{figure}

\subsection{Geometrical interpretation of capacity}

\label{subsec:Geometrical-interpretation}

The jacobian factor Eqn.\ref{eq:jacobian} can be rewritten as

\begin{equation}
\left\langle \frac{dw}{dx}\right\rangle _{x}=1-\frac{1}{2}W_{2}(Q,P)^{2},\label{eq:W2_dist}
\end{equation}
where $W_{k}$ ($k=2$ in above) is the Wasserstein-$k$ distance,
given by

\begin{equation}
W_{k}(Q,P)=\left[\int_{0}^{1}du\left(Q^{-1}(u)-P^{-1}(u)\right)^{k}\right]^{1/k}.\label{eq:Wk_dist}
\end{equation}

{[}In the following, we will make frequent use of both the probability
density function (PDF), and the cumulative distribution function (CDF).
We distinguish them by using upper case letters for CDFs, and lower
case letters for PDFs.{]} 

The Wasserstein distance measures the dissimilarity between two probability
distributions, and is the geodesic distance between points on the
manifold of probability distributions \cite{lott2006some,figalli2011optimal,chen2020optimal}.
Therefore, we can interpret Eqn.\ref{eq:single_pop} as predicting
that the reduction in memory capacity tracks the geodesic distance
we need travel from the standard normal distribution $P$ to the target
distribution $Q$ (Fig.\ref{fig:theory_illustration}(a)).

We demonstrate Eqn.\ref{eq:single_pop} and Eqn.\ref{eq:W2_dist}
with an instructive example. Let's consider a parameterized family
of normal distributions, with the second moment fixed to 1: $q(w)=\mathcal{N}(\sqrt{1-\sigma^{2}},\sigma^{2})$,
see Fig.\ref{fig:theory_illustration}(c). At $\sigma=1$, $q(w)$
is the standard normal distribution and we recover the unconstrained
Gardner capacity $\alpha_{0}(\kappa=0)=2$. As $\sigma\to0,$ $q(w)$
becomes a $\delta$-function at $1$ and $\alpha_{c}(\kappa)\to0$
(Fig.\ref{fig:theory_illustration}(c)). 

As evident in this simple example, perceptron capacity is strongly
affected by its weight distribution. Our theory enables prediction
of the shape of the distribution with optimal parameters within a
parameterized family of distributions. We apply our theory to a family
of biologically plausible distributions and compare our prediction
with experimentally measured distributions in Section \ref{sec:compare_experiment}.

\section{Optimal transport and the DisCo-SGD learning algorithm}

\label{sec:disco_algorithm}

Eqn.\ref{eq:single_pop} predicts the storage capacity for a perceptron
with a given weight distribution, but it does not specify a learning
algorithm for finding a solution to this non-convex learning problem.
Here we present a learning algorithm for perceptron learning with
a given weight distribution constraint. This algorithm will also serve
to test our theoretical predictions. For this purpose, we use optimal
transport theory to develop an SGD-based algorithm that is able to
find max-margin solutions that obey the prescribed distribution constraint.
Furthermore, we show that training can be interpreted as traveling
along the geodesic connecting the current empirical distribution and
the target distribution. 

Stochastic gradient descent (SGD) on a cross-entropy loss has been
shown to asymptotically converge to max-margin solutions on separable
data \cite{soudry2018implicit,nacson2019stochastic}. Given data $\left\{ \boldsymbol{\xi}^{\mu},\zeta^{\mu}\right\} _{\mu=1}^{P}$,
we use logistic regression to predict class labels from our perceptron
weights, $\hat{\zeta}^{\mu}=\sigma(\boldsymbol{w}^{t}\cdot\boldsymbol{\xi}^{\mu}),$where
$\sigma(z)=\left(1+e^{-z}\right)^{-1}$ and $\boldsymbol{w}^{t}$
is the weight at the $t$-th update. This defines an SGD update rule
:

\begin{equation}
w_{i}^{t+\delta t}\leftarrow w_{i}^{t}-\delta t\sum_{\mu}\xi_{i}^{\mu}(\hat{\zeta}^{\mu}-\zeta^{\mu}),\label{eq:SGD}
\end{equation}

where the $\mu$-summation goes from $1$ to $P$ for full-batch GD
and goes from $1$ to mini-batch size $B$ for mini-batches SGD (see
Appendix \ref{app:DiscoSGD} for more details). The theory of optimal
transport provides a principled way of transporting each individual
weight $w_{i}^{t}$ to a new value so that overall the new set of
weights satisfies the prescribed target distribution. In $1$-D, the
optimal transport plan $T^{*}$ has a closed-form solution in terms
of the current CDF $P$ and target CDF $Q$ \cite{thorpe2019introduction,ambrosio2013user}:
$T^{*}=Q^{-1}\circ P$, where $\circ$ denotes functional composition.
We demonstrate the optimal transport map in Fig.\ref{fig:theory_illustration}(b)
for the instructive example discussed in Section \ref{subsec:Geometrical-interpretation}. 

\begin{table}
\centering{}\includegraphics[scale=0.23]{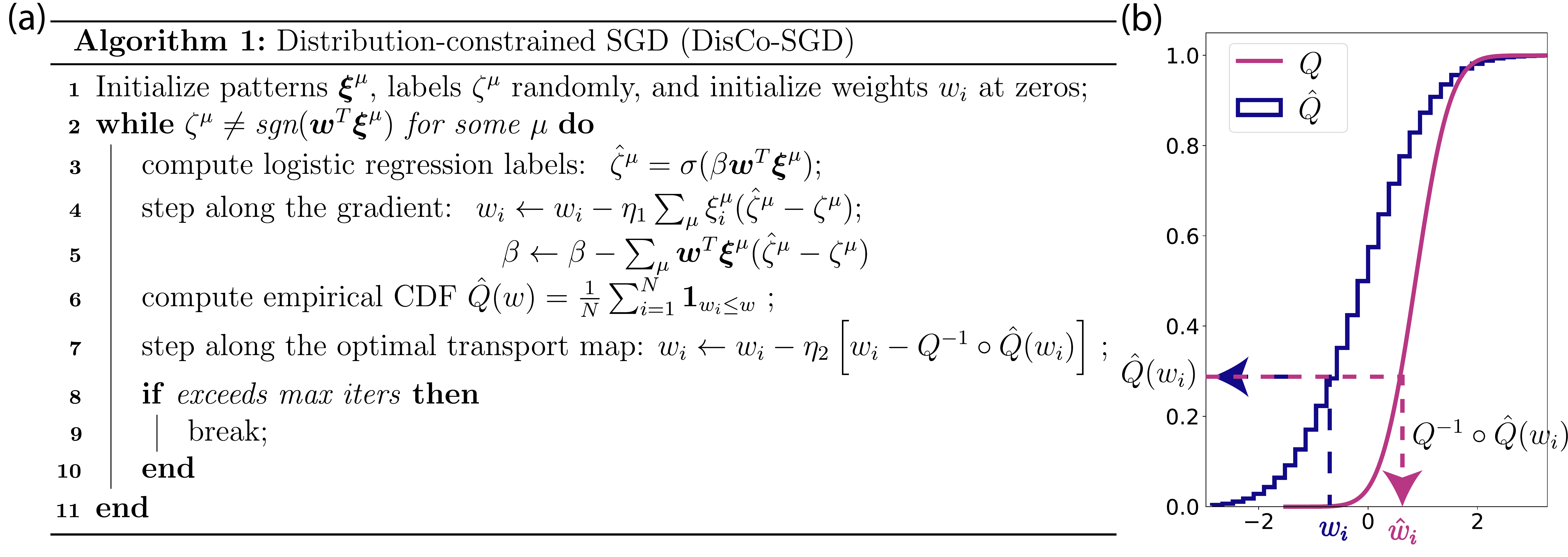}\caption{\label{tab:discoSGD}Disco-SGD algorithm. (a) We perform alternating
steps of gradient descent along the cross-entropy loss (Eqn.\ref{eq:SGD}),
followed by steps along the optimal transport direction (Eqn.\ref{eq:soft_constraint}).
(b) An illustration of Eqn.\ref{eq:w_hat}. For a given $w_{i}$,
we first compute its empirical CDF value $\hat{Q}(w_{i}),$then use
the inverse target CDF to transport $w_{i}$ to its new value, $\hat{w_{i}}=Q^{-1}\left(\hat{Q}(w_{i})\right)$.}
\end{table}

In order to apply $T^{*}$ to transport our weights $\left\{ w_{i}\right\} $
(omitting superscript $t$), we form the empirical CDF $\hat{Q}(w)=\frac{1}{N}\sum_{i=1}^{N}\mathbf{1}_{w_{i}\leq w}$,
which counts how many weights $w_{i}$ are observed below value $w$.
Then the new set of weights $\left\{ \hat{w}_{i}\right\} $ satisfying
target CDF $Q$ can be written as 

\begin{equation}
\hat{w}_{i}=Q^{-1}\circ\hat{Q}(w_{i}).\label{eq:w_hat}
\end{equation}

We illustrate Eqn.\ref{eq:w_hat} in action in Table \ref{tab:discoSGD}(b).

\begin{figure}
\centering{}\includegraphics[scale=0.4]{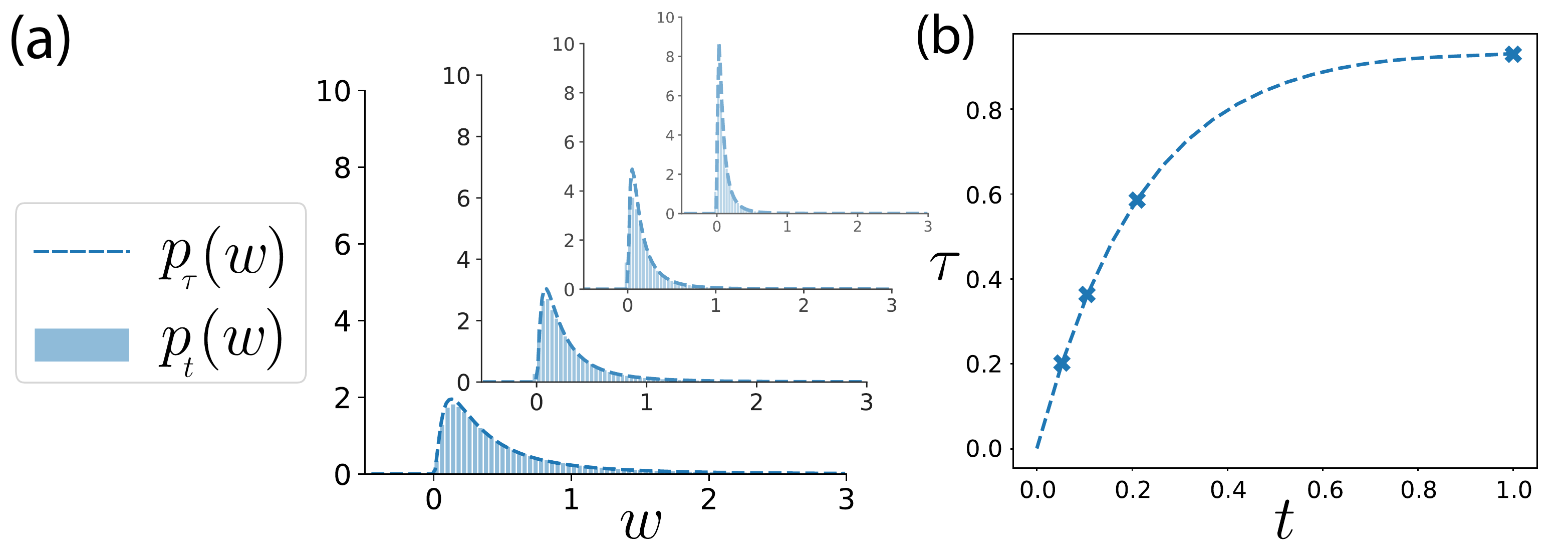}\caption{\label{fig:geodesic}Intermediate distributions during learning are
on the geodesic. (a) The solid histograms are the intermediate distribution
$p_{t}$ at different training time $t$ from the DisCo-SGD algorithm,
the dashed lines are geodesic distributions $p_{\tau}$ with the same
$W_{2}$ distance to the target distribution $Q$. From right to left
the training time advances, and the distributions transform further
away from the $\delta$-function initialization, and approach the
target distribution (a lognormal, in this example). (b) The geodesic
time $\tau$ as a function of the training time $t$. Location of
the crosses correspond to the distributions shown in (a).}
\end{figure}

However, performing such a one-step projection strongly interferes
with the cross-entropy objective, and numerically often results in
solutions that do not perfectly classify the data. Therefore, it would
be beneficial to have an incremental update rule based on Eqn.\ref{eq:w_hat}:

\begin{equation}
w_{i}^{\tau+\delta\tau}\leftarrow w_{i}^{\tau}+\delta\tau\left(\hat{w}_{i}-w_{i}^{\tau}\right),\label{eq:soft_constraint}
\end{equation}

where we have used a different update time $\tau$ to differentiate
with the cross-entropy update time $t$. 

We present our complete algorithm in Table \ref{tab:discoSGD}(a),
which we named `Distribution-constrained SGD' (DisCo-SGD) algorithm.
In the DisCo-SGD algorithm, we perform alternating updates on Eqn.\ref{eq:SGD}
and Eqn.\ref{eq:soft_constraint}, and identify $\delta t$ and $\delta\tau$
as learning rates $\eta_{1}$ and $\eta_{2}$. Note that in logistic
regression, the norm of the weight vector $||\boldsymbol{w}||$ is
known to increase with training and the max-margin solution is only
recovered at $||\boldsymbol{w}||\to\infty$. In contrast, imposing
a distribution constraint fixes the norm. Therefore, to allow a variable
norm, in Table \ref{tab:discoSGD} we include a trainable parameter
$\beta$ in our algorithm to serve as the norm of the weight vector.
This algorithm allows us to reliably discover linearly separable solutions
obeying the prescribed weight distribution $Q$.

Interestingly, Eqn.\ref{eq:soft_constraint} takes a similar form
to geodesic flows in Wasserstein space. Given samples $\left\{ w_{i}\right\} $
drawn from the initial distribution $P$ and $\left\{ \hat{w}_{i}\right\} $
drawn from the final distribution $Q$, samples $\left\{ w_{i}^{\tau}\right\} $
from intermediate distributions $P_{\tau}$ along the geodesic can
be calculated as $w_{(i)}^{\tau}=(1-\tau)w_{(i)}+\tau\hat{w}_{(i)}$,
where subscript $(i)$ denotes ascending order (see more in Appendix
\ref{app:Optimal-transport-theory}). For intermediate perceptron
weights $\boldsymbol{w}^{t}$ found by our algorithm, we can compute
its empirical distribution $p_{t}$ and compare with theoretical distribution
$p_{\tau}$ along the geodesic with the same $W_{2}$ distance to
the target distribution (see Appendix \ref{app:Optimal-transport-theory}
for how to calculate $p_{\tau}$). In Fig.\ref{fig:geodesic}(a),
we show that indeed the empirical distributions $p_{t}$ agree with
the geodesic distributions $p_{\tau}$ at geodesic time $\tau(t)$
(Fig.\ref{fig:geodesic}(a)). The relation between the geodesic time
$\tau$ and the SGD update time $t$ is shown in Fig.\ref{fig:geodesic}(b).
The interplay between the cross-entropy objective and the distribution
constraint is manifested in the rate at which the distribution moves
along the geodesic between the initial distribution and the target
one. 

\begin{figure}
\centering{}\includegraphics[scale=0.26]{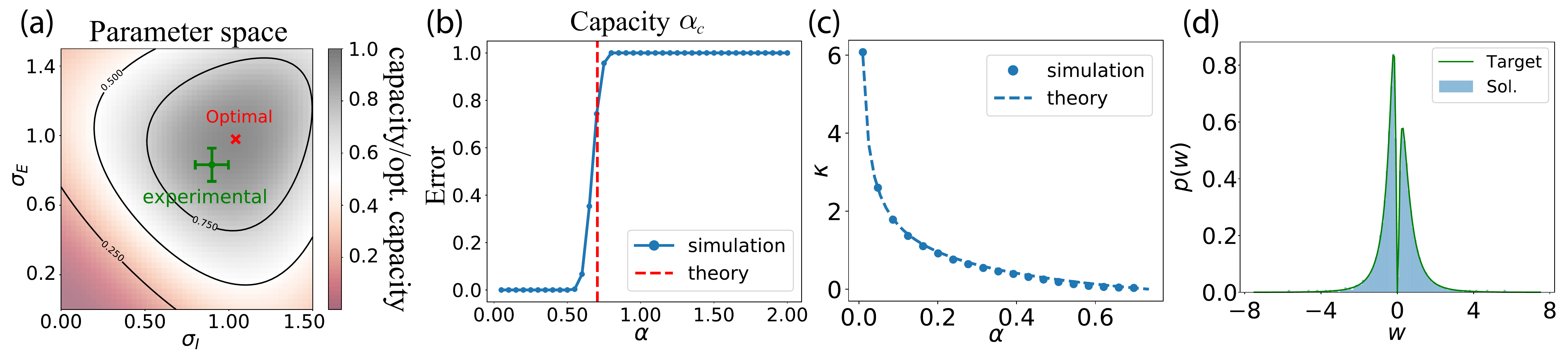}\caption{\label{fig:Experimental-landscape}Biologically-realistic distribution
and parameter landscape. (a) Capacity (normalized by the optimal value
in the landscape) as a function of the lognormal parameters $\sigma_{E}$
and $\sigma_{I}$. Experimental value is shown in green with error
bars, and optimal capacity is shown in red. (b)-(d) (theory from Eqn.\ref{eq:capacity_M_pop}
and simulations from DisCo-SGD): (b) Determination of capacity; (c)
Max-margin $\kappa$ at different load $\alpha$, which is the same
as $\alpha_{c}(\kappa)$; (d) Example weight distribution obtained
in simulation.}
\end{figure}

\section{Biologically-realistic distribution (E/I balanced lognormals) and
experimental landscape}

\label{sec:compare_experiment}

In order to apply our theory to the more biologically-realistic cases,
we generalize our theory from a single prescribed distribution to
an arbitrary number of input subpopulations each obeys its own distribution.
We consider a perceptron that consists of $M$ synaptic populations
$\boldsymbol{w}^{m}$ indexed by $m$, each constrained to satisfy
its own weight distribution $w_{i}^{m}\sim q_{m}(w^{m})$. We denote
the overall weight vector as $\boldsymbol{w}\equiv\{\boldsymbol{w}^{m}\}_{m=1}^{M}\in\mathbb{R}^{N\times1}$,
where the total number of weights is $N=\sum_{m=1}^{M}N_{m}$. In
this case, the capacity Eqn.\ref{eq:single_pop} is generalized to
(See Appendix \ref{app:capacity_M_pop} for detailed derivation):

\begin{equation}
\alpha_{c}(\kappa)=\alpha_{0}(\kappa)\left[\sum_{m}^{M}g_{m}\left\langle \frac{dw^{m}}{dx}\right\rangle _{x}\right]^{2},\label{eq:capacity_M_pop}
\end{equation}

where $g_{m}=N_{m}/N$ is the fraction of weights in this population.
Eqn. \ref{eq:capacity_M_pop} allows us to investigate the parameter
space of capacity with biologically-realistic distributions and compare
with the experimentally measured values. In particular, we are interested
the case with two synaptic populations that models the excitatory/inhibitory
synpatic weights of a biological neuron, hence, $m=E,I$. We model
the excitatory/inhibitory synaptic weights as drawn from two separate
lognormal distributions $(g_{I}=1-g_{E}$): $w_{i}^{E}\sim\text{\ensuremath{\frac{1}{\sqrt{2\pi}\sigma_{E}w^{E}}\exp\left\{ -\frac{(\ln w^{E}-\mu_{E})^{2}}{2\sigma_{E}^{2}}\right\} }}$
and $w_{i}^{I}\sim\text{\ensuremath{\frac{1}{\sqrt{2\pi}\sigma_{I}w^{I}}\exp\left\{ -\frac{(\ln w^{I}-\mu_{I})^{2}}{2\sigma_{I}^{2}}\right\} }}.$

We also demand that the mean synaptic weights satisfy the E/I balance
condition \cite{van1996chaos,van1998chaotic,tsodyks1995rapid,van2005irregular,rubin2017balanced,mongillo2018inhibitory,chapeton2012efficient}
$g_{E}\left\langle w^{E}\right\rangle =g_{I}\left\langle w^{I}\right\rangle $
as is often observed in cortex connectomic experiments \cite{anderson2000orientation,wehr2003balanced,okun2008instantaneous,poo2009odor,atallah2009instantaneous}.
With the E/I balance condition and fixed second moment, the capacity
is a function of the lognormal parameters $\sigma_{E}$ and $\sigma_{I}$.
In Fig.\ref{fig:Experimental-landscape}(a) we map out the 2d parameter
space of $\sigma_{E}$ and $\sigma_{I}$ using Eqn.\ref{eq:capacity_M_pop},
and find that the optimal choice of parameters which yields the maximum
capacity solution is close to the experimentally measured values in
a recent connectomic studies in mouse primary auditory cortex \cite{levy2012spatial}. 

In order to test our theory's validity on this estimated distribution
of synaptic weights, we perform DisCo-SGD simulation with model parameters
$\sigma_{E}$ and $\sigma_{I}$ fixed to their experimentally measured
values. Both the capacity (Fig.\ref{fig:Experimental-landscape}(b)),
max-margin $\kappa$ at different load (Fig.\ref{fig:Experimental-landscape}(c)),
and the empirical weights found by the algorithm (Fig.\ref{fig:Experimental-landscape}(d))
are in good agreement with our theoretical prediction.

\section{Generalization performance}

\label{sec:Generalization}

\subsection{Distribution-constrained learning as circuit inference}

A central question in computational neuroscience is how underlying
neural circuits determine its computation. Recently, thanks to new
parallelized functional recording technologies, simultaneous recordings
of the activity of hundreds of neurons in response to an ensemble
of inputs are possible \cite{ahrens2013whole,berenyi2014large}. An
interesting challenge is to infer the structural connectivity from
the measured input-output activity patterns. It is interesting to
ask how are these stimuli-response relations related to the underlying
structure of the circuit \cite{real2017neural,liu2017inference}.
In the following, we try to adress this circuit reconstruction task
in a simple setup where a student perceptron tries to learn from a
teacher perceptron \cite{seung1992statistical,engel2001statistical}.
In this setup, the teacher is considered to be the underlying ground-truth
neural circuit. The student is attempting to infer the connection
weights of this ground-truth circuit by observing a series of input-output
relations generated by the teacher. After learning is completed, one
can assess the faithfulness of the inference by comparing the teacher
and student. The teacher-student setup is also a well-known ‘toy model’
for studying generalization performance\textcolor{red}{{} }\cite{loureiro2021learning,lee2021continual,matiisen2019teacher}.
In this case since the learning data are generated by the teacher,
the overlap between teacher and student determines the generalization
performance of the learning. Here we ask to what extent prior knowledge
of the teacher weight distribution helps in learning the rule and
how this knowledge can be incorporated in learning. A similar motivation
may arise in other contexts, in which there is a prior knowledge about
the weight distribution of an unknown target linear classifier. 

Let's consider the teacher perceptron, $\boldsymbol{w}_{t}\in\mathbb{R}^{N}$,
drawn from some ground-truth distribution $p_{t}$. Given random inputs
$\boldsymbol{\xi}^{\mu}$ with $p(\xi_{i}^{\mu})=\mathcal{N}(0,1)$,
we generate labels by $\zeta^{\mu}=\text{sgn}(\boldsymbol{w}_{t}\cdot\boldsymbol{\xi}^{\mu}/||\boldsymbol{w}_{t}||+\eta^{\mu})$,
where $\eta^{\mu}$ is input noise and $\eta^{\mu}\sim\mathcal{N}(0,\sigma^{2})$.
We task the student perceptron $\boldsymbol{w}_{s}$ to find the max-margin
linear classifier for data $\{\boldsymbol{\xi}^{\mu},\zeta^{\mu}\}_{\mu=1}^{p}$:
$\max\kappa:\zeta^{\mu}\boldsymbol{w}_{s}\cdot\boldsymbol{\xi}^{\mu}\geq\kappa||\boldsymbol{w}_{s}||$.
Let's define the teacher-student overlap as

\begin{equation}
R=\frac{\boldsymbol{w}_{s}\cdot\boldsymbol{w}_{t}}{\left\Vert \boldsymbol{w}_{s}\right\Vert \left\Vert \boldsymbol{w}_{t}\right\Vert },\label{eq:overlap}
\end{equation}

which is a measure the faithfulness of the circuit inference. The
student's generalization error is then related to the overlap by $\varepsilon_{g}=1/\pi\arccos\left(R/\sqrt{1+\sigma^{2}}\right)$
\cite{seung1992statistical,engel2001statistical}.

As a baseline, let's first consider a totally uninformed student (without
any structural knowledge of the teacher), learning from a teacher
with a given (in particular non-Gaussian) weight distribution. In
this case, we can determine the overlap $R$ (Eqn.\ref{eq:overlap})
as a function of load $\alpha$ by solving the replica symmetric mean
field self-consistency equations as in \cite{seung1992statistical,engel2001statistical}.
An example of such learning for a lognormal teacher distribution is
shown in Fig.\ref{fig:compare_learnings}(a) (`unconstrained') for
the noiseless case ($\sigma=0$). Note that in the presence of noise
in the labels $(\sigma\neq0)$, $\alpha$ is bounded by $\alpha_{c}(\sigma)$
, since max-margin learning of separable data is assumed. The case
with nonzero $\sigma$ is presented in Appendix \ref{app:generalization_noise}.
In this unconstrained case, the student's weight distribution evolves
from a Gaussian for low $\alpha$ to one which increasingly resembles
the teacher distribution for large $\alpha$ (Fig.\ref{fig:compare_learnings}(b)). 

Next, we consider a student with information about the signs of the
individual teacher weights. We can apply this knowledge as a constraint
and demand that the signs of individual student weights agree with
that of the teacher's. The additional sign-constraints require a modification
of replica calculation in \cite{seung1992statistical,engel2001statistical},
which we present in Appendix \ref{app:generalization_sign_const}.
Surprisingly, we find both analytically and numerically that if the
teacher weights are not too sparse, the max-margin solution generalizes
poorly: after a single step of learning (with random input vectors),
the overlap, $R$, drops substantially from its initial value (see
`sign-constrained' in Fig.\ref{fig:compare_learnings}(a)). The source
of the problem is that, due to the sign constraint, max-margin training
with few examples yields a significant mismatch between the student
and teacher weight distributions. After only a few steps of learning,
half of the student's weights are set to zero, and the student's distribution,
$p(w_{s})=1/2\delta(0)+1/\sqrt{2\pi}\exp\{-w_{s}^{2}/4\}$, deviates
significantly from the teacher's distribution (see more in Appendix
\ref{app:generalization_sparsification}). The discrepancy between
the teacher and student weight distributions therefore suggest that
we should incorporate distribution-constraint into learning.

\begin{figure}
\centering{}\includegraphics[scale=0.24]{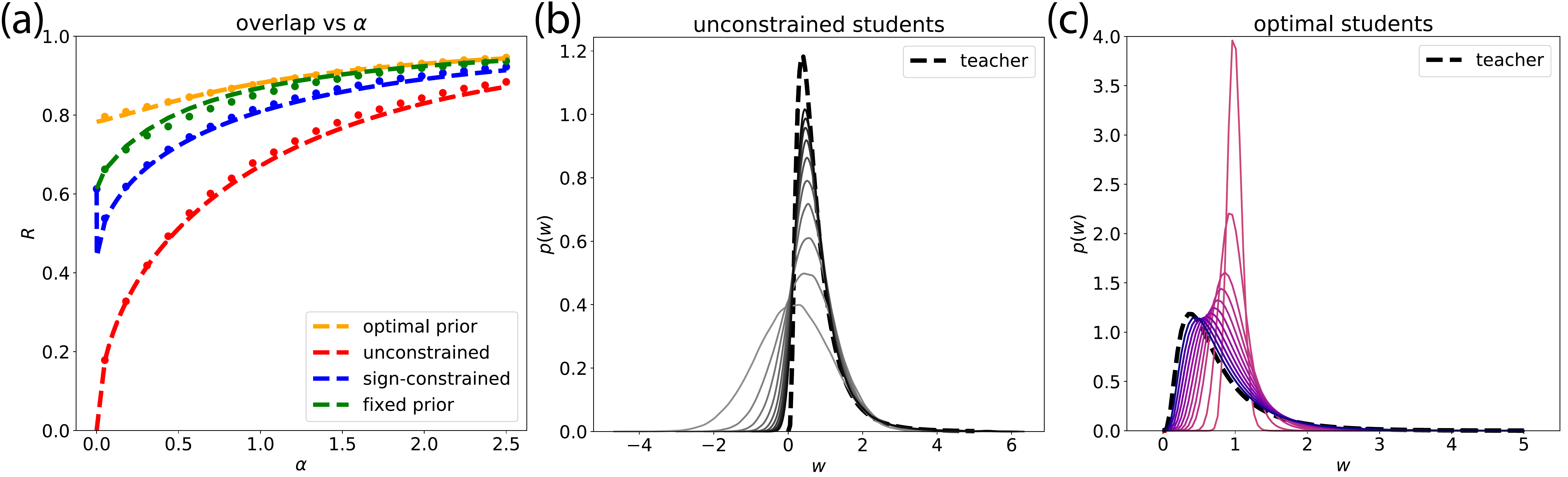}\caption{\label{fig:compare_learnings}Compare different learning paradigms.
(a) Teacher-student overlap $R$ , or equivalently the generalization
error $\varepsilon_{g}=1/\pi\arccos R$, as a function of load $\alpha$
in different learning paradigms. Dashed lines are from theory, and
dots are from simulation. Note that there is an initial drop of the
overlap in sign-constrained learning due to sparsification of weights.
(b)-(c) The darker color curves correspond to larger $\alpha$, and
dashed line is teacher distribution (same in both cases). (b) Distribution
of an unconstrained student evolves from normal distribution toward
the teacher distribution. (c) Optimal student prior evolves from a
$\delta$-function toward the teacher distribution.}
\end{figure}

\subsection{Distribution-constrained learning outperforms unconstrained and sign-constrained
learning }

Let's consider the case that the student weight are constrained to
some \textit{prior} distribution $q_{s}(w_{s})$, while the teacher
obeys a distribution $p_{t}(w_{t}),$for an arbitrary pair $q_{s},p_{t}$.
We can write down the Gardner volume $V_{g}$ for generalization as
in the capacity case (Eqn.\ref{eq:volume}):

\begin{equation}
V_{g}=\frac{\int d\boldsymbol{w}_{s}\left[\prod_{\mu=1}^{P}\Theta\left(\text{sgn}\left(\frac{\boldsymbol{w}_{t}\cdot\boldsymbol{\xi}^{\mu}}{||\boldsymbol{w}_{t}||}+\eta^{\mu}\right)\frac{\boldsymbol{w}_{s}\cdot\boldsymbol{\xi}^{\mu}}{||\boldsymbol{w}_{s}||}-\kappa\right)\right]\delta(||\boldsymbol{w}_{s}||^{2}-N)\delta\bigg(\int dk\left(\hat{q}(k)-q(k)\right)\bigg)}{\int d\boldsymbol{w}_{s}\delta(||\boldsymbol{w}_{s}||^{2}-N)}.\label{eq:volume-generalization}
\end{equation}

To obtain ensemble average of system over different realizations of
the training set, we perform the quenched average of $\log V_{g}$
over the patterns $\boldsymbol{\xi}^{\mu}$ and teacher $\boldsymbol{w}_{t}$,
and consider the thermodynamic limit of $N,P\to\infty$ and $\alpha=\frac{P}{N}$
stays $\mathcal{O}(1)$. We use the replica trick similar to \cite{seung1992statistical,engel2001statistical}.
Overlap $R$ (Eqn.\ref{eq:overlap}) can be determined as a function
of load $\alpha$ by solving the replica symmetric mean field self-consistency
equations in Appendix \ref{app:generalization_dist_const}. In this
distribution-constrained setting, we can perform numerical simulations
with DisCo-SGD algorithm (Table \ref{tab:discoSGD}) to find such
weights and compare with the predictions of our theory. 

Now we ask if the student has a \textit{prior} on the teacher's weight
distribution $p_{t}$, whether incorporating this knowledge in training
will improve generalization performance. One might be tempted to conclude
that the optimal prior distribution the student should adopt is always
that of the teacher's, i.e., $q_{s}=p_{t}$. We call this learning
paradigm `fixed prior', and show that its generalization performance
is better than that of the unconstrained and sign-constrained case
(Fig.\ref{fig:compare_learnings}(a)). However, instead of using a
fixed prior for the student, we can in fact choose the \textit{optimal
prior} distribution $p_{s}^{*}$ at different load $\alpha$. This
presents a new learning paradigm we called `optimal prior'. In Fig.\ref{fig:compare_learnings}(a),
we show that choosing optimal priors at different $\alpha$ achieves
the overall best generalization performance compared with all other
learning paradigms. For a given parameterized family of distributions,
our theory provides a way to analytically obtain the optimal prior
$p_{s}^{*}$ as a function of $\alpha$ (Fig.\ref{fig:compare_learnings}(c)).
Note that unlike the unconstrained case (Fig.\ref{fig:compare_learnings}(b)),
the optimal prior starts from a $\delta$-function at $1$ at zero
$\alpha$, and asymptotically approaches the teacher distribution
$p_{t}$ as $\alpha\to\infty$.

\section{Summary and Discussion}

We have developed a statistical mechanical framework that incorporates
structural constraints (sign and weight distribution) into perceptron
learning. The synaptic weights in our perceptron learning satisfy
two key biological constraints: (1) individual synaptic signs are
not affected by the learning task (2) overall synaptic weights obey
a prescribed distribution. These constraints may arise also in neuromorphic
devices \cite{han2021alternative,truong2014new}. Under the replica-symmetry
assumption, we derived a novel form of distribution-constrained perceptron
storage capacity, which admits a simple geometric interpretation of
the reduction in capacity in terms of the Wasserstein distance between
the standard normal distribution and the imposed distribution. To
numerically test our analytic theory, we used tools from optimal transport
and information geometry to develop an SGD-based algorithm, DisCo-SGD,
in order to reliably find weights that satisfy such prescribed constraints
and correctly classify the data, and showed that training with the
algorithm can be interpreted as geodesic flows in the Wasserstein
space of distributions. It would be interesting to compare our theory
and algorithm to \cite{arjovsky2017wasserstein,sanjabi2018convergence}
where the Wasserstein distance is used as an objective for training
generative models. We applied our theory to the biologically realistic
case of of excitatory/inhibitory lognormal distributions that are
observed in the cortex, and found experimentally-measured parameters
close to the optimal parameter values predicted by our theory. We
further studied input-output rule learning where the target rule is
defined in terms of a weighted sum of the inputs, and asked to what
extent prior knowledge of the target distribution may improve generalization
performance. Using the teacher-student perceptron learning setup,
we showed analytically and numerically that distribution constrained
learning substantially enhances the generalization performance. In
the context of circuit inference, distribution constrained learning
provides a novel and reliable way to recover the underlying circuit
structure from observed input-output neural activities. In summary,
our work provides new strategies of incorporating knowledge about
weight distribution in neural learning and reveals a powerful connection
between structure and function in neural networks. Ongoing extensions
of the present work include weight distribution constraints in recurrent
and deep architectures as well as testing against additional connectomic
databases. 

\newpage

\section{Appendix}

\subsubsection*{Preliminaries}

Throughout the appendix, we make frequent use of Gaussian integrals.
We introduce short-hand notations $\int Dt\equiv\int\frac{dt}{\sqrt{2\pi}}e^{-t^{2}/2}$
and $H(x)\equiv\int_{x}^{\infty}Dt$. Also, when we do not specify
the integration range it is understood that we are integrating from
$-\infty$ to $\infty$. 

\subsection{Capacity supplemental materials}

\label{app:capacity}

\subsubsection{Replica calculation of distribution-constrained capacity}

\label{app:capacity_dist_const}

In this section, we present the replica calculation of the distribution-constrained
storage capacity of a perceptron. 

As described in main text Eqn.2, we need to perform a quenched average
$\left\langle \cdot\right\rangle $ over the patterns $\boldsymbol{\xi}^{\mu}$
and labels $\zeta^{\mu}$ for $\log V$, which can be carried out
using the replica trick, $\left\langle \log V\right\rangle =\lim_{n\to0}(\left\langle V^{n}\right\rangle -1)/n$.
Following \cite{gardner1988optimal,gardner1988space}, we consider
first integer $n$, and at the end perform analytic continuation of
$n\to0$. The replicated Gardner volume is:

\begin{equation}
V=\frac{\prod_{\alpha=1}^{n}\int d\boldsymbol{w}^{\alpha}\left[\prod_{\mu=1}^{P}\Theta\left(\zeta^{\mu}\frac{\boldsymbol{w}^{\alpha}\cdot\boldsymbol{\xi}^{\mu}}{||\boldsymbol{w}^{\alpha}||}-\kappa\right)\right]\delta(||\boldsymbol{w}^{\alpha}||^{2}-N)\delta\left(\int dk\bigg(\hat{q}(k)-q(k)\bigg)\right)}{\prod_{\alpha=1}^{n}\int d\boldsymbol{w}^{\alpha}\delta(||\boldsymbol{w}^{\alpha}||^{2}-N)}\label{appeq:volume_cap_dist}
\end{equation}

Let's rewrite the Heaviside step function using Fourier representation
of the $\delta$-function $\delta(x)=\int_{-\infty}^{\infty}\frac{dk}{2\pi}e^{ikx}$
as (defining $z_{\alpha}^{\mu}=\zeta^{\mu}\frac{\boldsymbol{w}^{\alpha}\cdot\boldsymbol{\xi}^{\mu}}{||\boldsymbol{w}^{\alpha}||}$)

\begin{equation}
\Theta\left(z_{\alpha}^{\mu}-\kappa\right)=\int_{\kappa}^{\infty}d\rho_{\alpha}^{\mu}\delta(\rho_{\alpha}^{\mu}-z_{\alpha}^{\mu})=\int_{\kappa}^{\infty}d\rho_{\alpha}^{\mu}\int\frac{dx_{\alpha}^{\mu}}{2\pi}e^{ix_{\alpha}^{\mu}(\rho_{\alpha}^{\mu}-z_{\alpha}^{\mu})}.
\end{equation}

Note that now all the $\boldsymbol{\xi}^{\mu},\zeta^{\mu}$ dependence
is in $e^{-ix_{\alpha}^{\mu}z_{\alpha}^{\mu}}$. We perform the average
with respect to $\xi_{i}^{\mu}\sim p(\xi_{i}^{\mu})=\mathcal{N}(0,1)$
and $p(\zeta^{\mu})=\frac{1}{2}\delta(\zeta^{\mu}+1)+\frac{1}{2}\delta(\zeta^{\mu}-1)$
(also note that $||\boldsymbol{w}^{\alpha}||=\sqrt{N}$):

\begin{equation}
\begin{split}\left\langle \prod_{\mu\alpha}e^{-ix_{\alpha}^{\mu}z_{\alpha}^{\mu}}\right\rangle _{\xi\eta} & =\prod_{\mu j}\left\langle \exp\left\{ -\frac{i}{\sqrt{N}}\zeta^{\mu}\xi_{j}^{\mu}\sum_{\alpha}x_{\alpha}^{\mu}w_{j}^{\alpha}\right\} \right\rangle _{\xi\zeta}\\
 & =\prod_{\mu i}\left\langle \exp\left\{ -\frac{(\zeta^{\mu})^{2}}{2N}\sum_{\alpha\beta}x_{\alpha}^{\mu}x_{\beta}^{\mu}w_{i}^{\alpha}w_{i}^{\beta}\right\} \right\rangle _{\zeta}\\
 & =\prod_{\mu}\exp\left\{ -\frac{1}{2N}\sum_{\alpha\beta}x_{\alpha}^{\mu}x_{\beta}^{\mu}\sum_{i}w_{i}^{\alpha}w_{i}^{\beta}\right\} .
\end{split}
\label{appeq:cap_dist_energy_x_z}
\end{equation}

Introducing the replica overlap parameter $q_{\alpha\beta}=\frac{1}{N}\sum_{i}w_{i}^{\alpha}w_{i}^{\beta}$,
and notice that the $\mu$ index gives $P$ identical copies of the
same integral. We can suppress the $\mu$ indices and write

\begin{equation}
\begin{split}\left\langle \prod_{\mu\alpha}\Theta\left(z_{\alpha}^{\mu}-\kappa\right)\right\rangle _{\xi\zeta} & =\left[\int_{\kappa}^{\infty}\left(\prod_{\alpha}\frac{d\rho_{\alpha}dx_{\alpha}}{2\pi}\right)e^{K}\right]^{P}\end{split}
,
\end{equation}

where

\begin{equation}
K=i\sum_{\alpha}x_{\alpha}\rho_{\alpha}-\frac{1}{2}\sum_{\alpha\beta}q_{\alpha\beta}x_{\alpha}x_{\beta}\label{appeq:cap_dist_K_replica}
\end{equation}

captures all the data dependence in the quenched free energy landscape,
and therefore it is called the `energetic' part of the free energy.
In contrast, the $\delta$-functions in Eqn.\ref{appeq:volume_cap_dist}
are called `entropic' part because they regulate what kind of weights
are considered in the version space (space of viable weights).

\subsubsection*{The entropic part}

\begin{equation}
\begin{split}\delta(Nq_{\alpha\beta}-\sum_{i}w_{i}^{\alpha}w_{i}^{\beta}) & =\int\frac{d\hat{q}_{\alpha\beta}}{2\pi}\exp\left\{ iN\hat{q}_{\alpha\beta}q_{\alpha\beta}-i\hat{q}_{\alpha\beta}\sum_{i}w_{i}^{\alpha}w_{i}^{\beta}\right\} \end{split}
.
\end{equation}

Note that the normalization constraint $\delta(||\boldsymbol{w}^{\alpha}||^{2}-N)$
is automatically satisfied by requiring $q_{\alpha\alpha}=1$. Using
replica-symmetric ansatz: $\hat{q}_{\alpha\beta}=-\frac{i}{2}(\Delta\hat{q}\delta_{\alpha\beta}+\hat{q}_{1})$,
and $q_{\alpha\beta}=(1-q)\delta_{\alpha\beta}+q$, we have

\begin{equation}
iN\sum_{\alpha\beta}\hat{q}_{\alpha\beta}q_{\alpha\beta}=\frac{nN}{2}\left[\Delta\hat{q}+\hat{q}_{1}(1-q)\right]+\mathcal{O}(n^{2}).
\end{equation}

and

\begin{equation}
\begin{split}-i\sum_{\alpha\beta}\hat{q}_{\alpha\beta}\sum_{i}w_{i}^{\alpha}w_{i}^{\beta} & =-\frac{1}{2}(\Delta\hat{q}+\hat{q}_{1})\sum_{\alpha}\sum_{i}(w_{i}^{\alpha})^{2}-\frac{1}{2}\hat{q}_{1}\sum_{(\alpha\beta)}\sum_{i}w_{i}^{\alpha}w_{i}^{\beta}\\
 & =-\frac{1}{2}\Delta\hat{q}\sum_{\alpha}\sum_{i}(w_{i}^{\alpha})^{2}-\frac{1}{2}\hat{q}_{1}\sum_{i}\left(\sum_{\alpha}w_{i}^{\alpha}\right)^{2}\\
 & \HSTeq-\frac{1}{2}\Delta\hat{q}\sum_{\alpha}\sum_{i}(w_{i}^{\alpha})^{2}+\sqrt{-\hat{q}_{1}}\sum_{i}t_{i}\left(\sum_{\alpha}w_{i}^{\alpha}\right),
\end{split}
\end{equation}

where in the last step HST denotes Hubbard-Stratonovich transformation
$\int\frac{dt}{\sqrt{2\pi}}e^{-t^{2}/2}e^{bt}=e^{b^{2}/2}$ that we
use to linearize the quadratic term at the cost of introducing an
auxiliary Gaussian variable $t$ to be averaged over later.

Recall that $\hat{q}(k)=\int e^{ikw}\hat{p}(w)=\frac{1}{N}\sum_{i}^{N}e^{ikw_{i}^{\alpha}}$,
the distribution constraint becomes 

\begin{equation}
\begin{split}\delta\bigg(\int dk\left(\hat{q}(k)-q(k)\right)\bigg) & =\delta\left(\int dk\left(\frac{1}{N}\sum_{i}^{N}e^{ikw_{i}^{\alpha}}-q(k)\right)\right)\\
 & =\int\frac{d\hat{\lambda}_{\alpha}(k)}{2\pi}\exp\left\{ \int dki\hat{\lambda}_{\alpha}(k)\left(\sum_{i}e^{ikw_{i}^{\alpha}}-Nq(k)\right)\right\} .
\end{split}
\label{appeq:cap_dist_dist_const_delta}
\end{equation}

Note that $\sum_{i}\int dki\hat{\lambda}_{\alpha}(k)e^{ikw_{i}^{\alpha}}=2\pi i\sum_{i}\lambda_{\alpha}(-w_{i}^{\alpha})$
by inverse Fourier transform. Next,

\begin{equation}
\begin{split}-iN\int dk\hat{\lambda}_{\alpha}(k)q(k)= & -iN\int dk\left(\int dwe^{ikw}\lambda_{\alpha}(w)\right)\left(\int dw'e^{ikw'}q(w')\right)\\
 & =-2\pi iN\int dwdw'\lambda_{\alpha}(w)q(w')\delta(w+w')\\
 & =-2\pi iN\int dwq(w)\lambda_{\alpha}(-w).
\end{split}
\end{equation}

Now we can write down the full free energy. We ignore overall constant
coefficients such as $2\pi$'s and $i$'s in the integration measure,
which become irrelevant upon taking the saddle-point approximation.
We also leave out the denominator of $V$, as it does not depend on
data and is an overall constant. Note that under the replica-symmetric
ansatz the replica index $\alpha$ gives $n$ identical copies of
the same integral and thus the replica index $\alpha$ can be suppressed
(same for synaptic index $i$):

\begin{equation}
\left\langle V^{n}\right\rangle =\int dqd\hat{\lambda}(k)d\Delta\hat{q}d\hat{q}_{1}e^{nN(G_{0}+G_{1})},\label{appeq:cap_dist_free_energy}
\end{equation}

where (please note that $q$ is replica overlap, and $q(w)$ is the
imposed target distribution)

\begin{equation}
\begin{split}G_{0} & =\frac{1}{2}\Delta\hat{q}+\frac{1}{2}\hat{q}_{1}(1-q)-2\pi i\int dwq(w)\lambda(-w)+\left\langle \log Z(t)\right\rangle _{t},\\
Z(t) & =\int dw\exp\left\{ 2\pi i\lambda(-w)-\frac{1}{2}\Delta\hat{q}w^{2}+\sqrt{-\hat{q}_{1}}tw\right\} .
\end{split}
\end{equation}

Note that integrals in Eqn.\ref{appeq:cap_dist_free_energy} can be
evaluated using saddle-point approximation in the thermodynamic limit
$N\to\infty$.

Redefining $2\pi i\lambda(-w)-\frac{1}{2}\Delta\hat{q}w^{2}\to-\lambda(w)$
and $-\hat{q}_{1}\to\hat{q}_{1}$, we have 

\begin{equation}
\begin{split}G_{0} & =\frac{1}{2}\Delta\hat{q}-\frac{1}{2}\hat{q}_{1}(1-q)+\int dwq(w)\lambda(w)-\frac{1}{2}\Delta\hat{q}\int dwq(w)w^{2}+\left\langle \log Z(t)\right\rangle _{t},\\
Z(t) & =\int dw\exp\left\{ -\lambda(w)+\sqrt{\hat{q}_{1}}tw\right\} .
\end{split}
\label{appeq:G0_cap_dist_qneq1}
\end{equation}

We seek the saddle-point solution for $G_{0}$ with respect to the
order parameters $\Delta\hat{q}$, $\lambda(w)$, and $\hat{q}_{1}$:

\begin{equation}
\begin{split}0=\frac{\partial G_{0}}{\partial\Delta\hat{q}} & \Rightarrow1=\int dwq(w)w^{2}=\left\langle w^{2}\right\rangle _{q(w)},\end{split}
\label{appeq:second_moment}
\end{equation}
\begin{equation}
\begin{split}0=\frac{\partial G_{0}}{\partial\lambda(w)} & \Rightarrow q(w)=\left\langle \frac{1}{Z(t)}\exp\left\{ -\lambda(w)+\sqrt{\hat{q}_{1}}tw\right\} \right\rangle \end{split}
.\label{appeq:saddle_pt_1_qneq1}
\end{equation}

We observe that the saddle-point equation Eqn.\ref{appeq:second_moment}
fixes the second moment of the imposed distribution $q(w)$ to 1 and
therefore can be thought of as a second moment constraint. $G_{0}$
now simplifies to

\begin{equation}
G_{0}=-\frac{1}{2}\hat{q}_{1}(1-q)+\int dwq(w)\lambda(w)+\left\langle \log Z(t)\right\rangle _{t}.
\end{equation}

The remaining $\hat{q}_{1}$ saddle-point equation is a bit more complicated,

\begin{equation}
\begin{split}0=\frac{\partial G_{0}}{\partial\hat{q}_{1}} & =-\frac{1}{2}(1-q)+\frac{t}{2\sqrt{\hat{q}_{1}}}\left\langle \frac{1}{Z(t)}\int dww\exp\left\{ -\lambda(w)+\sqrt{\hat{q}_{1}}tw\right\} \right\rangle _{t}\end{split}
\end{equation}

Integration by parts for the second term in rhs:
\begin{equation}
\begin{split}1-q= & \frac{1}{\sqrt{\hat{q}_{1}}}\int Dt\frac{1}{Z}\sqrt{\hat{q}_{1}}\int dww^{2}\exp\left\{ -\lambda(w)+\sqrt{\hat{q}_{1}}tw\right\} \\
- & \frac{1}{\sqrt{\hat{q}_{1}}}\int Dt\frac{1}{Z^{2}}\sqrt{\hat{q}_{1}}\left(\int dww\exp\left\{ -\lambda(w)+\sqrt{\hat{q}_{1}}tw\right\} \right)^{2}\\
= & \left\langle \left\langle w^{2}\right\rangle _{f(w)}\right\rangle _{t}-\left\langle \left\langle w\right\rangle _{f(w)}^{2}\right\rangle _{t},
\end{split}
\end{equation}

where in the last step we have defined an induced distribution $f(w)=Z(t)^{-1}\exp\left\{ -\lambda(w)+\sqrt{\hat{q}_{1}}tw\right\} $.
Since the second moments are fixed to 1, we have

\begin{equation}
q=\left\langle \left\langle w\right\rangle _{f(w)}^{2}\right\rangle _{t},\label{appeq:saddle_pt_2_qneq1}
\end{equation}

which gives a nice interpretation of $q$ in terms of the average
overlap of $w$ in the induced distribution $f(w)$.

\subsubsection*{Limit $q\to1$}

We are interested in the critical load $\alpha_{c}$ where the version
space (space of viable weights) shrinks to a single point, i.e., there
exists only one viable solution. Since $q$ measures the typical overlap
between weight vectors in the version space, the uniqueness of the
solution implies $q\to1$ at $\alpha_{c}$. In this limit, the order
parameters $\left\{ \hat{q}_{1},\lambda(w)\right\} $ diverges and
we need to re-derive the saddle point equations Eqn.\ref{appeq:saddle_pt_1_qneq1}
and Eqn.\ref{appeq:saddle_pt_2_qneq1} in terms of the undiverged
order parameters $\left\{ u,r(w)\right\} $:

\begin{equation}
\hat{q}_{1}=\frac{u^{2}}{(1-q)^{2}};\qquad\lambda(w)=\frac{r(w)}{1-q}.
\end{equation}

Now $G_{0}$ becomes

\begin{equation}
G_{0}=\frac{1}{1-q}\left\{ -\frac{1}{2}u^{2}+\int dwq(w)r(w)+(1-q)\left\langle \log Z(t)\right\rangle _{t}\right\} ,
\end{equation}

and 

\begin{equation}
Z(t)\,=\int dw\exp\frac{1}{1-q}\left\{ -r(w)+utw\right\} .
\end{equation}

We can perform a saddle-point approximation for the $w$ integral
in $Z(t)$ at the saddle value $w$ such that $r'(w)=ut$:

\begin{equation}
Z(t)=\exp\left\{ \frac{-r(w)+utw}{1-q}\right\} .\label{appeq:cap_dist_Zt_1}
\end{equation}

Then

\begin{equation}
G_{0}=\frac{1}{1-q}\left\{ -\frac{1}{2}u^{2}+\int dwq(w)r(w)-\left\langle r(w)\right\rangle _{t}+u\left\langle tw\right\rangle \right\} .
\end{equation}

Let's use integration by parts to rewrite

\begin{equation}
\begin{split}\int dwq(w)r(w) & =-\int Q(w)r'(w)dw\\
\left\langle r(w)\right\rangle _{t} & =\int\frac{dt}{\sqrt{2\pi}}e^{-t^{2}/2}r(w)=-\int P(t)r'(w)dw,
\end{split}
\label{appeq:cap_dist_Zt_3}
\end{equation}

where $Q(w)$ is the CDF of the imposed distribution $q(w)$ and $P(t)=\frac{1}{2}\left[1+\text{Erf}(\frac{t}{\sqrt{2}})\right]$
is the normal CDF. 

Now the saddle-point equation

\begin{equation}
0=\frac{\partial G_{0}}{\partial r'(w)}\Rightarrow Q(w)=P(t)\label{appeq:cap_dist_saddle_1}
\end{equation}

determines $w(t)$ implicitly. The $u$ equation gives

\begin{equation}
0=\frac{\partial G_{0}}{\partial u}\Rightarrow u=\left\langle tw\right\rangle _{t}=\left\langle \frac{dw}{dt}\right\rangle _{t}\label{eq:appeq:cap_dist_saddle_2}
\end{equation}

where in the last equality we have used integration by parts. Using
Eqn.\ref{appeq:cap_dist_saddle_1}-\ref{eq:appeq:cap_dist_saddle_2}
$G_{0}$ is simplified to

\begin{equation}
G_{0}=\frac{1}{2(1-q)}\left\langle \frac{dw}{dt}\right\rangle _{t}^{2}.\label{appeq:cap_dist_G0}
\end{equation}

\subsubsection*{The energetic part}

We would like to perform a similar procedure as shown above, to Eqn.\ref{appeq:cap_dist_K_replica}
using the replica-symmetric ansatz. We observe that the effect of
the distribution constraint is entirely captured in $G_{0}$ and therefore
$G_{1}$ is unchanged compared with the standard Gardner calculation
of perceptron capacity. We reproduce the calculation here for completeness.

Under the replica-symmetric ansatz $q_{\alpha\beta}=(1-q)\delta_{\alpha\beta}+q$,
Eqn.\ref{appeq:cap_dist_K_replica} becomes

\begin{equation}
\begin{split}K & =i\sum_{\alpha}x_{\alpha}\rho_{\alpha}-\frac{1-q}{2}\sum_{\alpha}x_{\alpha}^{2}-\frac{q}{2}\left(\sum_{\alpha}x_{\alpha}\right)^{2}\\
 & \HSTeq i\sum_{\alpha}x_{\alpha}\rho_{\alpha}-\frac{1-q}{2}\sum_{\alpha}x_{\alpha}^{2}-it\sqrt{q}\sum_{\alpha}x_{\alpha}.
\end{split}
\end{equation}

where we have again used the Hubbard-Stratonovich transformation to
linearize the quadratic piece. Performing the Gaussian integrals in
$x_{\alpha}$ (define $\alpha=\frac{P}{N}$),

\begin{equation}
nG_{1}=\alpha\log\left[\left\langle \int_{\kappa}^{\infty}\frac{d\rho}{\sqrt{2\pi(1-q)}}\exp\left\{ -\frac{(\rho+t\sqrt{q})^{2}}{2(1-q)}\right\} \right\rangle _{t}^{n}\right].
\end{equation}

At the limit $n\to0$,

\begin{equation}
nG_{1}=\alpha n\left\langle \log\left[\int_{\kappa}^{\infty}\frac{d\rho}{\sqrt{2\pi(1-q)}}\exp\left\{ -\frac{(\rho+t\sqrt{q})^{2}}{2(1-q)}\right\} \right]\right\rangle _{t}.
\end{equation}

Perform the Gaussian integral in $\rho$ and define $\tilde{\kappa}=\frac{\kappa+t\sqrt{q}}{\sqrt{1-q}}$,
we have
\begin{equation}
G_{1}=\alpha\int Dt\log H(\tilde{\kappa}).
\end{equation}

At the limit $q\to1,\alpha\to\alpha_{c}$, $\int_{-\infty}^{\infty}Dt$
is dominated by $\int_{-\kappa}^{\infty}Dt$, and $H(\tilde{\kappa})\to\frac{1}{\sqrt{2\pi}\tilde{\kappa}}e^{-\tilde{\kappa}^{2}/2}$.
The $\mathcal{O}\left(\frac{1}{1-q}\right)$ (leading order) contribution
gives

\begin{equation}
G_{1}=-\frac{1}{2(1-q)}\alpha_{c}\int_{-\kappa}^{\infty}Dt(\kappa+t)^{2}.\label{appeq:cap_dist_G1_1pop}
\end{equation}

Let $G=G_{0}+G_{1}$. As $n\to0$, $\left\langle V^{n}\right\rangle =e^{n\left(NG\right)}\to1+n\left(NG\right)$,
and $\left\langle \log V\right\rangle =\lim_{n\to0}\frac{\left\langle V^{n}\right\rangle -1}{n}=NG$. 

Combining with Eqn.\ref{appeq:cap_dist_G0} (relabel $t\leftrightarrow x$
to distinguish between the two auxiliary Gaussian variables), we have 

\begin{equation}
\left\langle \log V\right\rangle =\frac{N}{2(1-q)}\left[\left\langle \frac{dw}{dx}\right\rangle _{x}^{2}-\alpha_{c}\int_{-\kappa}^{\infty}Dt(\kappa+t)^{2}\right]\label{appeq:cap_dist_logV}
\end{equation}

Capacity $\alpha_{c}$ is reached when Eqn.\ref{appeq:volume_cap_dist}
goes to zero. We arrive at the distribution-constrained capacity 

\begin{equation}
\alpha_{c}(\kappa)=\alpha_{0}(\kappa)\left\langle \frac{dw}{dx}\right\rangle _{x}^{2},\label{appeq:cap_dist_formula}
\end{equation}
where $\alpha_{0}(\kappa)=\left[\int_{-\kappa}^{\infty}Dt(\kappa+t)^{2}\right]^{-1}$
is the unconstrained capacity. 

\subsubsection*{Instructive Examples}

(1) Standard normal distribution $w\sim\mathcal{N}(0,1)$. 

In this case $w=x$ and $\alpha_{c}(\kappa)=\alpha_{0}(\kappa)$.

(2) Normal distribution with nonzero mean $w\sim\mathcal{N}(\mu,\sigma^{2}).$
This is the example discussed in the main text Fig.1.

In this case $w=\mu+\sigma x$ and $\mu^{2}+\sigma^{2}=1$ due to
the second moment constraint Eqn.\ref{appeq:second_moment}. Then
$\alpha_{c}(\kappa)=\sigma^{2}\alpha_{0}(\kappa).$

(3) Lognormal distribution $w\sim\frac{1}{\sqrt{2\pi}w}\exp\left\{ -\frac{(\ln w-\mu)^{2}}{2\sigma^{2}}\right\} .$

In this case $w=e^{\mu+\sigma x}$ where $\mu=-\sigma^{2}.$ $\alpha_{c}(\kappa)=\sigma^{2}e^{-\sigma^{2}}\alpha_{0}(\kappa)$.

\subsubsection*{Geometrical interpretation}

Note that although the Jacobian factor $\left\langle \frac{dw}{dx}\right\rangle _{x}$
takes a simple form, in practice sometimes it might not be the most
convenient form to use. Integrating by parts ($p(x)=\mathcal{N}(0,1)$),

\begin{equation}
\left\langle \frac{dw}{dx}\right\rangle _{x}=\int dxp(x)wx
\end{equation}

Now define $u=P(x)$ so that $du=p(x)dx$ and $w=Q^{-1}(P(x))=Q^{-1}(u)$,
we can express the Jacobian in terms of the CDFs

\begin{equation}
\left\langle \frac{dw}{dx}\right\rangle _{x}=\int_{0}^{1}du\left(Q^{-1}(u)P^{-1}(u)\right)
\end{equation}

Furthermore, 

\begin{equation}
\begin{split}\left\langle \frac{dw}{dx}\right\rangle _{x} & =\frac{1}{2}\left[\int_{0}^{1}du\left(Q^{-1}(u)\right)^{2}+\int_{0}^{1}du\left(P^{-1}(u)\right)^{2}-\int_{0}^{1}du\left(Q^{-1}(u)-P^{-1}(u)\right)^{2}\right]\\
 & =\frac{1}{2}\left[2-W_{2}(P,Q)^{2}\right],
\end{split}
\end{equation}

where we have used second moments equal to $1$ and the definition
of the Wasserstein-$k$ distance in the second equality. Therefore,
we have arrived at the geometric interpretation of the Jacobian term

\begin{equation}
\left\langle \frac{dw}{dx}\right\rangle _{x}=1-\frac{1}{2}W_{2}(P,Q)^{2}.
\end{equation}

\subsubsection{Theory for an arbitrary number of synaptic subpopulations}

\label{app:capacity_M_pop}

In this section, we generalize our theory in the above section to
the set up of a perceptron with $M$ synaptic populations indexed
by $m$, $\boldsymbol{w}^{m}$, such that each $w_{i}^{m}$ satisfies
its own distributions constraints $w_{i}^{m}\sim q_{m}(w^{m})$. We
denote the overall weight vector as $\boldsymbol{w}\equiv\{\boldsymbol{w}^{m}\}_{m=1}^{M}\in\mathbb{R}^{N\times1}$,
where the total number of weights is $N=\sum_{m=1}^{M}N_{m}$. The
replica overlap now becomes $q_{\alpha\beta}=\frac{1}{N}\sum_{m}^{M}\sum_{i}^{N_{m}}w_{i}^{m\alpha}w_{i}^{m\beta}.$
The distribution constraint becomes (see Eqn.\ref{appeq:cap_dist_dist_const_delta}
for the case of $M=1$)

\begin{equation}
\prod_{m}\delta\left(\int dk^{m}\left(\frac{1}{N_{m}}\sum_{i}^{N_{m}}e^{ik^{m}w_{i}^{m\alpha}}-q_{m}(k^{m})\right)\right).
\end{equation}

We introduce $\hat{q}_{\alpha\beta},\lambda_{m}(k)$ to write the
$\delta$-functions into Fourier representations, and use replica-symmetric
ansatz $\hat{q}_{\alpha\beta}=-\frac{i}{2}(\Delta\hat{q}\delta_{\alpha\beta}+\hat{q}_{1})$,
and $q_{\alpha\beta}=(1-q)\delta_{\alpha\beta}+q$ as before. After
similar manipulations that lead to Eqn.\ref{appeq:G0_cap_dist_qneq1},
the entropic part of the free energy becomes ($g_{m}=N_{m}/N$ is
the fraction of weights in $m$-th population)

\begin{equation}
\begin{split}G_{0}= & \frac{1}{2}\Delta\hat{q}-\frac{1}{2}\hat{q}_{1}(1-q)+\sum_{m}g_{m}\int dw^{m}q_{m}(w^{m})\lambda_{m}(w^{m})\\
 & -\frac{1}{2}\Delta\hat{q}\sum_{m}g_{m}\int dw^{m}q_{m}(w^{m})\left(w^{m}\right)^{2}+\sum_{m}g_{m}\left\langle \log Z_{m}(t)\right\rangle _{t},\\
Z_{m}(t)= & \int dw^{m}\exp\left\{ -\lambda_{m}(w^{m})+\sqrt{\hat{q}_{1}}tw^{m}\right\} .
\end{split}
\end{equation}

Now the second moment constraint $0=\partial G_{0}/\partial\Delta\hat{q}$
(Eqn.\ref{appeq:second_moment}) becomes the weighted sum of second
moments from each population:

\begin{equation}
1=\sum_{m}g_{m}\int dw^{m}q_{m}(w^{m})\left(w^{m}\right)^{2}=\sum_{m}g_{m}\left\langle \left(w^{m}\right)^{2}\right\rangle _{q_{m}}.
\end{equation}

We take the $q\to1$ limit as before: 

\begin{equation}
\hat{q}_{1}=\frac{u^{2}}{(1-q)^{2}};\qquad\lambda_{m}(w^{m})=\frac{r_{m}(w^{m})}{1-q}.
\end{equation}

Use saddle-point approximation for $Z_{m}(t)$ and integrate by parts
as in Eqn.\ref{appeq:cap_dist_Zt_1}-\ref{appeq:cap_dist_Zt_3}, the
entropic part becomes

\begin{equation}
G_{0}=\frac{1}{1-q}\left\{ -\frac{1}{2}u^{2}+\sum_{m}g_{m}r'_{m}(w^{m})\left[P(x)-Q_{m}(w^{m})\right]+u\sum_{m}g_{m}\left\langle tw^{m}\right\rangle _{t}\right\} .
\end{equation}

Now the saddle-point equation for order parameters $r'_{m}(w^{m})$
and $u$ gives

\begin{equation}
\begin{split}P(x) & =Q_{m}(w^{m})\\
u & =\sum_{m}g_{m}\left\langle tw^{m}\right\rangle _{t}=\sum_{m}g_{m}\left\langle \frac{dw^{m}}{dt}\right\rangle _{t}.
\end{split}
\end{equation}

Therefore,

\begin{equation}
G_{0}=\frac{1}{2(1-q)}\left[\sum_{m}g_{m}\left\langle \frac{dw^{m}}{dt}\right\rangle _{t}\right]^{2}.
\end{equation}

The energetic part (Eqn.\ref{appeq:cap_dist_Zt_1}) remains unchanged
and thus (relabel $t\leftrightarrow x$) 

\begin{equation}
\alpha_{c}(\kappa)=\alpha_{0}(\kappa)\left[\sum_{m}g_{m}\left\langle \frac{dw^{m}}{dx}\right\rangle _{x}\right]^{2}.
\end{equation}

\subsubsection*{E/I balanced lognormals}

Now we specialize to the biologically realistic E/I balanced lognormal
distributions described in the main text. We are interested the case
with two synaptic populations $m=E,I$ that models the excitatory/inhibitory
synpatic weights of a biological neuron. $w_{i}^{E}\sim\text{\ensuremath{\frac{1}{\sqrt{2\pi}\sigma_{E}w^{E}}\exp\left\{ -\frac{(\ln w^{E}-\mu_{E})^{2}}{2\sigma_{E}^{2}}\right\} }}$
and $w_{i}^{I}\sim\text{\ensuremath{\frac{1}{\sqrt{2\pi}\sigma_{I}w^{I}}\exp\left\{ -\frac{(\ln w^{I}-\mu_{I})^{2}}{2\sigma_{I}^{2}}\right\} }}$.
Let's denote the E/I fractions as $g_{E}=r$ and $g_{I}=1-r$. The
CDF of the lognormals are given by

\begin{equation}
\begin{split}Q_{m}(w^{m})= & H\left[\frac{1}{\sigma_{m}}\left(\mu_{m}-\ln w^{m}\right)\right].\end{split}
\end{equation}

The corresponding inverse CDF is 

\begin{equation}
Q_{m}^{-1}(u)=\exp\left\{ \mu_{m}-\sigma_{m}H^{-1}(u)\right\} .
\end{equation}

The capacity is therefore
\begin{equation}
\begin{split}\alpha_{c} & =\alpha_{0}\left[\sum_{m}g_{m}\int_{0}^{1}duQ_{m}^{-1}(u)P^{-1}(u)\right]^{2}\\
 & =\alpha_{0}\left[r\int_{0}^{1}duH^{-1}(u)\exp\left\{ \mu_{E}-\sigma_{E}H^{-1}(u)\right\} +(1-r)\int duH^{-1}(u)\exp\left\{ \mu_{I}-\sigma_{I}H^{-1}(u)\right\} \right]^{2}.
\end{split}
\end{equation}

This model has five parameters $\left\{ r,\sigma_{E},\sigma_{I},\mu_{E},\mu_{I}\right\} $.
We use values of $r$ reported in experiments (the ratio between of
E. connections found and I. connections found). 

We also have two constraints. The E/I balanced constraint $g_{E}\left\langle w^{E}\right\rangle _{q_{E}}=g_{I}\left\langle w^{I}\right\rangle _{q_{I}}$:

\begin{equation}
re^{\mu_{E}+\frac{1}{2}\sigma_{E}^{2}}=(1-r)e^{\mu_{I}+\frac{1}{2}\sigma_{I}^{2}},
\end{equation}

and the second moment constraint $1=\sum_{m}g_{m}\left\langle \left(w^{m}\right)^{2}\right\rangle _{q_{m}}$:

\begin{equation}
1=re^{2(\mu_{E}+\sigma_{E}^{2})}+(1-r)e^{2(\mu_{I}+\sigma_{I}^{2})}.
\end{equation}

Therefore there are two free parameters left and we choose to express
$\mu_{E}$ and $\mu_{I}$ in terms of the rest:

\begin{equation}
\begin{split}\mu_{I}= & -\frac{1}{2}\sigma_{I}^{2}-\ln(1-r)-\frac{1}{2}\ln\left[\frac{e^{\sigma_{I}^{2}}}{1-r}+\frac{e^{\sigma_{E}^{2}}}{r}\right]\\
\mu_{E}= & -\frac{1}{2}\sigma_{E}^{2}-\ln r-\frac{1}{2}\ln\left[\frac{e^{\sigma_{I}^{2}}}{1-r}+\frac{e^{\sigma_{E}^{2}}}{r}\right].
\end{split}
\end{equation}

The parameter landscape is plotted against the two free parameters
$\sigma_{E}$ and $\sigma_{I}$. Here we report comparisons across
different experiments \cite{levy2012spatial,avermann2012microcircuits,holmgren2003pyramidal,molnar2008complex,thomson2002synaptic,yang2013development}
similar to main text Fig.4 (Fig.4 (a) is included here for reference).
Note that despite the apparently different shape of distributions,
all the experimentally measured parameter values are within the first
quantile of the optimal values predicted by our theory.

\begin{figure}
\centering{}\includegraphics[scale=0.12]{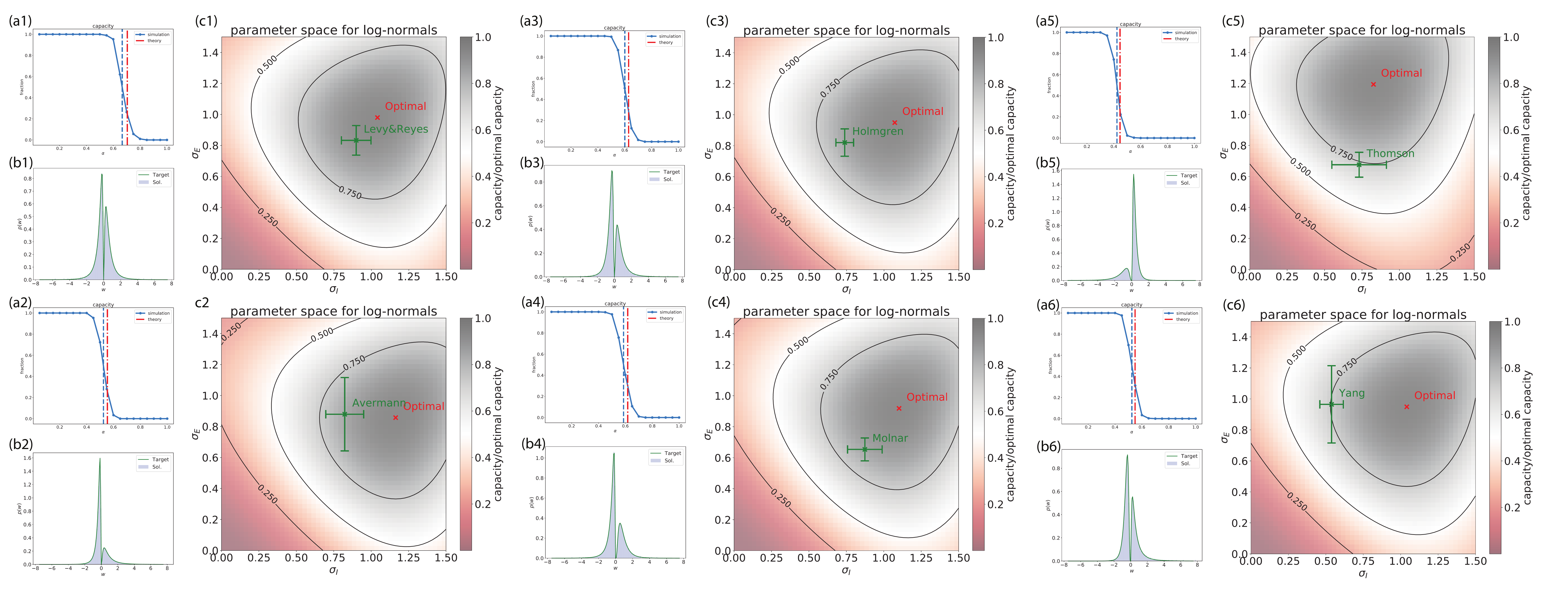}\caption{\label{appfig:Experimental_panel} Additional parameter landscape
for the biologically-realistic distribution. (a)-(b) (theory from
main text Eqn.10 and simulations from DisCo-SGD): (a) Determination
of capacity; (b) Example of weight distribution obtained in simulation.
(c) Capacity (normalized by the optimal value in the landscape) as
a function of the lognormal parameters $\sigma_{E}$ and $\sigma_{I}$.
Experimental value is shown in green with error bars, and optimal
capacity is shown in red. }
\end{figure}

\subsubsection{Capacity for biased inputs and sparse label}

\label{app:capacity_biased_sparse}

In this section, we generalized our theory in Section \ref{app:capacity_dist_const}
to the set up of nonzero-mean input patterns $\boldsymbol{\xi}^{\mu}$
and sparse labels $\zeta^{\mu}$:

\begin{equation}
\begin{split}p(\xi_{i}^{\mu})= & \mathcal{N}(m,1-m^{2})\\
p(\zeta^{\mu})= & f\delta(\zeta^{\mu}-1)+(1-f)\delta(\zeta^{\mu}+1).
\end{split}
\end{equation}

In this case, we need to include a bias in the perceptron $\hat{\zeta}^{\mu}=\text{sgn}(\frac{\boldsymbol{w}\cdot\boldsymbol{\xi}^{\mu}}{\left\Vert \boldsymbol{w}\right\Vert }-b)$
to be able to correctly classify patterns in general. 

Note that $m=0$ and $f=1/2$ reduces to the case in Section \ref{app:capacity_dist_const}.
We observe due to the multiplicative relation between the Jacobian
term and the original Gardner capacity in Eqn.\ref{appeq:cap_dist_formula},
entropic effects (such as distribution constraints and sign-constraints)
factors with the energetic effects (such as the nonzero mean inputs
and sparse labels), and they don't interfere with each other. Therefore,
the calculations for nonzero mean inputs and sparse labels are identical
with the original Gardner case. Here we only reproduce the calculation
for completeness. Readers already familiar with this calculation should
skip this part. 

The analog of Eqn.\ref{appeq:cap_dist_energy_x_z} reads (define the
local fields as $h_{i}^{\mu}=\sum_{\alpha}x_{\alpha}^{\mu}w_{i}^{\alpha}$)

\begin{equation}
\begin{split}\prod_{\mu\alpha}\left\langle e^{-\frac{i}{\sqrt{N}}x_{\alpha}^{\mu}\zeta^{\mu}\boldsymbol{\xi}^{\mu}\cdot\boldsymbol{w}^{\alpha}}\right\rangle _{\xi\zeta}= & \prod_{\mu i}\left\langle \exp\left\{ -\frac{i}{\sqrt{N}}\zeta^{\mu}\xi_{i}^{\mu}h_{i}^{\mu}\right\} \right\rangle _{\xi\zeta}\\
= & \prod_{\mu i}\left\langle \exp\left\{ -\frac{im}{\sqrt{N}}\zeta^{\mu}h_{i}^{\mu}-\frac{1}{2N}(1-m^{2})\left(h_{i}^{\mu}\right)^{2}\right\} \right\rangle _{\zeta}\\
= & \prod_{\mu}\left\langle \exp\left\{ -im\zeta^{\mu}\sum_{\alpha}x_{\alpha}^{\mu}M_{\alpha}-\frac{1-m^{2}}{2}\sum_{\alpha\beta}x_{\alpha}^{\mu}x_{\beta}^{\mu}q_{\alpha\beta}\right\} \right\rangle _{\zeta},
\end{split}
\end{equation}

where in the second equality we have carried out the Gaussian integral
in $\boldsymbol{\xi}^{\mu}$ and in the third equality we introduced
the order parameters 

\begin{equation}
q_{\alpha\beta}=\frac{1}{N}\sum_{i}w_{i}^{\alpha}w_{i}^{\beta},\qquad M_{\alpha}=\frac{1}{\sqrt{N}}\sum_{i}w_{i}^{\alpha}.
\end{equation}

Now the full energetic term becomes

$\begin{aligned}\left\langle \Theta\left(\frac{1}{\sqrt{N}}\zeta^{\mu}\boldsymbol{\xi}^{\mu}\cdot\boldsymbol{w}^{\alpha}-b\zeta^{\mu}-\kappa\right)\right\rangle _{\xi\zeta}\qquad\qquad\qquad\qquad\\
=\prod_{\mu}\left\langle \int_{\kappa+b\zeta^{\mu}}^{\infty}\frac{d\lambda_{\alpha}^{\mu}}{2\pi}\int dx_{\alpha}^{\mu}\exp\left\{ -im\zeta^{\mu}\sum_{\alpha}x_{\alpha}^{\mu}M_{\alpha}-\frac{1-m^{2}}{2}\sum_{\alpha\beta}x_{\alpha}^{\mu}x_{\beta}^{\mu}q_{\alpha\beta}\right\} \right\rangle _{\zeta}\\
=f\prod_{\mu}\int_{\kappa+b}^{\infty}\frac{d\lambda_{\alpha}^{\mu}}{2\pi}\int dx_{\alpha}^{\mu}\exp\left\{ i\sum_{\alpha}x_{\alpha}^{\mu}\left(\lambda_{\alpha}^{\mu}-mM_{\alpha}\right)-\frac{1-m^{2}}{2}\sum_{\alpha\beta}x_{\alpha}^{\mu}x_{\beta}^{\mu}q_{\alpha\beta}\right\} \\
+(1-f)\prod_{\mu}\int_{\kappa-b}^{\infty}\frac{d\lambda_{\alpha}^{\mu}}{2\pi}\int dx_{\alpha}^{\mu}\exp\left\{ i\sum_{\alpha}x_{\alpha}^{\mu}\left(\lambda_{\alpha}^{\mu}+mM_{\alpha}\right)-\frac{1-m^{2}}{2}\sum_{\alpha\beta}x_{\alpha}^{\mu}x_{\beta}^{\mu}q_{\alpha\beta}\right\} \\
=f\prod_{\mu}\int_{\frac{\kappa+b-mM_{\alpha}}{\sqrt{1-m^{2}}}}^{\infty}\frac{d\lambda_{\alpha}^{\mu}}{2\pi}\int dx_{\alpha}^{\mu}\exp\left\{ i\sum_{\alpha}x_{\alpha}^{\mu}\lambda_{\alpha}^{\mu}-\frac{1}{2}\sum_{\alpha\beta}x_{\alpha}^{\mu}x_{\beta}^{\mu}q_{\alpha\beta}\right\} \\
+(1-f)\prod_{\mu}\int_{\frac{\kappa-b+mM_{\alpha}}{\sqrt{1-m^{2}}}}^{\infty}\frac{d\lambda_{\alpha}^{\mu}}{2\pi}\int dx_{\alpha}^{\mu}\exp\left\{ i\sum_{\alpha}x_{\alpha}^{\mu}\lambda_{\alpha}^{\mu}-\frac{1}{2}\sum_{\alpha\beta}x_{\alpha}^{\mu}x_{\beta}^{\mu}q_{\alpha\beta}\right\} .
\end{aligned}
$

Now $G_{1}$ becomes 

\begin{equation}
\begin{split}G_{1}= & \frac{1}{1-q}\left\{ f\int_{\frac{\kappa-b+mM}{\sqrt{1-m^{2}}}}^{\infty}Dt\left(t+\frac{\kappa+b-mM}{\sqrt{1-m^{2}}}\right)^{2}+(1-f)\int_{\frac{-\kappa-b-mM}{\sqrt{1-m^{2}}}}^{\infty}Dt\left(t+\frac{\kappa-b+mM}{\sqrt{1-m^{2}}}\right)^{2}\right\} \end{split}
.
\end{equation}

Note that the hat-variables $\hat{M}_{\alpha}$ conjugated with $M_{\alpha}$
are in subleading order to $\hat{q}_{\alpha\beta}$ in the thermodynamic
limit, and therefore $G_{0}$ is unchanged. Let $v=M-b/m$, we have
now the capacity

\begin{equation}
\begin{split}\alpha_{c}(\kappa)= & \left\langle \frac{dw}{dx}\right\rangle _{x}^{2}\left[f\int_{\frac{-\kappa+mv}{\sqrt{1-m^{2}}}}^{\infty}Dt\left(t+\frac{\kappa-mv}{\sqrt{1-m^{2}}}\right)^{2}+(1-f)\int_{\frac{-\kappa-mv}{\sqrt{1-m^{2}}}}^{\infty}Dt\left(t+\frac{\kappa+mv}{\sqrt{1-m^{2}}}\right)^{2}\right]^{-1}\end{split}
,
\end{equation}

where the order parameter $v$ needs to be determined from the saddle-point
equation

\begin{equation}
f\int_{\frac{-\kappa+mv}{\sqrt{1-m^{2}}}}^{\infty}Dt\left(t+\frac{\kappa-mv}{\sqrt{1-m^{2}}}\right)=(1-f)\int_{\frac{-\kappa-mv}{\sqrt{1-m^{2}}}}^{\infty}Dt\left(t+\frac{\kappa+mv}{\sqrt{1-m^{2}}}\right).
\end{equation}

In Fig.\ref{appfig:nonzero_mean_inputs} we numerically solve $\alpha_{c}(\kappa)$
for different values of $m$ and $f$. 

\begin{figure}
\centering{}\includegraphics[scale=0.5]{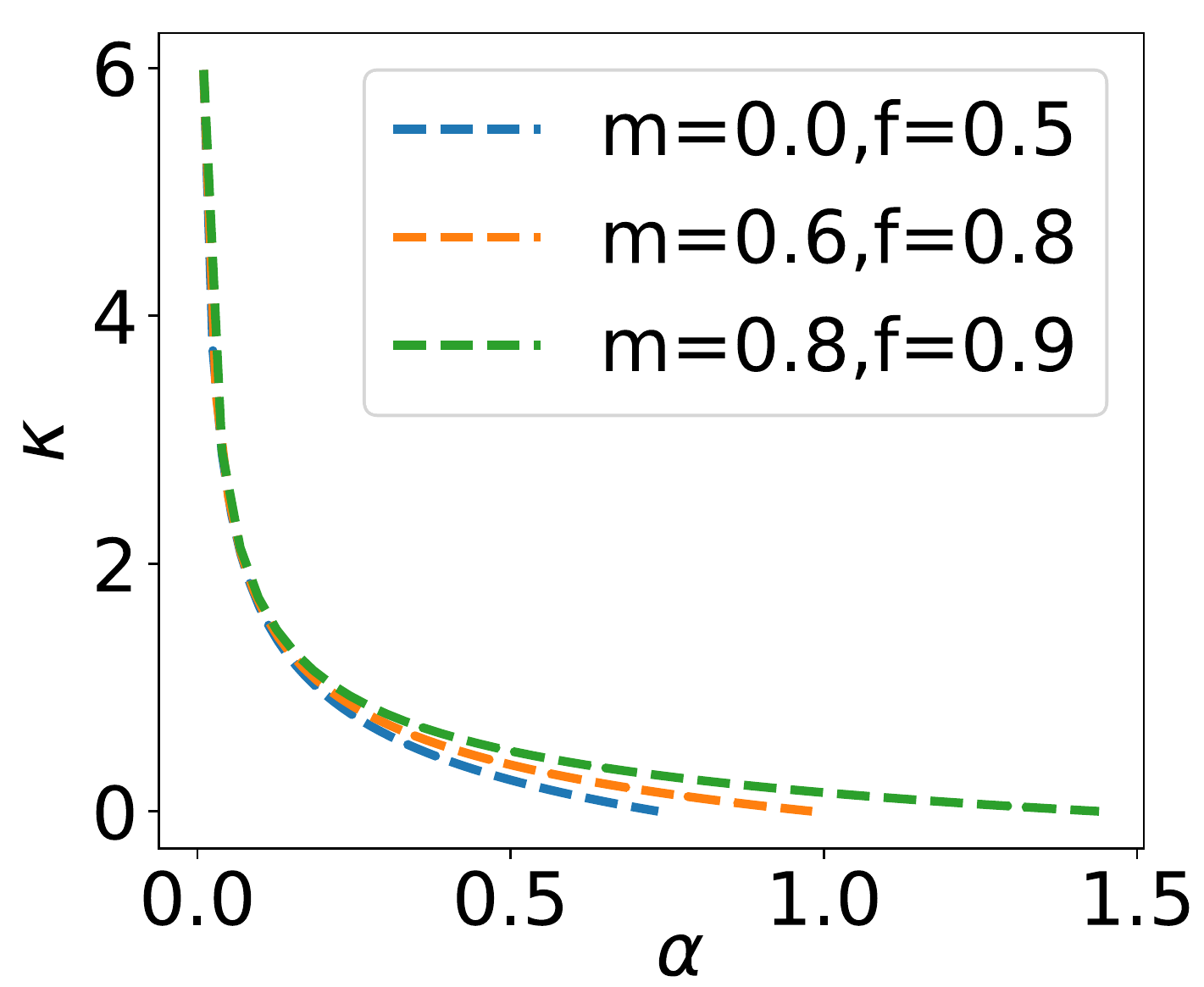}\caption{\label{appfig:nonzero_mean_inputs}$\alpha_{c}(\kappa)$ for different
values of input mean $m$ and label sparsity $f$. Note that the blue
curve corresponds to the vanilla case shown in main text Fig.4(c).}
\end{figure}

\subsection{Optimal transport theory}

\label{app:Optimal-transport-theory}

In recent years, Wasserstein distances has found diverse applications
in fields ranging from machine learning \cite{arjovsky2017wasserstein,frogner2015learning,montavon2016wasserstein}
to geophysics \cite{engquist2013application,engquist2016optimal,chen2018quadratic,metivier2016measuring,metivier2016optimal}.
In optimal transport theory, the Wasserstein-$k$ distance arise as
the minimal cost one needs to pay in transporting one probability
distribution to another, when the moving cost between probability
masses are measured by the $L_{k}$ norm \cite{villani2009optimal}.
When one equips the probability density manifold with the Wasserstein-$2$
distance as metric, it becomes the Wasserstein space, a Riemannian
manifold of real-valued distributions with a constant nonnegative
sectional curvature \cite{lott2006some,figalli2011optimal,chen2020optimal}.
Note that in our statistical mechanical theory main text Eqn.3-5,
the Wasserstein-$2$ distance naturally arises in the mean-field limit
without assuming any a priori transportation cost.

Here we briefly review the theory of optimal transport. Intuitively,
optimal transport concerns the problem of finding the shortest path
of morphing one distribution into another. In the following, we will
use the \textit{Monge} formulation \cite{thorpe2019introduction,ambrosio2013user}.

Given probability distributions $P$ and $Q$ with supports $X$ and
$Y$, we say that $T:X\to Y$ is a transport map from $P$ to $Q$
if the \textit{push-forward }of $P$ through $T,$ $T_{\#}P$, equals
$Q$:

\begin{equation}
Q=T_{\#}P\equiv P(T^{-1}(Y)).\label{eq:pushforward}
\end{equation}

Eqn.\ref{eq:pushforward} can be understood as moving probability
masses $x\in X$ from distribution $P$ to $y\in Y$ according to
transportation map $T$, such that upon completion the distribution
over $Y$ becomes $Q$.

We are interested in finding a transportation plan that minimizes
the transportation cost as measured by some distance function $d:X\times Y\to\mathbb{R}$
:

\begin{equation}
C(T;d)=\int_{X}d(T(x),x)p(x)dx\qquad\text{s.t.}\;T_{\#}P=Q.\label{eq:transportation_cost}
\end{equation}

The transportation plan that minimizes Eqn.\ref{eq:transportation_cost}
is called the optimal transport plan $T^{*}=\text{argmin}_{T}C(T;d)$.
When the distance function $d$ is chosen to be the $L_{k}$ norm,
the minimal cost becomes the Wasserstein-$k$ distance:

\begin{equation}
W_{k}(P,Q)=\inf_{T}C(T;L_{k})|_{T_{\#}P=Q}.\label{eq:Wass-p}
\end{equation}

In $1$-dimension, the Wasserstein-$k$ distance has a closed form
given by main text Eqn.6, and the optimal transport map has an explicit
formula in terms of the CDFs: $T^{*}=Q^{-1}\circ P$. An example of
the optimal transport map and how it moves probability masses between
distributions is given in Fig.\ref{appfig:optimal_transport_lognormal}
for transport between $p(w)=\mathcal{N}(0,1)$ and $q(w)=\frac{1}{\sqrt{2\pi}\sigma w}\exp\left\{ \frac{(\ln w-\mu)^{2}}{2\sigma^{2}}\right\} .$
Note that in this case, the optimal transport plan is simply $T^{*}(x)=e^{\mu+\sigma x}$.

\begin{figure}
\centering{}\includegraphics[scale=0.25]{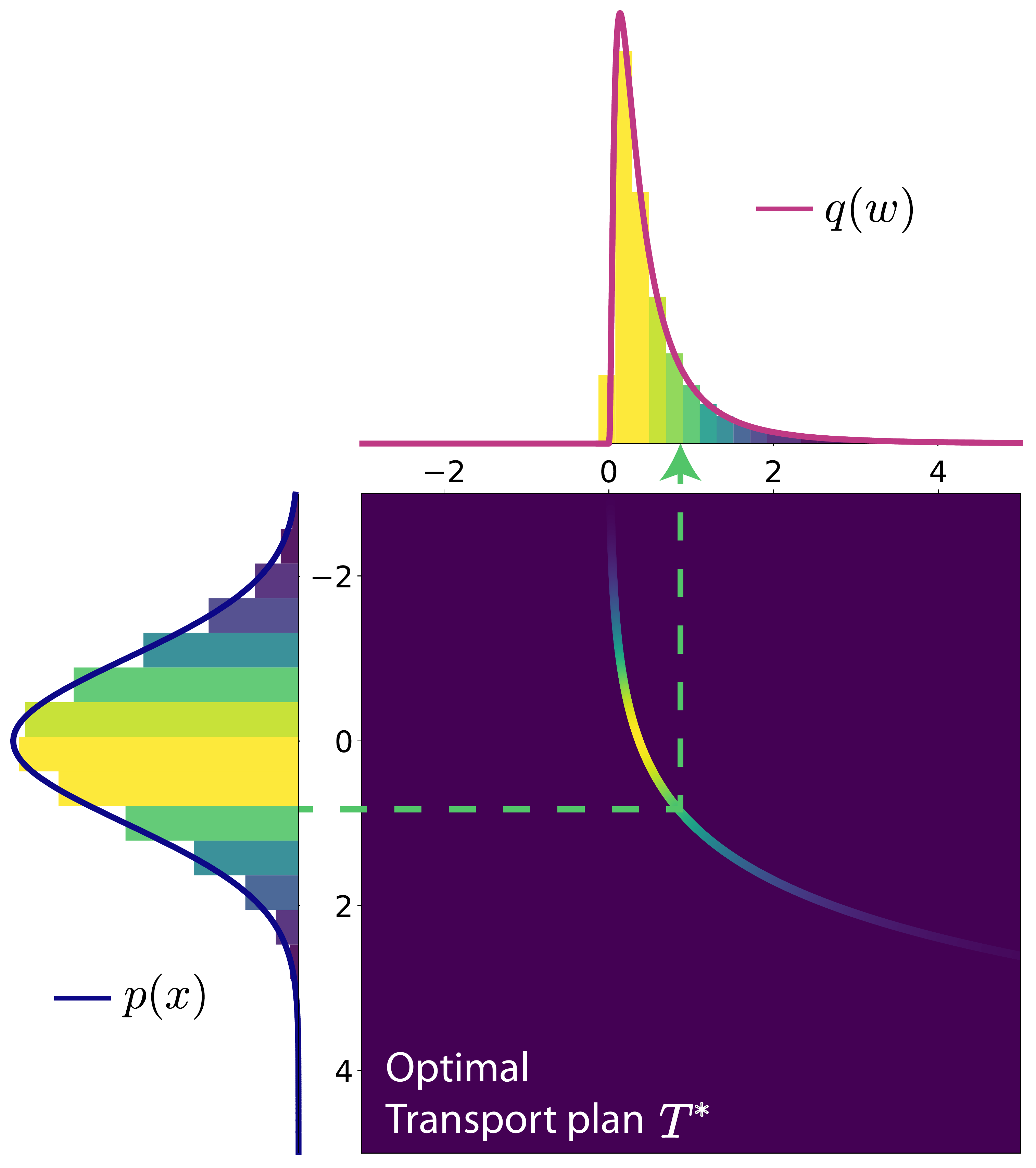}\caption{\label{appfig:optimal_transport_lognormal}An example optimal transport
plan from standard normal, $p(x)$, to a lognormal distribution $q(w)$.
The optimal transport plan $T^{*}$ is plotted in between the distributions.
$T^{*}$ moves $p(x)$ units of probability mass $x$ to location
$w$, as indicated by the dashed line, and the colors are chosen to
reflect the amount of probability mass to be transported.}
\end{figure}

Now consider the manifold $\mathcal{M}$ of real-valued probability
distributions, where points on this manifold are probability measures
that admits a probability density function. When endowed with the
$W_{k}$ metric, $(\mathcal{M},W_{k})$ becomes a metric space and
is in particular a geodesic space \cite{thorpe2019introduction,ambrosio2013user}.
We can explicitly construct the geodesics connecting points on $\mathcal{M}$.
We parameterize the geodesic by the \textit{geodesic time} $\tau\in[0,1].$
Then given $T^{*}$ an optimal transport plan, the intermediate probability
distributions along the geodesic take the following form \cite{thorpe2019introduction}:

\begin{equation}
P_{\tau}=\left((1-\tau)\text{Id}+\tau T^{*}\right)_{\#}P\label{eq:geodesic}
\end{equation}

where $\text{Id}$ is the identity map and $P_{\tau}$ is a constant
speed geodesic connecting $P_{\tau=0}=P$ and $P_{\tau=1}=Q$.

For the discrete case, we can describe the sample $\left\{ w_{i}^{\tau}\right\} $
from $P_{\tau}$ in a simple manner in terms of the samples $\left\{ w_{i}\right\} $
drawn from $P$ and $\left\{ \hat{w}_{i}\right\} $ drawn from $Q$.
We can arrange the samples in the ascending order, or equivalently,
forming their order statistics $\left\{ x_{(i)}:x_{(1)}\leq...\leq x_{(N)}\right\} $,
which can be thought of as atoms in a discrete measure. Then in terms
of the order statistics, 

\begin{equation}
w_{(i)}^{\tau}=(1-\tau)w_{(i)}+\tau\hat{w}_{(i)}
\end{equation}

Upon infinitetesimal change in the geodesic time, $\tau\to\tau+\delta\tau$,
the geodesic flow becomes

\begin{equation}
w_{(i)}^{\tau+\delta\tau}=w_{(i)}^{\tau}+\delta\tau\left(\hat{w}_{(i)}-w_{(i)}\right)
\end{equation}

Specializing to the case discussed in main text Section 3, $w_{(i)}=w_{(i)}^{\tau=0}$
is the initialization for the perceptron weight and therefore just
a constant, we can promoted it $w_{(i)}\to w_{(i)}^{\tau}$ to fix
the overall scale in the perceptron weight, then we arrive at main
text Eqn.9.

\subsection{Generalization supplemental materials}

\label{app:generalization}

\subsubsection{Replica calculation of generalization with sign-constraint}

\label{app:generalization_sign_const}

In this section, we calculate the sign-constraint teacher-student
setup. To ease notation, let's denote the teacher perceptron $\boldsymbol{w_{t}}\equiv\boldsymbol{w}^{0}$
and the (replicated) student perceptron $\boldsymbol{w}_{s}^{a}\equiv\boldsymbol{w}^{a}.$
Given random inputs $\boldsymbol{\xi}^{\mu}$ with $p(\xi_{i}^{\mu})=\mathcal{N}(0,1)$,
we generate labels by $\zeta^{\mu}=\text{sgn}(\boldsymbol{w}^{0}\cdot\boldsymbol{\xi}^{\mu}/||\boldsymbol{w}^{0}||+\eta^{\mu})$,
where $\eta^{\mu}$ is input noise and $\eta^{\mu}\sim\mathcal{N}(0,\sigma^{2})$.
Let's denote the signs of the teacher perceptron as $s_{i}=\text{sgn}(w_{i}^{0}).$
The student perceptron's weights are constrained to have the same
sign as that of the teacher's, so we insert $\Theta(s_{i}w_{i}^{a})$
in the Gardner volume to enforce this constraint (we leave out the
denominator part of $V$ as it does not depend on data and is an overall
constant): 

\begin{equation}
\left\langle V^{n}\right\rangle _{\xi\eta w^{0}}=\prod_{\alpha=1}^{n}\left\langle \int_{-\infty}^{\infty}\frac{d\boldsymbol{w}^{a}}{\sqrt{2\pi}}\prod_{\mu=1}^{p}\Theta\left(\text{sgn}\left(\frac{\boldsymbol{w^{0}}\cdot\boldsymbol{\xi}^{\mu}}{||\boldsymbol{w^{0}}||}+\eta^{\mu}\right)\frac{\boldsymbol{w^{a}}\cdot\boldsymbol{\xi}^{\mu}}{||\boldsymbol{w^{a}}||}-\kappa\right)\prod_{i}^{N}\Theta(s_{i}w_{i}^{a})\right\rangle _{\xi\eta w^{0}}.\label{sc_original}
\end{equation}

We observe that upon redefining $s_{i}w_{i}^{a}\to w_{i}^{a},s_{i}\xi_{i}^{\mu}\to\xi_{i}^{\mu}$,
we can absorb the sign-constraints into the integration range of $w$
from $[-\infty,+\infty]$ to $[0,\infty]$:

\begin{equation}
\left\langle V^{n}\right\rangle _{\xi\eta w^{0}}=\prod_{\alpha=1}^{n}\left\langle \int_{0}^{\infty}\frac{d\boldsymbol{w}^{a}}{\sqrt{2\pi}}\prod_{\mu=1}^{p}\Theta\left(\text{sgn}\left(\frac{\boldsymbol{w^{0}}\cdot\boldsymbol{\xi}^{\mu}}{||\boldsymbol{w^{0}}||}+\eta^{\mu}\right)\frac{\boldsymbol{w^{a}}\cdot\boldsymbol{\xi}^{\mu}}{||\boldsymbol{w^{a}}||}-\kappa\right)\right\rangle _{\xi\eta w^{0}}.\label{sc_vol}
\end{equation}

Therefore, sign constraint amounts to restricting all the weights
to be positive. In the following, we denote $\int_{0}^{\infty}$as
$\int_{c}$. 

Let's define the local fields as
\begin{equation}
h_{\mu}^{a}=\frac{\boldsymbol{w^{a}}\cdot\boldsymbol{\xi}^{\mu}}{\sqrt{N}};\qquad h_{\mu}^{0}=\frac{\boldsymbol{w^{0}}\cdot\boldsymbol{\xi}^{\mu}}{\sqrt{N}}+\eta^{\mu}
\end{equation}

We leave the average over teacher $w^{0}$ to the very end.

\begin{equation}
\begin{split}\left\langle V^{n}\right\rangle _{\xi\eta} & =\prod_{\mu a}\int_{c}\frac{d\boldsymbol{w}^{a}}{\sqrt{2\pi}}\int dh_{\mu}^{a}\Theta\bigg(\text{sgn}(h_{\mu}^{0})h_{\mu}^{a}-\kappa\bigg)\left\langle \delta\left(h_{\mu}^{a}-\frac{\boldsymbol{w^{a}}\cdot\boldsymbol{\xi}^{\mu}}{\sqrt{N}}\right)\right\rangle _{\xi\eta}\\
 & =\int_{c}(\prod_{a=1}^{n}\frac{d\boldsymbol{w}^{a}}{\sqrt{2\pi}})\int\prod_{\mu a}\frac{dh_{\mu}^{a}d\hat{h}_{\mu}^{a}}{2\pi}\int\prod_{\mu}\frac{dh_{\mu}^{0}d\hat{h}_{\mu}^{0}}{2\pi}\prod_{\mu a}\Theta\bigg(\text{sgn}(h_{\mu}^{0})h_{\mu}^{a}-\kappa\bigg)\\
 & \times\bigg\langle\exp\bigg\{\sum_{\mu a}\bigg(i\hat{h}_{\mu}^{a}h_{\mu}^{a}-i\hat{h}_{\mu}^{a}\frac{\boldsymbol{w^{a}}\cdot\boldsymbol{\xi}^{\mu}}{\sqrt{N}}\bigg)+\sum_{\mu}\bigg(i\hat{h}_{\mu}^{0}h_{\mu}^{0}-i\hat{h}_{\mu}^{0}\frac{\boldsymbol{w^{0}}\cdot\boldsymbol{\xi}^{\mu}}{\sqrt{N}}-i\hat{h}_{\mu}^{0}\eta^{\mu}\bigg)\bigg\}\bigg\rangle_{\xi\eta}\\
 & =\int_{c}(\prod_{a=1}^{n}\frac{d\boldsymbol{w}^{a}}{\sqrt{2\pi}})\int\prod_{\mu a}\frac{dh_{\mu}^{a}d\hat{h}_{\mu}^{a}}{2\pi}\int\prod_{\mu}\frac{dh_{\mu}^{0}d\hat{h}_{\mu}^{0}}{2\pi}\prod_{\mu a}\Theta\bigg(\text{sgn}(h_{\mu}^{0})h_{\mu}^{a}-\kappa\bigg)\\
 & \times\exp\left\{ \sum_{\mu a}i\hat{h_{\mu}^{a}}h_{\mu}^{a}+\sum_{\mu}i\hat{h}_{\mu}^{0}h_{\mu}^{0}\right\} \\
 & \times\prod_{\mu}\exp\left\{ -\frac{1}{2N}\left[\sum_{a,b}\hat{h}_{\mu}^{a}\hat{h}_{\mu}^{b}\sum_{i}w_{i}^{a}w_{i}^{b}+N\left(\hat{h}_{\mu}^{0}\right)^{2}+2\sum_{a}\hat{h}_{\mu}^{a}\hat{h}_{\mu}^{0}\sum_{i}w_{i}^{a}w_{i}^{0}\right]\right\} ,
\end{split}
\end{equation}

where in the last step we perform the average over noise $\eta^{\mu}\sim\mathcal{N}(0,\sigma^{2})$
and patterns $p(\xi_{i}^{\mu})=\mathcal{N}(0,1)$, and make use of
the normalization conditions $\sum_{i}(w_{i}^{0})^{2}=N$ and $\sum_{i}(w_{i}^{a})^{2}=N$.

Now let's define order parameters
\begin{equation}
q_{ab}=\frac{1}{N}\sum_{i}w_{i}^{a}w_{i}^{b},\qquad R_{a}=\frac{1}{N}\sum_{i}w_{i}^{a}w_{i}^{0}.
\end{equation}

We introduce conjugate variables $\hat{q}_{ab}$ and $\hat{R}_{a}$
to write the $\delta$-functions into its Fourier representations,
and after some algebraic manipulations we can bring the Gardner volume
into the following form ($\alpha\equiv p/N$):

\begin{equation}
\begin{split}\langle\langle V^{n}\rangle\rangle_{\xi,z} & =\int(\prod_{a}d\hat{q}_{1}^{a})(\prod_{ab}dq^{ab}d\hat{q}^{ab})(\prod_{a}dR^{a}d\hat{R}^{a})e^{nNG}\end{split}
,
\end{equation}

where ($\bar{h}_{\mu}^{0}=\gamma h_{\mu}^{0};\quad\gamma=1/\sqrt{1+\sigma^{2}}$)

\begin{equation}
\begin{split}nG= & nG_{0}+\alpha nG_{E}\\
nG_{0}= & -\frac{1}{2}\sum_{ab}\hat{q}^{ab}q^{ab}-\sum_{a}\hat{R}^{a}R^{a}+n\left\langle \ln Z\right\rangle _{w^{0}},\\
Z= & \int_{c}\left(\prod_{a}\frac{dw_{i}^{a}}{\sqrt{2\pi}}\right)\exp\bigg\{\frac{1}{2}\sum_{a}\hat{q}_{1}^{a}(w_{i}^{a})^{2}+\frac{1}{2}\sum_{a\neq b}\hat{q}^{ab}w_{i}^{a}w_{i}^{b}+\sum_{a}\hat{R}^{a}w_{i}^{a}w_{i}^{0}\bigg\},\\
nG_{1}= & \ln\int\prod_{a}\frac{d\hat{h}^{a}dh^{a}}{2\pi}\int D\bar{h}^{0}\prod_{a}\Theta\bigg(\text{sgn}(\frac{\bar{h}^{0}}{\gamma})h^{a}-\kappa\bigg)\\
 & \times\exp\bigg\{ i\sum_{a}\hat{h}^{a}h^{a}-i\gamma\bar{h}^{0}\sum_{a}h^{a}R^{a}-\frac{1}{2}\sum_{a}(\hat{h}^{a})^{2}[1-(\gamma R^{a})^{2}]-\frac{1}{2}\sum_{a\neq b}\hat{h}^{a}\hat{h}^{b}(q^{ab}-\gamma^{2}R^{a}R^{b})\bigg\}.
\end{split}
\end{equation}

The energetic part $G_{1}$ is the same as the unconstrained case
in \cite{seung1992statistical,engel2001statistical}. After standard
manipulations, we have
\begin{equation}
G_{1}=2\int DtH\bigg(-\frac{\gamma Rt}{\sqrt{q-\gamma^{2}R^{2}}}\bigg)\ln H\bigg(\frac{\kappa-\sqrt{q}t}{\sqrt{1-q}}\bigg).
\end{equation}

\subsubsection*{Entropic part}

In this subsection, we perform the integrals in the entropic part,
and we will see novel terms coming from the constraint on the student's
integration range.

We start by assuming a replica-symmetric solution for the auxiliary
variables introduced in the Fourier decomposition of the $\delta$-functions,
\begin{equation}
\hat{R}^{a}=\hat{R};\qquad\hat{q}^{ab}=\hat{q}+(\hat{q}_{1}-\hat{q})\delta_{ab};\qquad\hat{q}_{1}^{a}=\hat{q}_{1};\qquad m_{i}^{a}=m_{i};\qquad\hat{m}_{i}^{a}=\hat{m}_{i},
\end{equation}

and $q_{ab}=(1-q)\delta_{ab}+q.$

Then the entropic part, 
\begin{equation}
\begin{split}Z & =\int\left(\prod_{a}\frac{dw_{i}^{a}}{\sqrt{2\pi}}\right)\exp\bigg\{\frac{1}{2}(\hat{q}_{1}-\hat{q})\sum_{a}(w_{i}^{a})^{2}+\hat{R}w_{i}^{0}\sum_{a}w_{i}^{a}+\frac{1}{2}\hat{q}(\sum_{a}w_{i}^{a})^{2}\bigg\}\\
 & \HSTeq\int Dt\int_{c}(\prod_{a}\frac{dw_{i}^{a}}{\sqrt{2\pi}})\exp\bigg\{\frac{1}{2}(\hat{q}_{1}-\hat{q})\sum_{a}(w_{i}^{a})^{2}+(\hat{R}w_{i}^{0}+t\sqrt{\hat{q}})\sum_{a}w_{i}^{a}\bigg\},
\end{split}
\end{equation}

where we have introduced Gaussian variable $t$ to linearize quadratic
term as usual. Now the integral becomes $n$ identical copies and
we can drop the replica index $a$,
\begin{equation}
G_{0}=-\frac{1}{2}\hat{q}_{1}+\frac{1}{2}\hat{q}q-\hat{R}R+\left\langle \ln Z\right\rangle _{t,w^{0}}.
\end{equation}

We can bring the log term into the form of an induced distribution
$f(w)$,

\begin{equation}
\begin{split}Z= & \int_{0}^{\infty}\frac{dw}{\sqrt{2\pi}}\exp\left[-f(w)\right]\\
f(w)= & \frac{1}{2}(\hat{q}-\hat{q}_{1})w^{2}-(\hat{R}w^{0}+t\sqrt{\hat{q}})w
\end{split}
.
\end{equation}

Under saddle-point approximation, we obtain a set of mean field self-consistency
equations for the order parameters: 
\begin{equation}
\begin{split}0=\frac{\partial G_{0}}{\partial\hat{q}_{1}} & \Rightarrow1=\left\langle \left\langle w^{2}\right\rangle _{f}\right\rangle _{t,w^{0}}\\
0=\frac{\partial G_{0}}{\partial\hat{R}} & \Rightarrow R=\left\langle w^{0}\left\langle w\right\rangle _{f}\right\rangle _{t,w^{0}}\\
0=\frac{\partial G_{0}}{\partial\hat{q}} & \Rightarrow q=\left\langle \left\langle w\right\rangle _{f}^{2}\right\rangle _{t,w^{0}}
\end{split}
,\label{appeq:sc_saddle_1-3_qneq1}
\end{equation}

\begin{equation}
\begin{split}0=\frac{\partial G_{1}}{\partial q} & \Rightarrow\hat{q}=-2\alpha\partial_{q}G_{1}\\
0=\frac{\partial G_{1}}{\partial R} & \Rightarrow\hat{R}=\alpha\partial_{R}G_{1}
\end{split}
.\label{appeq:sc_saddle_4-5_qneq1}
\end{equation}

\subsubsection*{$q\to1$ limit}

In this limit the order parameter diverges, and we define the new
set of undiverged order parameters as

\begin{equation}
\hat{R}=\frac{\tilde{R}}{1-q};\qquad\hat{q}=\frac{\tilde{q}^{2}}{(1-q)^{2}};\qquad\hat{q}-\hat{q}_{1}=\frac{\Delta}{1-q}.
\end{equation}

Then 

\begin{equation}
\begin{split}f(w)= & \frac{1}{1-q}\left[\frac{1}{2}\Delta w^{2}-(\tilde{R}w^{0}+t\tilde{q})w\right]\\
= & \frac{1}{1-q}\left[\frac{1}{2}\Delta\left(w-\frac{1}{\Delta}(\tilde{R}w^{0}+t\tilde{q})\right)^{2}-\frac{1}{2\Delta}(\tilde{R}w^{0}+t\tilde{q})^{2}\right].
\end{split}
\end{equation}

Then $\langle w\rangle_{f}=\frac{1}{\Delta}\left(\tilde{R}w^{0}+t\tilde{q}\right)$,
and the integral over the auxiliary variable is dominated by values
of $t$ such that $\tilde{R}w^{0}+t\tilde{q}>0$. In the following,
we denote $\left\langle \left[g(t)\right]_{+}\right\rangle _{t}$
as integrating over range of $t$ such that $g(t)>0$. Then the self-consistency
equations Eqn.\ref{appeq:sc_saddle_1-3_qneq1} take a compact form
(after rescaling order parameters $\tilde{R}\rightarrow\tilde{R}\Delta$
, $\tilde{q}\rightarrow\tilde{q}\Delta$)

\begin{equation}
\begin{split}1= & \frac{1}{\Delta}\left\langle \Theta(\tilde{R}w^{0}+t\tilde{q})\right\rangle _{t,w^{0}}\\
1= & \left\langle \left[\tilde{R}w^{0}+t\tilde{q}\right]_{+}^{2}\right\rangle _{t,w^{0}}\\
R= & \left\langle w^{0}\left[\tilde{R}w^{0}+t\tilde{q}\right]_{+}\right\rangle _{t,w^{0}}
\end{split}
,\label{appeq:sc_saddle_1-3_q=00003D1}
\end{equation}

Eqn.\ref{appeq:sc_saddle_1-3_qneq1} becomes ($\tilde{\kappa}=\kappa/\sqrt{1-\gamma^{2}R^{2}}$)

\begin{equation}
\begin{split}\tilde{R}\Delta & =\frac{\alpha\gamma}{\sqrt{2\pi}}\sqrt{1-\gamma^{2}R^{2}}\int_{-\tilde{\kappa}}^{\infty}Dt\bigg(\tilde{\kappa}+t\bigg)\\
\frac{\Delta}{2}\left(2-\tilde{q}^{2}\Delta-2\tilde{R}R\right)= & \alpha\int_{-\infty}^{\kappa}DtH\left(-\frac{\gamma Rt}{\sqrt{1-\gamma^{2}R^{2}}}\right)(\kappa-t)^{2}
\end{split}
.\label{appeq:sc_saddle_4-5_q=00003D1}
\end{equation}
The free energy is (recall that $\gamma=1/\sqrt{1+\sigma^{2}}$)

\begin{equation}
G=\frac{1}{2(1-q)}\left(\Delta-\tilde{q}^{2}-2\tilde{R}R+\frac{1}{\Delta}\left\langle \left[\tilde{R}w^{0}+t\tilde{q}\right]_{+}^{2}\right\rangle _{t,w^{0}}\right)-\alpha\int_{-\infty}^{\kappa}DtH\bigg(-\frac{\gamma Rt}{\sqrt{1-\gamma^{2}R^{2}}}\bigg)(\kappa-t)^{2}.
\end{equation}

\subsubsection{Replica calculation of generalization with distribution-constraint}

\label{app:generalization_dist_const}

In this subsection, we will consider the case where student weights
are constrained to some \textit{prior} distribution $q_{s}(w_{s})$,
while the teacher obeys a distribution $p_{t}(w_{t}),$for an arbitrary
pair $q_{s},p_{t}$. We can write down the Gardner volume $V_{g}$
for generalization as in the capacity case (main text Eqn.2):

\begin{equation}
V_{g}=\frac{\int d\boldsymbol{w}_{s}\left[\prod_{\mu=1}^{P}\Theta\left(\text{sgn}\left(\frac{\boldsymbol{w}_{t}\cdot\boldsymbol{\xi}^{\mu}}{||\boldsymbol{w}_{t}||}+\eta^{\mu}\right)\frac{\boldsymbol{w}_{s}\cdot\boldsymbol{\xi}^{\mu}}{||\boldsymbol{w}_{s}||}-\kappa\right)\right]\delta(||\boldsymbol{w}_{s}||^{2}-N)\delta\bigg(\int dk\left(\hat{q}(k)-q(k)\right)\bigg)}{\int d\boldsymbol{w}_{s}\delta(||\boldsymbol{w}_{s}||^{2}-N)}.\label{appeq:gen_dc_vol}
\end{equation}

We treat the distribution constraint $q_{s}(w)$ similar to Section
\ref{app:capacity_dist_const}. The entropic part of the free energy
becomes

\begin{equation}
\begin{split}G_{0}= & -\frac{1}{2}\hat{q}_{1}+\frac{1}{2}\hat{q}q-\hat{R}R+\int_{-\infty}^{\infty}dwq_{s}(w)\lambda(w)+\left\langle \ln Z\right\rangle _{t,w_{t}}\\
Z= & \int\frac{dw}{\sqrt{2\pi}}\exp\left[-f(w)\right]\\
f(w)= & \frac{1}{2}(\hat{q}-\hat{q}_{1})w^{2}-(\hat{R}w_{t}+t\sqrt{\hat{q}})w+\lambda(w)
\end{split}
.
\end{equation}
At the limit $q\rightarrow1$, we make the following ansatz

\begin{equation}
\hat{R}=\frac{\tilde{R}}{1-q};\quad\hat{q}=\frac{u^{2}}{(1-q)^{2}};\quad\hat{q}-\hat{q}_{1}=\frac{\Delta}{1-q};\quad\lambda(w)=\frac{r(w)}{1-q}.
\end{equation}
Then

\begin{equation}
\begin{split}G_{0}= & \frac{1}{\left(1-q\right)}\left(-\frac{1}{2}u^{2}+\frac{1}{2}\Delta-\tilde{R}R+\int dwq_{s}(w)r(w)\right)+\langle\ln Z\rangle_{t,w_{t}}\\
f(w)= & \frac{1}{1-q}\left(\frac{1}{2}\Delta w^{2}-(\tilde{R}w_{t}+ut)w+r(w)\right)
\end{split}
\end{equation}
We can absorb $\frac{1}{2}\Delta w^{2}$ into the definition of $r(w)$,
$\frac{1}{2}\Delta w^{2}+r(w)\to r(w)$, and $0=\partial G_{0}/\partial\Delta$
gives the second moment constraint, $1=\int dwq_{s}(w)w^{2}$. 

Then,

\begin{equation}
\begin{split}G_{0}= & \frac{1}{\left(1-q\right)}\left(-\frac{1}{2}u^{2}-\tilde{R}R+\int dwq(w)r(w)\right)+\langle\ln Z\rangle_{t,w_{t}}\\
f(w)= & \frac{1}{1-q}\left(r(w)-(\tilde{R}w_{t}+ut)w\right)
\end{split}
.
\end{equation}
Next, we perform a saddle-point approximation on the log-term in $G_{0}$,
\begin{equation}
Z=\int\frac{dw}{\sqrt{2\pi}}\exp\left[-f(w)\right]\approx\exp\left[-f(w_{s})\right],
\end{equation}
where $w_{s}$ is the saddle-point value for the weight, and is determined
implicitly by 
\begin{equation}
r'(w_{s})=\tilde{R}w_{t}+ut.
\end{equation}
Note that $r'(w_{s})$ is now an induced random variable from random
variables $w_{t}$ and $t$. For later convenience, we rescale $r'(w_{s})$
to define a new random variable $z$,
\begin{equation}
z\equiv u^{-1}r'(w_{s})=t+u^{-1}\tilde{R}w_{t}\equiv t+\varepsilon w_{t},
\end{equation}
where we have also defined 
\begin{equation}
\varepsilon\equiv u^{-1}\tilde{R}.
\end{equation}
The induced distribution on $z$ is then 
\begin{equation}
\tilde{p}(z)=\int Dt\int dw_{t}p(w_{t})\delta(z-t-\varepsilon w_{t}).
\end{equation}
Now the entropic part becomes
\begin{equation}
G_{0}=\frac{1}{\left(1-q\right)}\left(-\frac{1}{2}u^{2}-\tilde{R}R+\int dwq_{s}(w)r(w)+\langle(\tilde{R}w_{t}+ut)w_{s}\rangle_{t,w_{t}}-\langle r(w_{s})\rangle_{t,w_{t}}\right).
\end{equation}
Integrate by parts,
\begin{equation}
\int dwq(w)r(w)=-\int dwQ(w)r'(w),
\end{equation}

\begin{equation}
\begin{split}\langle r(w_{s})\rangle_{t,w_{t}}= & \int Dtdw_{t}p_{t}(w_{t})r(w_{s})\\
= & \int dz\delta(z-t-\varepsilon w_{t})\int Dtdw_{t}p_{t}(w_{t})r(w_{s})\\
= & \int dz\tilde{p}(z)r(w_{s})\\
= & -\int dz\tilde{P}(z)r'(w_{s})
\end{split}
.\label{appeq:113}
\end{equation}
Now $0=\partial G/\partial r'(w_{s})$ gives

\begin{equation}
Q(w_{s})=\tilde{P}(z).
\end{equation}
 which implicitly determines $w_{s}(z).$ 

Next,
\begin{equation}
0=\frac{\partial G}{\partial u}\Rightarrow u=\langle w_{s}(z)t\rangle_{t,w_{t}},
\end{equation}
\begin{equation}
0=\frac{\partial G}{\partial\tilde{R}}\Rightarrow R=\langle w_{s}(z)w_{t}\rangle_{t,w_{t}}.
\end{equation}

The free energy then simplifies to 
\begin{equation}
G=\frac{u^{2}}{2\left(1-q\right)}+\alpha G_{1}.
\end{equation}
The energetic part as $q\to1$ becomes (same as the unconstrained
and sign-constrained case)

\begin{equation}
G_{1}=-\frac{1}{1-q}\int_{-\infty}^{\kappa}DtH\bigg(-\frac{\gamma Rt}{\sqrt{1-\gamma^{2}R^{2}}}\bigg)(\kappa-t)^{2}.
\end{equation}
The remaining two saddle point equations are (1) the vanishing log-Gardner
volume and (2) $0=\partial G/\partial R$:
\begin{equation}
\frac{1}{2}u^{2}=\alpha\int_{-\infty}^{\kappa}DtH\bigg(-\frac{\gamma Rt}{\sqrt{1-\gamma^{2}R^{2}}}\bigg)(\kappa-t)^{2},
\end{equation}
\begin{equation}
\varepsilon u=\alpha\gamma\sqrt{\frac{2}{\pi}}\sqrt{1-\gamma^{2}R^{2}}\int_{-\tilde{\kappa}}^{\infty}Dt\bigg(\tilde{\kappa}+t\bigg).
\end{equation}

In summary, the order parameters $\left\{ R,\kappa,u,\varepsilon\right\} $
can be determined from a set of self-consistency equations:

\begin{equation}
\begin{split}u= & \langle w_{s}(z)t\rangle_{t,w_{t}}\\
R= & \langle w_{s}(z)w_{t}\rangle_{t,w_{t}}\\
\frac{1}{2}u^{2}= & \alpha\int_{-\infty}^{\kappa}DtH\bigg(-\frac{\gamma Rt}{\sqrt{1-\gamma^{2}R^{2}}}\bigg)(\kappa-t)^{2}\\
\varepsilon u= & \frac{2\alpha\gamma}{\sqrt{2\pi}}\sqrt{1-\gamma^{2}R^{2}}\int_{-\tilde{\kappa}}^{\infty}Dt\bigg(\tilde{\kappa}+t\bigg)
\end{split}
,
\end{equation}

where we have introduced $\tilde{\kappa}=\kappa/\sqrt{1-\gamma^{2}R^{2}}$,
an auxiliary normal variable $t\sim\mathcal{N}(0,1)$, and an induced
random variable $z\equiv t+\varepsilon w_{t}$ with induced distribution

\begin{equation}
\tilde{p}(z)=\int Dt\int dw_{t}p_{t}(w_{t})\delta(z-t-\varepsilon w_{t}).\label{appeq:induced_dist}
\end{equation}
Note that $w_{s}(z)$ can be determined implicitly by equating the
CDF of the induced variable $z$ and the distribution that the student
is constrained to: 
\begin{equation}
Q(w_{s})=\tilde{P}(z).
\end{equation}

\subsubsection*{Examples}

(1) Lognormal distribution

In the following, we solve $w_{s}(z)$ explicitly from the CDF equation
$Q(w_{s})=\tilde{P}(z)$. For a lognormal teacher,

\begin{equation}
p_{t}(w_{t})=\frac{1}{w_{t}}\frac{1}{\sqrt{2\pi}\sigma}\exp\left\{ -\frac{(\ln w_{t}-\mu)^{2}}{2\sigma^{2}}\right\} .
\end{equation}
The second moment constraint implies $\mu=-\sigma^{2}.$

The induced CDF of $z$ is
\begin{equation}
\tilde{P}(z)=\int_{-\infty}^{z}dz'\int_{-\infty}^{\infty}Dt\int_{0}^{\infty}dw_{t}p_{t}(w_{t})\delta(z'-t-\varepsilon w_{t}).
\end{equation}
Let $x=(\ln w-\mu)/\sigma,$

\begin{equation}
\begin{split}\tilde{P}(z)= & \int_{-\infty}^{z}dz'\int_{-\infty}^{\infty}Dt\int_{-\infty}^{\infty}Dx\delta(z'-t-\varepsilon e^{\mu+\sigma x})\\
= & \int_{-\infty}^{\infty}DxH(\varepsilon e^{\mu+\sigma x}-z)
\end{split}
.
\end{equation}
Now the CDF of $w_{s}$ is

\begin{equation}
\begin{split}Q_{s}(w_{s})= & \int_{-\infty}^{w_{s}}q_{s}(w)dw=H\left(-\frac{\ln w_{s}-\mu}{\sigma}\right)\end{split}
.
\end{equation}
Therefore, equating $\tilde{P}(z)$ and $Q_{s}(w_{s})$:
\begin{equation}
\int_{-\infty}^{\infty}DxH(\varepsilon e^{\mu+\sigma x}-z)=H\left(-\frac{\ln w_{s}-\mu}{\sigma}\right),
\end{equation}

We can solve for \textbf{$w_{s}(z)$} by (recall\textbf{ $z\equiv t+\varepsilon w_{t}$})
\begin{equation}
w_{s}(z)=\exp\left\{ \mu+\sigma H^{-1}\left(\int DxH(z-\varepsilon e^{\mu+\sigma x})\right)\right\} .
\end{equation}

Or in terms of error functions
\begin{equation}
w_{s}(z)=\exp\left\{ \mu+\text{\ensuremath{\sqrt{2}\sigma}erf}^{-1}\left(\int Dx\text{erf}\left(\frac{\varepsilon e^{\mu+\sigma x}-z}{\sqrt{2}}\right)\right)\right\} .
\end{equation}

We can also calculate the initial overlap (before any learning):
\begin{equation}
R_{0}=\left\langle \boldsymbol{w}_{t}\cdot\boldsymbol{w}_{s}\right\rangle _{p_{t}q_{s}}=e^{2\mu+\sigma^{2}}=e^{-\sigma^{2}}.
\end{equation}

(2) Uniform distribution 

Assuming that both the teacher and the student have a uniform distribution
in range $[0,\sigma].$

The second moment constraint fixes $\sigma=\sqrt{3}.$

We can solve (as in the lognormal example above), 

\begin{equation}
w_{s}(z)=\frac{1}{\varepsilon}\int_{-\infty}^{z}dz'\left(H(z'-\varepsilon\sigma)-H(z')\right).
\end{equation}

(3) Half-normal distribution

Assuming that both the teacher and the student has a half-normal distribution
$\frac{2}{\sqrt{2\pi}\sigma}\exp\left\{ -\frac{w^{2}}{2\sigma^{2}}\right\} $.

The second moment constraint fixes $\sigma=1$, and 

\begin{equation}
w_{s}(z)=\sigma H^{-1}\left\{ \frac{1}{2}-\int_{-\infty}^{\frac{z}{\sqrt{1+\sigma^{2}\varepsilon^{2}}}}DtH(-\sigma\varepsilon t)\right\} .
\end{equation}

\subsubsection*{Arbitrary number of synaptic subpopulations}

Just like in the case of Section \ref{app:capacity_M_pop}, we can
generalize our theory above to incorporate distribution constraints
with an arbitrary number of synaptic subpopulations. Let's consider
a student perceptron with $M$ synaptic populations indexed by $m$,
$\boldsymbol{w}^{m}$, such that each $w_{i}^{m}$ satisfies its own
distributions constraints $w_{i}^{m}\sim Q_{m}(w^{m})$. We denote
the overall weight vector as $\boldsymbol{w}\equiv\{\boldsymbol{w}^{m}\}_{m=1}^{M}\in\mathbb{R}^{N\times1}$.
The total number of weights is $N=\sum_{m=1}^{M}N_{m}$, and we denote
the fractions as $g_{m}=N_{m}/N$. Since the derivation is similar
to that of Section \ref{app:capacity_M_pop} and Section \ref{app:generalization_dist_const},
we will only present the results here.

As before, the order parameters $\left\{ R,\kappa,u,\varepsilon\right\} $
can be determined from a set of self-consistency equations:

\begin{equation}
\begin{split}u= & \sum_{m}g_{m}\langle w^{m}(z)t\rangle_{t,w_{t}}\\
R= & \sum_{m}g_{m}\langle w^{m}(z)w_{t}\rangle_{t,w_{t}}\\
\frac{1}{2}u^{2}= & \alpha\int_{-\infty}^{\kappa}DtH\bigg(-\frac{\gamma Rt}{\sqrt{1-\gamma^{2}R^{2}}}\bigg)(\kappa-t)^{2}\\
\varepsilon u= & \frac{2\alpha\gamma}{\sqrt{2\pi}}\sqrt{1-\gamma^{2}R^{2}}\int_{-\tilde{\kappa}}^{\infty}Dt\bigg(\tilde{\kappa}+t\bigg)
\end{split}
,
\end{equation}

where $\tilde{\kappa}=\kappa/\sqrt{1-\gamma^{2}R^{2}}$, $t\sim\mathcal{N}(0,1)$.
and an induced random variable $z\equiv t+\varepsilon w_{t}$ with
induced distribution the same as Eqn.\ref{appeq:induced_dist}.

Note that every $w^{m}(z)$ can be determined by equating the CDF
of the induced variable $z$ and the $m$-th distribution that $w^{m}(z)$
is constrained to: 
\begin{equation}
Q_{m}(w^{m})=\tilde{P}(z).
\end{equation}

\subsubsection{Sparsification of weights in sign-constraint learning}

\label{app:generalization_sparsification}

For unconstrained weights, max-margin solutions are considered beneficial
for generalization particularly for small size training sets. As a
first step toward biological plausibility, one can try to constraint
the sign of individual weights during learning (e.g., excitatory or
inhibitory). In the generalization error setup, we can impose a constraint
that the teacher and student have the same set of weight signs. Surprisingly,
we find both analytically and numerically that if the teacher weights
are not too sparse, the max-margin solution generalizes poorly: after
a single step of learning (with random input vectors), the overlap,
$R$, drops substantially from its initial value $R_{0}$ (by a factor
of $\sqrt{2}$ for a half-Gaussian teacher, see the blue curves in
Fig.\ref{appfig:sparsification}(a). 

We can verify this by calculating $R_{0}$ in two different ways.
As an example, in the following we consider the case where both the
teacher and student have half-normal distributions. 

(1) By definition, the overlap is $R=\frac{\boldsymbol{w}_{s}\cdot\boldsymbol{w}_{t}}{\left\Vert \boldsymbol{w}_{s}\right\Vert \left\Vert \boldsymbol{w}_{t}\right\Vert }$.
Since $\boldsymbol{w}_{s}$ and $\boldsymbol{w}_{t}$ are uncorrelated
before learning ($\alpha=0$), the initial overlap is then $R_{0}=\frac{\left\langle w_{s}\right\rangle \left\langle w_{t}\right\rangle }{||\boldsymbol{w}_{s}||\boldsymbol{w}_{t}||}$$=\frac{2}{\pi};$

(2) Take the $\alpha\to0$ limit in Eqn.\ref{appeq:sc_saddle_1-3_qneq1}
and Eqn.\ref{appeq:sc_saddle_4-5_qneq1} and calculate $R_{0+}=\lim_{\alpha\to0+}R(\alpha)$
$=\frac{\sqrt{2}}{\pi}$. 

Therefore, in this example $R_{0+}=R_{0}/\sqrt{2}.$

The source of the problem is that due to the sign constraint, max-margin
training with few examples yields a significant mismatch between the
student and teacher weight distributions. After only a few steps of
learning, half of the student's weights are set to zero, and the student's
distribution, $p(w_{s})=\frac{1}{2}\delta(0)+\frac{1}{\sqrt{2\pi}}\exp\{-\frac{w_{s}^{2}}{4}\}$,
deviates significantly from the teacher's half-normal distribution
(Fig.\ref{appfig:sparsification}(b)).

\begin{figure}
\centering{}\includegraphics[scale=0.18]{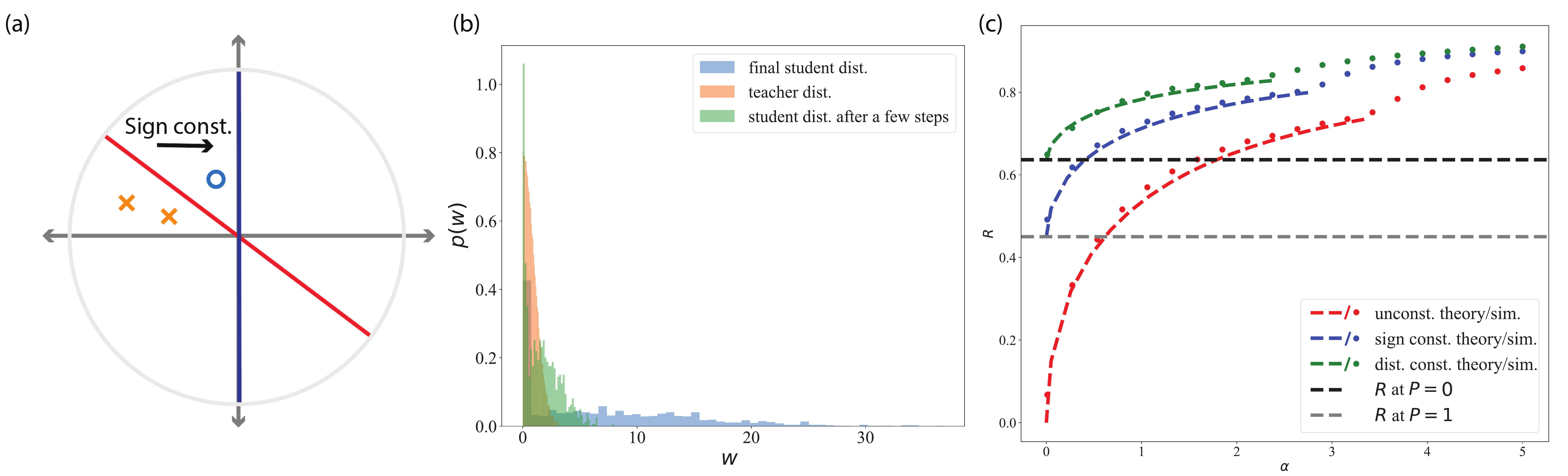}\caption{\label{appfig:sparsification} Sparsification of weights in sign-constraint
learning. (a) An illustration of weight sparsification. In this schematic,
the perceptron lives on this $1$-dimensional circle and $N=2$. Red
line denotes the hyperplane orthogonal to the perceptron weight before
sign-constraint, crosses and circles indicate examples in different
classes. Sign-constraint pushes the weights to the first quadrant,
which zeros half of the weights on average. Blue line indicates the
hyperplane obtained after the sign-constraint. (b) Sparsification
of weights due to max-margin training. After only a few iterations,
nearly half of the student weights are set to zero, and the distribution
deviates significantly from the teacher's distribution. (c) Teacher-student
overlap as a function of load $\alpha$ for different learning paradigms.
Dashed lines are from theory, and dots are from simulation. Note the
horizontal dashed lines show the initial drop in overlap from zero
example and to just a single example. In this case teacher has nonzero
noise, $\gamma=0.85$. }
\end{figure}

\subsubsection{Noisy teacher}

\label{app:generalization_noise}

We generate examples $\{\boldsymbol{\xi}^{\mu},\zeta^{\mu}\}_{\mu=1}^{P}$
from a teacher perceptron, $\boldsymbol{w}_{t}\in\mathbb{R}^{N}$:
$\zeta^{\mu}=\text{sgn}(\boldsymbol{w}_{t}\cdot\boldsymbol{\xi}^{\mu}/||\boldsymbol{w}_{t}||+\eta^{\mu})$,
where $\eta^{\mu}$ is input noise and $\eta^{\mu}\sim\mathcal{N}(0,\sigma^{2})$.
In this subsection we present additional numerical results for the
case when $\sigma\neq0$. As in previous sections, we define the noise
level parameter $\gamma=1/\sqrt{1+\sigma^{2}}.$

Our theory's prediction is confirmed by numerical simulation for a
wide range of teacher noise level $\gamma$ and teacher weight distributions
$P_{t}(w_{t})$. We find that distribution-constrained learning performs
consistently better all the way up to capacity (capacity in this framework
is due to teacher noise). For illustration, in Fig.\ref{appfig:vary_noise}
we show theory and simulation for fixed prior learning of three different
teacher distributions: uniform, half-normal, and lognormal.

\begin{figure}
\centering{}\includegraphics[scale=0.18]{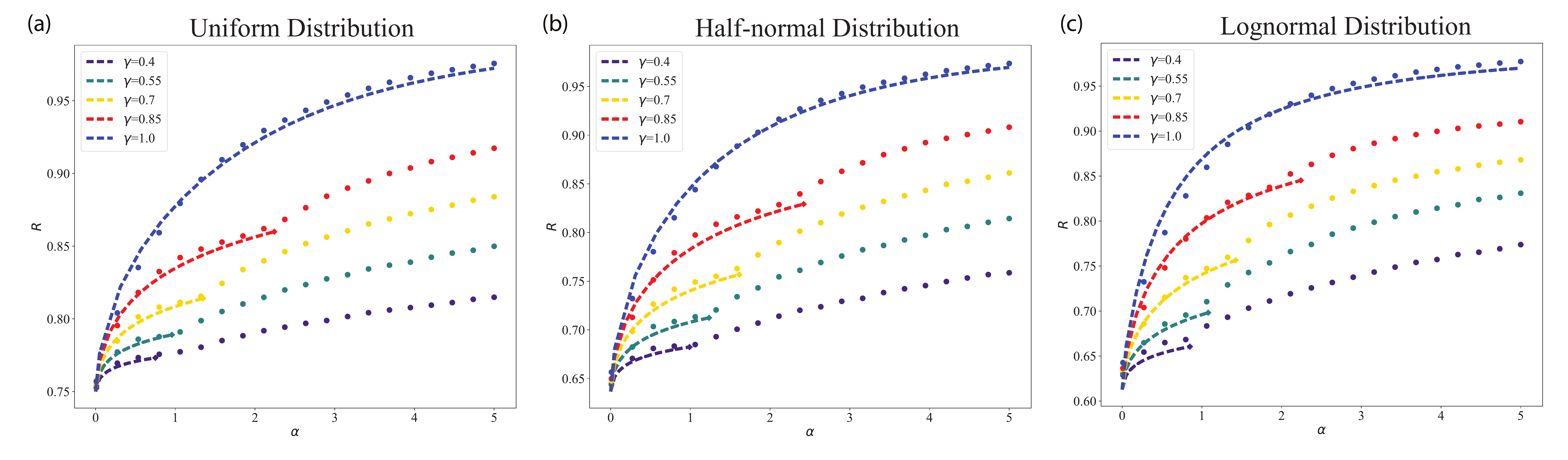}\caption{\label{appfig:vary_noise}Generalization (measured by overlap) performance
for different distributions and different noise levels in fixed prior
learning. From left to right: uniform, half-normal, and lognormal
distribution. In all cases the student is constrained to have the
same distribution as that of the teacher's. Dashed lines are from
theory and dots are from DisCo-SGD simulation.}
\end{figure}

\subsection{DisCo-SGD simulations}

\label{app:DiscoSGD}

\subsubsection*{Avoid vanishing gradients}

Note that we often observe a vanishing gradient in DisCo-SGD when
we choose a constant learning rate $\eta_{1}$, and in such cases
the algorithm tends to find poor margin $\kappa$ which deviates from
the max-margin value predicted from the theory. We find that scaling
$\eta_{1}$ with the standard deviation of the gradient solves this
problem:

\begin{equation}
\eta_{1}=\eta_{1}^{0}/\text{std}\left(\sum_{\mu}\xi_{i}^{\mu}(\hat{\zeta}^{\mu}-\zeta^{\mu})\right),
\end{equation}

where the standard deviation is computed across the synaptic index
$i$ and $\eta_{1}^{0}$ is a constant. 

\subsubsection*{Mini-batches}

For the capacity simulations, we always use full-batch in the SGD
update, so it is in fact simply gradient descent. However, in the
case of generalization, we find that training with mini-batches improves
the generalization performance, since it acts as an source of stochasticity
during training. In main text Fig.5 we use mini-batch size $B=0.8P$
($80\%$ of examples are used for each SGD update).

When we vary teacher's noise level, we find that scaling $B$ with
$\gamma$ improves the quality of the solutions, as measured by the
generalization performance (or equivalently, the teacher-student overlap).
Generally, the more noisy the teacher is, the smaller the mini-batches
should be. This is because smaller mini-batch size corresponds to
higher stochasticity, which helps overcoming higher teacher noise. 

\subsubsection*{Parameters}

All the capacity simulations are performed with the following parameters
$N=1000,\eta_{1}^{0}=0.01,\eta_{2}=0.6,t_{max}=10000,$ where $t_{max}$
is the maximum number of iterations of the DisCo-SGD algorithm. 

All results are averaged over 300 realizations.

In main text Fig.4, the experimental \cite{levy2012spatial} parameters
are $g_{E}=45.8\%,\sigma_{E}=0.833,\sigma_{I}=0.899$.

In main text Fig.5(a): We show the teacher-student overlap as a function
of $\alpha$. Dots are simulations performed with series of student
distribution from $\sigma_{s}=0.1$ to $\sigma_{s}=1.4$, where the
teacher distribution sits in the middle of this range, $\sigma_{t}=0.7$.
Each such simulation is performed with fixed $\sigma_{s}$ and varying
load $\alpha\in[0.05,2.5].$ In main text Fig.5(b): we show the empirical
weight distributions found by unconstrained perceptron learning for
$\alpha\in[0.05,10]$. In main text Fig.5(c) we show optimal student
distribution for $\alpha\in[0.05,2.5]$. Note that optimal prior learning
approaches the teacher distribution much faster than unconstrained
learning. 

All the generalization DisCo-SGD simulations are performed with the
same parameter as in the capacity DisCo-SGD simulations, but with
two additional parameter: teacher's noise level $\gamma$ and SGD
mini-batch size $B$. 

For the simulations in Fig.\ref{appfig:vary_noise} we use

$\gamma=0.4,B=0.2P;\gamma=0.55,B=0.4P;\gamma=0.7,B=0.6P;\gamma=0.85,B=0.8P;\gamma=1.0,B=P$
(noiseless case).

\subsection{Replica symmetry breaking}

\label{app:Replica-symmetry-breaking}

\subsubsection{Bimodal distributions}

In deriving the capacity formula, we have assumed replica-symmetry
(RS). It is well-known that replica-symmetry breaking occurs in the
Ising perceptron \cite{penney1993weight,bouten1998learning}, so it
is natural to ask to what extent our theory holds when approaching
the Ising limit. Let's consider a bimodal distribution with a mixture
of two normal distributions with non-zero mean centered around zero,

\[
p(w)=\frac{1}{2}\mathcal{N}(-\mu,\sigma)+\frac{1}{2}\mathcal{N}(\mu,\sigma)
\]
The second moment constraint requires $\mu^{2}+\sigma^{2}=1$.

We can gradually decrease the Gaussian width $\sigma$, or equivalently
$\mu=\sqrt{1-\sigma^{2}}$ (which we call `separation' in the following)
and compare the capacity theoretically predicted by the RS theory
and numerically found by the DisCo-SGD algorithm.

\begin{figure}
\centering{}\includegraphics[scale=0.32]{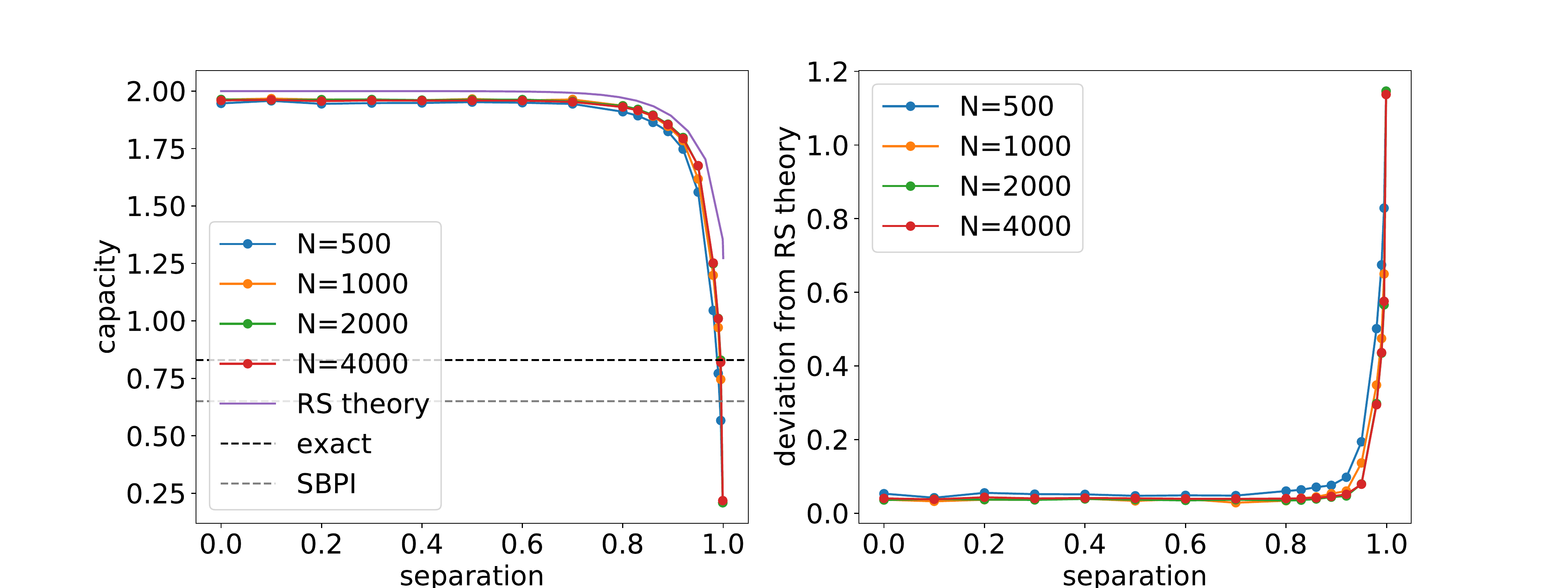}\caption{\label{appfig:replica_symmetry_breaking}Left: Capacity as a function
of separation for different size perceptrons. Dots are from DisCo-SGD
simulations and the `RS theory' line is from our theory. Exact values
for Ising perceptron and state-of-the-art numerical values are included
as well. Right: Deviation from the RS theory as a function of separation.
This is the same as subtracting the simulation values from the theoretical
predictions in the left figure. }
\end{figure}

In Fig.\ref{appfig:replica_symmetry_breaking} we can see that the
simulation agrees well with the RS theory until one gets very close
to the Ising limit ($\mu=1$). To understand finite size effects,
we extrapolate to the infinite size limit ($N\to\infty$) in Fig.\ref{appfig:finite_size_effects},
and found that the deviation from RS theory has a sharp transition
near $\mu=1$, marking the breakdown of the RS theory.

\begin{figure}
\centering{}\includegraphics[scale=0.32]{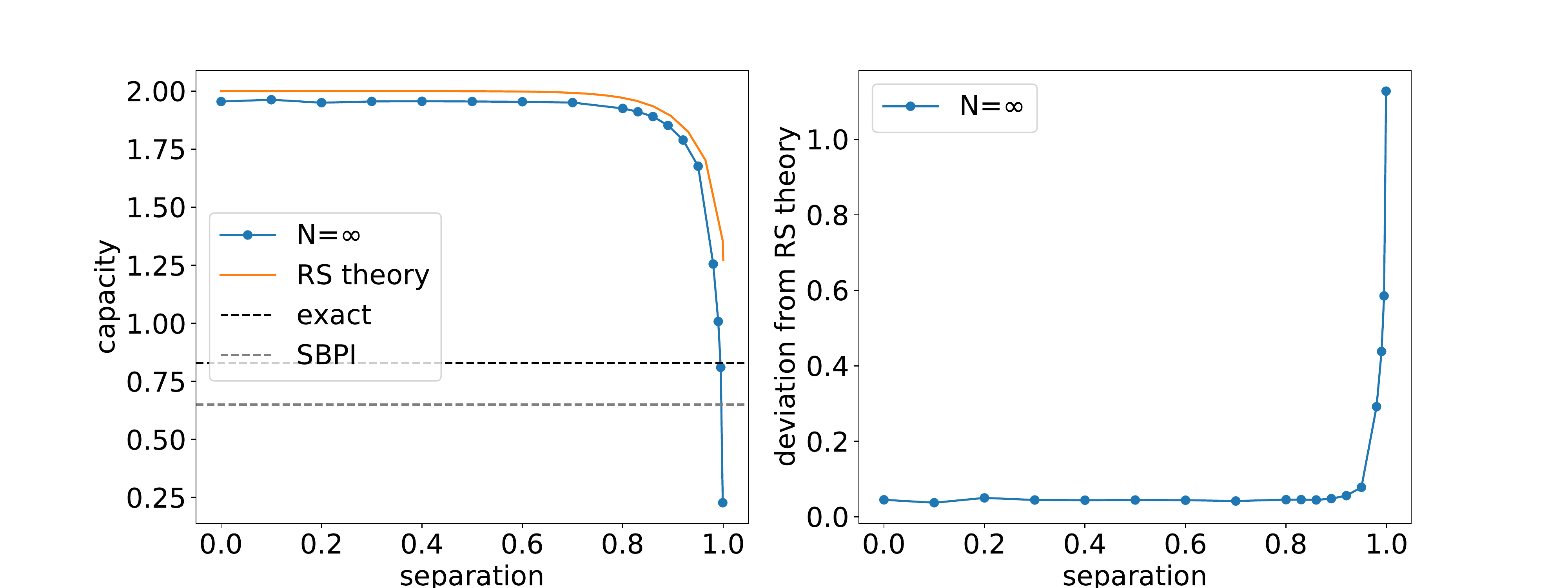}\caption{\label{appfig:finite_size_effects}Finite size effects. Left/Right:
we extrapolate simulation values in Fig.\ref{appfig:replica_symmetry_breaking}
Left/Right to infinite $N$. }
\end{figure}

\subsubsection*{Ising perceptron}

It is also interesting to compare our distribution-constrained RS
theory to the unconstrained RS theory. In this Ising limit,

\begin{equation}
q(w)=\frac{1}{2}\delta(w-1)+\frac{1}{2}\delta(w+1),
\end{equation}

and CDF
\begin{equation}
Q(w)=\frac{1}{2}\Theta(w-1)+\frac{1}{2}\Theta(w+1).
\end{equation}

Equating $Q(w)$ with the normal CDF $P(x)$ and solve for $w(x)$,
we find $w(x)=\text{sgn}(x)$. Then $dw/dx=2\delta(x)$ and $\left\langle \frac{dw}{dx}\right\rangle _{x}=\frac{2}{\sqrt{2\pi}}$.
Therefore,

\begin{equation}
\lim_{Ising}\alpha_{c}(\kappa=0)=\frac{4}{\pi},
\end{equation}

which is exactly the same as the prediction from the unconstrained
RS theory \cite{penney1993weight,bouten1998learning}. In contrast,
the exact capacity of Ising perceptron with replica-symmetry breaking
is $\alpha_{c}\approx0.83$. For comparison, we have included these
values in Fig.\ref{appfig:finite_size_effects}(a), as well as the
capacity found by the state-of-the-art supervised learning algorithm
(Stochastic Belief Propagation, SBPI \cite{baldassi2007efficient})
for Ising perceptron.

\subsubsection{Sparse distributions}

For a teacher with sparse distribution, $p(w_{t})=(1-\rho)\delta(w_{t})+\text{\ensuremath{\frac{\rho}{\sqrt{2\pi}\sigma_{t}w_{t}}\exp\left\{ -\frac{(\ln w_{t}-\mu_{t})^{2}}{2\sigma_{t}^{2}}\right\} }. }$We
found that the simulations start to deviate from the theory, and the
reason might be due to replica symmetry breaking. In Fig.\ref{appfig:sparse_teacher},
we use the optimal prior learning paradigm similar to main text Fig.5.
We see that our RS theory no longer gives accurate prediction of overlap
in this case.

\begin{figure}
\centering{}\includegraphics[scale=0.3]{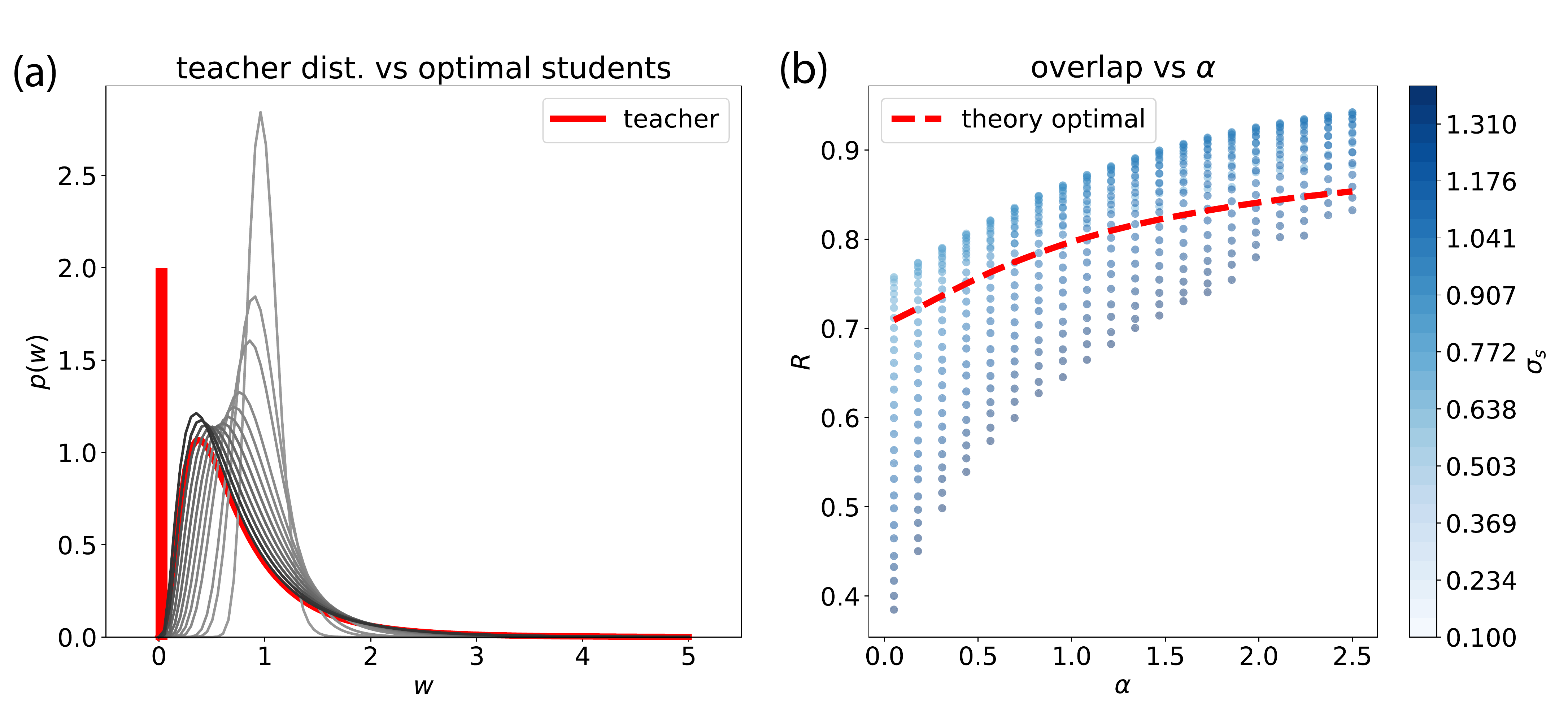}\caption{\label{appfig:sparse_teacher}Optimal student prior distribution as
a function of $\alpha$. (a) Gray curves correspond to a series of
optimal student distributions as a function of $\alpha$, with the
darker color representing larger $\alpha$. Red is teacher distribution.
(b) Overlap as a function of $\alpha$ for different student priors.
Red dashed line is the optimal overlap calculated from our replica-symmetric
theory. Dots are from DisCo-SGD simulations. For the same $\alpha$,
different color dots represent different overlaps obtained from simulations
with different $\sigma_{s}$. }
\end{figure}

%auto-ignore
\chapter{Generative modeling by feedforward neural networks}
\label{vae}
\section{Introduction}

Although data sampled from the natural world appear to be high-dimensional, their variations can usually be explained using a much smaller number of latent factors. Both biological and artificial information processing systems exploit such structure and learn explicit representations that are faithful to data generative factors, known commonly as disentangled representations \cite{bengio2013representation}.  For example, sparse coding, an influential model of the primary visual cortex, proposes that the visual cortex neurons are coding for latent variables of natural scenes: oriented edges \cite{olshausen1997sparse}. A very popular method of extracting latent variables is by using the bottleneck neurons of deep autoencoders \cite{hinton1994autoencoders,alemi2018information}. 
In this paper, we examine unsupervised learning of disentangled representations  in the context of variational inference and a generalization of the Variational Autoencoder (VAE)  \cite{kingma2013auto}, $\beta$-VAE, developed specifically for disentangled representation learning \cite{higgins2017beta}.

%Given their ubiquity, natural questions follow as to what the desirable utility of latent variables is in modelling tasks. Perhaps the most desirable quality is that latent variables should be able to generate data similar to the observed distribution, mapping to the true data generative factors in the real world. More formally, there should exist a distribution of latents $p(z)$ that, upon sampling, can generate data conditioned on the latent distribution that will exceedingly similar to the true data distribution. This relationship can be expressed as maximum likelihood using the law of total probability. 

We will adopt a probabilistic framework for latent-variable modeling of data \cite{kingma2019introduction}, where a generative model $p_{\btheta}(\x,\z)$ for data $\x$ and latent variables $\z$ is assumed:
\begin{align}
    \label{eqn:generative}
    p_{\btheta}(\x,\z) = \ p_{\btheta}(\x|\z)p(\z).
\end{align}
Here %$\z$ is a vector whose elements contain the latent variables $\lbrace z_1,\ldots,z_k\rbrace$, $\Omega$ is the support of $\z$, and 
$\btheta$ denotes the parameters of our model, $p_{\btheta}(\x|\z)$ models the stochastic process that generates the data given the latent variables, and $p(\z)$ is the prior on the latent variables. An interpretable and common choice for $p(\z)$, and the subject of our paper, is a factorized distribution $p(\z) = \prod_{i=1}^k p_i(z_i)$, which implies statistical independence. Examples of models with independent priors include popular methods such as Independent Component Analysis \cite{hyvarinen2000independent,khemakhem2019variational} and Principal Component Analysis \cite{tipping1999probabilistic}.

%Here, $p(z)$ represents the prior distribution and $p_{\theta}(x)$ is our true data distribution. Combining this probabilistic perspective with the independent coding definition introduced earlier, the intuition that naturally follows is that the latent variables should be uncorrelated. However, this is generally not true. In this paper, we will examine this tension in the context of the $\beta$-VAE architecture and demonstrate that Statistical decorrelation is not compatible with Bayesian inference. We introduce a general theory of how  the introduction of the parameter $\beta$ affects the dynamics of the variational inference framework, and examine it in both solvable linear models and more realistic, real world simulations.

While a common definition of learning disentangled representations has yet to be agreed upon \cite{bengio2013representation,kingma2013auto,locatello2018challenging,higgins2018towards}, extracting statistically independent latent factors is a natural choice \cite{bengio2013representation,hyvarinen2000independent} and is the definition we will adopt. Such a representation is efficient in that it carries no redundant information \cite{dayan2001theoretical}, and at the same time sufficient information to generate the data. 

In our probabilistic framework, the model posterior distribution $p_{\btheta}(\z|\x)$ allows inference of true latent variables. In principle, this could be used to form disentangled representations. However, model posterior is often intractable \cite{kingma2019introduction}, and variational methods are used to estimate it.

We focus on a state-of-the-art variational inference method for learning disentangled representations, $\beta$-VAE \cite{higgins2017beta}. The $\beta$-VAE training objective includes a hyperparameter, $\beta$, encapsulating the original VAE \cite{kingma2013auto} as a special case with choice $\beta = 1$. When $\beta$ is larger than unity, {\it conditional} independence of the learned representations at the bottleneck layer are enforced, corresponding to a conditional independence assumption on data generating latent variables, i.e. $p(\z|\x) = \prod_{i}p_i(z_i|\x)$ \cite{higgins2017beta}. However, as pointed above, a more natural assumption on latents is full statistical independence. Further, statistically independent latents are in general not conditionally independent. Given the popularity of VAEs in representation learning, it is important to understand the role of the $\beta$ hyperparameter in learning disentangled (statistically independent) latent variables.

Our main contributions are as follows:
\begin{enumerate}
    \item We provide general results about variational Bayesian inference in  $\beta$-VAE. Specifically, we prove that the $\beta$-VAE objective is non-increasing with increasing $\beta$, leading to worse reconstruction performance but more conditionally independent  representations.  Further, we argue that  latent variable inference performance generally tends to be non-monotonic in $\beta$.
    \item We introduce an analytically tractable model for $\beta$-VAE, specializing to statistically independent latent generative factors. %In this model, we can calculate expectations values with respect to data distribution exactly, thus enabling us to explicitly calculate all quantities of interest in the context of $\beta$-VAE. 
    We analytically calculate the optimality conditions for this model, and numerically find that there is an optimal $\beta$ for the best inference of latent variables.
    \item We test our insights from the general theorems and the analytically tractable model using a realistic $\beta$-VAE architecture, using a synthetic MNIST dataset. Simulations agree well with our theory. 
\end{enumerate}

The rest of this paper is organized as follows. In Section \ref{bVAE}, we provide a review of variational inference and $\beta$-VAE. In Section \ref{general}, we  prove several theorems about variational inference in the context of $\beta$-VAE. In Section \ref{analytical}, we introduce our analytical results. In Section \ref{numerical}, we test our insights from the general theorems and the tractable models using a $\beta$-VAE architecture on a synthetic MNIST dataset. Finally, in Section VI we discuss our results and present our conclusions.

\section{Variational Inference and $\beta$-VAE}\label{bVAE}

Inference of latent variables in probabilistic models is often an intractable calculation \cite{kingma2013auto,kingma2019introduction}. %In principle, the posterior distribution can be calculated from Bayes’ Rule 
%
%\begin{align}
%    \label{eqn:bayes}
%    p_{\btheta}(\z|\x) = \frac{p_{\btheta}(\x|\z)p(\z)}{p_{\btheta}(\x)},
%\end{align}
%
%however this is difficult because the data distribution $p_{\btheta}(\x)$ is often intractable \cite{kingma2019introduction}. %
Variational methods instead optimize over a set of tractable distributions, $q_{\bphi}(\z|\x)$, that best approximates $p_{\btheta}(\z|\x)$. We will refer to $q_{\bphi}(\z|\x)$ as the inference model. The difference between the two distributions can be quantified using the Kullback-Leibler (KL) divergence, which we call Model Inference Error (MIE):
\begin{align}
\label{eqn:model inference error}
   {\rm MIE} \equiv \mathbb{E}_{p(\mathbf{x})}\left[ D_{KL}(q_{\bphi}({\bf z}|{\bf x}) || p_{\theta}({\bf z}|{\bf x} ) )\right] . 
\end{align}
We distinguish between MIE and the True Inference Error (TIE),
\begin{align}
\label{eqn:true inference error}
   {\rm TIE} \equiv \mathbb{E}_{p(\mathbf{x})}\left[ D_{KL}(q_{\bphi}({\bf z}|{\bf x}) || p_{\text{g-t}}({\bf z}|{\bf x} ) )\right],
\end{align}
which can only be known when one has access to the underlying `ground-truth' data generative process and the ground-truth posterior, $p_{\text{g-t}}({\bf z}|{\bf x} )$.

VAEs fit the parameters of the probabilistic model and the variational distribution simultaneously.
A key identity in doing so is \cite{jordan1999introduction} 
\begin{align}
\label{eqn:main}
        \ln p_{\btheta}({\bf x}) &- D_{KL}(q_{\bphi}(\z|\x) || p_{\btheta}(\z|\x ) ) \nonumber \\
        &= \mathbb{E}_{q_{\bphi}(\z|\x)}\left[\ln p_{\btheta}(\x|\z)\right] -D_{KL}(q_{\bphi}(\z|\x) \| p(\z)).
\end{align}
 Model fitting is done by maximizing the data log-likelihood, $\ln p_{\btheta}(\x)$, under model parameters. Because the KL divergence is non-negative, the right hand side of \eqref{eqn:main} serves as a lower bound for $\ln p_{\btheta}(\x)$ and is called the Evidence Lower Bound (ELBO)
\begin{align}\label{cost}
    &{\rm ELBO}(\btheta,\bphi) \nonumber \\ &\qquad \equiv \mathbb{E}_{q_{\bphi}(\z | \x)}\left[\log p_{\btheta}(\x | \z)\right]- D_{KL}\left(q_{\bphi}(\z|\x) \| p(\z)\right).
\end{align}
%
% The difference between the log likelihood of the data distribution and the ELBO is Inference Error. Rearranging our terms, we can see that the ELBO satisfies the following. 
%
% \begin{align}
%     \ln p({\bf x}) = D_{KL}(q_{\phi}({\bf z}|{\bf x}) || p_{\theta}({\bf z}|{\bf x} ) ) + \mathcal{L}
% \end{align}
%
VAE parameterizes the distributions $p_{\btheta}(\x|\z)$ and $q_{\bphi}(\z|\x)$ with neural networks, and maximizes ELBO as a proxy for maximizing the data likelihood.

The neural network realization of the $p_{\btheta}(\x|\z)$ is referred to as a decoder \cite{kingma2013auto}. Once the VAE is trained, the decoder can be used as to generate new samples from the model data distribution \cite{kingma2013auto,doersch2016tutorial}. The term $\mathbb{E}_{q_{\bphi}(\z | \x)}\left[\log p_{\btheta}(\x | \z)\right]$ measures the reconstruction performance of the generative model. We will refer to it as the reconstruction objective.

The neural network realization of the inference model is referred to as an encoder \cite{kingma2013auto}. Its outputs constitute a bottleneck layer and  represent inferred latent variables. Note that the MIE calculated from this representation appears on the left hand side of \eqref{eqn:main}. 

$\beta$-VAE is an extension of the traditional VAE, where an extra, adjustable hyperparameter $\beta$ is placed in the training objective: 
\begin{align}\label{bVAEobj}
    \mathcal{L}(\btheta,\bphi;\beta) = \mathbb{E}_{q_{\bphi}(\z | \x)}\left[\log p_{\btheta}(\x | \z)\right]-\beta D_{K L}\left(q_{\bphi}(\z | \x) \| p(\z)\right).
\end{align}
Specifically, when  $\beta = 1$, the  $\beta$-VAE is equivalent to VAE and ${\rm ELBO}(\btheta,\bphi) = \mathcal{L}(\btheta,\bphi;1)$. 

Higher values of $\beta$ emphasizes the KL divergence between the inference model $q_{\bphi}(\z|\x)$ and the independent prior $p(\z)$ in the objective \eqref{bVAEobj}. Smaller values of the KL divergence favor a conditionally independent inference model. This can be used to learn disentangled representations of conditionally independent latent variables, whose probability distributions factorize when conditioned on data .

However, as alluded to in our introduction, in many cases of interest and application \cite{Huang_2018_ECCV,lample2017fader,karras2019style}, latent variables are conditionally dependent while being independent \cite{hyvarinen2000independent},\cite{tipping1999probabilistic}. We will encounter an analytically tractable case in Section \ref{analytical}. In such cases, it is not clear if a $\beta$ different than 1 helps learning a disentangled representation which extracts statistically independent latent factors. Our goal in the remaining of this paper is to examine this case analytically and numerically.

For convenience, we also attach a table of terms and corresponding mathematical expressions used throughout the paper (Table \ref{table:1}).

\begin{table}[h!]
\centering
\scalebox{1.02}{
\def\arraystretch{1.85}
\begin{tabular}{ |c|c|c| } 
\hline
\textbf{Term} & \textbf{Mathematical Expression} \\
\hline
Prior & $p(\z)$ \\ 
\hline
Model Posterior & $p_{\btheta}(\z|\x)$ \\ 
\hline
Ground-Truth Posterior & $p_{\text{g-t}}(\z|\x)$ \\ 
\hline
% Generator model & $p_{\btheta}(\x|\z)$ \\ 
% \hline
Inference Model & $q_{\bphi}(\z|\x)$ \\ 
\hline
Data Log-Likelihood & $\log p_{\btheta}(\x)$ \\ 
\hline
Reconstruction Objective & $\mathbb{E}_{q_{\bphi}(\z | \x)}\left[\log p_{\btheta}(\x | \z)\right]$ \\ 
\hline
\shortstack{Conditional Independence Loss} & $D_{K L}\left(q_{\bphi}(\z | \x) \| p(\z)\right)$  \\ 
\hline
$\rm{MIE}$ & $\mathbb{E}_{p(\mathbf{x})}[ D_{KL}(q_{\bphi}({\bf z}|{\bf x}) || p_{\theta}({\bf z}|{\bf x} ) )]$  \\ 
\hline
$\rm{TIE}$ & $\mathbb{E}_{p(\mathbf{x})}[ D_{KL}(q_{\bphi}({\bf z}|{\bf x}) || p_{\text{g-t}}({\bf z}|{\bf x} ) )]$  \\ 
\hline
\shortstack{Evidence Lower Bound \\ (ELBO)}  & \shortstack{$\mathbb{E}_{q_{\bphi}(\z | \x)}\left[\log p_{\btheta}(\x | \z)\right]$ \\ $- D_{K L}\left(q_{\bphi}(\z | \x) \| p(\z)\right)$ }\\ 
\hline
\end{tabular}}
\vskip 0.1in
\caption[Table caption text]{Table of terms and corresponding mathematical expressions.}
\label{table:1}
\end{table}

\section{How $\beta$ Affects Model Performance and Inference of Latent Variables}\label{general}

In this section, we provide general statements on the effect of the $\beta$ parameter on the representation learning and the generative functions of  $\beta$-VAE. We do this by proving propositions about how various terms in the identity \eqref{eqn:main} change as a function of $\beta$. Our first two propositions imply that increasing $\beta$ worsens the quality of reconstructed samples while improving conditional disentangling. While these points have been shown in simulations \cite{higgins2017beta,burgess2018understanding}, here we provide analytical statements. Our last proposition gives a handle on understanding behavior of MIE through ELBO.

% We can denote our VAE objective as
% \begin{align}
%     O(\theta,\phi,\beta) = \mathbb{E}_{q_{\phi}({\bf z}|{\bf x})}\left(\ln p_{\theta}({\bf x}|{\bf z})\right) -\beta D_{KL}(q_{\phi}({\bf z}|{\bf x}) || p({\bf z}))
% \end{align}

In the following, we will denote optimal parameters of a $\beta$-VAE that maximizes the objective \eqref{bVAEobj} by $\btheta^*$ and $\bphi^*$. They are given as a solution to 
\begin{align}\label{opt}
    \frac{\partial \mathcal{L}}{\partial \btheta} = \bm{0}, \qquad \frac{\partial \mathcal{L}}{\partial \bphi} = \bm{0}.
\end{align}
We denote the value of the optimal objective by
\begin{align}
    \mathcal{L}^*(\beta) \equiv \mathcal{L}(\btheta^*(\beta),\bphi^*(\beta),\beta),
\end{align}
and the value of ELBO at the optimal point by
\begin{align}
{\rm ELBO}^*(\beta) \equiv {\rm ELBO}  (\btheta^*(\beta),\bphi^*(\beta)).  
\end{align}

Our first proposition concerns the behavior of $\mathcal{L}^*(\beta)$ as a function of $\beta$.

\begin{proposition}\label{prop1}
The optimal value of the $\beta$-VAE objective, $\mathcal{L^*}(\beta)$, is non-increasing with increasing $\beta$:
\begin{align}
    \frac{\partial \mathcal{L^*}(\beta) }{\partial \beta} = - D_{K L}\left(q_{\bphi^*}(\mathbf{z} | \mathbf{x}) \| p(\mathbf{z})\right) \leq 0.
\end{align}
\end{proposition}
\begin{proof}
Follows from an application of the chain rule, the optimality conditions \eqref{opt}, and the nonegativity of the KL-divergence:
\begin{align}
    \frac{\partial \mathcal{L^*} }{\partial \beta} &= \left.\left(\frac{\partial \mathcal{L} }{\partial \btheta}\cdot \frac{\partial \btheta }{\partial \beta} + \frac{\partial \mathcal{L} }{\partial \bphi}\cdot \frac{\partial \bphi }{\partial \beta} + \frac{\partial \mathcal{L} }{\partial \beta}\right)\right|_{\btheta =\btheta^*, \bphi =\bphi^* } \nonumber \\
    &= -D_{KL}(q_{\bphi^*}({\bf z}|{\bf x}) || p({\bf z})) \leq 0.
\end{align}
\end{proof}

The next proposition shows how the two terms in $\mathcal{L}^*$ change with $\beta$. 

\begin{proposition}\label{prop2}  The KL divergence between the inference model and the prior is non-increasing with increasing $\beta$:
\begin{align}\label{eq:qp}
\frac{d}{d\beta} D_{KL}(q_{\phi^*}({\bf z}|{\bf x}) || p({\bf z})) \leq 0.
\end{align}
Together with Proposition \eqref{prop1}, this implies that
\begin{align}\label{eq:re}
\frac{d \, \mathbb{E}_{q_{\bphi^*}(\z | \x)}\left[\log p_{\btheta^*}(\x | \z)\right]}{d\beta} \leq 0
\end{align}
\end{proposition}
\begin{proof}
See Section \ref{pprop2}.
\end{proof}

The next proposition is about the behavior of ${\rm ELBO}^*$.

\begin{proposition}\label{prop3}
${\rm ELBO}^*$ is maximized at $\beta = 1$. %More specifically,
%
%\begin{align}
%   \frac{\partial{\rm ELBO}^*(\beta)}{\partial \beta} = (\beta-1)\frac{d}{d\beta} D_{KL}(q_{\bphi^*}(\z|\x) || p(\z)).
%\end{align}
\end{proposition}
\begin{proof} Note that by definition 
\begin{align}\label{lelbo}
\mathcal{L} = {\rm ELBO} + (1-\beta)D_{KL}(q_{\bphi}(\z|\x) || p({\bf z})).
\end{align}
By evaluating \eqref{lelbo} at $\btheta=\btheta^*$ and $\bphi=\bphi^*$, and the chain rule, we get:
\begin{align}
\frac{d\,{\rm ELBO^*(\beta)}}{d\beta}&= \frac{d}{d\beta} \left[\mathcal{L}^* -(1-\beta)D_{KL}(q_{\bphi^*}(\z|\x) || p({\bf z}))\right] \nonumber \\
%&=(-1+1)D_{KL}(q_{\phi^*}({\bf z}|{\bf x}) || p({\bf z})) - \frac{d}{d\beta}D_{KL}(q_{\phi^*}({\bf z}|{\bf x}) || p({\bf z})) \nonumber \\
&= (\beta-1) \frac{d}{d\beta}D_{KL}(q_{\bphi^*}({\bf z}|{\bf x}) || p({\bf z})).
\end{align}
The proposition follows from this result and \eqref{eq:qp}. 
\end{proof}

For simplicity of notation, we presented most of our formulas and propositions for a single data point. All our results generalize to the case where one averages over the data distribution $p(\x)$, or a finite training set. 

Inference of latent variables, measured by MIE, is affected by $\beta$ as well. In the $\beta \to \infty$ limit the inference model becomes more and more conditionally independent, deviating from the model posterior. Is the behavior monotonic? While MIE is not explicitly calculable, we can get a hint of its behavior by rearranging \eqref{eqn:main}, and evaluating it at the optimal $\beta$-VAE parameters:
\begin{align}\label{MIEELBO}
{\rm MIE}(\beta) = \Ep\left[\ln p_{\btheta^*(\beta)}(\x)-{\rm ELBO}^*(\beta)\right].
\end{align}
As reconstruction performance worsens with $\beta$, it is reasonable to expect that the data likelihood decreases with $\beta$. Because ELBO is non-monotonic with a maximum, even if the data log-likehood was monotonic with $\beta$, we can expect a non-monotonic behavior of MIE with an optimal value.  In the next section, we will see two specific examples of this.

\section{Analytical Results}\label{analytical}

In this section we demonstrate our general theory for two different analytically tractable cases. 

\subsection{$\beta$-VAE with a fixed decoder does not lead to better disentangling}
A simple case is when the decoder of the $\beta$-VAE is not trained. In our notation, this amounts to $\btheta$ being fixed. Then the $\beta$-VAE objective \eqref{bVAEobj} only trains the encoder network and the inference model, $q_{\bphi}(\z\|\x)$. We can deduce the behavior of MIE as a function of $\beta$ from \eqref{MIEELBO}. The data likelihood, $p_{\btheta}(\x)$, does not change as a function of training. ${\rm ELBO}^*$ is maximized at $\beta = 1$ from Proposition \ref{prop3}, which can be seen to apply to fixed $\btheta$. This means MIE is minimum at $\beta = 1$. In this case, $\beta=1$, or the original VAE is best at learning the true latent variables.

\subsection{Optimal $\beta$ values in an analytically tractable model}

Next, we present a tractable VAE model, in which we can explicitly calculate the $\beta$-dependence in every term in eq. \eqref{eqn:main}. 

We assume that our data $\x$ comes from mixing of ground truth latent variables (or sources) $\s \in \mathbb{R}^{k}$ through a mixing matrix $\A \in \mathbb{R}^{N \times k}$, then corrupted by noise $\bm{\eta} \in \mathbb{R}^{N}$, 
\begin{equation}
\label{source}
    \mathbf{x} = \mathbf{A} \mathbf{s} + \bm{\eta}.
\end{equation}
We assume $\s \sim \N(\bm{0},\I_k)$, $ \bm{\eta} \sim \N(\bm{0},\I_N)$. The data distribution is found to be,
\begin{equation}
\label{eqn:data}
    p(\x) = \N(\bm{0},\A\A^{\top}+\I_N) \equiv \N(\bm{0},\bSigma_{\x}).
\end{equation}
We denote a $d \times d$ identity matrix as $\I_d$. In this model we can calculate the ground-truth posterior exactly (see Section \ref{app:gtposterior} for details): 
    \begin{align}\label{eqn:true posterior}
        p_{\text{g-t}}(\s |\x) &= \mathcal{N}(\bmu_{\s|\x}, \bSigma_{\s|\x}), \nonumber \\
      \text{with}\;\;  \bmu_{\s|\x} &= (\A^{\top}\A + \mathbf{I}_k)^{-1}\A^{\top}\x \nonumber \\
      \text{and}\;\;  \bSigma_{\s|\x} &= (\A^{\top}\A + \mathbf{I}_k)^{-1}.
    \end{align}
Note that the covariance matrix of the posterior is non-diagonal. Even though the latent factors are statistically independent, when conditioned on data they are dependent. Therefore, we expect a non-trivial dependence of MIE and TIE on the hyperparameter $\beta$.

%, implying that the model posterior from the encoder could never truly match the true posterior, as the model posterior tends to be diagonal when ELBO is optimized. 

Our encoder $q_{\bphi}(\z|\x)$ contains a fully-connected layer $\{\W^{\mu},\bb^{\mu}\}$ with linear activation that codes for the mean $\bmu_{\z}$ of the latent variables $\z$, and a fully-connected layer $\{\W^{\sigma},\bb^{\sigma}\}$ with exponential activation that codes for the diagonal part of the covariance matrix $\bSigma_{\z}$. Given an input $\x \in \mathbb{R}^N$, we generate latent variables $\z\sim \N(\bmu_{\z},\bSigma_{\z}) \in \mathbb{R}^k$ by
\begin{align}
\label{eqn:musigma}
     \bmu_{\z} = \W^{\mu} \x + \bb^\mu, \quad \bSigma_{\z} = \diag ( \exp(\W^{\sigma} \x + \bb^\sigma)),
\end{align}
where the $\diag$ operation maps vectors in $\mathbb{R}^{k}$ to the diagonal of a diagonals matrix in $\mathbb{R}^{k\times k}$. The exponential nonlinearity in the definition of the covariance matrix acts elementwise and prevents negative covariances.

Our decoder consists of a single fully-connected layer $\{\D,\bb^{D}\}$ with linear activations. We assume the output $\y \in \mathbb{R}^N$ is normally distributed, $\y \sim \N (\D \z+\bb^D,\sigma_y^2 \I_N)$, where $\sigma_y^2$ is a hyperparameter. Without loss of generality, from now on we choose $\sigma^2_y =1$. 

The decoder defines  $p_{\btheta}(\x | \z)$. The full data likelihood can be calculated using the prior $p(\z) = \mathcal{N}(\bm{0},\I_k)$ through $p_{\btheta}(\x) = \int d\z\, p_{\btheta}(\x | \z)p(\z)$. With this setup, our decoder is fully capable of modeling the data generative process \eqref{eqn:data}, by choosing $\D = \A$, $\bb^D = \bm{0}$ and $\sigma^2_y = 1$. Any deviation from these parameters will be due to the encoder, or the inference model, deviating from the ground-truth distribution.

In order to solve this model, we integrate out data (i.e., performing $\mathbb{E}_{\x\sim p(\x)}[..]$, using eq. \eqref{eqn:data}) in the $\beta$-VAE objective in eq. \eqref{cost} to arrive at (see Section \ref{app:intoutdata} for details)
\begin{align}
\label{eqn:objective_nox}
	\mathcal{L}_{\beta} = &-\frac{1}{2} \bigg\{ \text{Tr} \bigg[ (\D \W^{\mu} - \mathbf{I}_N) \bSigma_{\x} (\D \W^{\mu} - \mathbf{I}_N)^{\top} \bigg] \nonumber \\
	&+ \beta \text{Tr} \bigg[ \W^{\mu} \bSigma_{\x} (\W^{\mu})^{\top} \bigg] \nonumber \\
	&+ \sum_i^k \left(\left[\D^{\top}\D\right]_{ii} + \beta\right) e^ {\frac{1}{2}\left[\W^{\sigma}\bSigma_{\x} (\W^{\sigma})^{\top}\right]_{ii} + b_i^{\sigma}} \nonumber \\
	&+ (\D \bb^{\mu} + \bb^D)^2 + \beta (\bb^{\mu})^2 - \beta \sum_i^k b^{\sigma}_i  \bigg\}.
\end{align}
We optimize over the network parameters, which amounts to setting the partial derivative of $\mathcal{L}_{\beta}$ with respect to $\{\W^{\mu},\bb^{\mu}, \W^{\sigma},\bb^{\sigma}, \D,\bb^{D}\}$ to zero. Upon simplifying, we find (see Section \ref{app:derivatives} for details)
\begin{equation}
\label{eqn:bmubD}
    \bb^{\mu} = \bb^{D}=0, 
\end{equation}{}
and the remaining equations are ($a=1,...,N;\; b=1,...,k$):
\begin{align}\label{eqn:equations}
	0 &= \left[\left( \D^{\top} (\D \W^{\mu} - \mathbf{I}_N) + \beta \W^{\mu} \right)\bSigma_{\x} \right]_{ab}, \nonumber \\
	0 &= \left[ (\D \W^{\mu} - \I_N) \bSigma_{\x} (\W^{\mu})^{\top} \right]_{ab} + D_{ab} e^{ \frac{1}{2}\left[\W^{\sigma}\bSigma_{\x} (\W^{\sigma})^{\top}\right]_{bb} + b^{\sigma}_b }, \nonumber \\
	\hspace{-1cm} 0 &= \left(\left[\D^{\top}\D\right]_{aa} + \beta \right) e^{ \frac{1}{2}\left[\W^{\sigma}\bSigma_{\x}\right]_{ab}W^{\sigma}_{ab} + b^{\sigma}_a} \left[\W^{\sigma}\bSigma_{\x}\right]_{ab}, \nonumber \\
	0 &= \left( \left[\D^{\top}\D\right]_{aa} + \beta \right) e^{  \frac{1}{2}\left[\W^{\sigma}\bSigma_{\x} (\W^{\sigma})^{\top}\right]_{aa} + b^{\sigma}_a } - \beta.
\end{align}
We can calculate the model posterior distribution $p_{\theta}(\z |\mathbf{x})$ at the network optimum, eqs. \eqref{eqn:bmubD} and \eqref{eqn:equations}. Using Bayes' rule we find (see Section \ref{app:modelposterior})
\begin{align}\label{eqn:model posterior}
        p_{\theta}(\z |\x) &= \mathcal{N}(\bmu_{\z|\x}, \bSigma_{\z|\x}), \nonumber \\
      \text{with}\;\;  \bmu_{\z|\x} &= (\D^{\top}\D + \I_k)^{-1}\D^{\top}\x  \nonumber \\
      \text{and}\;\;  \bSigma_{\z|\x} &= (\D^{\top}\D + \I_k)^{-1}.
\end{align}
Note that when $\A = \D$, eq. \eqref{eqn:model posterior} reduces to  eq. \eqref{eqn:true posterior}, and the model posterior matches with the ground-truth posterior.
We are interested in the inference errors MIE and TIE, eqs. \eqref{eqn:model inference error} and \eqref{eqn:true inference error}.
Upon integrating out the data, we find (see Section \ref{app:modelposterior} for derivations)
\begin{align}
\label{eqn:mie/tie}
    &{\rm MIE/TIE} \nonumber \\
    = &\frac{1}{2} \bigg\{ \sum_i^k E_{ii}^{-1}\exp\bigg[ \frac{1}{2}\bigg(\W^{\sigma}\bSigma_{\x}(\W^{\sigma})^{\top} \bigg)_{ii} + b^{\sigma}_i \bigg] - \sum_i^k b^{\sigma}_i \nonumber \\
    & + \text{Tr}\log \mathbf{E} + \text{Tr}\bigg[(\mathbf{F} - \W^{\mu})^{\top}\mathbf{E}^{-1}(\mathbf{F} - \W^{\mu})\bSigma_{\x} \bigg] -k \bigg\},
\end{align}
where for MIE
\begin{align}
         \E = (\D^{\top}\D+\I_k)^{-1}, \quad
     \F = (\D^{\top}\D+\I_k)^{-1}\D^{\top},
\end{align}
and for TIE
\begin{align}
    \E = (\A^{\top}\A + \mathbf{I}_k)^{-1},\quad
     \F = (\A^{\top}\A + \mathbf{I}_k)^{-1}\A^{\top}.
\end{align}
% \begin{align}{\textwidth}
% {\begin{aligned}
%     &\rm{MIE}
%     \begin{cases} 
%      \E = (\D^{\top}\D+\I_k)^{-1} \\
%      \F = (\D^{\top}\D+\I_k)^{-1}\D^{\top}
%   \end{cases}
%   \MoveEqLeft[-1]
% \end{aligned}}
% & \hskip -2em &
% \begin{aligned}
%     &{\rm TIE} 
%     \begin{cases} 
%      \E = (\A^{\top}\A + \mathbf{I}_k)^{-1} \\
%      \F = (\A^{\top}\A + \mathbf{I}_k)^{-1}\A^{\top}
%   \end{cases}
% \MoveEqLeft[-1]
% \end{aligned}
% \end{align}
%

%
\begin{figure}[t]
  \centering
  \includegraphics[width=12cm]{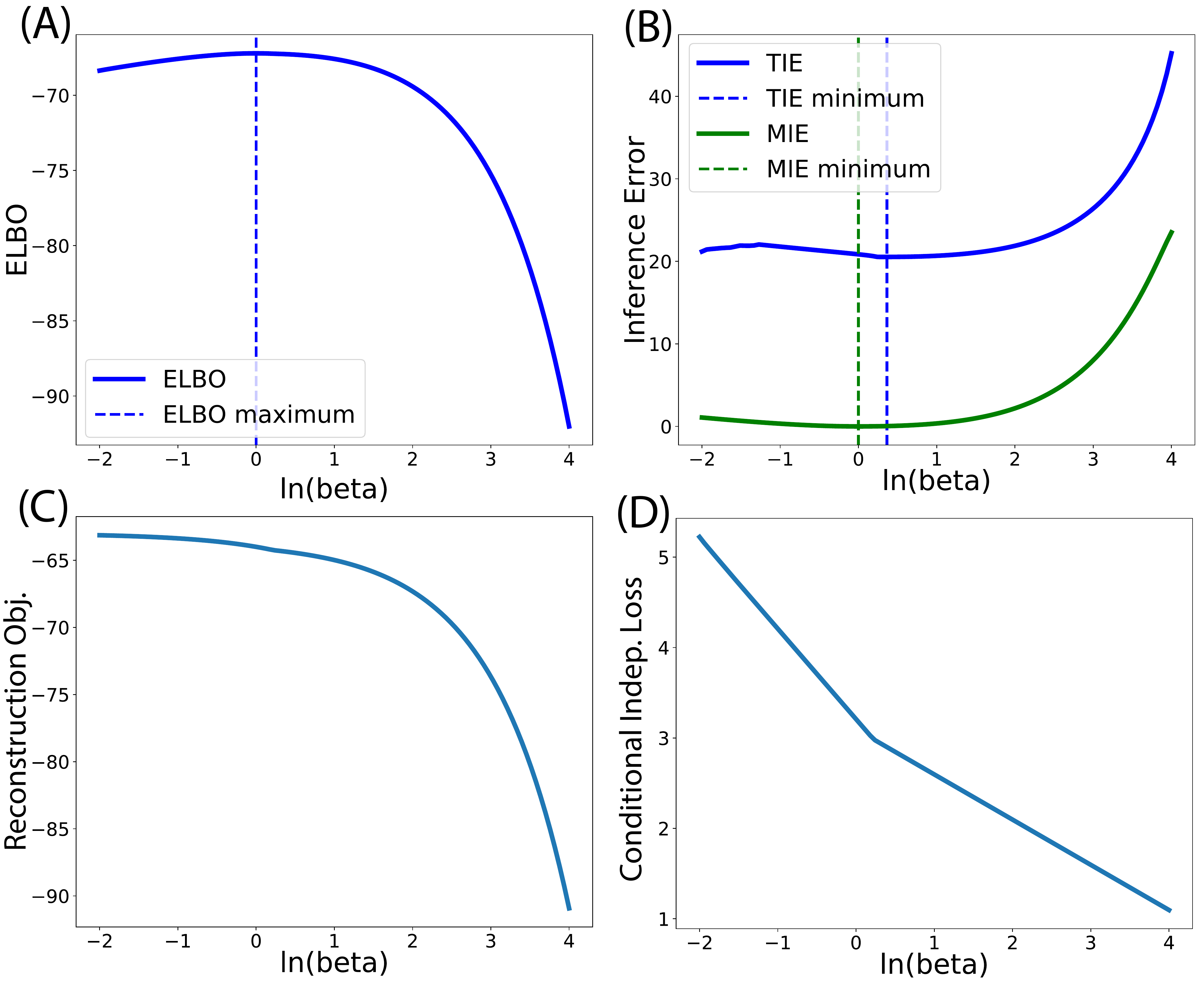}
  \caption{$\beta$-dependence of various quantities at the optimal parameter configuration of $\beta$-VAE. (A) ELBO as a function of $\beta$. (B) MIE/TIE as a function of $\beta$. (C) Reconstruction objective as a function of $\beta$. (D) Conditional Independence Loss as a function of $\beta$. In these plots, we averaged the plotted quantities over the data distribution.}
  \label{fig:panel}
\end{figure}

As an example, we numerically solve eq. \eqref{eqn:equations} for $N=128, k=2$, $A_{ij}=1/2(1+\delta_{ij})$, and use the optimal network parameters $\{\W^{\mu*},\bb^{\mu*}, \W^{\sigma*},\bb^{\sigma*}, \D^*,\bb^{D*}\}$ to calculate ELBO (Fig. \ref{fig:panel}(A)) and inference errors (Fig. \ref{fig:panel}(B)). We  see that ELBO is maximized at $\beta = 1$, while the inference error is not monotonically decreasing and has a minimum at some $\beta$. This confirms the theory we outlined earlier. Also, data log-likelihood is monotonically decreasing with $\beta$ (not shown). We further calculate individual terms in the ELBO: the reconstruction objective (Fig. \ref{fig:panel}(C)), $\mathbb{E}_{q_{\bphi}(\z | \x)}\left[\log p_{\btheta}(\x | \z)\right]$, and the conditional Independence Loss (Fig. \ref{fig:panel}(D)), $D_{KL}\left(q_{\bphi}(\z|\x) \| p(\z)\right)$. Indeed both terms are monotonically decreasing with $\beta$, confirming our propositions.

\section{Numerical Simulations}\label{numerical}

In this section, we examine a deep, nonlinear $\beta$-VAE on a synthetic dataset. The dataset is generated according to eq. \eqref{source} by mixing 10 MNIST digits, arranged as columns of a matrix $\A$, with ground truth sources, $\s \sim \N(\bm{0},\I_k)$, and subsequently adding a noise $\bm{\eta} \sim \N(\bm{0},\I_N)$. Other experimental setups and corresponding datasets that were explored are included in Section \ref{app_sims} (Fig. \ref{fig:panel3}).

%The dataset comprises of a single MNIST digit localized at different locations on a blank canvas. The location of the digit in a sample from our data $\x$ is determined by the same process outlined in eq. \eqref{source}, where ground truth sources, $\s$, are mixed and subsequently corrupted by a noise, $\bm{\eta}$. Here, we are assuming that $\s \sim \N(\bm{0},\I_k)$, $ \bm{\eta} \sim \N(\bm{0},\I_N)$. The resulting values correspond to a cartesian coordinate for the figure on the canvas.%

The encoder, $q_{\bphi}(\z|\x)$, consists of three feed-forward fully-connected layers with tanh activations, ending in two separate output layers encoding the mean of the latent variables $\z$, $\bmu_{\z}$, and the variance, $\bSigma_{\z}$. These are each parameterized by $k$ encoding units. The decoder, $p_{\btheta}(\x | \z)$, consists of three feed-forward fully-connected layers with tanh activation functions, which takes its input from the encoder, and outputs the reconstructed image. Model details are included in Section C.  %With inputs $\z$ from the posterior distribution generated by the encoder, the decoder generates realistic data reconstructions (Fig. \ref{fig:panel3}(E)). 

% Training the network to reconstruct inputs constrained by different values of $\beta$ results in different encoding results in the latent neurons, which can be observed qualitatively (Fig. \ref{fig:panel2}(E). When trained with higher beta values, reconstruction of the digits noticeable worsens, while units in the bottleneck seem to encode for structured, orthogonal axes of motion. When trained with lower $\beta$ values, reconstruction improves, but neurons are not encoding for what seem to be single generative factors. 

After training, we  calculate individual terms in the $\beta$-VAE objective and demonstrate their dependence on  $\beta$. These terms correspond to the Reconstruction Objective, (Fig. \ref{fig:panel2}(C)), and the conditional Independence Loss, (Fig. \ref{fig:panel2}(D)). As we observed in the analytically tractable case, and predicted by our theory, these terms are decreasing with $\beta$. Correspondingly, after being maximized around $\beta = 1$ the entire ELBO term decreases with $\beta$ (Fig. \ref{fig:panel2}(A)). We also calculate the TIE for the $\beta$-VAE at various $\beta$, which follows a non-monotonic trend and has an optimal $\beta$ (Fig. \ref{fig:panel2}(B)).

\begin{figure}
  \centering
  \includegraphics[width=12cm]{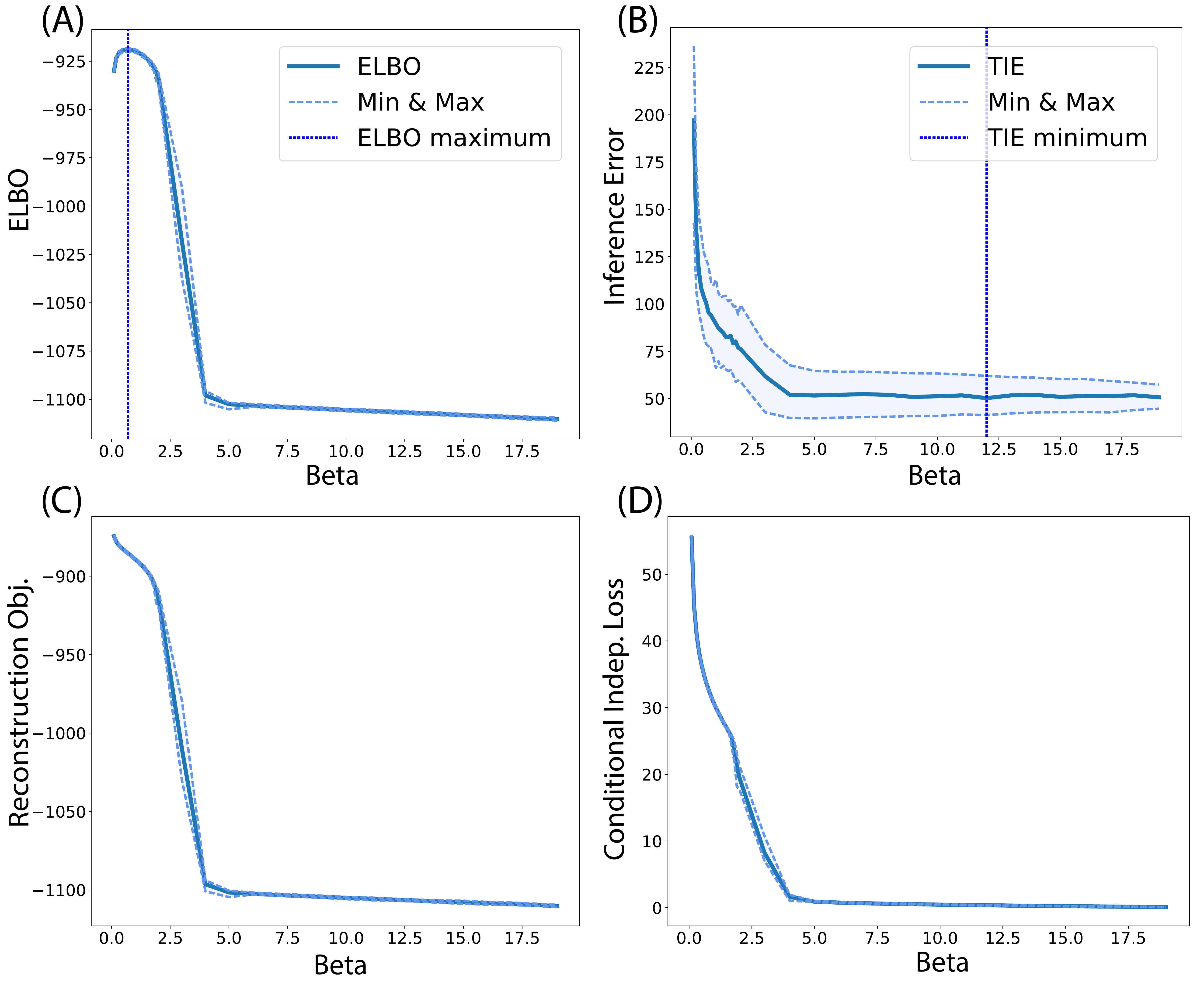}
 \caption{Values for error terms across 100 random initializations of the network. Solid line represents the average. Dashed lines around the solid line represent the minimum and maximum values, and vertical dashed line represent the extremum. (A) ELBO as a function of $\beta$. (B) TIE as a function of $\beta$. (C) Reconstruction Objective as a function of $\beta$. (D) Conditional Independence Loss as a function of $\beta$.  }
  \label{fig:panel2}
\end{figure}

\counterwithin{figure}{section}
\section{Proof of Proposition \ref{prop2}}\label{pprop2}

We prove a more general version of eq. \eqref{eq:qp} given in Prop. \ref{prop2}. Eq. \eqref{eq:re} follows from eq. \eqref{eq:qp} and Prop. \ref{prop1}.

\begin{proposition}
Consider an objective function given by a sum of two terms,
\begin{align}\label{Odef}
    O(\bkappa;\beta) = A(\bkappa) - \beta B(\bkappa),
\end{align}
to be maximized over parameters $\bkappa$, and $\beta$ is a hyperparameter. Let $\bkappa^*(\beta) = \underset{\bkappa}{\argmax} \, O(\bkappa,\beta)$. As $\beta$ increases $B(\kappa^*(\beta))$ is nonincreasing.
\end{proposition}

\begin{proof}
The proof uses contradiction. Let $\beta_2 > \beta_1$ and 
\begin{align} \label{kappa}
\bkappa_1 \equiv \bkappa^*(\beta_1),\qquad  \bkappa_2 \equiv \bkappa^*(\beta_2).
\end{align}
Then 
\begin{align}\label{Ok}
    O(\bkappa_1,\beta_1) &= O(\bkappa_1,\beta_2) + (\beta_2-\beta_1)B(\bkappa_1) \nonumber \\ 
    & \leq O(\bkappa_2,\beta_2) + (\beta_2-\beta_1)B(\bkappa_1),
\end{align}
where the first line is an identity, and the second line follows from the optimality of $\bkappa_2$ at $\beta=\beta_2$.

Now we assume $B(\bkappa_2) > B(\bkappa_1)$, and see that this leads to a contradiction.
\begin{align}
&O(\bkappa_2,\beta_2) + (\beta_2-\beta_1)B(\bkappa_1)  \nonumber \\
&\qquad <  O(\bkappa_2,\beta_2) + (\beta_2-\beta_1)B(\bkappa_2)  = O(\bkappa_2,\beta_1).
\end{align}     
The inequality follows from our assumption, and the equality from \eqref{Odef}. Combined with \eqref{Ok}, this implies
\begin{align}
    O(\bkappa_1,\beta_1) < O(\bkappa_2,\beta_1)
\end{align}
which contradicts \eqref{kappa}. Therefore if $\beta_2>\beta_1$, then $B(\kappa_2) \leq  B(\kappa_1)$.
\end{proof}

\section{Details of the analytically tractable $\beta$-VAE model}

\subsection{Integrating out data from the objective}
\label{app:intoutdata}
The full $\beta$-VAE objective is \eqref{bVAEobj} averaged with respect to the data distribution $p(\x)$:
\begin{align}
& L(\bm{\theta},\bm{\phi};\beta) \equiv \mathbb{E}_{p(\x)}\left[ \mathcal{L}(\bm{\theta},\bm{\phi};\beta)\right]  \nonumber \\
&=\mathbb{E}_{p(\x)} \left[\mathbb{E}_{q_{\bphi}(\z | \x)}\left[\log p_{\btheta}(\x | \z)\right]-\beta D_{K L}\left(q_{\bphi}(\z | \x) \| p(\z)\right) \right].
\end{align}
We first calculate $\mathbb{E}_{q_{\bphi}(\z | \x)}\left[\log p_{\btheta}(\x | \z)\right]$. We use the reparametrization trick: For $\z  \sim q_{\bphi}(\z | \x) = \N(\bmu_{\z},\bSigma_{\z})$, we can write $\z = \bmu_{\z} + \bSigma_{\z}^{1/2} \bm{\epsilon}$ with $\bm{\epsilon}\sim \N(\mathbf{0},\mathbf{I}_k)$. Then,
\begin{align*}
    &\mathbb{E}_{\z \sim q_{\bphi}(\mathbf{z} | \mathbf{x})} [ \log p_{\btheta}(\mathbf{x} | \mathbf{z})] \\
    &\quad=\mathbb{E}_{\z \sim q_{\bphi}(\mathbf{z} | \mathbf{x})} \bigg[ \log \mathcal{N}(\mathbf{x}; \mathbf{D}\z + \mathbf{b}^{D}, \mathbf{I}_N) \bigg] \\ 
    &\quad= \mathbb{E}_{\bm{\epsilon} \sim \mathcal{N}(0,1)} \bigg[\log \mathcal{N}(\mathbf{x}; \mathbf{D}(\bmu_{\z} +  \bSigma_{\z}^{1/2} \bm{\epsilon}) + \mathbf{b}^{D},  \mathbf{I}_N ) \bigg] \\
    &\quad= -\frac{N}{2}\log(2\pi) - \frac{1}{2}(\D \bmu_{\z} + \mathbf{b}^{D} - \mathbf{x})^2 \\
    &\qquad\quad- \frac{1}{2} \mathbb{E}_{\bm{\epsilon} \sim \mathcal{N}(0,1)}\bigg[\bm{\epsilon}^{\top} (\D \bSigma_{\z}^{1/2})^{\top} (\D \bSigma_{\z}^{1/2})\bm{\epsilon} \bigg].
\end{align*}
The last term can be calculated by the following useful trick. Let's introduce a source term $\J$ into the generating functional,
\begin{align}
   Z[\J] = \int \frac{d\z}{(2\pi)^{n/2} \sqrt{\det \bSigma_{\z}}} \exp \bigg(-\frac{1}{2}\z^{\top} \bSigma_{\z}^{-1}\z + \J^{\top} \A \z \bigg),
\end{align}
then differentiating with respect to the source, 
\begin{align}
 \bigg(\frac{\delta}{\delta \J} \bigg)^{\top} \bigg(\frac{\delta}{\delta \J} \bigg) Z[\J] \bigg\vert_{\J=0} =     \mathbb{E}_{\z \sim \mathcal{N}(0,\bSigma_{\z})}\left[\z^{\top} \A^{\top} \A \z \right] %&\text{Tr} (\A \bSigma_{\z} \A^{\top}).
\end{align}
On the other hand, we can perform the Gaussian integral in $Z[\J]$ to obtain,
\begin{align}
    Z[\J] = \exp \bigg\{\frac{1}{2}(\J\A)^{\top} \bSigma_{\z} (\J\A)  \bigg\}.
\end{align}
Then we arrive at 
\begin{align}
\label{eqn:correlation}
    \mathbb{E}_{\z \sim \mathcal{N}(0,\bSigma_{\z})}\left[\z^{\top} \A^{\top} \A \z \right] = \text{Tr} (\A \bSigma_{\z} \A^{\top})
\end{align}
Eq. \eqref{eqn:correlation} is central to the calculations of many results presented in the text.

Going back to the reconstruction objective, using eq. \eqref{eqn:correlation} we have (up to constants)
\begin{align}
    \mathbb{E}_{\z \sim q_{\bphi}(\mathbf{z} | \mathbf{x})} [ \log p_{\btheta}(\mathbf{x} | \mathbf{z})] &= - (\D \bmu_{\z} + \mathbf{b}^{D} - \mathbf{x})^2 \nonumber \\
    &\qquad\qquad- \text{Tr}(\D^{\top}\D \bSigma_{\z}).
\end{align}{}
Similarly we can calculate the conditional independence loss,
\begin{align}
    &D_{K L}(q_{\bphi}(\mathbf{z} | \mathbf{x}) \| p(\mathbf{z})) \nonumber \\
    &\qquad\qquad = -\frac{1}{2}\left(k + \text{Tr}\log \bSigma_{\z} - \bmu_{\z}^{\top} \bmu_{\z} - \text{Tr} \bSigma_{\z} \right).
\end{align}
Putting everything together, the objective function we want to maximize is (neglecting constant terms)
\begin{align}
    &L(\bm{\theta},\bm{\phi};\beta) =  \frac{1}{2}\mathbb{E}_{p(\mathbf{x})}\left[ - (\D \bmu_{\z} + \mathbf{b}^{D} - \mathbf{x})^{\top}(\D \bmu_{\z} + \mathbf{b}^{D} - \mathbf{x})\right. \nonumber \\
    &\quad - \beta \bmu_{\z}^{\top}\bmu_{\z} - \text{Tr}(\D^{\top}\D \bSigma_{\z}) + \beta \text{Tr}\log \bSigma_{\z} - \beta\text{Tr}\bSigma_{\z} ].
\end{align}
The expectation with respect to $\x$ amounts to performing Gaussian integrals in $\x$, as $\x \sim \N(\mathbf{0},\bSigma_{\x})$, and thus can be done exactly. After plugging in the definition of $\bmu_{\z}, \bSigma_{\z}$ from eq. \eqref{eqn:musigma}, and performing the $\x$ integrals, the result is given in eq. \eqref{eqn:objective_nox}.

\subsection{Taking derivatives of the objective}
\label{app:derivatives}
In order to take derivatives of eq. \eqref{eqn:objective_nox}, we unpack the indices (to ease the notation, we denote $\bSigma_{\x}$ as $\bSigma$, and follow the Einstein summation convention, repeated indices are to be summed over unless the summation is explicitly specified)
\begin{align}
	L = &-\frac{1}{2} \bigg\{ (D_{ij} W^{\mu}_{jk} - \delta_{ik}) \Sigma_{kl} (W^{\mu}_{ml} D_{im} - \delta_{il}) - \beta \sum_i b^{\sigma}_i \nonumber \\
	&+ \beta (W^{\mu})_{ij} \Sigma_{jk} W^{\mu}_{ik} +(D_{ij} b^{\mu}_j + b^D_i)^2 + \beta (b^{\mu}_i)^2 \nonumber \\
	&+ \sum_i \bigg( D^2_{li} + \beta \bigg) \exp \bigg(\frac{1}{2} W^{\sigma}_{ij} \Sigma_{jk} W^{\sigma}_{ik} + b_i^{\sigma} \bigg)  \bigg\}
\end{align}
Then,
\begin{align}
    0 &= \frac{\partial L}{\partial W^{\mu}_{ab}} = \left [ \left( \D^{\top} (\D \W^{\mu} - \mathbf{I}_N) + \beta \W^{\mu} \right) \bSigma \right]_{ab}, \\
    0 &= \frac{\partial L}{\partial b^{\mu}_{a}} = \left[ (\D \bb^{\mu} + \bb^D) \D + \beta \bb^{\mu} \right]_a, \\
    0 & = \frac{\partial L}{\partial D_{ab}} = \left[ (\D \W^{\mu} - \I_N) \bSigma (\W^{\mu})^{\top} \right]_{ab} \nonumber \\
    & \qquad \qquad  + \left[\D \bb^{\mu} + \bb^D\right]_a b^{\mu}_b + D_{ab} e^{ \frac{1}{2}\left[\W^{\sigma}\bSigma (\W^{\sigma})^{\top}\right]_{bb} + b^{\sigma}_b }, \\
    0 &= \frac{\partial L}{\partial b^{D}_{a}} = \left[\D \bb^{\mu} + \bb^D\right]_a \\
    0 &= \frac{\partial L}{\partial W^{\sigma}_{ab}} = \left(\left[\D^{\top}\D\right]_{aa} + \beta \right) e^{ \frac{1}{2}\left[\W^{\sigma}\bSigma\right]_{ab}W^{\sigma}_{ab} + b^{\sigma}_a } \left[\W^{\sigma}\bSigma\right]_{ab}, \\
    0 &= \frac{\partial L}{\partial b^{\sigma}_{a}} = \left( \left[\D^{\top}\D\right]_{aa} + \beta \right) e^ { \frac{1}{2}\left[\W^{\sigma}\C (\W^{\sigma})^{\top}\right]_{aa} + b^{\sigma}_a } - \beta.
\end{align}
From the $b^{\mu}_a$ and $b^D_a$ equations we can immediately see $\bb^{\mu} = \bb^D = 0$.  

\subsection{Derivation of the ground-truth posterior}
\label{app:gtposterior}
We observe that since both $\s$ and $\bm{\eta}$ are independently normally distributed in \eqref{source}, $\s$ and $\bm{\eta}$ are jointly normal, i.e., $p(\s,\bm{\eta})$ is a normal distribution. However, note that $p(\s,\bm{\eta})$ is just $p(\s,\x)$ up to a coordinate transformation, so $p(\s,\x)$ is also normal. Also, as $\s \in \mathbb{R}^{k}$, $\x \in \mathbb{R}^N$, $(\s, \x) \in \mathbb{R}^{N+k}$. We can think of $\s$ and $\x$ partition a $(N+k)$-dimensional normal distribution $p((\s,\x))$. Therefore, to find the conditional probability $p_{\text{g-t}}(\s |\x)$, we can just use the formula for conditioning multivariate normal distribution: 
\begin{equation}
    p_{\text{g-t}}(\s |\x) = \mathcal{N}(\bmu_{\s|\x}, \bSigma_{\s|\x}), 
\end{equation}
where 
\begin{equation}
    \begin{split}
        \bmu_{\s|\x} &= \bmu_{\s} + \text{Cov}(\x,\s)^{\top}(\bSigma_{\x})^{-1}(\x - \bmu_{\x}) \\
        \bSigma_{\s|\x} &= \bSigma_{\s} - \text{Cov}(\x,\s)^{\top}(\bSigma_{\x})^{-1}\text{Cov}(\x,\s).
    \end{split}{}
\end{equation}{}
Now specializing to our case \eqref{source}, 
\begin{equation}
\text{Cov}(\x,\s) = \text{Cov}(\A \s + \bm{\eta}, \s)= \A \bSigma_{\s} + \text{Cov}(\s,\bm{\eta})= \A.
\end{equation}
Note that $\bSigma_{\x} = \A \A^{\top} + \mathbf{I}_N$, then 
\begin{equation}
        \bmu_{\s|\x} = \A^{\top}(\A \A^{\top} + \mathbf{I}_N)^{-1}\x = (\A^{\top}\A + \mathbf{I}_k)^{-1}\A^{\top}\x,
\end{equation}
where in the second equality we have used the \textit{matrix push-through identity}:
For any matrices $\mathbf{U} \in \mathbb{R}^{N\times k}$,$\mathbf{V} \in \mathbb{R}^{k \times N}$,
\begin{equation}
\label{eqn:push}
    (\mathbf{I}_N + \mathbf{U}\mathbf{V})^{-1}\mathbf{U} = \mathbf{U}(\mathbf{I}_k + \mathbf{V}\mathbf{U})^{-1}.
\end{equation}
Now the covariance,
\begin{equation}
    \begin{aligned}
        \bSigma_{\s|\x} &= \mathbf{I}_k - \A^{\top}(\A \A^{\top} + \mathbf{I}_N)^{-1} \A \\
        &= \mathbf{I}_k - \A^{\top}\A (\A^{\top} \A + \mathbf{I}_k)^{-1} \\
        &= \mathbf{I}_k - \A^{\top}\A [(\A^{\top}\A)^{-1} - (\A^{\top}\A )^{-1}(\A^{\top}\A + \mathbf{I}_k)^{-1}] \\
        &= (\A^{\top}\A + \mathbf{I}_k)^{-1},
    \end{aligned}{}
\end{equation}
where in the third equality we have used the \textit{Woodbury matrix identity}: For any invertible matrix $\B \in \mathbb{R}^{N \times N}$ and size compatible matrices $\mathbf{U} \in \mathbb{R}^{N\times k}$ and $\mathbf{V} \in \mathbb{R}^{k \times N}$:
\begin{equation}
\label{eqn:woodbury}
    (\B+\U\V)^{-1} = \B^{-1} - \B^{-1}\U(\mathbf{I}_k+\V\B^{-1}\U)^{-1}\V\B^{-1}.
\end{equation}{}

\subsection{Derivation of the model posterior}
\label{app:modelposterior}
Our goal is to use the Bayes rule to calculate the model posterior, $p_{\btheta}(\x|\z) = p_{\btheta}(\x|\z)p(\z)/p_{\btheta}(\x).$

In order to do so, we first need to calculate the \textit{evidence} $p_{\btheta}(\x)$,
\begin{align}
    p_{\btheta}(\x) &= \int_{\mathbb{R}^{k}} d\z p_{\btheta}(\x|\z) p(\z) \\
    &= \int_{\mathbb{R}^{k}} d\z \N (\D \z, \mathbf{I}_N)\N(\mathbf{0},\mathbf{I}_k) \\
    &= \N(\mathbf{0},(\D\D^{\top}+\mathbf{I}_N)),
\end{align}
where in the third equality we have used eq.s \eqref{eqn:push} and \eqref{eqn:woodbury} to simplify. 
Therefore, 
\begin{equation}
    p_{\btheta}(\x|\z) = \frac{\N (\D \z, \mathbf{I}_N)\N(\mathbf{0},\mathbf{I}_k)}{\N(\mathbf{0},(\D\D^{\top}+\mathbf{I}_N))}
\end{equation}
After some simplifications  using eq.s \eqref{eqn:push} and \eqref{eqn:woodbury}, we arrived at
\begin{align}
    p_{\btheta}(\x|\z) &= \N((\D\D^{\top}+\mathbf{I}_N)^{-1}\D^{\top}, (\D\D^{\top}+\mathbf{I}_N)^{-1}) \nonumber \\
    &\equiv \mathcal{N}(\bmu_{\z|\x}, \bSigma_{\z|\x}).
\end{align}
\subsection{Derivation of $\rm{MIE/TIE}$}
First let's consider $\rm{MIE}$. Let
\begin{align}
    \bmu_{\z|\x} \equiv \F \x, \qquad 
    \bSigma_{\z|\x} \equiv \E.
\end{align}{}
Then, we can write $\rm{MIE}$ as 
\begin{align}
    \rm{MIE} =& \mathbb{E}_{p(\x)}\big[ D_{K L}(q_{\bphi}(\mathbf{z} | \mathbf{x}) \| p_{\btheta}(\x|\z)) \big] \nonumber \\
    = &\frac{1}{2}\mathbb{E}_{p(\x)}\bigg[(\F\x - \bmu_{\z})^{\top}\E^{-1}(\F\x - \bmu_{\z}) \nonumber \\
    & + \text{Tr}(\E^{-1}\bSigma_{\z}) - \log \bigg(\frac{\det\bSigma_{\z}}{\det\E} \bigg) - k \bigg].
\end{align}
Plugging in eq. \eqref{eqn:musigma} and performing the $\x$ Gaussian integrals as in Section \ref{app:intoutdata}, we arrive at eq. \eqref{eqn:mie/tie}.

Note that at network optimum, our model posterior $p_{\btheta}(\z|\x)$ equals to the ground-truth posterior $p(\s|\x)$ upon changing $\D$ to $\A$. Therefore, we just need to replace $\D$ by $\A$ in the above derivation to obtain the results for $\rm{TIE}$.

\section{Simulation Details}\label{app_sims}
%\subsection{Model Architecture}
The deep neural network models used for the numerical experiments  task used the same overall architecture. The encoder is a feed forward network with 3 hidden layers, with 256, 200, and 200 units. 2 parallel hidden layers with 2 neurons parameters the mean and variance for $k=2$ latent variables. The decoder consists of 3 feed-forward hidden layers with 200, 200, and 256 units, then outputs the reconstructed image. 
%\subsection{Training Details}
The network was trained for 1000 epochs over the entire synthetic dataset, comprising of 1000 examples. We used a tanh activation function used along with Adam Optimization \cite{kingma2014adam} with a learning rate of 1e-3. Experiments were repeated across 300 realizations for each $\beta$ value. Results shown were averaged over the whole set of realizations. 

%\subsection{Error Term Calculations}
%Error terms and objectives were calculated across 300 realizations of trained networks with $\beta$={0...350}. 
The Reconstruction Objective was calculated for each trained model through generating 1000 samples from the encoder, passing them to the decoder to approximately calculate  $\mathbb{E}_{q_{\bphi}(\z | \x)}\left[\log p_{\btheta}(\x | \z)\right]$,
%,. The Reconstruction Objective, $\mathbb{E}_{q_{\bphi}(\z | \x)}\left[\log p_{\btheta}(\x | \z)\right]$, was calculated 
and averaging over the data $\x$. The Conditional Independence Loss was calculated directly using the Tensorflow Distributions library's native KL Divergence method. The ELBO was calculated by numerically taking the difference of these two terms, and the $\beta$-VAE objective was an extension of this with the hyperparameter $\beta$ included. The Inference Error was calculated numerically using the modelled $\bmu_{\z}$ and $\bSigma_{\z}$ and estimating $p(\x)$ from mini-batches.

In Fig. \ref{fig:panel3}, we show results on another simulation consistent with our findings.

%\subsection{Data Analysis}
% Statistics were calculated from the 300 realizations, and include minimum and maximum values for error terms, along with the mean, median, and standard deviation.

\begin{figure}
  \centering
  \includegraphics[width=12cm]{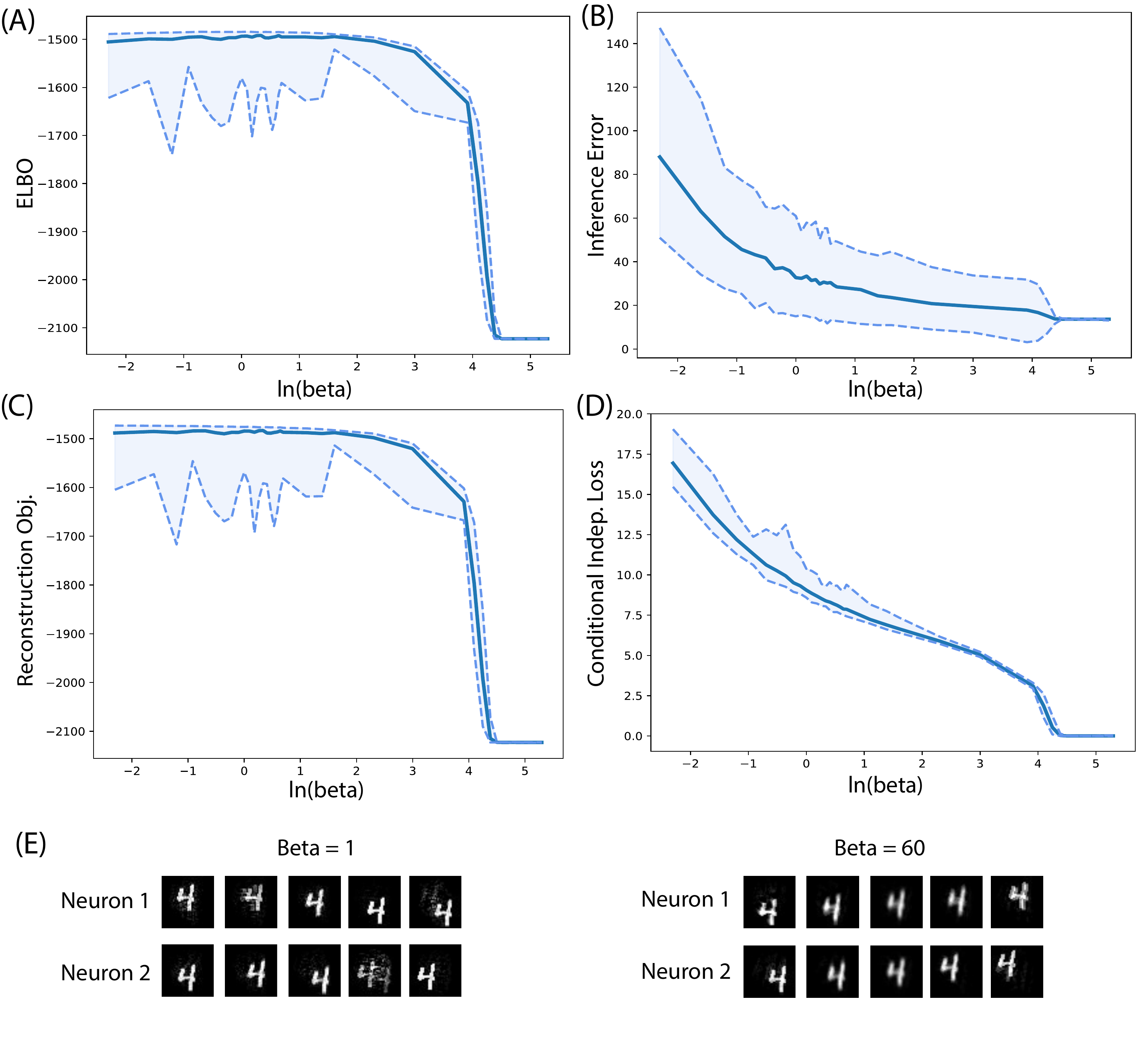}
\caption{Values for error terms across 300 random initializations of the network for a synthetic dataset, which comprises of a single MNIST digit localized at different locations on a blank canvas. The cartesian coordinate of the digit in a sample from our data, $\x$, is determined by eq. \eqref{source}, with $A_{ij} = 2\delta_{ij} + 0.73$,
% $\A = \left( 
% \begin{array}{cc} 
% 2.73 & 0.73 \\ 0.73 & 2.73 
% \end{array} 
% \right)$
$\s \sim \N(\bm{0},\I_k)$, $ \bm{\eta} \sim \N(\bm{0},\I_N)$,  $N=k=2$.  Dashed lines represent the minimum and maximum values, and solid line represents the average. (A) ELBO as a function of $\beta$. (B) TIE as a function of $\beta$. Its minima over various random initialization follow a non-monotonic trend. (C) Reconstruction objective as a function of $\beta$. (D) Conditional independence loss as a function of $\beta$.  (E) Traversal of latent encoding in bottleneck neurons for small and large $\beta$. One neuron is held fixed while the other is modulated to generate reconstructions. Reconstruction of the digits noticeable worsens with higher $\beta$, while units in the bottleneck encode for structured, orthogonal axes of motion. } 
\label{fig:panel3}
\end{figure}

\section{Discussion and Conclusion}\label{discussion}

In this chapter, we examined the learning of disentangled representations by extracting statistically independent latent variables in $\beta$-VAE. We proved general theorems on variational Bayesian inference in the context of $\beta$-VAE and introduced an analytically tractable $\beta$-VAE model. We also performed experiments on synthetic datasets to test our insights from the general theorems and the tractable model, and found good agreements.

$\beta$-VAE enforces conditional independence of its units at the bottleneck layer. This preference is not compatible with independence of latent variables, and therefore may lead to an optimal value of $\beta$ for latent variable inference.

% Notably, our proof that the $\beta$-VAE objective does not increase with $\beta$ coupled with the latent variable inference behaving non-monotonically in $\beta$ points to a tension in achieving optimal disentangling results while minimizing inference error. In the $\beta$-VAE framework, the natural notion of disentangling as full statistical independence is not preserved. Conditional independence in the form of uncorrelated latent variables, while fulfilled by the $\beta$-VAE paradigm, may not amount to a satisfying definition of disentangling. 

There are other perspectives on what constitutes a disentangled representation not addressed in this paper\cite{bengio2013representation,burgess2018understanding}, including definitions not statistical in nature, instead taking into account the manifold structure and symmetry transformations in data  \cite{bengio2013representation,dicarlo2007untangling,higgins2018towards}. Other deep learning approaches to disentangling include the adversarial setting \cite{denton2017unsupervised, tran2017disentangled,john2018disentangled}. Disentangled representations have also been studied in supervised and semi-supervised contexts \cite{siddharth2017learning}.
%auto-ignore
\chapter{Non-equilibrium dynamics of recurrent neural networks}
\label{attractor}

\section*{Introduction}
Dynamical attractors have found much use in neuroscience as models for carrying out computation and signal processing \citep{Poucet2005-nm}. While point-like neural attractors and analogies to spin glasses have been widely explored \citep{Hopfield1982-fb,Amit1985-ls}, an important class of experiments are explained by `continuous attractors' where the collective dynamics of strongly interacting neurons stabilizes a low-dimensional family of activity patterns. Such continuous attractors have been invoked to explain experiments on motor control based on path integration \citep{Seung1996-ck,Seung2000-bk}, head direction \citep{Kim2017-zs} control, 
spatial representation in grid or place cells \citep{Yoon2013-nl,OKeefe1971-jt,Colgin2010-rd,Wills2005-hz,Wimmer2014-og, Pfeiffer2013-qn}, amongst other information processing tasks \citep{Hopfield2015-wt,Roudi2007-pm,Latham2003-fr,Burak2012-bu}. 
%A neural network can be programmed with one or more attractors during a learning process that changes the strength of synaptic interactions between neurons. 

These continuous attractor models are at the fascinating intersection of dynamical systems and neural information processing. The neural activity in these models of strongly interacting neurons is described by an emergent collective coordinate \citep{Yoon2013-nl,Wu2008-iw,Amari1977-hf}. This  collective coordinate stores an internal representation \citep{Sontag2003-lk,Erdem2012-uf} of the organism's state in its external environment, such as position in space \citep{Pfeiffer2013-qn,McNaughton2006-xq} or head direction \citep{Seelig2015-uu}.

However, such internal representations are useful only if they can be driven and updated by external signals that provide crucial motor and sensory input  \citep{Hopfield2015-wt, Pfeiffer2013-qn,Erdem2012-uf,Hardcastle2015-as,Ocko2018-gv}.  Driving and updating the collective coordinate using external sensory signals opens up a variety of capabilities, such as path planning \citep{Ponulak2013-op,Pfeiffer2013-qn}, correcting errors in the internal representation or in sensory signals \citep{Erdem2012-uf,Ocko2018-gv}, and the ability to resolve ambiguities in the external sensory and motor input \citep{Hardcastle2015-as,Evans2016-pr,Fyhn2007-ys}.

In all of these examples, the functional use of attractors requires interaction between external signals and the internal recurrent network dynamics. However, with a few significant exceptions
 \citep{fung2015fluctuation,mi2014spike,Wu2008-iw,Wu2005-sw,Monasson2014-nu,Monasson2013-kx,Burak2012-bu}, most theoretical work has either been in the limit of no external forces and strong internal recurrent dynamics, or in the limit of strong external forces where the internal recurrent dynamics can be ignored \citep{Moser2017-dj,Tsodyks1999-px}.

Here, we study continuous attractors in neural networks subject to external driving forces that are neither small relative to internal dynamics, nor adiabatic. We show that the physics of the emergent collective coordinate sets limits on the maximum speed at which  internal representations can be updated by external signals.

Our approach begins by deriving simple classical and statistical laws satisfied by the collective coordinate of many neurons with strong, structured interactions that are subject to time-varying external signals, Langevin noise, and quenched disorder. Exploiting these equations, we demonstrate two simple principles; (a) an `equivalence principle' that %relates the effective dynamics of the emergent collective coordinate in the stationary frame and the co-moving frame, and in particular 
predicts how much the internal representation lags a rapidly moving external signal, (b)  under externally driven conditions, quenched disorder in network connectivity can be modeled as a state-dependent effective temperature. 
Finally, we apply these results to place cell networks and derive a non-equilibrium driving-dependent memory capacity, complementing numerous earlier works on memory capacity in the absence of external driving.

\section*{Collective coordinates in continuous attractors}
We study $N$ interacting neurons following the formalism presented in \citep{Hopfield2015-wt},
\begin{equation}
\label{eqn:eom_zero}
\frac{di_n}{dt} = -\frac{i_n}{\tau} + \sum_{k = 1}^{N} J_{nk}f(i_k) + I^{ext}_n(t) +\eta_{int}(t),
\end{equation}
where $f(i_k) = (1+e^{-i_k/i_0})^{-1}$ is the neural activation function that represents the firing rate of neuron $k$, and $i_n$ is an internal excitation level of neuron $n$ akin to the membrane potential. We consider synaptic connectivity matrices with two distinct components, 
\begin{equation}
J_{ij} = J_{ij}^{0}+ J_{ij}^{d}.
\label{eqn:JijBreakdown}
\end{equation}

As shown in Fig.\ref{fig:schematics}, $J_{ij}^{0}$ encodes the continuous attractor. We will focus on $1$-D networks with $p$-nearest neighbor excitatory interactions to keep bookkeeping to a minimum: $J^{0}_{ij} = J(1 - \epsilon)$ if neurons $|i-j| \leq p$, and $J_{ij}^{0} = - J \epsilon$ otherwise. The latter term, $- J \epsilon$, with $0\leq \epsilon \leq 1$, represents long-range, non-specific inhibitory connections %that effectively impose a soft constraint on the number of neurons firing at any given time,
as frequently assumed in models of place cells \citep{Monasson2013-kx,Hopfield2010-lf}, head direction cells \citep{Chaudhuri2016-lh} and other continuous attractors \citep{Seung2000-bk,Burak2012-bu}. 

The disorder matrix $J_{ij}^{d}$ represents random long-range connections, a form of quenched disorder \citep{Sebastian_Seung1998-vm,Kilpatrick2013-go}. %We will derive our results for a general disorder matrix; later we will apply our results to place cell models where such long-range disorder is generated from encoding multiple spatial environments in one network.
Finally, $I_n^{ext}(t)$ represents external driving currents from e.g. sensory and motor input possibly routed through other regions of the brain. The Langevin noise $\eta_{int}(t)$ represents private noise internal to each neuron \citep{Lim2012-ek,Burak2012-bu}. %with $\langle \eta_{int}(t) \eta_{int}(0) \rangle = C_{int}\delta(t)$.

\begin{figure}
\begin{center}
    \includegraphics[width=0.8\linewidth]{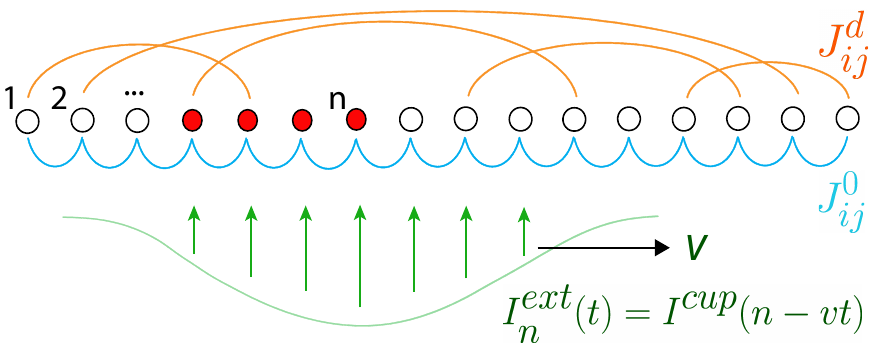}
\end{center}
\caption{ 
The effective dynamics of neural networks implicated in head direction and spatial memory is described by a continuous attractor. Consider $N$ neurons connected in a 1-D topology, with local excitatory connections between $p$ nearest neighbors (blue), global inhibitory connections (not shown), and random long-range disorder (orange). Any activity pattern quickly condenses into a `droplet' of contiguous firing neurons (red) of characteristic size; the droplet center of mass $\bar{x}$ is a collective coordinate parameterizing a continuous attractor. The droplet can be driven by space and time-varying external currents $I^{ext}_{n}(t)$ (green). \label{fig:schematics}} 
\end{figure}

%\textbf{Collective coordinate}
A neural network %with $p$-nearest neighbor interactions%
like Eqn.\eqref{eqn:eom_zero} qualitatively resembles a similarly connected network of Ising spins %; the inhibitory connections impose a (soft) constraint on the number of neurons that can be firing at any given time and hence similar to working %
at fixed magnetization \citep{Monasson2013-pn}. 
At low noise, the activity in such a system will condense \citep{Monasson2013-kx,Hopfield2010-lf} to a localized `droplet', since interfaces between firing and non-firing neurons are penalized by $J(1-\epsilon)$. 
%{\bf this should be $J(1-\epsilon)$, right?} 
The center of mass of such a droplet, $\bar{x} \equiv \frac{\sum_n n f(i_n) }{\sum_n f(i_n)}$ is an emergent collective coordinate that approximately describes the stable low-dimensional neural activity patterns of these $N$ neurons. Fluctuations about this coordinate have been extensively studied \citep{Wu2008-iw,Burak2012-bu,Hopfield2015-wt,Monasson2014-nu}.

\section*{Space and time dependent external signals}

We focus on how space and time-varying external signals, modeled here as external currents $I^{ext}_n(t)$ can drive and reposition the droplet along the attractor. We will be primarily interested in a cup-shaped current profile that moves at a constant velocity $v$, i.e., $I^{ext}_n(t) = I^{cup}(n - vt)$ where 
$I^{cup}(n) = d(w-|n|), n \in [-w,w]$, $I^{cup}(n) = 0$ otherwise. Such a localized time-dependent drive could represent landmark-related sensory signals \citep{Hardcastle2015-as}; see Discussions. 

%In addition to such positional information, continuous attractors often also receive velocity information \citep{Major2004-ku,McNaughton2006-xq,Seelig2015-uu,Ocko2018-gv}. We do not consider such input in the main text but extend our analysis to velocity integration in the Appendix. %, we accommodate such driving and show that $v$ can be interpreted as the difference in velocity implied by position and velocity information manipulated in rece
%; such signals are modeled \citep{Burak2009-ch,Hopfield2010-lf} as a time-independent anti-symmetric $A^{0}_{ij}$ added on to $J^{0}_{ij} \to J^0_{ij} + A^{0}_{ij}$. 

% that `tilts' the continuous attractor, so the droplet moves with a velocity proportional to $A^0_{ij}$. 

%Such velocity integration (or `dead-reckoning') will inevitably accumulate errors that are then corrected using direct positional information modeled by $I^{ext}_n(t)$ \citep{Hardcastle2015-as}.

%In the Appendix, we find that in the presence of $A_{ij}$, the velocity $v$ of $I^{ext}(t)$ can be interpreted as the difference in velocity implied by positional and velocity information, which has been manipulated in virtual reality experiments \citep{Major2004-dz,Colgin2010-rd,Seelig2015-uu, Ocko2018-gv, Campbell2018-wt}. Therefore, for simplicity here we set $A_{ij}=0$.

The effective dynamics of the collective coordinate $\bar{x}$ in the presence of currents $I^{ext}_n(t)$ can be obtained by computing the effective force on the droplet of finite size. We find that
\begin{equation}
\gamma \dot{\bar{x}} = - \partial_{\bar{x}} V^{ext}(\bar{x},t), 
\end{equation}
where $V^{ext}(\bar{x},t)$ is a piecewise quadratic potential $V^{cup}(\bar{x}-vt)$ for currents $I^{ext}_n(t) = I^{cup}(n - vt)$, and $\gamma$ is the effective drag coefficient of the droplet. (Here, we neglect rapid transients of timescale $\tau$ \citep{Wu2008-iw}.) 

The strength of the external signal is set by the depth $d$ of the cup $I^{cup}(n)$. Previous studies have explored the $d = 0$ case, i.e., undriven diffusive dynamics of the droplet \citep{Burak2012-bu,Monasson2014-nu, Monasson2013-pn,Monasson2015-nl} or the large $d$ limit \citep{Hopfield2015-wt} when the internal dynamics can be ignored. Here we focus on an intermediate regime, $d < d_{max}$ where internal representations are updated continuously by the external currents, without any jumps \citep{Ponulak2013-op,Pfeiffer2013-qn,Erdem2012-uf}.

In fact, as shown in the Appendix, we find a threshold signal strength $d_{max}$ beyond which the external signal destabilizes the droplet,  instantly `teleporting' the droplet from any distant location to the cup without continuity along the attractor, erasing any prior positional information held in the internal representation.

We focus here on $d < d_{max}$, a regime with continuity of internal representations. Such continuity is critical for many applications such as path planning \citep{Ponulak2013-op,Pfeiffer2013-qn,Erdem2012-uf} and resolving local ambiguities position within the global context \citep{Hardcastle2015-as,Evans2016-pr,Fyhn2007-ys}. In this regime, the external signal updates the internal representation with finite `gain' \citep{Fyhn2007-ys} and can thus fruitfully combine information in both the internal representation and the external signal. Other applications that simply require short-term memory storage of a strongly fluctuating variable may not require this continuity restriction.

\begin{figure}
\begin{center}
\includegraphics[width=0.8\linewidth]{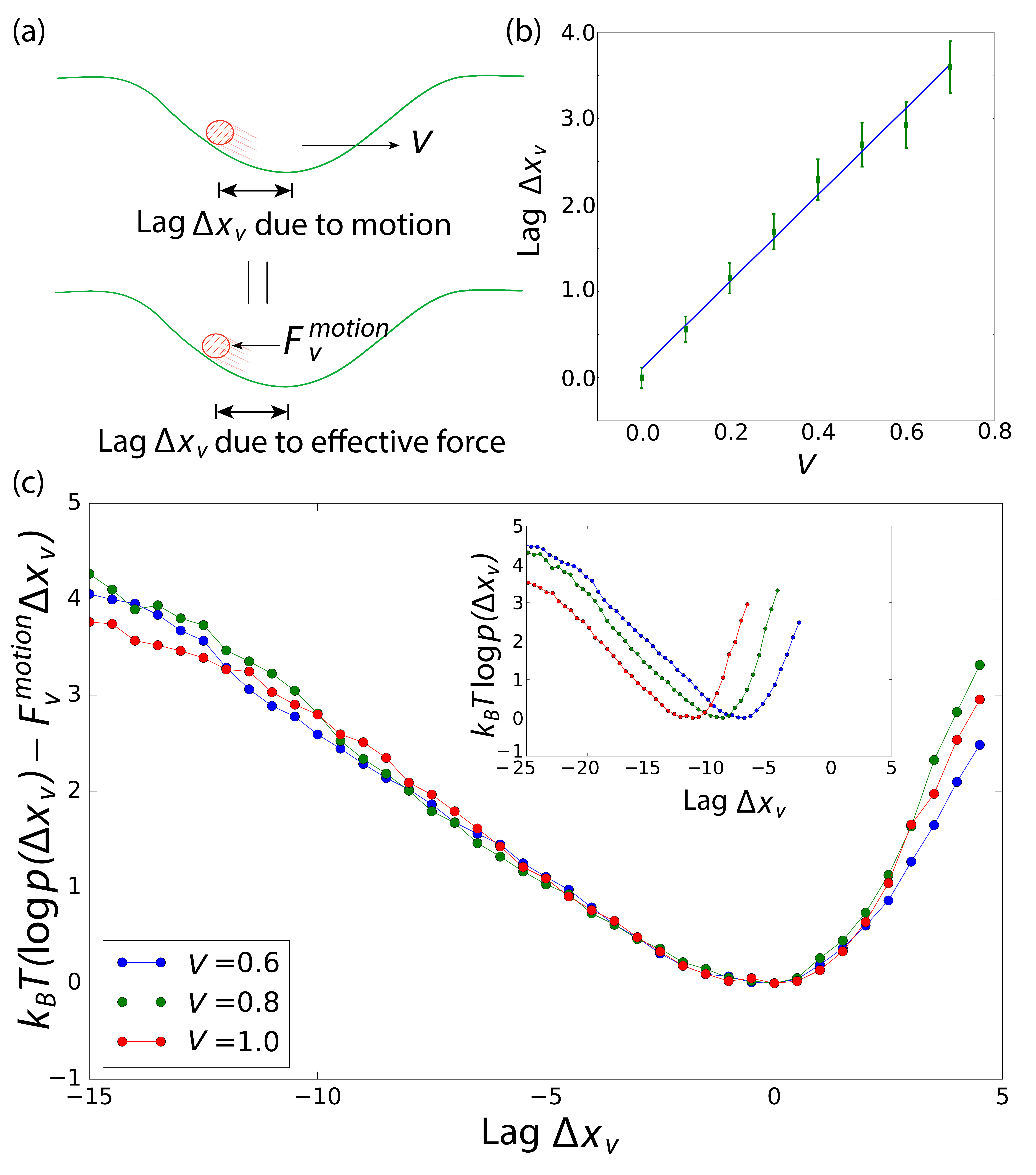}
\end{center}
\caption{(a) The mean position and fluctuations of the droplet driven by currents $I_{n}^{ext} = I^{cup}(n-vt)$ are described by an `equivalence' principle; in a frame co-moving with $I_n^{cup}(t)$ with velocity $v$, we simply add an effective force $F^{motion}_{v} = \gamma v$ where $\gamma$ is a drag coefficient. (b) This prescription correctly predicts that the droplet lags the external driving force by an amount linearly proportional to velocity $v$, as seen in simulations. (c) Fluctuations of the driven droplet's position, due to internal noise in neurons, are also captured by the equivalence principle. If $p(\Delta x_v)$ is the probability of finding the droplet at a lag $\Delta x_v$, we find that $k_B T \log p(\Delta x_v) - k_B T F^{motion}_{v} \Delta x_v$ is independent of velocity and can be collapsed onto each other (with fitting parameter $T$).  (Inset: $\log p(\Delta x_v)$ before subtracting $F^{motion}_{v} x$.)
\label{fig:equiv principle}}
\end{figure}

\subsection*{Equivalence principle}
We first consider driving the droplet in a network at constant velocity $v$ using an external current $I^{ext}_n = I^{cup}(n - vt)$. We allow for Langevin noise but no disorder in the couplings $J^{d} = 0$ in this section. For very slow driving ($v \to 0$), the droplet will settle into and track the bottom of the cup. When driven at a finite velocity $v$, the droplet cannot stay at the bottom since there is no net force exerted by the currents $I^{ext}_n$ at that point.  

Instead, the droplet must lag the bottom of the moving external drive by an amount $ \Delta x_v  = \bar{x} - v t$ such that the slope of the potential $V^{cup}$ provides an effective force $F_v^{motion} \equiv \gamma v$ needed to keep the droplet in motion at velocity $v$. That is, %the lag $\Delta x_v$ when averaged over a long trajectory, must be,
\begin{equation} 
-\partial_{\bar{x}} V^{cup}(\langle \Delta x_v \rangle)  = F_v^{motion} \equiv \gamma v.
\label{eqn:equivalance}
\end{equation}
%This equation is effectively an `equivalence' principle for over-damped motion -- in analogy with inertial particles accelerated in a potential, the droplet lags to a point where the slope of the driving potential provides sufficient force to keep the droplet in motion at that velocity. 
This equation, which we call an `equivalence principle' in analogy with inertial particles in an accelerated frame, is verified by simulations in Fig.~\ref{fig:equiv principle}b.  Similar results on a lag between driving forces and the response were obtained in earlier works \citep{fung2015fluctuation,mi2014spike}.

% verifies that the average lag $\langle \Delta x_v \rangle$ depends on velocity in a way described by Eqn.~\ref{eqn:equivalance}.

In fact, we find that the the above `equivalence' principle predicts the entire distribution $p(\Delta x_v)$ of fluctuations of the lag $\Delta x_v$ due to Langevin noise; see Fig.\ref{fig:equiv principle}c. By binning the lag $\Delta x_v(t)$ for trajectories of the droplet obtained from repeated numerical simulations, we determined $p(\Delta x_v)$, the occupancy of the droplet in the moving frame of the drive. As detailed in the Appendix, data for different velocities collapses using an effective temperature scale $T$, verifying that %that $\log p(\Delta x_v)$ for different velocities corresponds to the same quadratic potential $V^{cup}$ plus a velocity-dependent linear potential, $-F^{motion}_v \Delta x_v$, in agreement with the equivalence principle. That is,
\begin{equation}
\label{eqn:fluc lag}
k_B T \log p(\Delta x_v) = -(V^{{cup}}(\Delta x_v) - F^{{motion}}_v \Delta x_v),
\end{equation}
%for some effective temperature scale $T$ for the collective coordinate $\bar{x}$, ultimately set by $\eta_{int}(t)$. () 

%As a result, the $\log p(\Delta x_v)$ for different velocities collapse onto each other upon subtracting the linear potential due to the motion force, as shown in Fig.\ref{fig:equiv principle}c. 

Our results here are consistent with the fluctuation-dissipation result obtained in \citep{Monasson2013-kx} for driven droplets. In summary, in the co-moving frame of the driving signal, the droplet's position $\Delta x_v$ fluctuates as if it were in thermal equilibrium in the modified potential $V^{eff} = V^{cup} - F^{motion}_v \Delta x_v$.  

\section*{Speed limits on updates of internal representation}
The simple `equivalence principle' implies a striking bound on 
%These results for the distribution of the lag $\Delta x_v$, captured by a simple `equivalence principle', imply a striking restriction 
the update speed of internal representations. A driving signal cannot drive the droplet at velocities greater than some $v_{crit}$ if the predicted lag for $v > v_{crit}$ is larger than the cup. In the Appendix, we find $v_{crit} = 2d(w+R)/3\gamma$, where $2R$ is the droplet size.

Larger driving strength $d$ increases $v_{crit}$, but as was previously discussed, we require $d < d_{max}$ in order to retain continuity and stability of the internal representation. %, i.e. to prevent teleportation of the activity bump. 
Hence, we find an absolute upper bound on the fastest external signal that can be tracked by the internal representation,

\begin{equation}
v^* = \kappa p J\gamma^{-1},
\label{eqn:vfund}
\end{equation}
where $p$ is the range of interactions, $J$ is the synaptic strength, $\gamma^{-1}$ is the mobility or inverse drag coefficient of the droplet, and $\kappa$ is a dimensionless $\mathcal{O}(1)$ number.
%$\kappa = 8/3(c^{-1}-2)$ is a dimensionless $\mathcal{O}(1)$ number.

%\subsection*{Equivalence principle with Langevin noise}
%Before adding quenched disorder to the coupling matrix, we first simulate the system at finite temperature with only the bare connectivity matrix $J^0_{ij}$ in Eqn.\eqref{eqn:JijBreakdown}. 
%In this case, the coordinate $\bar{x}$ fluctuates around its mean value. 

\begin{figure}
\begin{center}
     \includegraphics[width=0.8\linewidth]{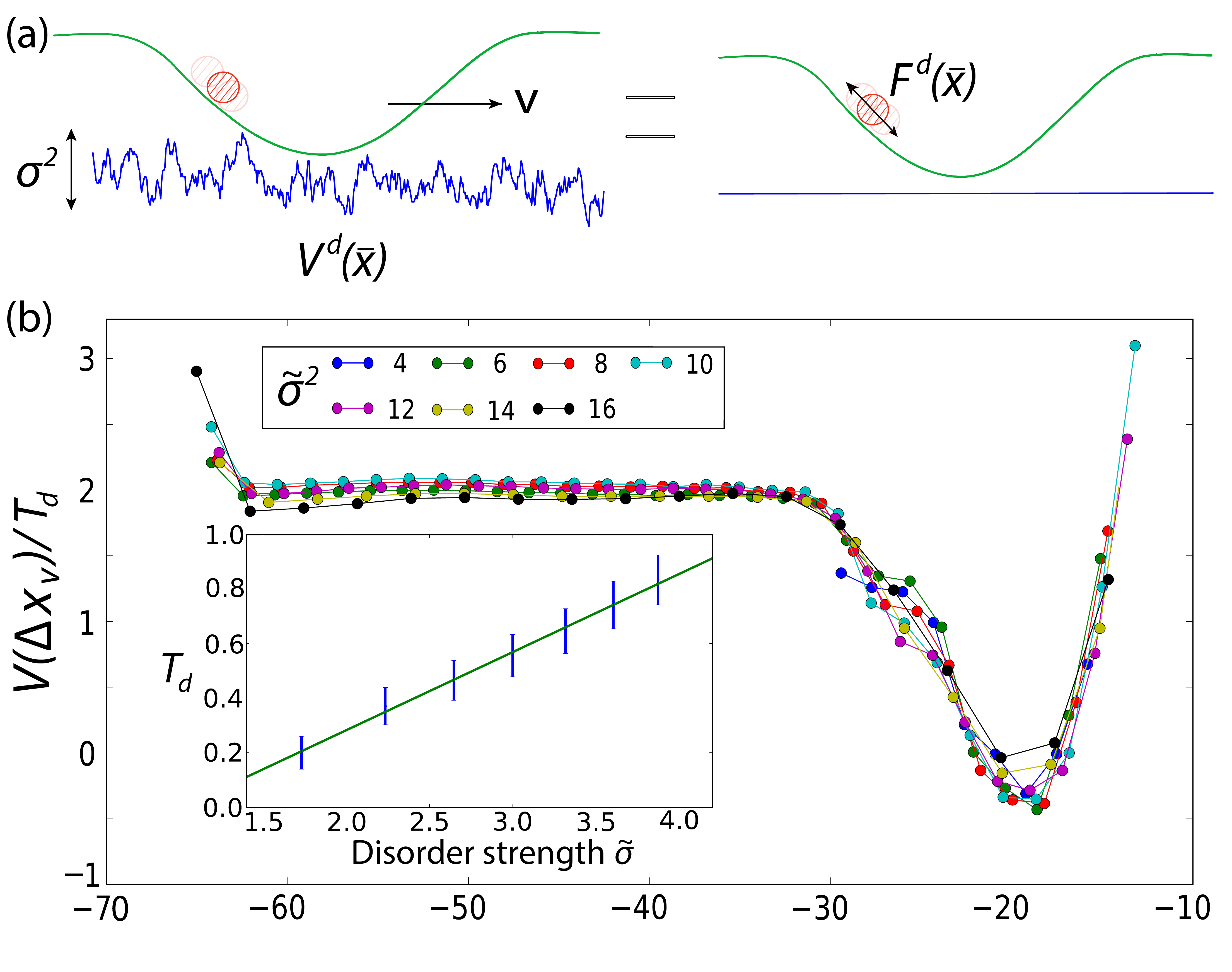}
\end{center}
\caption{Disorder in neural connectivity is well-approximated by an effective temperature $T_d$ for a moving droplet. (a) Long-range disorder breaks  the degeneracy of the continuous attractor, creating a rough landscape. A droplet moving at velocity $v$ in this rough landscape experiences random forces. (b) The fluctuations of a moving droplet's position, relative to the cup's bottom, can be described by an effective temperature $T_{d}$. We define a potential $V(\Delta x_v) = - k_B T_{d} \log p(\Delta x_v)$ where $p(\Delta x_v)$ is the probability of the droplet's position fluctuating to a distance $\Delta x_v$ from the peak external current. We find that $V(\Delta x_v)$ corresponding to different amounts of disorder $\tilde{\sigma}^2$ (where $\tilde{\sigma}^2$ is the average number of long-ranged disordered connections per neuron in units of $2p$), can be collapsed by the one fitting parameter $T_{d}$. (inset) $T_{d}$ is linearly proportional to the strength of disorder $\tilde{\sigma}$. \label{fig:Tbumpy}}
\end{figure}

\begin{figure}
\begin{center}
    \includegraphics[width=0.8\linewidth]{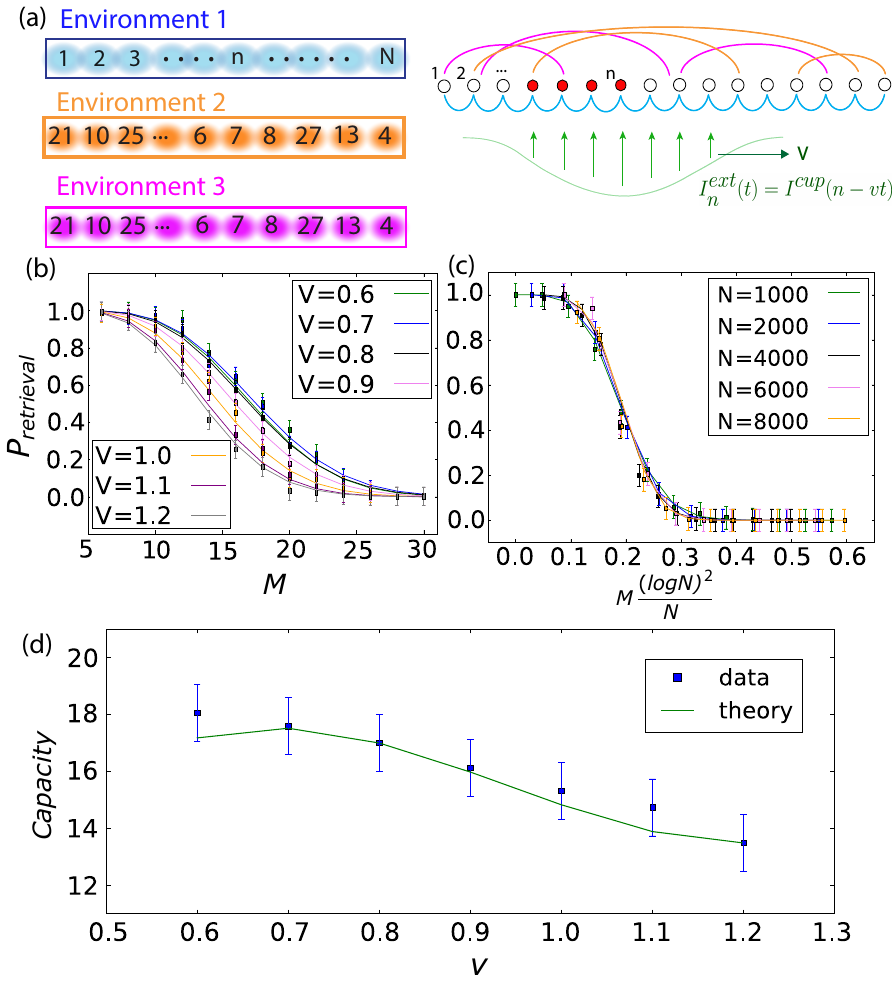}
\end{center}
\caption{Non-equilibrium capacity of place cell networks limits retrieval of spatial memories at finite velocity. (a) Place cell networks model the storage of multiple spatial memories in parts of the hippocampus by coding multiple continuous attractors in the same set of neurons. Neural connections encoding spatial memory 2,3,\ldots act like long range disorder for spatial memory 1. Such disorder, through an increased effective temperature, reduces the probability of tracking a finite velocity driving signal. (b)  The probability of successful retrieval, $P_{retrieval}$, decreases with the number of simultaneous memories $M$ and velocity $v$ (with $N=4000,p=10,\epsilon=0.35,\tau=1,J=100,d=10,w=30$ held fixed). (c) $P_{retrieval}$ simulation data collapses when plotted against $M/(N/(\log N)^2)$ (parameters same as (b) with $v=0.8$ held fixed and $N$ varies).
(d) The non-equilibrium capacity $M_c$ as a function of retrieval velocity $v$.\label{fig:capacity}}
\end{figure}

\section*{Disordered connections and effective temperature}

We now consider the effect of long-range quenched disorder $J_{ij}^{d}$ in the synaptic matrix \citep{Sebastian_Seung1998-vm,Kilpatrick2013-go}, which breaks the exact degeneracy of the continuous attractor, creating an effectively rugged landscape, $V^{d}(\bar x)$, %$V^{d}(\bar x) =  -\frac{1}{2} \sum_{nk}J^{d}_{nk}f(i_k^{\bar{x}})f(i_n^{\bar{x}})$  
as shown schematically in Fig. \ref{fig:Tbumpy} and computed in the Appendix. %, where $\{i_k^{\bar{x}}\}$ denotes the current configuration where the c.o.m is at $\bar{x}$.%We find that the equivalence principle still holds - there is an effective potential $V$ and the center of mass sits at a point such that $\partial_x V = \gamma v$. 
When driven by a time-varying external signal, $I^{ext}_i(t)$, the droplet now experiences a net potential $V^{ext}(\bar{x},t) + V^{d}(\bar{x})$. The first term causes motion with velocity $v$ and a lag predicted by the equivalence principle, 
%The second term $V^{d}(\bar{x})$ is difficult to handle in general. 
and for sufficiently large velocities $v$, the effect of the second term can be modeled as effective Langevin white noise. To see this, note that $V^{d}(\bar{x})$ is uncorrelated on length scales larger than the droplet size; hence for large enough droplet velocity $v$, the forces $F^{d}(t) \equiv -\partial_{\bar{x}}V^d\vert_{\bar{x} = \bar{x}(t)}$ due to disorder are effectively random and uncorrelated in time. More precisely, let $\sigma^2 = \text{Var}(V^d(\bar{x}))$. In the Appendix, we compute $F^{d}(t)$ and show that $F^d(t)$ has an auto-correlation time, %$C(t,t') \equiv \langle F^{d}(t) F^{d}(t')\rangle = \sigma^2/4R^2 \exp(-|t-t'|/\tau_{cor})$, where
$\tau_{cor} = 2R/v$ due to the finite size of the droplet.  %, where $\sigma^2$ is the variance of $F^{d}$. On timescales larger than $\tau_{cor}$, $F^{d}(t)$ acts like white noise %and we effectively have $C(t,t') = T_d \delta(t-t')$

Thus, on longer timescales, $F^{d}(t)$ is uncorrelated and can be viewed as Langevin noise for the droplet center of mass $\bar{x}$, associated with a disordered-induced temperature $T_d$. Through repeated simulations with different amounts of disorder $\sigma^2$, we inferred the distribution $p(\Delta x_v)$ of the droplet position in the presence of such disorder-induced fluctuations; see Fig. \ref{fig:Tbumpy}. The data collapse in Fig. \ref{fig:Tbumpy}b confirms that the effect of disorder (of size $\sigma^2$) on a rapidly moving droplet can indeed be modeled by an effective disorder-induced temperature $T_{d} \sim \sigma \tau_{cor}$.
%As we see in Fig. \ref{fig:Tbumpy}, the occupancies inferred for different disorder size collapse onto each other after rescaling by $T_d$. 
(For simplicity, we assume that internal noise $\eta_{int}$ in Eqn.\eqref{eqn:eom_zero} is absent here.Note that in general $\eta_{int}$ will also contribute to $T_d$. Here we focus on the contribution of disorder to an effective temperature $T_d$ since internal noise $\eta_{int}$ has been considered in prior works \citep{fung2015fluctuation}.)  % if we assume a temperature $T_{d} \propto \sigma e^{-v\tau}$. %, where %we choose $\tau_{obs} = 4w\tau$ throughout 
%\textcolor{red}{In the case when there are fluctuations present in the strength of the external signal, $T_d$ then describes the mean of the disorder-induced temperature. }

Thus, the disorder $J_{ij}^d$ effectively creates thermal fluctuations about the lag predicted by the equivalence principle; such fluctuations may carry the droplet out of the driving cup $I^{cup}(n-vt)$ and prevent successful update of the internal representation. We found that this effect can be quantified by a simple Arrhenius-like law, \begin{equation}
\label{eqn:DeltaE}
r \sim \exp(- \Delta E(v,d) / k_{B} T_{d})
\end{equation}

where $\Delta E(v,d)$ is the energy gap between where the droplet sits in the drive and the escape point, predicted by the equivalence principle, and $T_d$ is the disorder-induced temperature.  Thus, given a network of $N$ neurons, the probability of an external drive moving the droplet successfully across the network is proportional to $\exp(- r N)$. (Note that $r$ depends on $N$ in a way such that $\exp(-r N)$ becomes a step function as $N\to \infty$: i.e., always successful below a critical amount of disorder (capacity), and always failing beyond this capacity.)

\section*{Implications: Memory capacity of driven place cell networks}

The capacity of a neural network to encode multiple memories has been studied in numerous contexts since Hopfield's original work \citep{Hopfield1982-fb}. While specifics differ \citep{Amit1985-as,Battaglia1998-eg,Monasson2013-pn,Hopfield2010-lf}, the capacity is generally set by the failure to retrieve a specific memory because of the effective disorder in neural connectivity due to other stored memories.

However, these works on capacity do not account for non-adiabatic external driving. Here, we use our results to determine the capacity of a place cell network \citep{OKeefe1971-jt,Battaglia1998-eg,Monasson2013-pn} to both encode and manipulate memories of multiple spatial environments at a finite velocity. Place cell networks \citep{ Tsodyks1999-px,Monasson2013-kx, Monasson2015-nl, Monasson2014-nu, Monasson2013-pn} encode memories of multiple spatial environments as multiple continuous attractors in one network. Such networks have been used to describe recent experiments on place cells and grid cells in the hippocampus \citep{Yoon2013-nl, Hardcastle2015-as,Moser2014-lw}.

In experiments that expose a rodent to different spatial environments $\mu = 1,\ldots M$ \citep{Alme2014-qc, Moser2017-dj, Kubie1991-ge}, the same place cells $i=1,\ldots N$ are seen having `place fields' in different spatial arrangements $\pi^\mu(i)$ as seen in Fig.\ref{fig:capacity}A, where $\pi^\mu$ is a permutation specific to environment $\mu$. Consequently, Hebbian plasticity suggests that each environment $\mu$ would induce a set of synaptic connections $J_{ij}^{\mu}$ that corresponds to the place field arrangement in that environment; i.e., $J_{ij}^{\mu} = J(1-\epsilon)$ if $|\pi^{\mu}(i) - \pi^{\mu}(j)| < p $. That is, each environment corresponds to a $1$-D network when the neurons are laid out in a specific permutation $\pi^\mu$. The actual network has the sum of all these connections $J_{ij} = \sum_{\mu=1}^M J_{ij}^\mu$ over the $M$ environments the rodent is exposed to.

While $J_{ij}$ above is obtained by summing over $M$ structured environments, from the perspective of, say, $J_{ij}^{1}$, the remaining $J_{ij}^{\mu}$ look like long-range disordered connections. We will assume that the permutations $\pi^{\mu}(i)$ corresponding to different environments are random and uncorrelated, a common modeling choice with experimental support \citep{Hopfield2010-lf,Monasson2015-nl, Monasson2014-nu,Alme2014-qc,Moser2017-dj}. Without loss of generality, we assume that $\pi^{0}(i) = i$ (blue environment in Fig.\ref{fig:capacity}.) Thus, $J_{ij} = J_{ij}^{0} + J_{ij}^{d}, J_{ij}^{d} = \sum_{\mu = 1}^{M-1} J_{ij}^{\mu}$. The disordered matrix $J^{d}_{ij}$ then has an effective variance $\sigma^2 \sim (M-1)/N$. Hence, we can apply our previous results to this system. 
Now consider driving the droplet with velocity $v$ in Environment 1 using external currents. The probability of successfully updating the internal representation over a distance $L$ is given by $P_{retrieval} = e^{- rL/v}$, where $r$ is given by Eqn.\eqref{eqn:DeltaE}.

In the thermodynamic limit $N \to \infty$, with $w, p, L/N$ held fixed, $P_{retrieval}$ becomes a Heaviside step function $\Theta (M_c-M)$ at some critical value $M_c$ given by
\begin{equation}
\label{eqn:capacity}
%M_c \sim \bigg[\frac{\Delta E(v,d)}{e^{-\omega v\tau}}\bigg]^2 \frac{N}{(\log N)^2}
M_c \sim \bigg[v\Delta E(v,d)\bigg]^2 \frac{N}{(\log N)^2}
\end{equation}
for the largest number of memories that can be stored and retrieved at velocity $v$. $\Delta E(v,d) = (4dw-3\gamma v-2dR)(-v\gamma+2dR)/4d $. Fig.\ref{fig:capacity} shows that our numerics agree well with this formula, showing a novel dependence of the capacity of a neural network on the speed of retrieval and the strength of the external drive.Note that the fact that Eqn.\eqref{eqn:capacity} scales sublinearly in $N$ reflects our choice of `perfect' retrieval in the definition of successful events. As in earlier works \cite{Hopfield1982-fb,hertz1991introduction,Amit1985-ls,Amit1985-as}, the precise definition of capacity can change capacity by $\log$ factors.

%Experimental implications

\section*{Conclusion}

Thus we have considered continuous attractors in neural networks driven by localized time-dependent currents $I^{cup}(n-vt)$;  in recent experiments, such currents can represent landmark-related sensory signals \citep{Hardcastle2015-as} when a rodent is traversing a spatial environment at velocity $v$, or signals that update the internal representation of head direction \citep{Seelig2015-uu}. Several recent experiments have controlled the effective speed of visual stimuli in virtual reality environments \citep{Meshulam2017-ue, Aronov2017-jy,Kim2017-zs, Turner-Evans2017-ch}. Other experiments have probed crosstalk between memories of multiple spatial environments \citep{Alme2014-qc}. Our results predict an error rate that rises with speed and with the number of  environments.

While our analysis used specific functional forms for, e.g., the current profile $I^{cup}(n-vt)$, our bound simply reflects the finite response time in moving emergent objects, much like moving a magnetic domain in a ferromagnet using space and time varying fields. Thus we expect our bound to hold qualitatively for other related forms \citep{Hopfield2015-wt}.%, including those used in recent experiments.

In addition to positional information considered here, continuous attractors are known to also receive velocity information \citep{Major2004-ku,McNaughton2006-xq,Seelig2015-uu,Ocko2018-gv}. We do not consider such input in the main text but extend our analysis to velocity integration in the Appendix. 

% Discussion
In summary, we found that the non-equilibrium statistical mechanics of a strongly interacting neural network can be captured by a simple equivalence principle and a disorder-induced temperature for the network's collective coordinate. Consequently, we were able to derive a velocity-dependent bound on the number of simultaneous memories that can be stored and retrieved from a network. We discussed how these results, based on general theoretical principles on driven neural networks, allow us to connect robustly to recent time-resolved experiments in neuroscience\citep{ Kim2017-zs,Turner-Evans2017-ch,Hardcastle2015-as,hardcastle2017multiplexed,campbell2018principles} on the response of neural networks to dynamic perturbations.

\section*{Appendix}

\section*{Equations for the collective coordinate}
\label{app:eqn_cc}
% Introduce place cell model  
As in the main text, we model $N$ interacting neurons as,

\begin{equation}
\begin{split}
\label{eqn:eom}
\frac{di_n}{dt} &= -\frac{i_n}{\tau} + \sum_{k = 1}^{N} J_{nk}f(i_k) + I_n^{ext}(t) + \eta^{int}_n(t),
\\
&\text{where} \; f(i) = \frac{1}{1+e^{-i/i_0}}.
\end{split}
\end{equation}
The synaptic connection between two different neurons $i,j$ is $J_{ij} = J(1 - \epsilon)$ if neurons $i$ and $j$ are separated by a distance of at most $p$ neurons, and $J_{ij}= - J \epsilon$ otherwise, and note that we set the self-interaction to zero. The internal noise is a white noise, $\langle \eta^{int}_n(t) \eta^{int}_n(0) \rangle = C_{int}\delta(t)$ with an amplitude $C_{int}$. $I_n^{ext}(t)$ are external  driving currents discussed below.

% Collective coordinate
Such a quasi 1-d network with $p$-nearest neighbor interactions resembles a similarly connected network of Ising spins at fixed magnetization in its behavior; the strength of inhibitory connections $\epsilon$ constrains the total number of neurons $2R$ firing at any given time to $2R \sim p \epsilon^{-1}$. It was shown \citep{Hopfield2010-lf,Monasson2013-kx, Monasson2014-nu, Monasson2013-pn} that below a critical temperature $T$, the $w$ firing neurons condense into a contiguous droplet of neural activity, minimizing the total interface between firing and non-firing neurons. Such a droplet was shown to behave like an emergent quasi-particle that can diffuse or be driven around the continuous attractor. We define the center of mass of the droplet as,
\begin{equation}
\bar{x} \equiv \sum_n n f(i_n).
\end{equation}
The description of neural activity in terms of such a collective coordinate $\bar{x}$ greatly simplifies the problem, reducing the configuration space from the $2^N$ states for the $N$ neurons to $N$-state consists of the center of mass of the droplet along the continuous attractor \citep{Wu2008-iw}. Computational abilities of these place cell networks, such as spatial memory storage, path planning and pattern recognition, are limited to parameter regimes in which such a collective coordinate approximation holds (e.g., noise levels less than a critical value $T < T_c$) .

% Driving forces
The droplet can be driven by external signals such as sensory or motor input or input from other parts of the brain. We model such external input by the currents $I_{n}^{ext}$ in Eqn.\ref{eqn:eom}; for example, sensory landmark-based input \citep{Hardcastle2015-as} when an animal is physically in a region covered by place fields of neurons $i, i+1,\ldots, i+z$, currents $I^{ext}_{i}$ through $I^{ext}_{i+z}$ can be expected to be high compared to all other currents $I^{ext}_j$. Other models of driving in the literature include adding an anti-symmetric component $A_{ij}$ to synaptic connectivities  $J_{ij}$ \citep{Ponulak2013-op}; we consider such a model in Appendix \ref{Aij}.

Let $\{i_k^{\bar{x}}\}$ denote the current configuration such that the droplet is centered at location $\bar{x}$. The Lyapunov function of the neural network is given by\citep{Hopfield2015-wt}, 

\begin{equation}
	\label{eqn:lyapunov}
	\begin{split}
		\mathcal{L}[\bar{x}] &\equiv \mathcal{L}[f(i_k^{\bar{x}})] \\
		&= \frac{1}{\tau}\sum_{k} \int_0^{f(i_k^{\bar{x}})}f^{-1}(x)dx\\
		& -\frac{1}{2} \sum_{n,k}J_{nk}f(i_k^{\bar{x}})f(i_n^{\bar{x}})- \sum_k f(i_k^{\bar{x}})I^{ext}_k(t).
	\end{split}
\end{equation}

%We can think of the Lyapunov function as an energy function that the dynamics of the system tends to minimize
%, with the caveat that this 'energy' does not correspond to a Boltzmann distribution. 
In a minor abuse of terminology, we will refer to terms in the Lyapunov function as energies, even though energy is not conserved in this system. For future reference, we denote the second term $V_J(\bar{x}) = -1/2\sum_{nk}J_{nk}f(i_k^{\bar{x}})f(i_n^{\bar{x}})$, which captures the effect of network synaptic connectivities. Under the `rigid bump approximation' used in \citep{Hopfield2015-wt},i.e., ignoring fluctuations of the droplet, we find,

\begin{eqnarray}
V_J(\bar{x}) &= -\frac{1}{2}\sum_{n,k} f(i_n^{\bar{x}}) J_{nk} f(i_k^{\bar{x}}) \\ &\approx -\frac{1}{2}\sum_{\substack{|n-\bar{x}| \le R, \\ |k-\bar{x}| \le R}} f(i_n^{\bar{x}}) J_{nk} f(i_k^{\bar{x}}).
\end{eqnarray}

For a quasi 1-d network with $p$-nearest neighbor interactions and no disorder, $V_J(\bar{x})$ is constant, giving a smooth continuous attractor. However, as discussed later, at the presence of disorder, $V_J(\bar{x})$ has bumps (i.e. quenched disorder) and is no longer a smooth continuous attractor.

To quantify the effect of the external driving, we write the third term in Eqn.\eqref{eqn:lyapunov},
 \begin{eqnarray}
 V^{ext}(\bar{x},t) &=& -\sum_k I^{ext}_k(t) f(i_k^{\bar{x}}) \\ 
 &\approx & -\sum_{|k-\bar{x}| < R} I^{ext}_k(t) f(i_k^{\bar{x}})
 \end{eqnarray}
 
Thus, the external driving current $I_n^{ext}(t)$ acts on the droplet through the Lyapunov function $V^{ext}(\bar{x},t)$. Hence we define
\begin{equation}
 F^{ext}(\bar{x},t) = -\partial_{\bar{x}} V^{ext}(\bar{x},t)  \\
\end{equation}
to be the external force acting on the droplet center of mass. %Past work has considered noisy or constant $I^{ext}$. %The simplest current with a non-constant spatial profile is a ramp of the form $I_n^{ext} = n$ where $F^{ext}(\bar{x},t)$ is a constant.

\subsection*{Fluctuation and dissipation}
\label{fluctuation-dissipation}

We next numerically verify that the droplet obeys a fluctuation-dissipation-like relation by driving the droplet using external currents $I^{ext}$ and comparing the response to diffusion of the droplet in the absence of external currents.

%for an effective temperature $T$,

We use a finite ramp as the external driving, $I^{ext}_n = n$ with $n < n_{max}$, and $I^{ext}_n = 0$ otherwise (see Fig.\ref{fig:fluctuation-dissipation}(a)). We choose $n_{max}$ to be such that it takes considerable time for the droplet to relax to its steady-state position at the end of the ramp. We notice that for different slopes of the $I^{ext}_n$, the droplet have different velocities, and it is natural to define a mobility of the droplet, $\mu$, by $v = \mu f$, where $f$ is the slope of $I^{ext}_n$.
Next, we notice that on a single continuous attractor the droplet can diffuse because of internal noise in the neural network. Therefore, we can infer the diffusion coefficient $D$ of the droplet from $\langle x^2 \rangle = 2Dt$ for a collection of diffusive trajectories (see Fig.\ref{fig:fluctuation-dissipation}(b)), where we have used $x$ to denote the center of mass $\bar{x}$ for the droplet to avoid confusion.

In Fig.\ref{fig:fluctuation-dissipation}(c) we numerically verify that $\mu$ and $D$ depend on parameters $\tau$ and $R$ in the same way, i.e. $D$ and $\mu$ are both proportional to $1/\tau$ and independent of $R$. This suggest that $D \propto \mu$, if we call the proportionality constant to be $k_BT$, then we have a fluctuation-dissipation-like relation,
\begin{equation}
\label{eqn:FDT}
D = \mu k_B T.
\end{equation} 

Note that Eqn.\eqref{eqn:FDT} has also been derived for the case of binary neurons with a hard constraint on the number of firing population \citep{Monasson2013-kx}.

\begin{figure}
\begin{center}
    \includegraphics[width=0.8\linewidth]{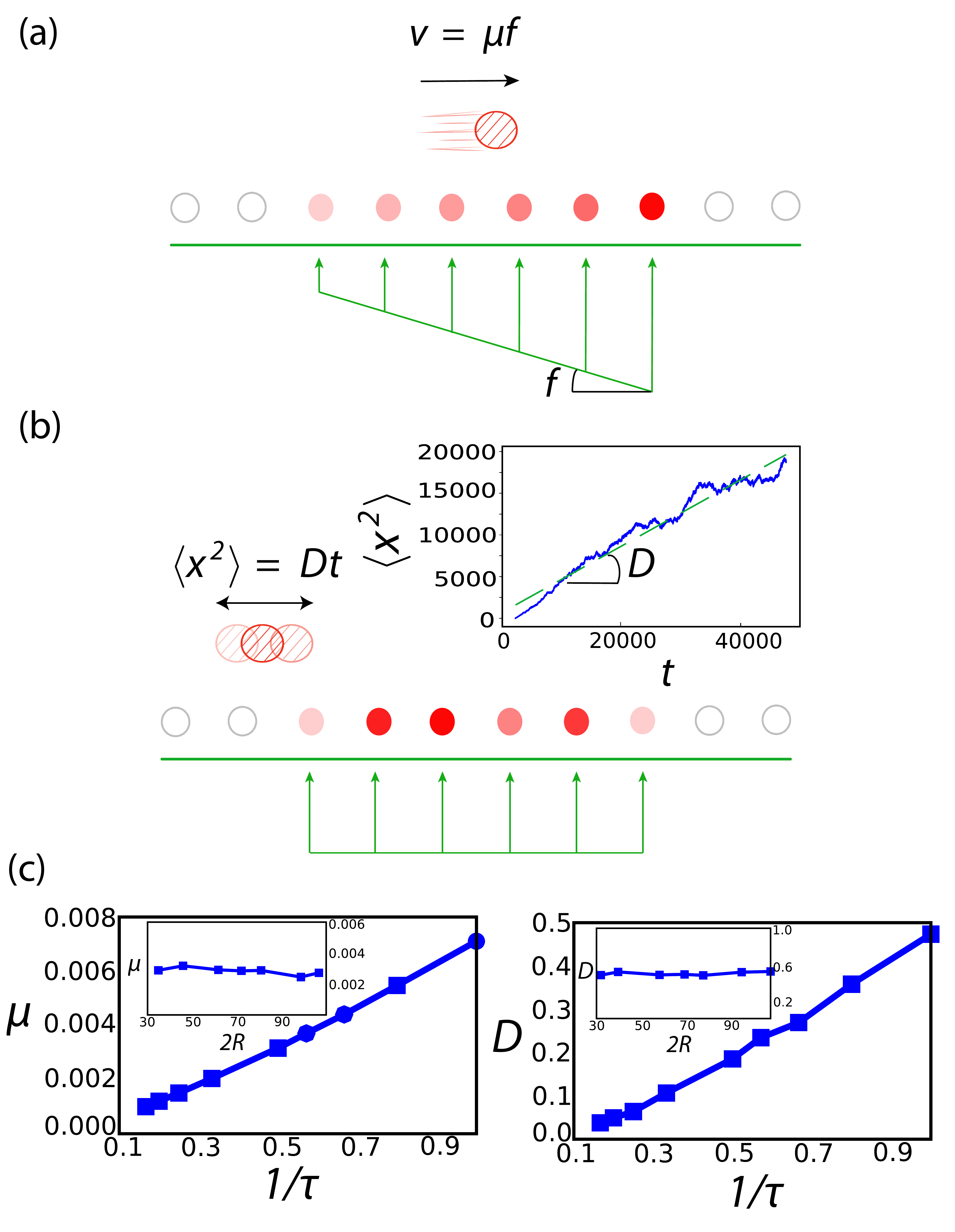}
\end{center}
\caption{(a) Schematics of the droplet being driven by a linear potential (ramp), illustrating the idea of mobility. Green lines are inputs, red dots are active neurons, the more transparent ones represent earlier time.
(b) Schematics of the droplet diffusing under an  input with no gradient, %(not necessary, but a good comparison with the linear potential),
giving rise to diffusion. Inset is the plot of mean-squared distance vs time, clearly showing diffusive behavior. Note here we have changed the droplet c.o.m. position $\bar{x}$ as $x$ to avoid confusion with the mean-position.
(c) Comparison between mobility $\mu = \gamma^{-1}$ and diffusion coefficient $D$. Both $\mu$ and $D$ depend on blob size and $\tau$ in the same way, and thus $D$ is proportional to $\mu$. \label{fig:fluctuation-dissipation}}
\end{figure}

\section{Space and time dependent external driving signals}

%The simplest current with a non-constant spatial profile is a ramp of the form $I_n^{ext} = n$ where $F_{ext}(\bar{x},t)$ is a constant.
%However, the corresponding potential is highly non-local, as it spans the entire system. Also, as $n$ becomes large, the size of the current will exceed the maximum signal size $d_{max}$ given in Eqn.\eqref{eqn:dmax} as we will show later, and the droplet will spontaneously 'teleport' to the large $n$ region of the attractor and the collective coordinate description is no longer valid. 

%Thus we only consider external inputs that are finite everywhere as $N \to \infty$. 
We consider the model of sensory input used in the main text: $I^{cup}(n) = d(w-|n|), n \in [-w,w]$, $I^{cup}(n) = 0$ otherwise.
We focus on time-dependent currents $I^{ext}_n(t) = I^{cup}(n - vt)$.
Such a drive was previously considered in \citep{Wu2005-sw}, albeit without time dependence. Throughout the paper, we refer to $w$ as the linear size of the drive, $d$ as the depth of the drive, and set the drive moving at a constant velocity $v$. From now on, we will go to the continuum limit and denote $I^{ext}_n(t) = I^{ext}(n,t) \equiv I^{ext}(x,t)$.

As an example, for $v=0$ (in this case, $\Delta x_v = \bar{x}$) we can write down the potential $V^{ext}$ for the external driving signal $I^{cup}(x) = d(w-|x|)$ by evaluating it at a stationary current profile $f(i_k^{\bar{x}}) = 1\; \text{if} \; |k-\bar{x}| \leq R, =0 \; \text{otherwise}$,
\begin{equation}
\label{eqn:Vext}
	V^{ext}(\bar{x}) = \begin{cases} 
		V_1(\bar{x}), & |\bar{x}| \leq R\\
		V_2(\bar{x}), & |\bar{x}| > R,
	\end{cases}
\end{equation}

where
\begin{equation}
	\begin{split}
		V_1(\bar{x}) &= -d\bigg [(R-\bar{x})(w-\frac{R-\bar{x}}{2}) + (R+\bar{x})(w-\frac{w+\bar{x}}{2})\bigg ]
		\\
		V_2(\bar{x}) &= -\frac{d}{2}(R+w-\bar{x})^2.
	\end{split}
\end{equation}

We plot $V^{ext}$ given by Eqn.\eqref{eqn:Vext} vs the c.o.m. position of droplet in Fig.\ref{fig:Vext panel}(a).

\begin{figure}
\begin{center}
    	\includegraphics[width=0.8\linewidth]{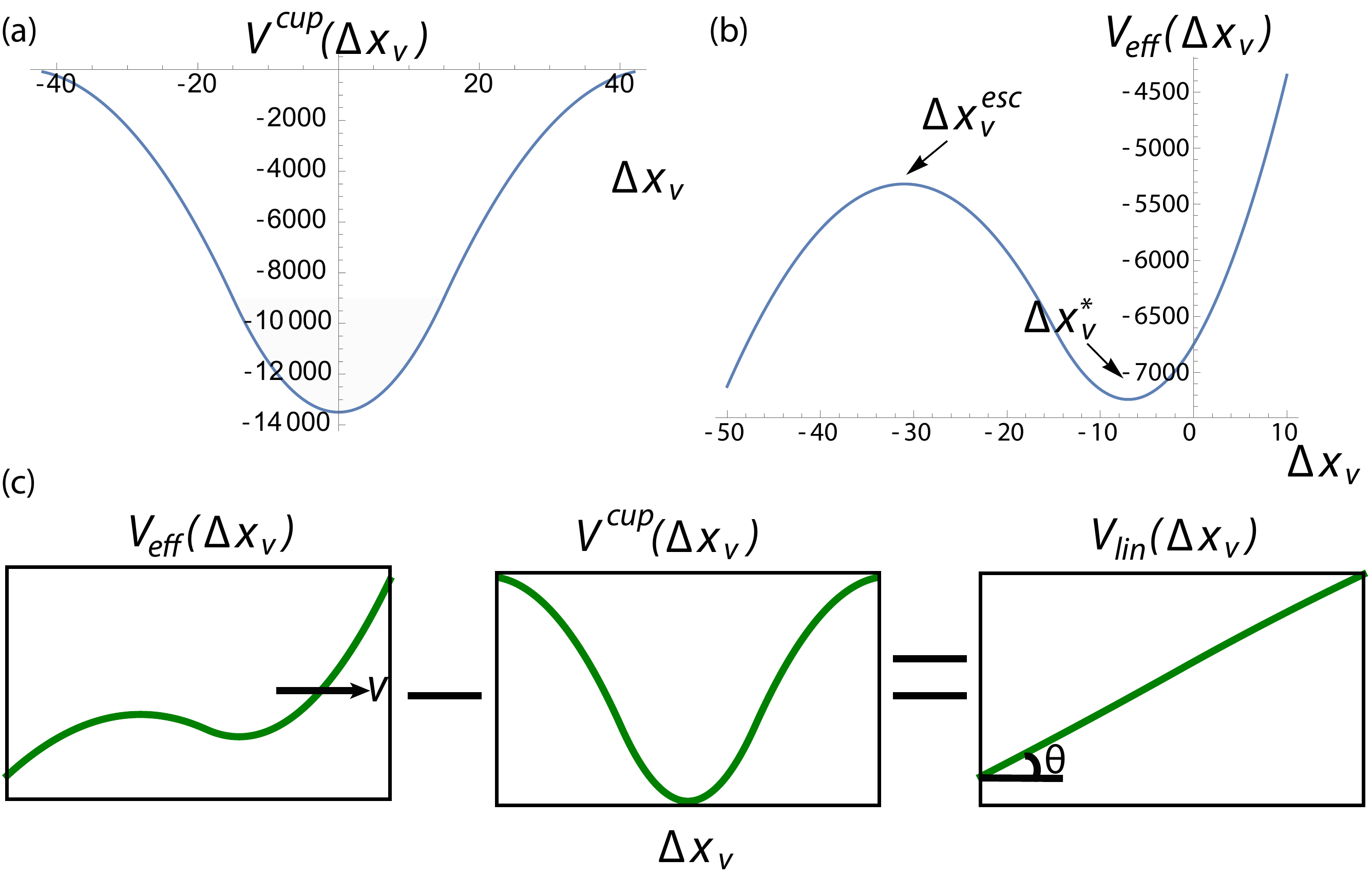}
\end{center}
	\caption{(a) $V^{ext}$ for external driving signal $I^{cup}(x,t)$ with $v=0$, plotted from Eqn.\eqref{eqn:Vext} with $d=20$, $R = 15$, $w = 30$. (b) Effective potential $V_{eff}$ experienced by the droplet for a moving cup-shaped external driving signal, plotted from Eqn.\eqref{eqn:Veff} with $d = 10$, $R= 15$, $w=30$, $\gamma v = 140$. (c) Schematic illustrating the idea of the equivalence principle (main text Eqn.(4)). The difference between the effective potential, $V_{eff} \equiv -k_BT\log p(\Delta x_v)$, experienced by a moving droplet, and that of a stationary droplet, $V^{cup}$, is a linear potential, $V_{lin} = -F_v^{motion} \Delta x_v$. The slope $\theta$ of the linear potential $V_{lin} = - F_v^{motion}\Delta x_v$ is proportional to velocity as $F_v^{motion} = \gamma v$. \label{fig:Vext panel}}
\end{figure}

\subsection*{A thermal equivalence principle}
The equivalence principle we introduced in the main text allows us to compute the steady-state position and the effective new potential seen in the co-moving frame. Crucially, the fluctuations of the collective coordinate are described by the potential obtained through the equivalence principle. The principle correctly predicts both the mean (main text Eqn.(4)) and the fluctuation (main text Eqn.(5)) of the lag $\Delta x_v$. Therefore, it is actually a statement about the equivalence of effective dynamics in the rest frame and in the co-moving frame. Specializing to the drive $I^{cup}(x,t)$, the equivalence principle predicts that the effective potential felt by the droplet (moving at constant velocity $v$) in the co-moving frame equals the effective potential in the stationary frame shifted by a linear potential, $V_{lin} = -F_v^{mot} \Delta x_v$, that accounts for the fictitious forces due to the change of coordinates (see Fig.\ref{fig:Vext panel}(c)).

Since we used \eqref{eqn:Vext} for the cup shape and the lag $\Delta x_v$ depends linearly on $v$, we expect that the slope of the linear potential $V_{lin}$ also depends linearly on $v$. Here the sign convention is chosen such that $V_{lin}<0$ corresponds to droplet moving to the right. 

\section{Speed limit for external driving signals}

In the following, we work in the co-moving frame with velocity $v$ at which the driving signal is moving. We denote the steady-state c.o.m. position in this frame to be $\Delta x^*_v$, and a generic position to be $\Delta x_v$.  

When $v>0$, the droplet will sit at a steady-state position $\Delta x_v^*<0$, equivalence principle says we should subtract a velocity-dependent linear potential $F^{mot}_v \Delta x_v = \gamma v \Delta x_v$ from $V^{ext}$ to account for the motion,
\begin{equation}
	\label{eqn:Veff}
	V_{eff}(\Delta x_v) = V^{cup}(\Delta x_v) - \gamma v \Delta x_v.
\end{equation}

We plot $V_{eff}$ vs $\Delta x_v$ in Fig.\ref{fig:Vext panel}(b). Notice that there are two extremal points of the potential, corresponding to the steady-state position, $\Delta x^*_v$, and the escape position, $\Delta x^{esc}_v$,

\begin{equation}
	\begin{split}
		\label{eqn:positions}
		\Delta x^*_v &= \gamma v/2d
		\\
		\Delta x^{esc}_v &= (dw-\gamma v + dR)/d.
	\end{split}
\end{equation}

We are now in position to derive $v_{crit}$ presented in the main text. We observe that as the driving velocity $v$ increases, $\Delta x^*_v$ and $\Delta x^{esc}_v$ will get closer to each other, and there will be a critical velocity such that the two coincide. %This can be seen in Fig.\ref{fig:vcrit}. 

By simply equating the expression for $x_{esc}$ and $x^*$ and solve for $v$, we found that
\begin{equation}
	\label{eqn:vcrit}
	v_{crit} = \frac{2d(w+R)}{3\gamma}.
\end{equation}

\subsection*{Steady-state droplet size}

Recall that the Lyapunov function of the neural network is given by \eqref{eqn:lyapunov},

\begin{equation}
\begin{split}
\mathcal{L}[\bar{x}]
&= \frac{1}{\tau}\sum_{k} \int_0^{f(i_k^{\bar{x}})}f^{-1}(x)dx\\
& + V_J(\bar{x})  + V^{ext}(\bar{x},t),
\end{split}
\end{equation}

Compared to the equation of motion \eqref{eqn:eom}, we see that the first term corresponds to the decay of neurons in the absence of interaction from neighbors (decay from 'on' state to 'off' state), and the second term corresponds to the interaction $J_{nk}$ term in the e.o.m, and the third term corresponds to the $I_n^{ext}$ in the e.o.m. Since we are interested in the steady-state droplet size, and thus only interested in the neurons that are 'on', the effect of the first term can be neglected (also note that $1/\tau \ll J_{ij}$, when using the Lyapunov function to compute steady-state properties, the first term can be ignored).

To obtain general results, we also account for long-ranged disordered connections $J^d_{ij}$ here. We assume $J^d_{ij}$ consists of random connections among all the neurons. We can approximate these random connections as random permutations of $J^0_{ij}$ and the full $J_{ij}$ is the sum over $M-1$ such permutations plus $J^0_{ij}$.

For the cup-shaped driving and its corresponding effective potential, Eqn.\eqref{eqn:Veff}, we are interested in the steady-state droplet size under this driving, so we first evaluate $V_{eff}$ at the steady-state position $\Delta x^*_v$ in Eqn.\eqref{eqn:positions}. To make the $R$-dependence explicit in the Lyapunov function, we evaluate $\mathcal{L}(\bar{x})$ under the 'rigid bump approximation' used in \citep{Hopfield2015-wt}, i.e., assuming $f(i_k^{\bar{x}}) = 1$ for $|k-\bar{x}| \leq R$, and $=0$ otherwise. 

%After some tedious and not-so-illuminating calculations (which are left as an exercise for the reader), 
We find that for $M-1$ sets of disorder interactions, the Lyapunov function is

\begin{equation}
	\begin{split}
		\label{eqn:blobenergy}
		\mathcal{L}[f(i_k^{\bar{x}})] &= J\bigg [ (\epsilon R^2 - (\epsilon + 2p)R + \frac{p(p+1)}{2} 
		\\
		&- pm(2R-p)^2    \bigg ] + \frac{(\gamma v)^2}{4d} + Rd(R-2w),
	\end{split}
\end{equation}
where we have defined the reduced disorder parameter $m = (M-1)/N$ and have used the equivalence principle in main text Eqn.(4) to add an effective linear potential to take into account the motion of the droplet.

Next, we note that the steady-state droplet size corresponds to a local extremum of the Lyapunov function. Extremizing Eqn.\eqref{eqn:blobenergy} with respect to droplet radius $R$, we obtain the steady-state droplet radius as a function of the external driving parameters $d,w$, and the reduced disorder parameter $m$, 

\begin{equation}
	\label{eqn:blobsize}
	R(d,w,m) = \frac{2p-4p^2m+2wd/J+\epsilon}{2d/J - 8pm + 4\epsilon},
\end{equation}

where we observe that in the formula the only dimensionful parameters $d$ and $J$ appears together to ensure the overall result is dimensionless. Our result for $R$ reduces to $	R_0 = \frac{p}{2\epsilon} + \frac{1}{4}$ by setting $M=1$ and $d=w=0$.
%Compared with the case without disordered interactions and external driving signals, 

%In fact, Eqn.\eqref{eqn:blobsize} already predicts a memory capacity, 
%\begin{equation}
%	\label{eqn:unstable}
%	m_c = \frac{\epsilon}{2p} + \frac{d}{4Jp},
%\end{equation}

%at which $R(d,w,m)$ blows up and the the droplet becomes unstable. This is similar to the capacity found in previous analysis \citep{Monasson2014-nu}\citep{Monasson2013-pn}, but with correction given by external signal strength $d$. 

\subsection*{Upper limit on external signal strength}

Here we present the calculation for maximal driving strength $I^{ext}$ %that leads up to the fundamental bound on external tracking velocity $v_{max}$ in main text Eqn.\eqref{eqn:vfund}. Intuitively, we have a bound on the maximal signal strength because when the driving signal (in the form of either a ramp or a cup) is too strong, 
beyond which the activity droplet will 'teleport' -- i.e., disappears at the original location and re-condense at the location of the drive, even if these two locations are widely separated. From now on, we refer to this maximal signal strength as the 'teleportation limit'. We can determine this limit by finding out the critical point where the energy barrier of breaking up the droplet at the original location is zero.

For simplicity, we assume that initially the cup-shaped driving signal is some distance $x_0$ from the droplet, and not moving (the moving case can be solved in exactly the same way by using equivalence principle and going to the co-moving frame of the droplet). We consider the following three scenarios during the teleportation process: $(1)$ the initial configuration: the droplet have not yet teleported, and stays at the original location with radius $R(0,0,m)$; $(2)$ the intermediate configuration: where the activity is no longer contiguous, giving a droplet with radius $\delta(d,w,m)$ at the center of the cup, and another droplet with radius $R(d,w,m)-\delta(d,w,m)$ at the original location (when teleportation happens, the total firing neurons changes from $R(0,0,m)$ to $R(d,w,m)$); $(3)$ the final configuration: the droplet has successfully teleported to the center of the cup, with radius $R(d,w,m)$. The three scenarios are depicted schematically in Fig.\ref{fig:teleportation}.

\begin{figure}
\begin{center}
    	\includegraphics[width=0.8\linewidth]{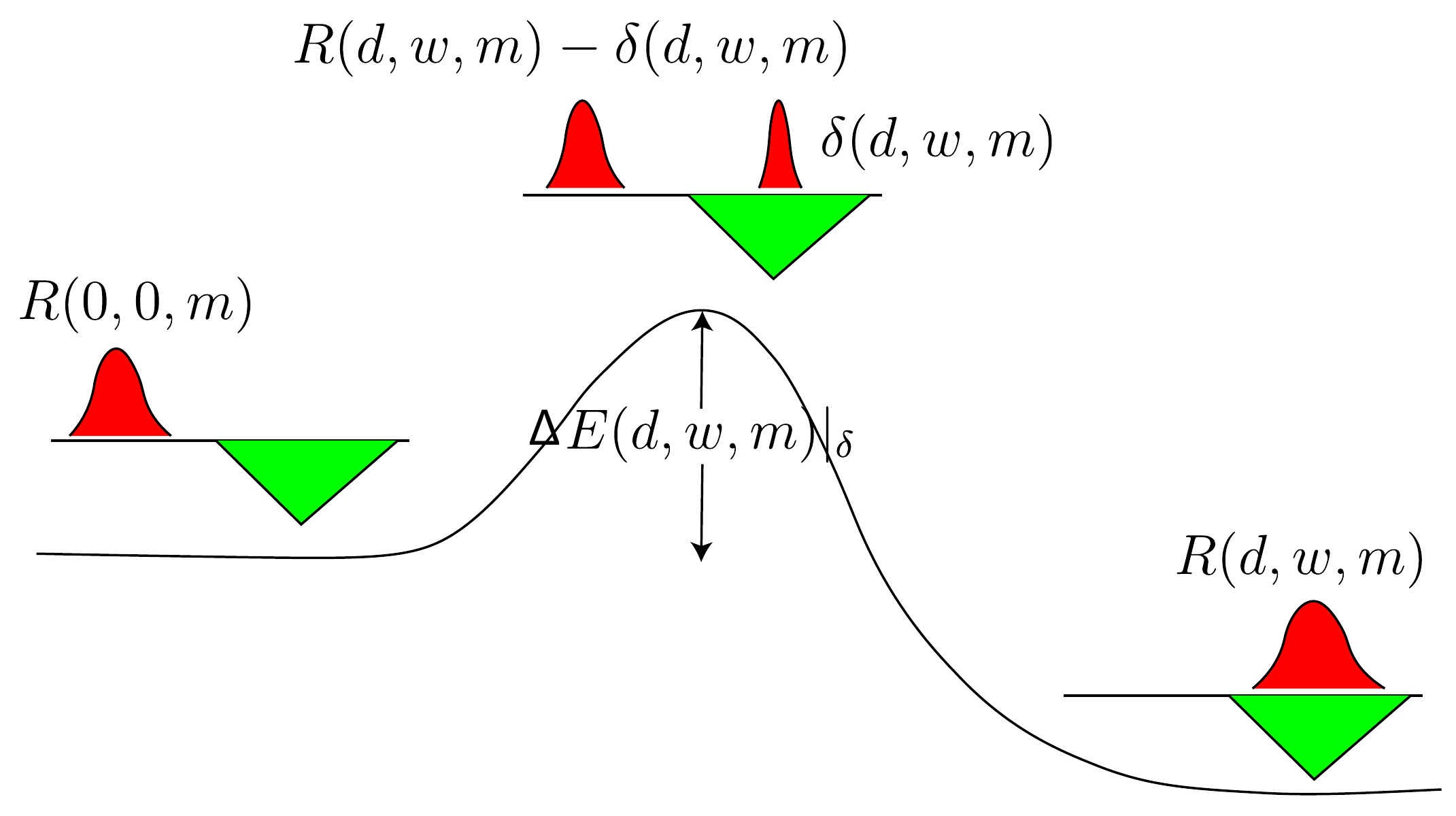}
\end{center}
	\caption{Schematics of three scenarios during a teleportation process. A initial configuration: the droplet is outside of the cup. A energetically unfavorable intermediate configuration that is penalize by $\Delta E$: the droplet breaks apart into two droplets, one outside the cup and one inside the cup; a final configuration with lowest energy: the droplet inside the cup grows to a full droplet while the droplet outside shrinks to zero size. Above each droplet is its corresponding radius $R$.\label{fig:teleportation}}
\end{figure}

The global minimum of the Lyapunov function corresponds to scenario $(3), $ However, there is an energy barrier between the initial configuration $(1)$ and final configuration $(3)$, corresponding to the $V_{eff}$ difference between initial configuration $(1)$ and intermediate configuration $(2)$. We would like to find the critical split size $\delta_c(d,w,m)$ that maximize the difference in $V_{eff}$, which corresponds to the largest energy barrier the network has to overcome in order to teleporte from $(1)$ to $(3)$. For the purpose of derivation, in the following we would like to rename $\mathcal{L}[f(i_k^m)]$ in Eqn.\eqref{eqn:blobenergy} as $E_0(d,w,m)\rvert_{R(d,w,m)}$ to emphasize its dependence on the external driving parameters and disordered interactions. The subscript $0$ stands for the default one-droplet configuration, and it is understood that $E_0(d,w,m)$ is evaluated at the network configuration of a single droplet at location $m$ with radius $R(d,w,m)$.

The energy for $(1)$ is simply $E_0(0,0,m)$, and the energy for $(3)$ is $E_0(d,w,m)$. However, the energy for $(2)$ is not just the sum of $E_0$ from the two droplets. Due to global inhibitions presented in the network, when there are two droplets, there will be an extra interaction term, when we evaluate the Lyapunov function with respect to this configuration. The interaction energy between two droplets in Fig.\ref{fig:teleportation} is

\begin{equation}
	E_{int}(m)\rvert_{R,\delta} = 4JR\delta(\epsilon - 2pm).
\end{equation}

Therefore, the energy barrier for split size $\delta$ is
\begin{equation}
	\begin{split}
		\Delta E&(d,w,m)\rvert_{\delta}  \\
		&= E_0(0,0,m)\rvert_{R(d,w,m)-\delta} + E_0(d,w,m)\rvert_{\delta}
		\\
		&+ E_{int}(m)\rvert_{R(d,w,m),\delta} - E_0(0,0,m)\rvert_{R(0,0,m)}.
	\end{split}
\end{equation}

Therefore, maximizing $\Delta E$ with respect to $\delta$, we find
\begin{equation}
	\begin{split}
		\delta_c &= \frac{dw}{d-8Jpm + 4J\epsilon}
	\end{split}
\end{equation}

Now we have obtained the maximum energy barrier during a teleportation process, $\Delta E\rvert_{\delta_c}$. A spontaneous teleportation will occur if $\Delta E\rvert_{\delta_c} \leq 0$, and this in turn gives a upper bound on external driving signal strength $d \leq d_{max}$ one can have without any teleportation spontaneous occurring. 

%The brute force way would be to solve  for the critical teleportation strength $d_{max}$. 
We plot the numerical solution of $d_{max}$ obtained from solving $\Delta E(d_c,w,m)\rvert \delta_c = 0$, compared with results obtained from simulation in Fig.\ref{fig:dmax}, and find perfect agreement. 

%However, we cannot obtain the analytical formula for $d_{max}$ this way: it turns out that this equation is a quintic polynomial in $d_{max}$ and does not have a nice closed-form solution. So we have to approach it differently. 

We also obtain an approximate solution by observing that the only relevant scale for that the critical split size $\delta_c$ is the radius of the droplet, $R$. We set $\delta_c = cR$ for some constant $0 \leq c \leq 1$. In general, $c$ can depend on dimensionless parameters like $p$ and $\epsilon$. Empirically we found the constant to be about 0.29 in our simulation. 

The droplet radius $R$ is a function of $d,w,m$ as we see in Eqn.\eqref{eqn:blobsize}, but to first order approximation we can set $R$ = $R^*$ for some steady-state radius $R^*$. Then we can solve

\begin{equation}
	\label{eqn:dmax}
	d_{max}(M) = \frac{4J(\epsilon-2pm)}{w/cR^* -1}.
\end{equation}

Note that the denominator is positive because $w > R$ and $0 \leq c \leq 1$. The simulation result also confirms that the critical split size $\delta_c$ stays approximately constant. 
We have checked that the dependence on parameters $J,w, m$ in Eqn.\eqref{eqn:dmax} agrees with the numerical solution obtained from solving $E_{bar}(d_c,w,m)\rvert \delta_c = 0$, up to the undetermined constant $c$.

\begin{figure}
\begin{center}
    	\includegraphics[width=0.8\linewidth]{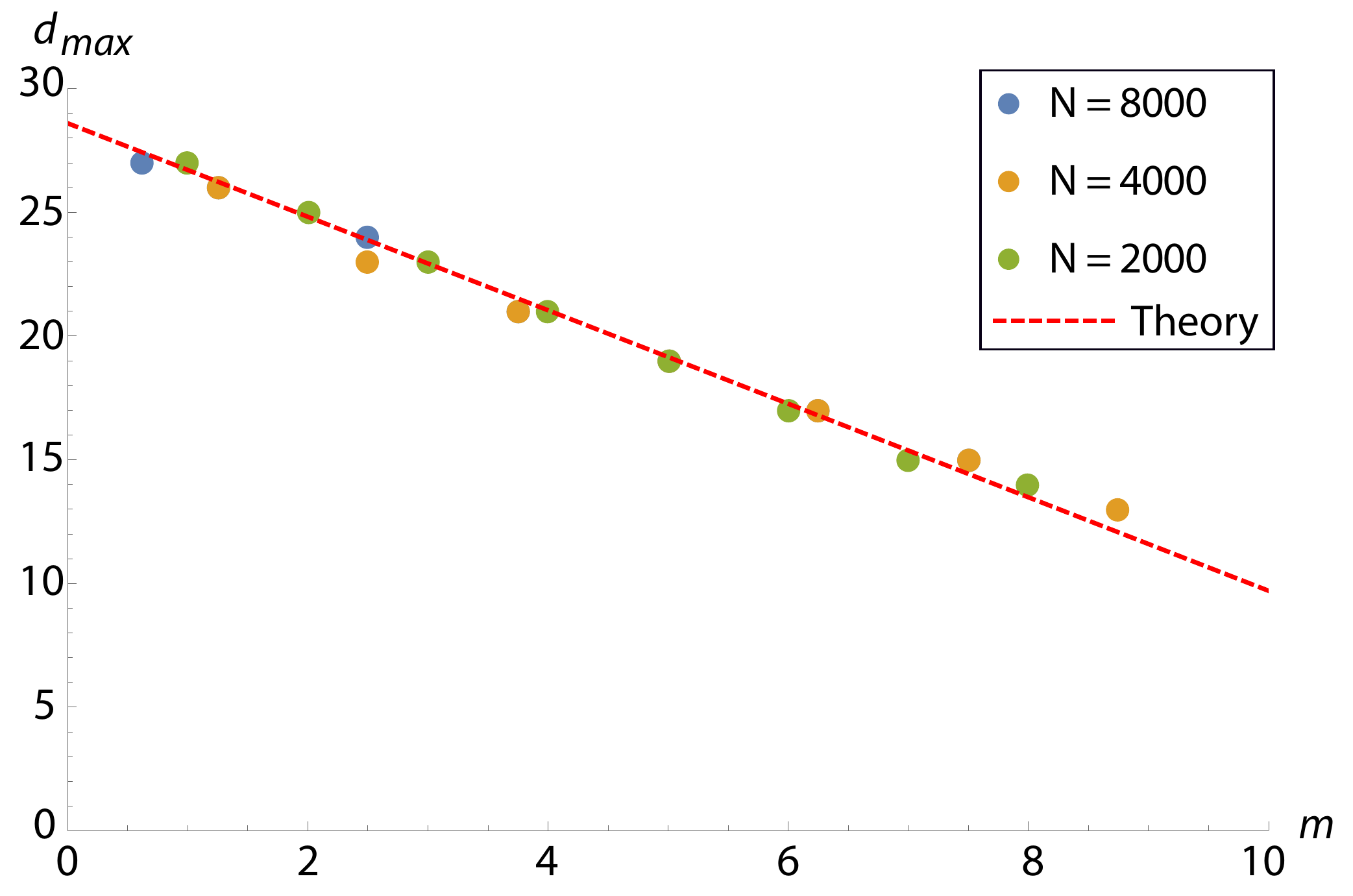}
\end{center}
	\caption{Teleportation depth $d_{max}$ plotted against disorder parameter $m$. The dots are data obtained from simulations for different $N$ but with $p=10$, $\epsilon=0.35$, $\tau=1$, $J=100$, and $w=30$ held fixed. The dotted line is the theoretical curve plotted from solving $\Delta E(d_c,w,m)\rvert \delta_c = 0$ for $d_{c}$ numerically. \label{fig:dmax}}
\end{figure}

%Note that Eqn.\eqref{eqn:dmax} vanishes at $m = \epsilon/2p$, which again signals the instability of the droplet because even for tiny external signal strength, it will spontaneous 'teleporte' to the localized driving signal. This is the leading order result of Eqn.\eqref{eqn:unstable} because we have ignored the $m$ dependence in $R$ in deriving Eqn.\eqref{eqn:dmax}. 

\subsection*{Speed limit on external driving}
\label{app:v_fund}
Recall that given a certain signal strength $d$, there is an upper bound on how fast the driving can be, Eqn.\eqref{eqn:vcrit}. Then in particular, for $d_{max}$, we obtain an upper bound on how fast external signal can drive the network,
\begin{equation}
	v_{max} = \frac{8J(w+R^*)(\epsilon-2pm)}{3\gamma(w/cR^* -1)}.
\end{equation}

For $w \gg R^*$, we can approximate
\begin{equation}
	v_{max} \approx \frac{16JcR^*(\epsilon/2-pm)}{3\gamma},
\end{equation}

In the absence of disorder, $m = 0$, the maximum velocity is bounded by

\begin{equation}
	\label{eqn:vbound_1}
	v_{max} \leq \frac{8c}{3}\frac{\epsilon JR^*}{\gamma} \leq \frac{8c}{3}\frac{\epsilon JR_{max}}{\gamma}.
\end{equation}

Recall that in Eqn.\eqref{eqn:dmax}, we have
% * <luzhiyue@hotmail.com> 2018-09-21T16:33:15.482Z:
% 
% How to get C18 from C17?
% 
% ^.
\begin{equation}
	\begin{split}
		R(d,w\gg R,0) &\leq R(d_{max}, w\gg R, 0)
		\\
		&= \frac{p}{2\epsilon} + \frac{1}{4} + 2cR^* + \mathcal{O}(\frac{R}{w})
		\\
		& \lessapprox \frac{p}{2\epsilon} + 2cR_{max},
	\end{split}
\end{equation}

where in the second line we have used \eqref{eqn:blobsize} for $d=d_{max}$, $m=0$, and $w \gg R$. 
Upon rearranging, we have

\begin{equation}
R_{max} \lessapprox \frac{1}{1-2c}\frac{p}{2\epsilon}.
\end{equation}

Plugging in Eqn.\eqref{eqn:vbound_1}, we have

\begin{equation}
	v_{max} \leq \frac{8c}{3}\frac{\epsilon JR_{max}}{\gamma} \lessapprox \frac{8}{3(c^{-1}-2)} \frac{Jp}{\gamma}.
\end{equation}

Therefore, we have obtained an fundamental limit on how fast the droplet can move under the influence of external signal, namely,
\begin{equation}
	v_{fund} = \kappa Jp\gamma^{-1},
\end{equation}

where $\kappa = 8/3(c^{-1}-2)$ is a dimensionless $\mathcal{O}(1)$ number.

\section{Path integration and velocity input}

\label{Aij}
Place cell networks \citep{Ocko2018-gv} and head direction networks \citep{Kim2017-zs} are known to receive information both about velocity and landmark information. Velocity input can be modeled by adding an anti-symmetric part $A_{ij}$ to the connectivity matrix $J_{ij}$, which effectively 'tilts' the continuous attractor. 

Consider now 
\begin{equation}
J_{ij} = J^0_{ij} + J^d_{ij} + A^0_{ij}, 
\end{equation}
where $A^0_{ij} = A$, if $0< i-j \leq p$; $-A$, if $0< j-i \leq p$; and $0$ otherwise.  

The anti-symmetric part $A^0_{ij}$ will provide a velocity $v$ that is proportional to the size $A$ of $A^0_{ij}$ for the droplet (See Fig.\ref{fig:vA}). In the presence of disorder, we can simply go to the co-moving frame of velocity $v$ and the droplet experiences an extra disorder-induced noise $\eta_A$ in addition to the disorder induced temperature $T_d$. 

We found that $\langle \eta_{A}(t)\eta_{A}(0)\rangle \propto \tilde{\sigma}\delta(t)$ (See Fig.\ref{fig:asym}), where $\tilde{\sigma}^2$ is the average number of disordered connection per neuron in units of $2p$.

Therefore, all our results in the main text applies to the case when both the external drive $I^{ext}(x,t)$ and the anti-symmetric part $A^0_{ij}$ exists. Specifically, we can just replace the velocity $v$ used in the main text as the sum of the two velocities corresponding to $I^{ext}(x,t)$ and $A^0_{ij}$.

\begin{figure}
\begin{center}
    	\includegraphics[width=0.6\linewidth]{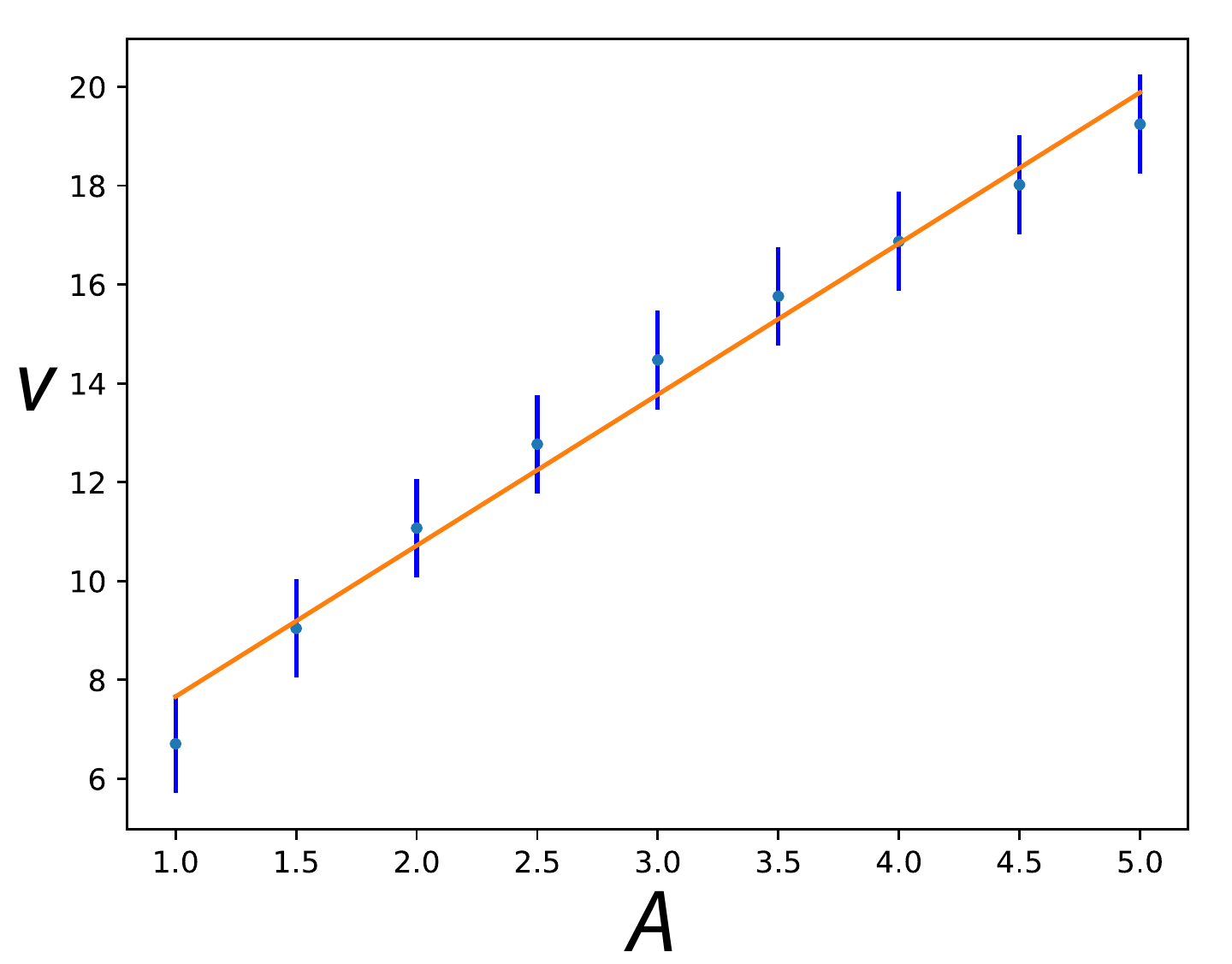}
\end{center}
	\caption{Velocity of droplet $v$ plotted against the size $A$ of the anti-symmetric matrix. We hold all other parameters fixed with the value same as in Fig.\ref{fig:dmax}. \label{fig:vA}}
\end{figure}

\begin{figure}
\begin{center}
    	\includegraphics[width=0.8\linewidth]{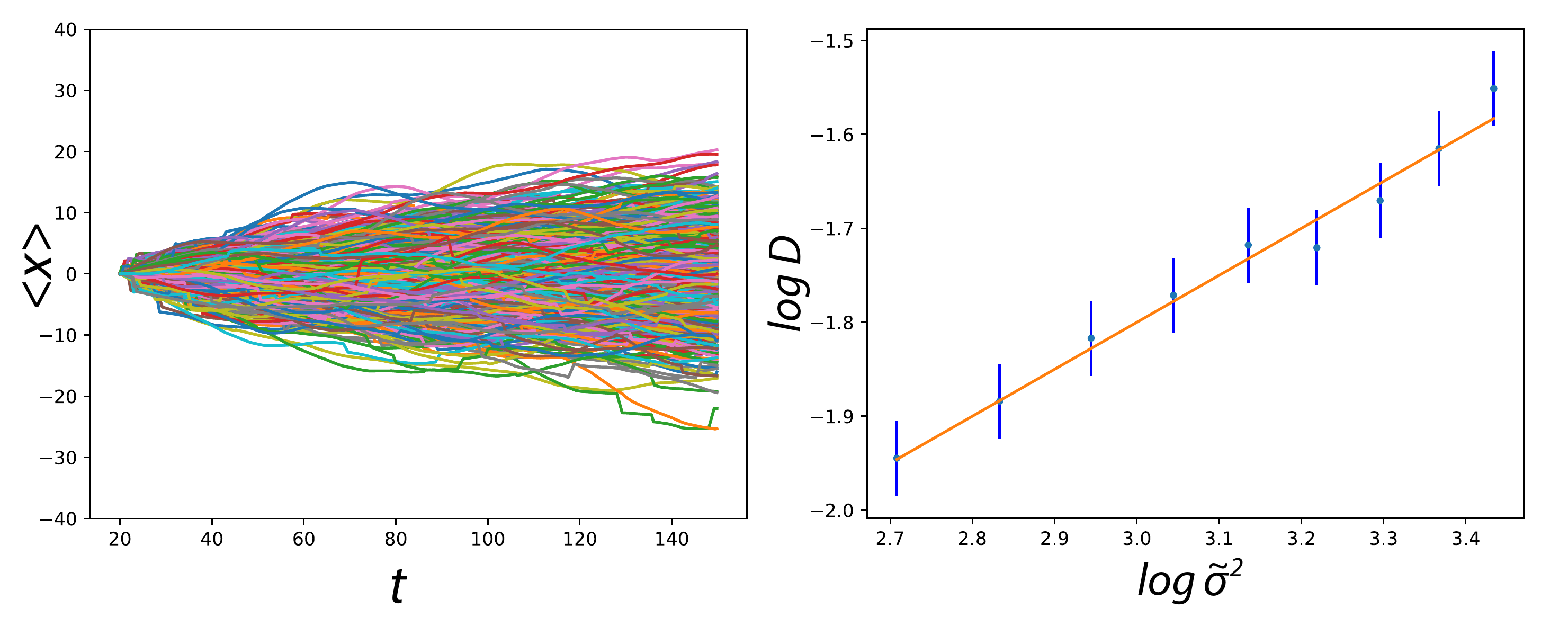}
\end{center}
	\caption{\textbf{Left}: At fixed $A=5$, a collection of 500 diffusive trajectories in the co-moving frame at velocity $v$, where $v$ is taken to be the average velocity of all the trajectories. We can infer the diffusion coefficient $D$ from the variance of these trajectories as Var$(x)= 2Dt$. \textbf{Right}: log$D$ plotted against log$\tilde{\sigma}^2$. The straight line has slope $1/2$, corresponding to $D \propto \tilde{\sigma}$. \label{fig:asym}}
\end{figure}

\section{Quenched Disorder - driving and disorder-induced temperature}
\label{Tbumpy}

\subsection{Disordered connections and disordered forces}
% Multiple environments
%When a place cell network is exposed to multiple environments, each of length $L$, neurons are assigned place fields that are unrelated to the assignments in other environments. The synaptic matrix is the sum of the matrices $J_{ij}^\alpha$ corresponding to each of the environments. Experiments show that there is no specific relationship between the maps in different environments; two neurons with nearby place fields in one environment seemingly have no bearing on their place fields in other environments. Hence, we model different environments as random permutations.

From now on, we start to include disorder connections $J^d_{ij}$ in addition to ordered connections $J^0_{ij}$ that corresponds to the nearest $p$-neighbor interactions. We assume $J^d_{ij}$ consists of random connections among all the neurons. These random connections can be approximated as random permutations of $J^0_{ij}$, such that the full $J_{ij}$ is the sum over $M-1$ such permutations plus $J^0_{ij}$. 

We `clip' the $J_{ij}$ matrix according to the following rule for each entry when summing over $J^0_{ij}$ and $J^d_{ij}$,
\begin{equation}
\label{eqn:clip}
\begin{split}
J&(1-\epsilon) + J(1-\epsilon) \to J(1-\epsilon) \\
J&(1-\epsilon) + J(-\epsilon) \to J(1-\epsilon)  \\
J&(-\epsilon) + J(-\epsilon) \to J(-\epsilon).
\end{split}
\end{equation}
Therefore, adding more disorder connections to $J_{ij}$ amounts to changing the inhibitory $-J\epsilon$ entries to the exitory $J(1-\epsilon)$. 

We would like to characterize the effect of disorder on the system. Under the decomposition $J_{ij} = J^0_{ij} + J^d_{ij}$, we can define a (quenched) disorder potential %that originates from some disorder energy $V_{d}$,
\begin{equation}
V^{d}(\bar{x}) \equiv V^{d}[f(i_k^{\bar{x}})] = -\frac{1}{2} \sum_{nk}J^{d}_{nk}f(i_k^{\bar{x}})f(i_n^{\bar{x}}),
\end{equation}

that captures all the disorder effects on the network.

Its corresponding disorder-induced force is then given by
\begin{equation}
\label{eqn:Fd}
 F^d(\bar{x}) = -\partial_{\bar{x}} V^{d}(\bar{x}).
\end{equation}

\subsection{Variance of disorder forces}

We compute the distribution of $V^d(\bar{x})$ using a combinatorial argument as follows. 

Under the rigid droplet approximation, calculating $V^d(\bar{x})$ amounts to summing all the entries within a $R$-by-$R$ diagonal block sub-matrix $J^{(\bar{x})}_{ij}$ within the full synaptic matrix $J_{ij}$ (recall that $V^d(\bar{x}) \propto \sum_{nk}f(i_n^{(\bar{x})})J_{nk}f(i_k^{(\bar{x})})$). Each set of disorder connection is a random permutation of $J^0_{ij}$, and thus has the same number of excitatory entries as $J^0_{ij}$, namely $2pN$. Since the inhibitory connections do not play a role in the summation by the virtue of \eqref{eqn:clip}, it suffices to only consider the effect of adding excitatory connections in $J^d_{ij}$ to $J^0_{ij}$.
 
There are $M-1$ sets of disordered connections in $J^d_{ij}$, and each has $2pN$ excitatory connections. Now suppose we add these $2pN(M-1)$ excitatory connections one by one to $J^0_{ij}$. Each time an excitatory entry is added to an entry $y$ in the $R$-by-$R$ block $J^{(\bar{x})}_{ij}$, there are two possible situations depending on the value of $y$ before addition: if $y = J(1-\epsilon)$ (excitatory), the addition of an excitatory connection does not change the value of $y$ because of the clipping rule in \eqref{eqn:clip}; if $y = -J\epsilon$ (inhibitory), the addition of an excitatory connection to $y$ changes $y$ to $J(1-\epsilon)$. In the latter case the value of $V^d(\bar{x})$ is changed because the summation of entries within $J^{(\bar{x})}_{ij}$ has changed, while in the former case $V^d(\bar{x})$ stays the same. (Note that if the excitatory connection is added outside $J^{(\bar{x})}_{ij}$, it does not change $V^d(\bar{x})$ and thus can be neglected.)

We have in total $2pN(M-1)$ excitatory connections to be added, and in total $(2R-p)^2$ potential inhibitory connections in the $R$-by-$R$ block $J^{(\bar{x})}_{ij}$ to be `flipped' to an excitatory connection. We are interested in, after adding all the $2pN(M-1)$ excitatory connections how many inhibitory connections are changed to excitatory connections, and the corresponding change in $V^d(\bar{x})$. 

%The general solution for this problem depends on the order of the outcome, because every time an inhibitatory connection gets flipped the subsequent probability of successfully flipping an inhibitory connections also changes, and it is difficult to write down an analytical solution. 

We can get an approximate solution if we assume that the probability of flipping an inhibitory connection does not change after subsequent addition of excitatory connections, and stays constant throughout the addition of all the $2pN(M-1)$ excitatory connections. This requires $2pN(M-1) \ll N^2$, i.e., $M \ll N$, which is a reasonable assumption since the capacity can not be $\mathcal{O}(N)$. 

For a single addition of exitatory connection, the probability of successfully flipping an inhibitory connection within $J^{(\bar{x})}_{ij}$ is proportional to the fraction of the inhibitory connections within $J^{(\bar{x})}_{ij}$ over the total number of entires in $J^0_{ij}$, 
\begin{equation}
q(\text{flip}) = \frac{(2R-p)^2}{N^2}.
\end{equation}

So the probability of getting $n$ inhibitory connections flipped is
\begin{equation}
P(n) = {2pN(M-1)\choose n} q^n (1-q)^{2pN(M-1)-n}.
\end{equation}

In other words, the distribution of flipping $n$ inhibitory connections to excitatory connections after adding $J^d_{ij}$ to $J^0_{ij}$ obeys $n \sim B(2pN(M-1),q)$. The mean is then
\begin{equation}
\begin{split}
\langle n \rangle &= 2pN(M-1)q = 2p(2R-p)^2 \bigg(\frac{M-1}{N}\bigg) \\
&= (2R-p)^2 2pm,
\end{split}
\end{equation}

where we have defined the reduced disorder parameter $m \equiv (M-1)/N$. The variance is 

\begin{equation}
\begin{split}
\langle n^2 \rangle &= 2pN(M-1)q(1-q) \\
&= 2pN(M-1) \frac{(2R-p)^2}{N^2} \bigg(1-\frac{(2R-p)^2}{N^2}\bigg) \\
&\approx (2R-p)^2 2pm,
\end{split}
\end{equation}

where in the last line we have used $N \gg 2R-p$.

Since changing $n$ inhibitory connections to $n$ exitory connections amounts to changing $V^d(\bar{x})$ by $-1/2 (J(1-\epsilon) - J(-\epsilon)) = -J/2$, we have

\begin{equation}
\label{eqn:Vd}
\text{Var}(V^d(\bar{x})) \equiv \sigma^2 = J^2(R-p/2)^2 pm.
\end{equation}

\subsection{Disorder temperature from disorder-induced force}

We focus on the case where $I^{ext}_n$ gives rise to a constant velocity $v$ for the droplet (as in the main text). In the co-moving frame, the disorder-induced force $F^d(\bar{x})$ acts on the c.o.m. like random kicks with correlation within the droplet size. For fast enough velocity those random kicks are sufficiently de-correlated and become a white noise at temperature $T_d$. 

To extract this disorder-induced temperature $T_d$, we consider the autocorrelation of $F^{d}[\bar{x}(t)]$ between two different c.o.m. location $\bar{x}(t)$ and $\bar{x}'(t')$ (and thus different times $t$ and $t'$),

\begin{equation}
C(t,t') \equiv \langle F^{d}[\bar{x}(t)] F^{d}[\bar{x}(t')]\rangle,
\end{equation}

where the expectation value is averaging over different realizations of the quenched disorder. 

Using \eqref{eqn:Fd}, we have
\begin{eqnarray}
C(t,t') &=  \langle \partial_{\bar{x}} V^d(\bar{x}) \partial_{\bar{x}'}V^d(\bar{x}') \rangle \\
&=\partial_{\bar{x}}\partial_{\bar{x}'} \langle  V^d(\bar{x}) V^d(\bar{x}') \rangle.
\end{eqnarray}

Within time $t-t'$, if the droplet moves a distance less than its size $2R$, then $V^{d}$ computed at $t$ and $t'$ will be correlated because $f(i_k^{\bar{x}})$ and $f(i_k^{\bar{x'}})$ have non-zero overlap. Therefore, we expect the autocorrelation function $\langle  V^d(\bar{x}) V^d(\bar{x}') \rangle$ behaves like the 1-$d$ Ising model with finite correlation length $\xi = 2R$ (up to a prefactor to be fixed later), 

\begin{equation}
\langle  V^d(\bar{x}) V^d(\bar{x}') \rangle \sim \exp (-\frac{|\bar{x}-\bar{x}'|}{\xi}).
\end{equation}

%\begin{equation}
%C(t,t') \sim \frac{\sigma^2}{\xi^2} \exp \bigg(-\frac{|\bar{x}-\bar{x}'|}{\xi}\bigg).

%\end{equation}

Hence, $C(t,t') \sim \exp \bigg(-\frac{|\bar{x}-\bar{x}'|}{\xi}\bigg)$. Now going to the co-moving frame, we can write the c.o.m. location as before, $\Delta x_v = \bar{x} - vt$, so the autocorrelation function becomes

%\begin{equation}
%\begin{split}
%C(t,t') &\sim \frac{\sigma^2}{\xi^2} \exp \bigg(-\frac{|(\Delta x_v + vt) - (\Delta x'_v +vt')|}{\xi}\bigg) \\
%&= \frac{\sigma^2}{\xi^2} \exp \bigg(-\frac{|v(t-t') + (\Delta x_v - \Delta x'_v)|}{\xi}\bigg) \\
%&\approx \frac{\sigma^2}{\xi^2} \exp \bigg(-\frac{v|t-t'|}{\xi}\bigg),
%\end{split}
%\end{equation}

\begin{equation}
\begin{split}
C(t,t') &\sim  \exp \bigg(-\frac{|(\Delta x_v + vt) - (\Delta x'_v +vt')|}{\xi}\bigg) \\
&=  \exp \bigg(-\frac{|v(t-t') + (\Delta x_v - \Delta x'_v)|}{\xi}\bigg) \\
&\approx \exp \bigg(-\frac{v|t-t'|}{\xi}\bigg),
\end{split}
\end{equation}

where in the last line we have used that the droplet moves much faster in the stationary frame than the c.o.m. position fluctuates in the co-moving frame, so $v(t-t') \gg \Delta x_v - \Delta x'_v$. 

Now let us define the correlation time to be $\tau_{cor} = \xi/v = 2R/v$. Then

\begin{equation}
%C(t,t') \sim \frac{\sigma^2}{4R^2} \exp \bigg(-\frac{|t-t'|}{\tau_{cor}}\bigg).
C(t,t') \sim \exp \bigg(-\frac{|t-t'|}{\tau_{cor}}\bigg).
\end{equation}

For $T \equiv |t-t'|\gg \tau_{cor}$, we want to consider the limiting behavior of $C(t,t')$ under an integral. Note that 

\begin{equation}
\begin{split}
&\int_{0}^{T} dt  \int_{0}^{T} dt' \exp \bigg(-\frac{|t-t'|}{\tau_{cor}}\bigg) \\
&= \tau_{cor}[2(T-\tau_{cor})+2\tau_{cor}e^{-T/\tau_{cor}}] \\
&\approx 2\tau_{cor}T \;\;\;\;\quad (\mbox{if } T \gg \tau_{cor}).
\end{split}
\end{equation}

Therefore, we have for $T \gg \tau_{cor}$, 

\begin{equation}
\begin{split}
&\int_{0}^{T} dt  \int_{0}^{T} dt' \exp \bigg(-\frac{|t-t'|}{\tau_{cor}}\bigg) \\
&= 2\tau_{cor} \int_{0}^{T} dt  \int_{0}^{T} dt' \delta (t-t').
\end{split}
\end{equation}

So we can write
\begin{equation}
\exp \bigg(-\frac{|t-t'|}{\tau_{cor}}\bigg) \to 2\tau_{cor} \delta(t-t'),
\end{equation}

and it is understood that this holds in the integral sense. Therefore, for $T \gg \tau_{cor}$, we expect $F^{d}(x)$ to act like uncorrelated white noise and we can write,

\begin{equation}
C(t,t') = T_d \delta(t-t') \propto \tau_{cor} \delta(t-t')
\end{equation}
where $T_d$ is a measure of this disorder-induced white noise. 

To deduce the form of disorder temperature $T_d$, we present the uncollapsed occupancies $- \log p(\Delta x_v) = V(\Delta x_v)/k_B T_d$ (described in the caption of main text Fig.3) in Fig.\ref{fig:T_bumpy_uncollapsed}. Compare with main text Fig.3, we can see that $T_d$ successfully captures the effect of disorder on the statistics of the emergent droplet if, 

\begin{equation}
\label{eqn:Td}
T_d = \tilde{k} \tau_{cor} \sigma,
\end{equation}
where $\sigma$ is given in \eqref{eqn:Vd} and $\tilde{k}$ is a fitting constant.

\begin{figure}
\begin{center}
    	\includegraphics[width=0.8\linewidth]{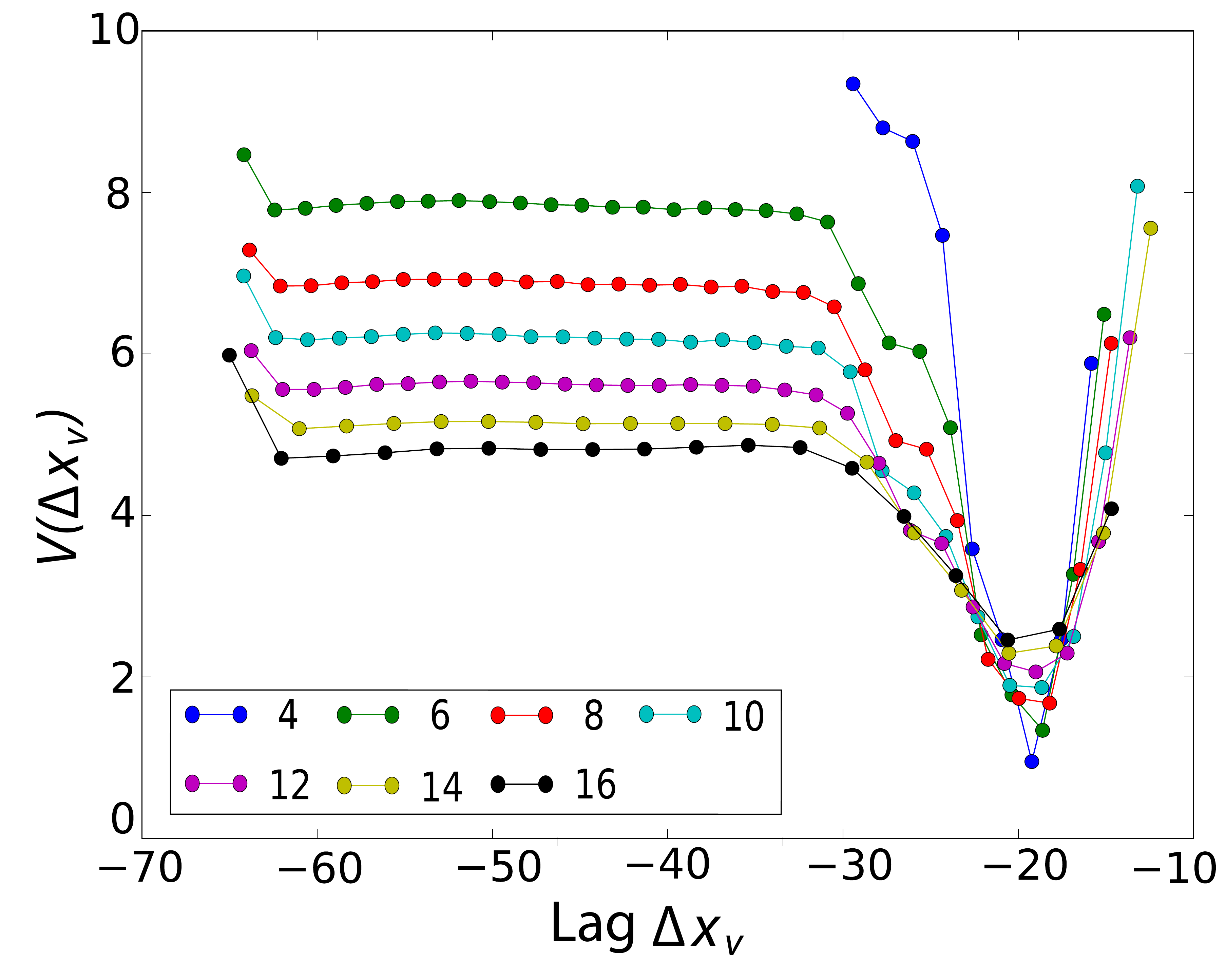}
\end{center}
	\caption{Uncollapsed data for the occupancies $- \log p(\Delta x_v)$ for different amounts of long ranged disordered connections.  Parameters same as in main text Fig.3 (see the last section of SI for further details). \label{fig:T_bumpy_uncollapsed}}
\end{figure}

\section{Derivation of the memory capacity for driven place cell network}
\label{pret}
In this section, we derive the memory capacity for driven place cell network described in the last section of the paper, namely, main text Eqn.(8).

Our continuous attractor network can be applied to study the place cell network. We assume a 1-dimensional physical region of length $L$. We study a network with $N$ place cell neurons and assume each neuron has a place field of size $d = 2 p L/N$ that covers the region $[0,L]$ as a regular tiling. The $N$ neurons are assumed to interact as in the leaky integrate-and-fire model of neurons. The external driving currents $I^{ext}(x,t)$ can model sensory input when the mouse is physically in a region covered by place fields of neurons $i, i+1,\ldots, i+z$, currents $I^{ext}_{i}$ through $I^{ext}_{i+z}$ can be expected to be high compared to all other currents $I^{ext}_j$, which corresponds to the cup-shape drive we used throughout the main text.

It has been shown in past work that the collective coordinate in the continuous attractor survives to multiple environments provided the number of stored memories $m < m_c$ is below the capacity $m_c$ of the network. Under capacity, the neural activity droplet is multistable; that is, neural activity forms a stable contiguous droplet as seen in the place field arrangement corresponding to any one of the $m$ environments. Note that such a contiguous droplet will not appear contiguous in the place field arrangement of any other environment. Capacity was shown to scale as $m_c = \alpha(p/N, R) N$ where $\alpha$ is an $O(1)$ number that depends on the size of the droplet $R$ and the range of interactions $p$.  However, this capacity is about the intrinsic stability of droplet and does not consider the effect of rapid driving forces.

%\subsection*{Barrier energy for escape}
%\label{app:deltaE}

When the droplet escapes from the driving signal, it has to overcome certain energy barrier. This is the difference in $V_{eff}$ between the two extremal points $\Delta x^*_v$ and $\Delta x^{esc}_v$. Therefore, we define the barrier energy to be $\Delta E = V_{eff}(x^{esc}_v) - V_{eff}(\Delta x^*_v)$, and we evaluate it using Eqn.\eqref{eqn:Veff} and Eqn.\eqref{eqn:positions},
\begin{equation}
\label{eqn:deltaE}
\Delta E(v,d) = \frac{(4dw-3\gamma v -2dR)(-\gamma v+2dR)}{4d}.
\end{equation}

Note this is the result we used in main text Eqn.(8).

As in the main text, the escape rate $r$ is given by the Arrhenius law,

\begin{equation}
r \sim \exp(-\frac{\Delta E(v,d)}{k_B T_{d}}).
\end{equation}

The total period of time of an external drive moving the droplet across a distance $L$ ($L\leq N$, but without loss of generality, we can set $L = N$) is $T = L/v$. We can imagine chopping $T$ into infinitesimal intervals $\Delta t$ st the probability of successfully moving the droplet across $L$ without escaping is,

\begin{equation}
\begin{split}
P_{retrieval} &= \lim_{\Delta t \to 0} (1-r \Delta t)^{\frac{T}{\Delta t}} 
\\
&= e^{-rT} = e^{-rN/v}
\\
&= \exp(-\frac{N}{v}e^{-\Delta E(v,d)/k_B T_{d}}).
\\
\end{split}
\end{equation}

$T_{d}$ is given by Eqn.\eqref{eqn:Td}

\begin{equation}
\begin{split}
T_d &=  \frac{2\tilde{k}R J(R-p/2)\sqrt{pm}}{v} \\
&\equiv k\sqrt{m}v^{-1},
\end{split}
\end{equation}

where in the last step we have absorbed all the constants (assuming $R$ is constant over different $m$'s) into the definition of $k$. 
Now we want to find the scaling behavior of $m$ s.t. in the thermodynamic limit ($N\to \infty$), $P_{retrieval}$ becomes a Heaviside step function $\Theta (m_c-m)$ at some critical memory $m_c$. With the aid of some hindsight, we try

\begin{equation}
m = \frac{\alpha^2}{(\log N) ^2},
\end{equation}

then in the thermodynamic limit,

\begin{equation}
\begin{split}
\lim_{N\to \infty} P_{retrieval} &= \lim_{N\to \infty} \exp(-\frac{N}{v}e^{-\log N v\Delta E(v,d)/\alpha k_B k})
\\
&= \lim_{N\to \infty} \exp (-\frac{N}{v} N^{-v\Delta E(v,d)/\alpha k_B k})
\\
&= \lim_{N\to \infty} \exp(-\frac{1}{v}N^{1-v\Delta E(v,d)/\alpha k_B k})
\\
&= \begin{cases} 
      1, & \alpha < v \Delta E(v,d)/k_B k \\
      0, & \alpha > v \Delta E(v,d)/k_B k
\end{cases}
\end{split}
\end{equation}

Therefore, we have arrive at the expression for capacity $m_c$, or in terms of $M = m_c N +1 \approx m_c N (N\gg 1)$,

\begin{equation}
M_c = \bigg[\frac{v\Delta E(v,d)}{k_B k}\bigg]^2 \frac{N}{(\log N)^2},
\end{equation}

or

\begin{equation}
M_c \sim \bigg[v\Delta E(v,d)\bigg]^2 \frac{N}{(\log N)^2}.
\end{equation}

\subsection*{Numerics of the place cell network simulations}
In this section, we explain our simulations in main text Fig.4 in detail. 

Recall that we only determine the Arrhenius-like escape rate $r$ up to an overall constant, we can absorb it into the definition of $\Delta E(v,d)$ (given by Eqn.\eqref{eqn:deltaE}) as an additive constant $a$,

\begin{equation}
r = \exp\bigg \{{-\frac{\Delta E(v,d)+a}{k_B k v\sqrt{(M-1)/N}}}\bigg \}.
\end{equation}

Then the theoretical curves corresponds to
\begin{equation}
\label{eqn:pretrieval}
P_{retrieval} = e^{-Nr/v}
\end{equation}

Therefore, our model Eqn.\eqref{eqn:pretrieval} has in total three parameters to determine $\gamma$, $k$, and $a$. In Fig.\ref{fig:barrier_collapse} we determine the parameters by collapsing data (see details of the collapse in below and in caption), and find that the best fit is found provided $\gamma = 240.30, k = 5255.0k_B^{-1}, a = -0.35445$. Henceforth we fix these three parameters to these values.

\begin{figure}
\begin{center}
    	\includegraphics[width=0.8\linewidth]{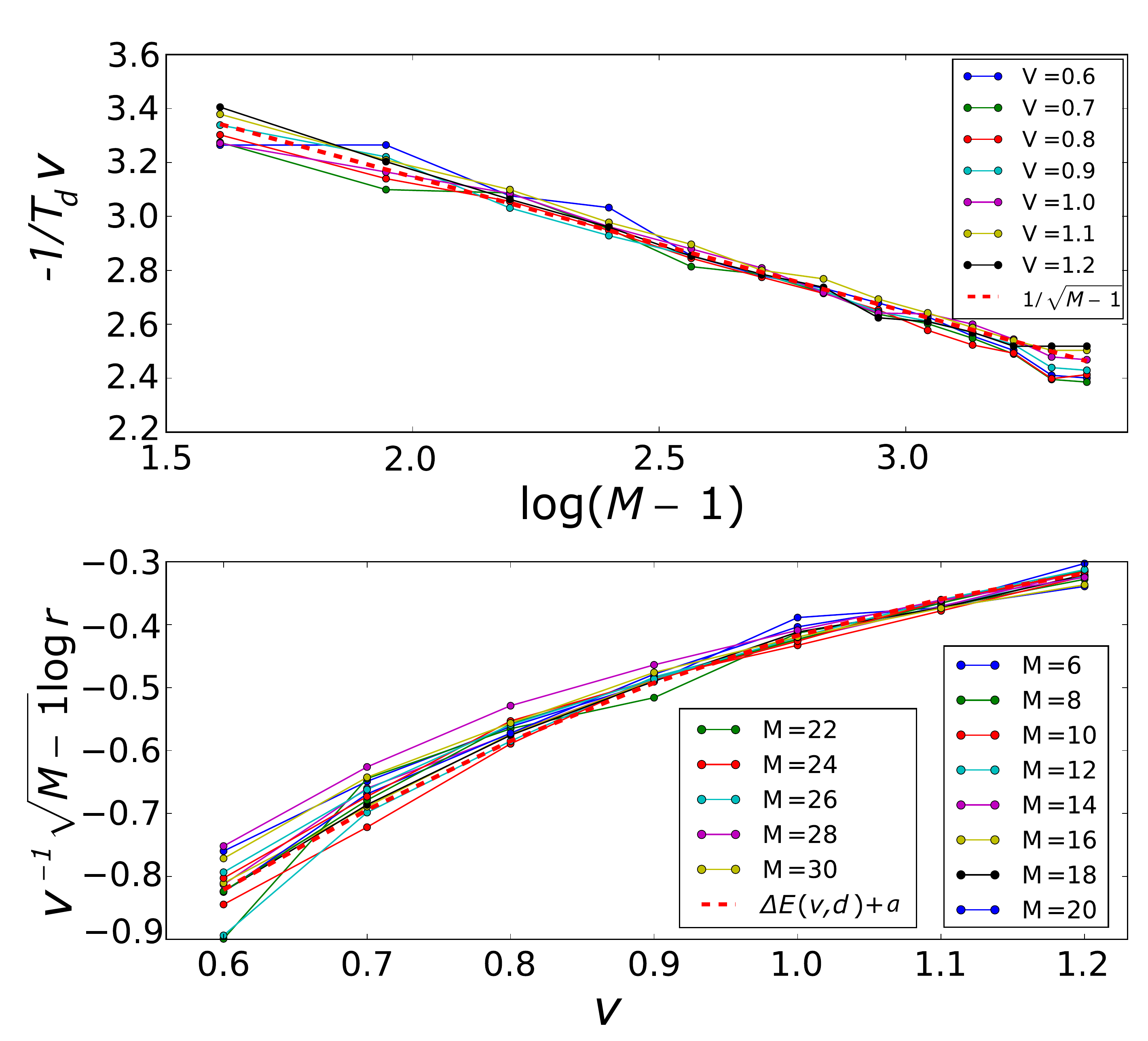}
\end{center}
\caption{\textbf{Top}: Plotting $-1/T_d v = \log\{ v^{-1} \log r /[\Delta E(v,d) + a] \}$ against $\log (M-1)$. Different solid lines corresponds to data with different $v$, and the dashed line corresponds to the $(M-1)^{-1/2}$ curve. \textbf{Bottom}: Plotting $v^{-1} \log r \sqrt{M-1} \propto \Delta E(v,d)$ against $v$. Different solid lines corresponds to data with different $M$, and dashed line corresponds to the $\Delta E(v,d)+a$ curve. \label{fig:barrier_collapse}}. 
\end{figure}

In Fig.\ref{fig:barrier_collapse} bottom, we offset the effect of $M$ by multiplying $v^{-1}\log r$ by $\sqrt{M-1}$, and we see that curves corresponding to different $M$ collapse to each other, confirming the $\sqrt{M-1}$ dependence in $T_d$. The collapsed line we are left with is just the $v$-dependence of $\Delta E(v,d)$, up to overall constant.

In Fig.\ref{fig:barrier_collapse} top, we offset the effect of $v$ in $T_d$ by multiplying $v^{-1}$ to $ \log r /[\Delta E(v,d) + a]$. We see that different curves corresponding to different $v$'s collapse to each other, confirming the $v^{-1}$ dependence in $T_d$. The curve we are left with is the $M$ dependence in $T_d$, which we see fits nicely with the predicted $\sqrt{M-1}$.  

In main text Fig.4(b) we run our simulation with the following parameters held fixed: $N=4000, \; p=10, \; \epsilon=0.35, \; \; \tau=1, \; J=100, \; d=10, \; w=30$. Along the same curve, we vary $M$ from $6$ to $30$, and the series of curves corresponds to different $v$ from $0.6$ to $1.2$. 

In  main text Fig.4(c) we hold the following parameters fixed: $p=10,\;  \epsilon=0.35, \; \tau=1, \; J=100, \; d=10, \; w=30, \; v=0.8$. Along the same curve, we vary $M/\frac{N}{(\log N)^2}$ from $0.1$ to $0.6$, and the series of curves corresponds to different $N$ from $1000$ to $8000$. 

In both  main text Fig.4(b)(c) the theoretical model we used is Eqn.\eqref{eqn:pretrieval} with the same parameters given above.

In main text Fig.4(d) we re-plot the theory and data from main text Fig.4(b) in the following way: for the theoretical curve, we find the location where $P_{retrieval} = 0.5$, and call the corresponding $M$ value theoretical capacity; for the simulation curve, we extrapolate to where $P_{retrieval} = 0.5$, and call the corresponding $M$ value, the simulation capacity.

For all simulation curves above, we drag the droplet from one end of the continuous attractor to the other end of the attractor, and run the simulation for 300 times. We then measure the fraction of successful events (defined as the droplet survived in the cup throughout the entire trajectory of moving) and failed events (defined as the droplet escape from the cup at some point before reaching the other end of the continuous attractor). We then define the simulation $P_{retreival}$ as the fraction of successful events.

%auto-ignore
\chapter{Discriminative learning by driven spin glasses}
\label{spinglass}
\section{Introduction}
Systems given many degrees of freedom
can learn and remember patterns of forces
that propel them far from equilibrium.
Such behaviors have been predicted and observed in many settings,
from charge-density waves~\cite{Coppersmith_97_Self,Povinelli_99_Noise}
to non-Brownian suspensions~\cite{Keim_11_Generic,Keim_13_Multiple,Paulsen_14_Multiple}, 
polymer networks~\cite{Majumdar_18_Mechanical}, 
soap-bubble rafts~\cite{Mukherji_19_Strength},
and macromolecules~\cite{Zhong_17_Associative}.
Such learning holds promise for engineering materials
capable of memory and computation.
This potential for applications, with experimental accessibility and ubiquity,
have earned these classical nonequilibrium many-body systems much attention recently~\cite{Keim_19_Memory}.
We measure many-body learning 
using a neural network (NN) that undergoes representation learning,
a type of machine learning.
Our toolkit detects and quantifies many-body learning
more thoroughly and precisely
than thermodynamic measures used to date.

One of the % most extensively characterized 
best-characterized instances of learning by driven matter
involves a  spin glass.
The spins are classical and interact randomly.
Consider applying fields from a set $\{ \vec{A}, \vec{B}, \vec{C} \}$,
which we call a \emph{drive}.
As the driving proceeds, the spins flip, absorbing work.
In a certain parameter regime,
the absorbed power shrinks adaptively:
The spins migrate toward a corner of configuration space 
where their configuration withstands the drive's insults.
Consider then imposing fields absent from the original drive.
Subsequent spin flips will absorb more work
than if the field belonged to the original drive.
% The work absorption reflects the drive's novelty: The spin glass has learned the drive.
Insofar as a simple, low-dimensional property of the material 
can be used to discriminate between 
drive inputs that fit a pattern and drive inputs that do not, 
we say that the material has learned the drive.

Learning behavior has been quantified with 
properties commonplace in thermodynamics.
Examples include work, magnetization, and strain.
This thermodynamic characterization has provided insights
but suffers from two shortcomings.
First, the types of thermodynamic properties vary from system to system.
For example, work absorption characterizes the spin glass's learning;
strain characterizes polymer networks'.
A more general approach would facilitate comparisons and standardize analyses.
Second, thermodynamic properties are useful 
for characterizing macroscopic equilibrium states.
But such properties are not necessarily the best
for describing the far-from-equilibrium systems that learn.

Over the past several years, machine learning has revolutionized 
the quantification of learning~\cite{Nielsen_15_Neural,Goodfellow_16_Deep}.
Machine learning calls for application
to the learning of drive patterns by many-body systems.

Parallels between statistical mechanics and 
certain machine-learning components 
have been known for decades~\cite{Engel_01_Statistical,Nielsen_15_Neural}.
For example, Boltzmann machines resemble
particles exchanging heat with thermal baths.
Parallels between \emph{representation learning} and statistical mechanics
have enjoyed less attention
(though one parallel was proposed in~\cite{Alemi_18_TherML}).
Figure~\ref{fig_VAE_SM_Parallel}(a) illustrates 
representation learning~\cite{Bengio_12_Representation}:
A high-dimensional variable $X$ is inputted into a NN.
The NN compresses relevant information 
into a low-dimensional variable $Z$.
The NN then decompresses $Z$ into a prediction $\hat{Y}$
of a high-dimensional variable $Y$.
If $Y = X$, the NN is an autoencoder,
mimicking the identity function.
The latent variable $Z$ acts as a bottleneck.
The bottleneck's size controls a tradeoff
between the memory consumed and the prediction's accuracy.
We call the NNs that perform representation learning
\emph{bottleneck NNs}.

% Figure: VAE-stat-mech parallel
\begin{figure}[hbt]
\centering
\includegraphics[width=.5\textwidth, clip=true]{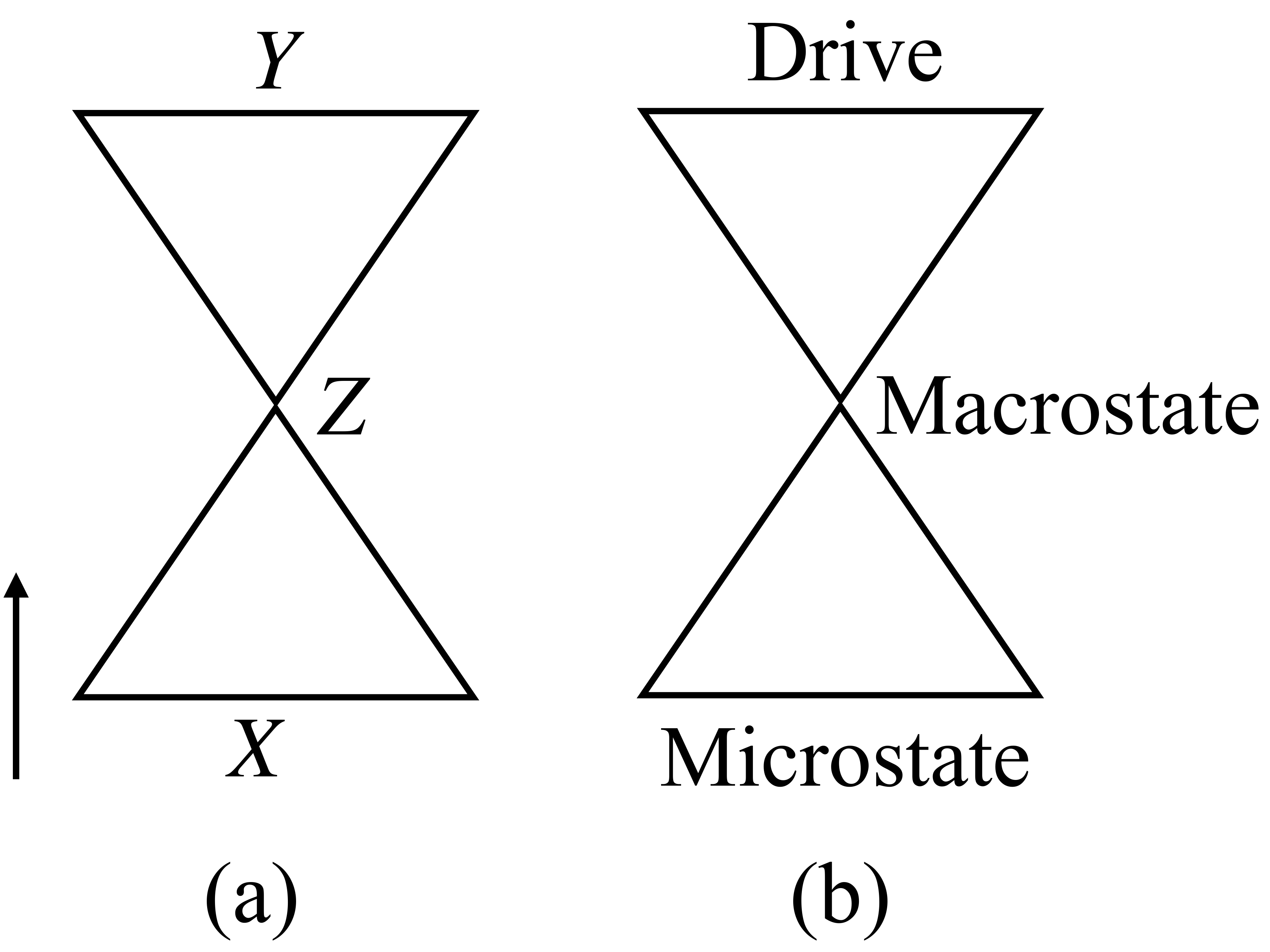}
\caption{\caphead{Parallel between two structures:}
(a) Structure of a bottleneck neural network, which performs representation learning.
(b) Structure of a far-from-equilibrium-statistical-mechanics problem.}
\label{fig_VAE_SM_Parallel}
\end{figure}

Representation learning, we argue, shares its structure with
problems in which a strong drive forces a many-body system
[Fig.~\ref{fig_VAE_SM_Parallel}(b)].
The system's microstate, like $X$, occupies a high-dimensional space.
A macrostate synopsizes the microstate in a few numbers, 
such as particle number and magnetization.
This synopsis parallels $Z$.
If the system has learned the drive, the macrostate encodes the drive.
One may reconstruct the drive from the macrostate,
as a bottleneck NN reconstructs $Y$ from $Z$.

Applying this analogy, we use representation learning 
to measure how effectively a far-from-equilibrium many-body system
learns a drive.
We illustrate with numerical simulations of the spin glass, 
whose learning has been characterized with work absorption~\cite{Gold_19_Self}.
However, our methods generalize to other platforms.
Our measurement scheme offers three advantages:
\begin{enumerate}
   \item 
   Bottleneck NNs register learning behaviors 
   more thoroughly, reliably, and precisely
   than work absorption.

   \item 
   Our framework applies to a wide class of
   strongly driven many-body systems.
   The framework does not rely on 
   strain, work absorption, susceptibility, etc.
   Hence our toolkit can characterize spins, suspensions, polymers, etc.

   \item 
   Our approach unites a machine-learning sense of learning
   with the statistical mechanical sense.
   This union is conceptually satisfying.
   
\end{enumerate}
We apply representation learning to measure 
 classification, memory capacity, discrimination, and novelty detection.
Our techniques can be extended to other facets of learning,
such as prediction and the decomposition of a drive into constituents.
% Reference about the above line: email chain “Comments on manuscript” — message received from Sarah on 9/24/19

Most of our measurement schemes have the following structure:
The many-body system is trained with 
a drive (e.g., fields $\vec{A}$, $\vec{B}$, and $\vec{C}$).
Then, the system is tested (e.g., with a field $\vec{D}$).
Training and testing are repeated in many trials.
Configurations realized % during the statistical mechanical training 
are used to train a bottleneck NN.
In some cases, the NN then receives data from
the statistical mechanical testing.
Finally, we analyze the NN's latent space and/or predictions.

The rest of this paper is organized as follows.
Section~\ref{sec_Setup} introduces the bottleneck NN that we use
and the spin-glass example.
In Sec.~\ref{sec_Results},
we prescribe how to quantify, using representation learning, 
the learning of a drive by a many-body system.
Section~\ref{sec_Discussion} closes with a discussion:
We decode our NN's latent space in terms of thermodynamic variables,
argue for our techniques' feasibility,
and detail opportunities engendered by this study.
% our unification of statistical mechanical learning and representation learning.

%
%
%
%
%
\section{Setup: Representation-learning model
and spin-glass example}
\label{sec_Setup}

This section introduces two toolkits applied in Sec.~\ref{sec_Results}:
(i) Section~\ref{sec_Intro_NNs} details the bottleneck NN we use.
(ii) Section~\ref{sec_Spin_Glass} details the spin glass
with which we illustrate statistical mechanical learners.

\subsection{Representation-learning model}
\label{sec_Intro_NNs}

% A reference: email chain "Specification of NNs" and "Re: Specification of NNs" -- exchanged with Weishun Sept.-Oct. 2019. The chain is now split into multiple threads.

This section overviews our architecture;
details appear in App.~\ref{sec_NN_Details}.
This paper's introduction identifies a parallel between
thermodynamic problems and bottleneck NNs (Fig.~\ref{fig_VAE_SM_Parallel}).
In the thermodynamic problem, $Y \neq X$ represents the drive.
We could design a bottleneck NN that predicts drives from configurations $X$.
But the NN would need to undergo supervised learning,
if built according to today's standards.
During supervised learning, the NN would receive tuples 
(configuration, label of drive that generated the configuration).
Receiving drive labels would give the NN
information not directly accessible to the many-body system.
The NN's predictive success would not necessarily reflect
only learning by the many-body system.
Hence we design a bottleneck NN that performs unsupervised learning,
receiving just configurations.

This NN is a \emph{variational autoencoder} (VAE)~\cite{Kingma_13_Auto,JR_14_Stochastic,Doersch_16_Tutorial}.
A VAE is a generative model:
It receives samples $x$ from a distribution over
the possible values of $X$,
learns about the distribution, and generates new samples.
The NN approximates the distribution,
using Bayesian variational inference (App.~\ref{sec_NN_Details}).
% "VAE is a NN whose objective function (loss function) is a variational lower bound (coming from Bayesian variational inference) of the maximum-likelihood estimation objective."
% (Reference: email chain "Specification of NNs" and "Re: Specification of NNs" -- message received from Weishun on 9/30/19)
The parameters are optimized during training facilitated by backpropagation.
% References for the following three lines: 
% (1) Meeting notes --> Youngsters - 9/26/19 --> p. 11
% (2) Email chain "Specification of NNs" -- message sent by Weishun on 10/1/19
% The VAE regularizes its latent-space representation
% by incorporating noise.
% This noise effectively raises the sample size seen by the VAE,
% increasing the representation's robustness.

Our VAE has five fully connected hidden layers,
% <-- "Fully connected layers means that the output of the layer is \phi(Wx + b), where x is the input vector, \phi is the non-linear activation function, W is the weight matrix and b is the bias. Here, fully connected just means that the input x is simply multiplied by the weight matrix, and not being convolved with a kernel or something fancy."
% <-- Reference: email chain "Specification of NNs" -- message sent by Weishun on 9/24/19
% Definition #2 of "fully connected": {Every neuron in layer j+1} receives input from {every neuron in layer j}.
% <-- Reference: https://arxiv.org/abs/1807.10300 — Suppl. Inf. — p. 10
with neuron numbers 200-200-(number of $Z$ neurons)-200-200.
% Reference about the next 3 sentences: Meeting notes --> Weishun - 9/3/19 - paper draft
We usually restrict $Z$ to 2-4 neurons.
This choice facilitates the visualization of the latent space
and suffices to quantify our spin glass's learning.
Growing the number of degrees of freedom,
and the number of drives, may require more dimensions.
But our study suggests that the number of dimensions needed
$\ll$ the system size.

% Each latent-space neuron represents a continuous-valued real number in $( -\infty,  \infty )$.
The latent space is visualized in Fig.~\ref{fig_Latent_Space}.
Each neuron corresponds to one axis
and represents a continuous-valued real number.
The VAE maps each inputted configuration
to one latent-space dot. % , which we have colored.
Close-together dots correspond to
configurations produced by the same field,
if the spin glass and NN learn well.
We illustrate this clustering in Fig.~\ref{fig_Latent_Space}
by coloring each dot according to the drive that produced it.

%
% Figure: Latent space
% 
% Source of figure: email chain "Decoding latent space" -- message sent by Weishun on 1/6/20
\begin{figure}[hbt]
\centering
\includegraphics[width=.6\textwidth, clip=true]{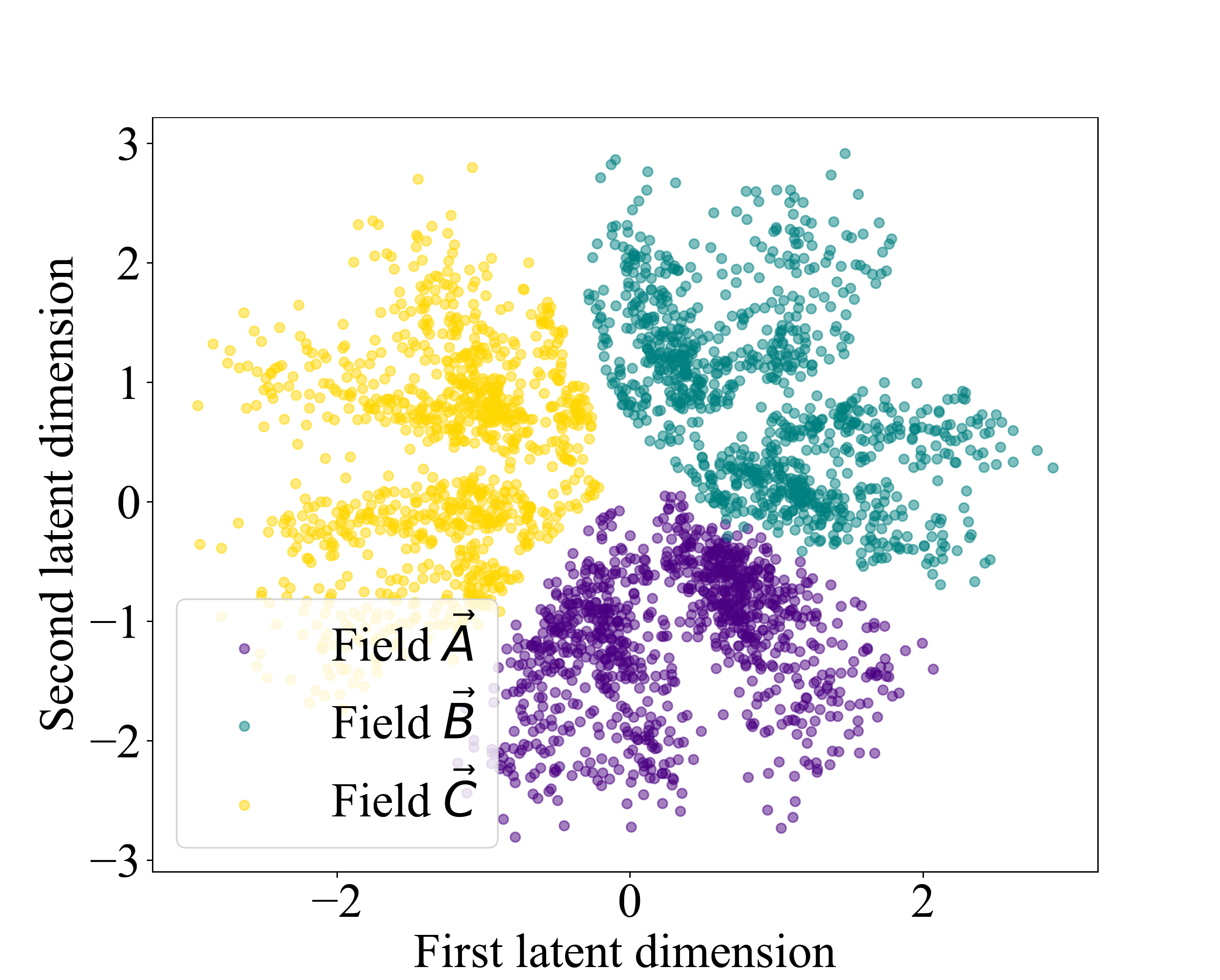}
\caption{\caphead{Visualization of latent space:}
The latent space $Z$ consists of two neurons, $Z_1$ and $Z_2$.
A variational autoencoder (VAE) formed this latent space
while training on configurations assumed by a 256-spin glass
during repeated exposure to three fields, $A$, $B$, and $C$. 
The VAE mapped each configuration to a dot in latent-space.
We color each dot in accordance with
the field that produced the configuration.
Same-color dots cluster together:
The VAE identified which configurations resulted from the same field.}
\label{fig_Latent_Space}
\end{figure}
\subsection{Spin glass}
\label{sec_Spin_Glass}

% References: (1) Jacob's paper
% (2) Meeting notes --> Youngsters - 9/5/19 --> p. 3
% (3) Email chain "Confirming the setup" --> message sent by Jacob on 9/23/19

A spin glass exemplifies the statistical mechanical learner.
We introduce the spins, Hamiltonian, and heat exchange below.
We model the time evolution, define work and heat,
and describe the initialization procedure.
Finally, we describe a parameter regime
in which the spin glass learns effectively.
Section~\ref{sec_Not_Enslaved_Or_Frozen}
distinguishes robust learning from superficially similar behaviors.

We adopt the model in~\cite{Gold_19_Self}.
Simulations are of $\Sites = 256$ classical spins.
The $j^\th$ spin occupies one of two possible states:
$s_j = \pm 1$.

The spins couple together and experience an external magnetic field.
Spin $j$ evolves under a Hamiltonian
\begin{align}
   \label{eq_Hamiltonian_j}
   H_j(t)
   =  \sum_{k \neq j}  J_{jk}  s_j  s_k
   +  A_j(t) s_j ,
\end{align}
and the spin glass evolves under
\begin{align}
   \label{eq_Hamiltonian}
   H(t)  
   =  \frac{1}{2} \sum_{j = 1}^\Sites  H_j(t) 
\end{align}
at time $t$.
We call the first term in Eq.~\eqref{eq_Hamiltonian_j} the \emph{interaction energy} 
and the second term the \emph{field energy}.
The couplings $J_{j k}  =  J_{kj}$ are defined in terms of 
an Erd\"{o}s-R\'enyi random network:
Nodes $j$ and $k$ have some probability $p$ 
of sharing an edge, for all $j$ and $k \neq j$.
% Reference for the above line: Meeting notes --> Youngsters - 9/5/19 --> p. 3
We identify nodes with spins and identify edges with couplings.
Each spin couples to eight other spins, on average.
The nonzero couplings $J_{j k}$ are selected according to 
a normal distribution of standard deviation 1.
% References for the above line: (1) Meeting notes --> Youngsters - 9/5/19 --> p. 3
% (2) Email chain "Confirming the setup" --> message sent by Jacob on 9/23/19

The $A_j(t)$ in Eq.~\eqref{eq_Hamiltonian_j} 
is defined as follows.
At time $t$, the spin glass experiences a field $\{ A_j(t) \}_j$.
$A_j(t)$ represents the magnitude and sign of the field at spin $j$.
All fields point along the same direction (conventionally labeled as the $z$-axis), 
so we simplify the vector notation $\vec{A}_j$ to $A_j$.
Elsewhere in the text, we simplify $\{ A_j (t) \}_j$
to the capital Latin letter $A$ (or $B$, or $C$, etc.). 
% References for the following line: 
% (1) Meeting notes --> Youngsters - 9/5/19 --> p. 3
% (2) Email chain "Are you ok?" --> message written by NYH on 1/9/20 -- contains copy of a message sent by Jacob
Each $A_j(t)$ is selected according to 
a normal distribution of standard deviation 3.
% Reference for the following lines: email chain "Confirming the setup" -- message sent by Jacob on 9/23/19
The field changes every 100 seconds. % , unless otherwise noted.
To train the spin glass, we construct a drive
by forming several random fields $\{ A_j \}_j$.
We randomly select a field from the set, then apply the field.
We repeat these two steps 299 times, unless otherwise noted
(Fig.~\ref{fig_Drive_Protocol}).

% Figure: Drive protocol
\begin{figure}[hbt]
\centering
\includegraphics[width=.75\textwidth, clip=true]{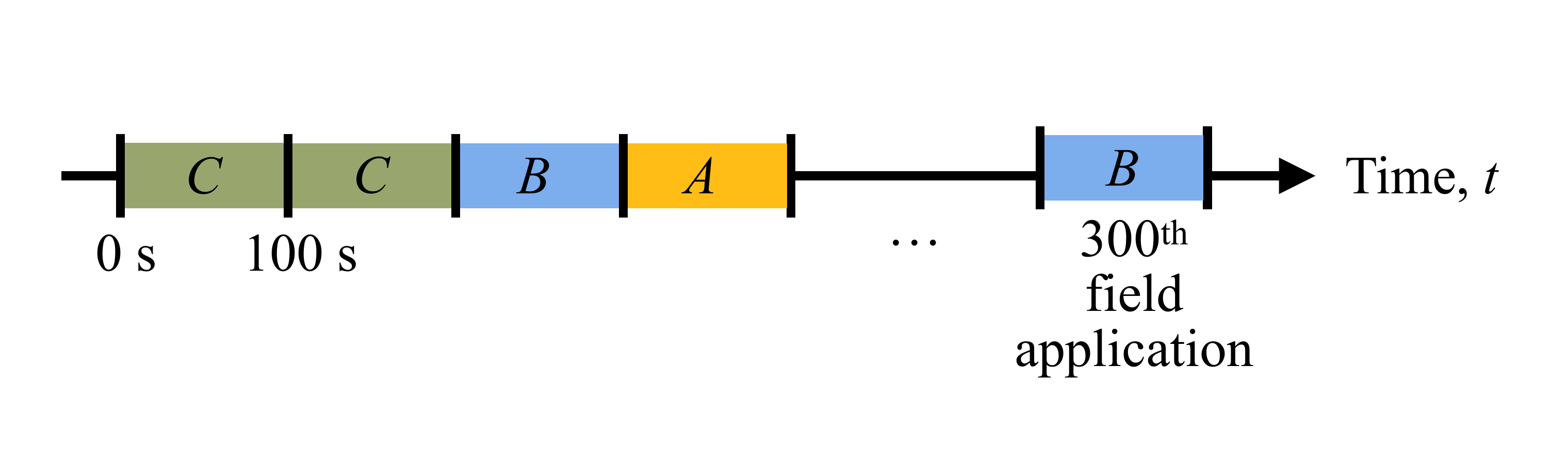}
\caption{\caphead{Driving protocol:} 
The drive consists of the set $\{A, B, C\}$ of fields.
A field is selected randomly from the drive and applied for 100 s,
and then this process is repeated.}
\label{fig_Drive_Protocol}
\end{figure}

The spin glass exchanges heat with
a bath at a temperature $T = 1 / \beta$.
We set Boltzmann's constant to one: $\kB = 1$.
Energies are measured in units of Kelvins (K).
To flip, a spin must overcome an energy barrier of height $B$.
Spin $j$ tends to flip at a rate 
\begin{align}
   \label{eq_Flip_Rate}
   \omega_j 
   =  e^{\beta [ H_j(t) - B]}
   / (1 \text{ second})  \, .
\end{align}
% Reference for the "per second" above: email chain "Jacob - spin-flip rate" -- message sent by Jacob on 1/29/20
% "H_j is -1 on average [after thermalization], though it varies quite a bit from spin to spin. This would result in exp(beta(H_j - B)) = ~10^-7, which is long enough that spins pretty much aren't going to flip without the assistance of the driving field."
Equation~\eqref{eq_Flip_Rate} has the form of Arrhenius's law
and obeys detailed balance.
Each spin flips once every $10^7$ s, on average.
We model the evolution with discrete 100-s time intervals,
using the Gillespie algorithm.
% <-- Similar to kinetic Monte Carlo (https://en.wikipedia.org/wiki/Gillespie_algorithm)

% Reference: Meeting notes --> Youngsters - 9/5/19
The spins absorb work when the field changes,
and they dissipate heat while flipping,
as we now detail.
Consider changing the field from $\{ A_j (t) \}$ to $\{ A'_j(t) \}$. 
The change in the spin glass's energy equals
the work absorbed by the spin glass:
\begin{align}
   W  :=  \sum_{j = 1}^\Sites  
   \left[  A'_j(t)  -  A_j(t)  \right]  s_j .
\end{align}
To define heat, we suppose that spin $k$ flips at time $t$:
$s_k  \mapsto  s'_k  =  - s_k$. 
The spin glass dissipates an amount $Q$ of heat equal to 
the negative of the change in the spin glass's energy:
\begin{align}
   Q  
   & :=  -  \frac{1}{2}  \sum_{j \neq k}  \left[
   J_{jk}  s_j  (s'_k - s_k)
   +  A_k(t)  (s'_k  -  s_k)  \right]  \\
   % % %
   & =  \sum_{j \neq k}  J_{j k}  s_j  s_k
   +  2 h_k (\alpha_t)  s_k .
\end{align}
Our discussion is cast in terms of the absorbed power,
$W / ( \text{100 s} )$. 

The spin glass is initialized in a uniformly random configuration $C$.
Then, the spins relax in the absence of any field for 100,000 seconds.
% Reference for the following line: email chain "Are you ok?" --> message written by NYH on 1/9/20 -- contains copy of a message sent by Jacob
The spin glass navigates to near a local energy minimum.
% <-- How does Jacob know that the configuration is near an energy minimum?
% "I know it's near a local energy minimum because all of the spins individually have negative energy, unless they are very close to 0 energy in which case they are occasionally positive some of the time."
% "The system as a whole navigates to near a local energy minimum, and you can verify that because each individual spin has negative energy (with the exception of some that are so close to 0 energy that they flip-flop)."
% Reference: email chain "Jacob - spin-flip rate" -- messages sent on 1/28/20
If a protocol is repeated in multiple trials, 
all the trials begin with the same $C$.

In a certain parameter regime, the spin glass learns its drive effectively,
even according to the absorbed power~\cite{Gold_19_Self}.
Consider training the spin glass on a drive $\{ A, B, C \}$.
The spin glass absorbs much work initially.
If the spin glass learns the drive, the absorbed power declines.
If a dissimilar field $D$ is then applied, the absorbed power spikes.
The spin glass learns effectively when
$\beta = 3$ K$^{-1}$ and $B = 4.5$ K~\cite{Gold_19_Self}.
These parameters define a Goldilocks regime:
The temperature is high enough,
and the barriers are low enough,
that the spin glass can explore phase space.
But $T$ is low enough, and the barriers are high enough,
that the spin glass is not hopelessly peripatetic.

\section{How to detect and quantify \\ a many-body system's \\ learning of a drive,
using representation learning}
\label{sec_Results}

This section shows how to quantify four facets of learning.
Section~\ref{sec_Classify} concerns
the many-body system's ability to classify drives; 
Sec.~\ref{sec_Capacity}, memory capacity;
Sec.~\ref{sec_Discriminate}, discrimination of similar fields;
and Sec.~\ref{sec_ROC}, novelty detection.
At the end of each section, we synopsize 
the technique introduced in boldface.
These four techniques illustrate how representation learning 
can be applied to quantify features of learning.
Other features may be quantified along similar lines.
% Opportunities include more-general discrimination and
% the identification of a field's components.\footnote{
% < f >
% {Cite another paper of Jeremy's? Reference: email chain “Comments on manuscript” — message received from Sarah on 9/24/19.}}
% < /f >
% Reference for the foregoing 3 sentences: email chain “Comments on manuscript” — message received from Sarah on 9/24/19
Code used can be found at the online repository~\cite{Github_repo}.

\subsection{Classification: Which drive is this?}
\label{sec_Classify}

% Reference for this section:
% Meeting notes --> Weishun, Jacob, Jeremy - 10/17/19

% References, partially for earlier version of this section:
% Meeting notes —> Youngsters - 9/26/19 —> p. 12
% Meeting notes —> Group skype - 9/19/19 —> p. 3 -- explains MAP estimation and scoring

% References for even earlier original version of this section:
% (1) 4/8/19 meeting notes 
% (2) Folder of 4/22/19 meeting notes —> Weishun report
% (3) Reference about VAE: Meeting notes —> folder “Weishun, Jeremy - 8/22/19” —> subfolder “Weishun report - 8/19/19” —> notes “Weishun replies + notes from skype (as well as other notes in the folder and subfolder)

% Detecting novelty, as in the previous section,
% requires one bit of information: 
% a ``yes'' or a ``no'' in response to ``Is this field new?''
% Classification requires more information.
A system classifies a drive by identifying the drive as 
one of many possibilities.
A VAE, we find, reflects more of a spin glass's classification ability
than absorbed power does.

We illustrate with the spin glass.
We generated random fields $A$, $B$, $C$, $D$, and $E$.
From 4 of the fields, we formed the drive $\mathcal{D}_1  :=  \{A, B, C, D\}$.
On the drive, we trained the spin glass in each of 1,000 trials.
In each of 1,000 other trials, we trained a fresh spin glass on
a drive $\mathcal{D}_2  :=  \{A, B, C, E\}$.
We repeated this process for each of the 5 possible 4-field drives.
Ninety percent of the trials were randomly selected for training our NN.
The rest were used for testing.

We used the VAE to gauge the spin glass's classification of the drives:
We identified the configurations 
occupied by the spin glass at a fixed time $t$ in the training trials.
On these configurations, we trained the VAE.
The VAE populated the latent space with dots
(as in Fig.~\ref{fig_Latent_Space})
whose density formed a probability distribution.

We then showed the VAE a time-$t$ configuration
from a test trial.
The VAE compressed the configuration into a latent-space point.
We calculated which drive most likely, 
according to the probability density,
generated the latent-space point.
The calculation was \emph{maximum} a posteriori \emph{estimation} 
(MAP estimation) (see~\cite{Bishop_06_Pattern} and App.~\ref{app_MAP}).
Here, the MAP estimation is equivalent to maximum-likelihood estimation.
We performed this testing and estimation for each trial in the test data.
The fraction of trials in which MAP estimation succeeded
forms the \emph{score}.
We scored the classification at each of many times $t$.
The score is plotted against $t$ in Fig.~\ref{fig_Classification},
as the blue, upper curve.

% Figure: MAP-estimate scores for simple classification protocol
% Reference: Email chain "Were you able to get ahold of Jacob?" --> message sent by Weishun on 12/19/19
% Earlier reference: Meeting notes --> Weishun, Jeremy - 10/28/19 --> the meeting notes and subfolder "Plots - classific'n score" --> drive_classification_score_random_IC
\begin{figure}[hbt]
\centering
\includegraphics[width=.6\textwidth, clip=true]{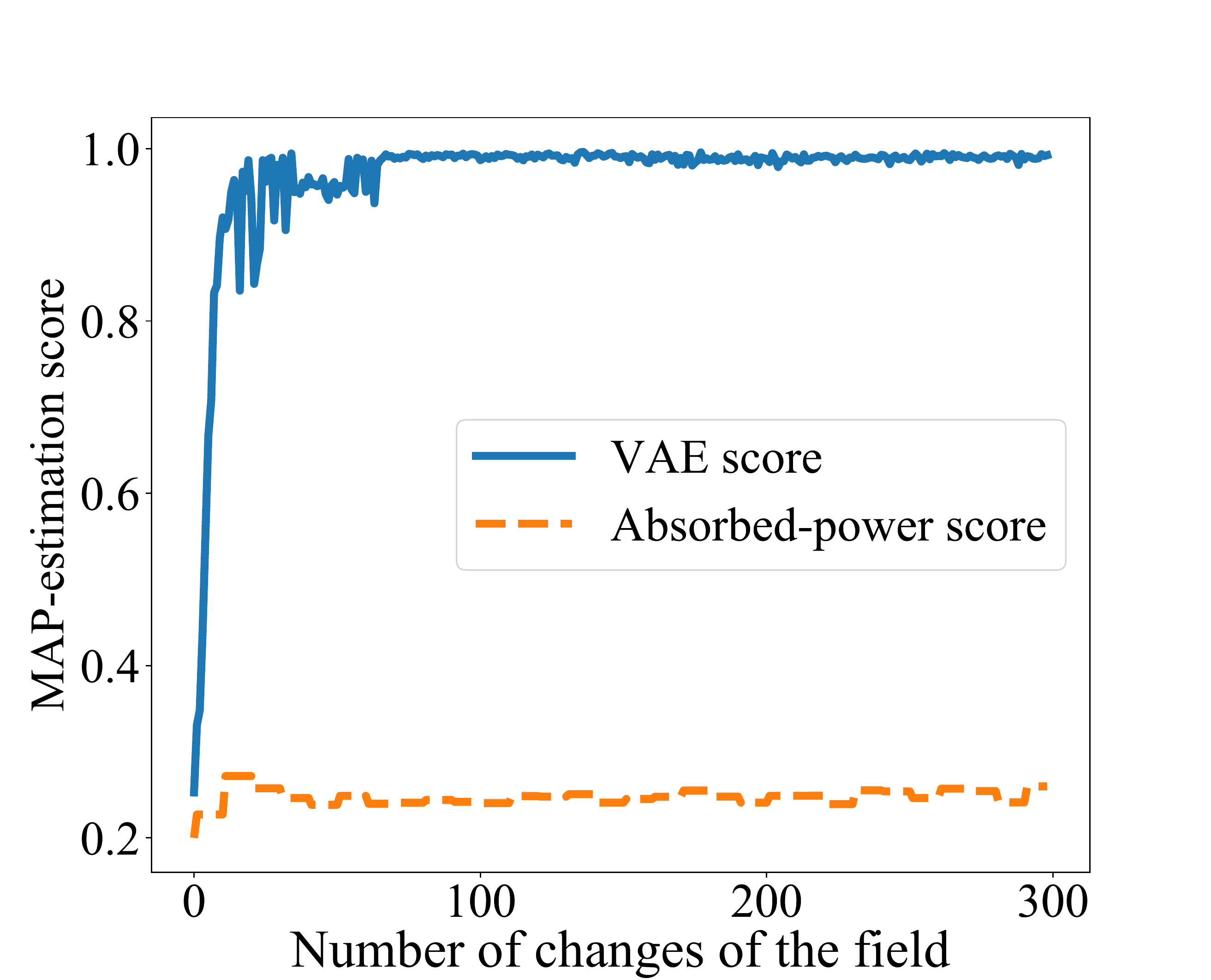}
\caption{\caphead{Quantification of a many-body system's classification ability:}
A spin glass classified a drive as one of five possibilities.
The system's classification ability was defined as
the score of the maximum \emph{a posteriori} (MAP) estimation
performed with a variational autoencoder (VAE) (blue, upper curve).
We compare with the score of MAP estimation performed with
absorbed power (orange, lower curve).
The VAE score rises to near the maximal value, 1.
The thermodynamic score remains slightly above 
the random-guessing score, $0.20$.
Hence the VAE detects more of the spin glass's classification ability.
}
\label{fig_Classification}
\end{figure}

The absorbed power reflects 
the spin glass's classification of the drives as follows.
For each drive $\mathcal{D}$ and each time $t$,
we histogrammed the power absorbed
while $\mathcal{D}$ was applied at $t$
in a VAE-training trial.
Then, we took a trial from the test set
and identified the power $\mathcal{P}$ absorbed at $t$.
We inferred which drive most likely, 
according to the histograms, produced $\mathcal{P}$. 
The guess's score appears as the orange, lower curve
in Fig.~\ref{fig_Classification}.

The score maximizes at 1.00 if the drive is always guessed accurately.
The score is lower-bounded by $1 / (\text{number of drives}) = 0.20$,
which results from random guessing.
In Fig.~\ref{fig_Classification}, each score grows 
over a time scale of tens of field switches.
The absorbed-power score begins at 0.20\footnote{
% < f >
\label{foot_Why_Not_0.2}
The VAE's score begins close to 0.20. The slight distance from 0.20, 
we surmise, comes from stochasticity of three types:
the spin glass's initial configuration, the MAP estimation,
and stochastic gradient descent. 
Stochasticity of only the first two types affects the absorbed power's score.
}
% < /f >
and comes to fluctuate around 0.25.
The VAE's score comes to fluctuate slightly below 1.00.
Hence the VAE reflects more of the spin glass's classification ability
than the absorbed power does.

\textbf{A many-body system's ability to classify drives is quantified with 
the score of MAP estimates calculated from a VAE's latent space.}

\subsection{Memory capacity: How many drives can be remembered?}
\label{sec_Capacity}

% References about NN estimate of capacity: 
% (1) Files from others --> Weishun - ML capacity - 11/25/19
% (2) Meeting notes —> Group skype - 9/19/19 —> 
% (i) p. 4 and 
% (ii) subfolder “Weishun's data - MAP”

% Section~\ref{sec_ROC} featured a many-body system 
% that identifies novel fields as novel.
% The system must remember the training fields to detect novelty.
How many fields can a many-body system remember?
A VAE, we find, registers a greater capacity 
than absorbed power registers.\footnote{
% < f >
% Reference: Meeting notes --> Marzen - 10/13/19
We use the term ``memory capacity'' in the physical sense of~\cite{Keim_19_Memory}.
A more specific, technical definition of ``memory capacity''
is used in reservoir computing~\cite{Jaeger_02_Short}.
}
% < /f >
Hence the VAE reflects statistical mechanical learning,
at high field numbers,
that the absorbed power does not.

%% How Jacob and Weishun selected fields to leave out:
% ~"we have in total 50 fields, and we form 5 drives by leaving out field 1-10, 11-20, …, 41-50. [ . . . ] [T]his protocol [ . . . ] makes sure that the drives are distinctive enough from each other."
% Reference: email chain "VAE app -- manuscript files" -- message sent by Weishun on 1/2/20
We illustrated by constructing 50 random fields.
We selected 40 to form a drive $\mathcal{D}_1$,
selected 40 to form a drive $\mathcal{D}_2$,
and repeated until forming 5 drives.
We trained the spin glass on $\mathcal{D}_j$
in each of 1,000 trials, for each of $j = 1, 2, \ldots 5$.
Ninety percent of the trials were designated as VAE-training trials;
and 10\%, as VAE-testing trials.

The choice of 50 fields is explained in App.~\ref{app_Capacity_Work}:
Fifty fields exceed the spin-glass capacity registered by the absorbed power.
We aim to show that 50 fields do not exceed the capacity
registered by the VAE:
The VAE identifies spin-glass learning missed by the absorbed power.

We used representation learning to quantify the spin glass's capacity as follows.
For a fixed time $t$, we collected 
the configurations occupied by the spin glass at $t$
in the VAE-training trials.
On these configurations, the VAE performed unsupervised learning.
The VAE populated its latent space with dots
that formed five clusters.
Then, we fed the VAE the configuration occupied at $t$ during a test trial.
The VAE formed a new dot in latent space.
We MAP-estimated the drive that, according to the VAE,
most likely generated the dot (Sec.~\ref{sec_Classify}).
The fraction of test trials in which the VAE guessed correctly
constitutes the VAE's score.
The score is plotted against $t$ in Fig.~\ref{fig_Capacity},
as the blue, upper curve.

% Figure: Capacity according to VAE
% Reference: Email chain "Were you able to get ahold of Jacob?" --> message sent by Weishun on 12/19/19
% Old reference: Email chain “Capacity Study” — plot sent by Weishun on 11/25/19
% Conversation continued in email chain “Re: Capacity study - responses” — initiated by NYH on 11/26/19
% Old reference, in which the data that used to be here was first reported: Meeting notes —> Group skype - 9/19/19 —> subfolder “Weishun's data - MAP”
\begin{figure}[hbt]
\centering
\includegraphics[width=.6\textwidth, clip=true]{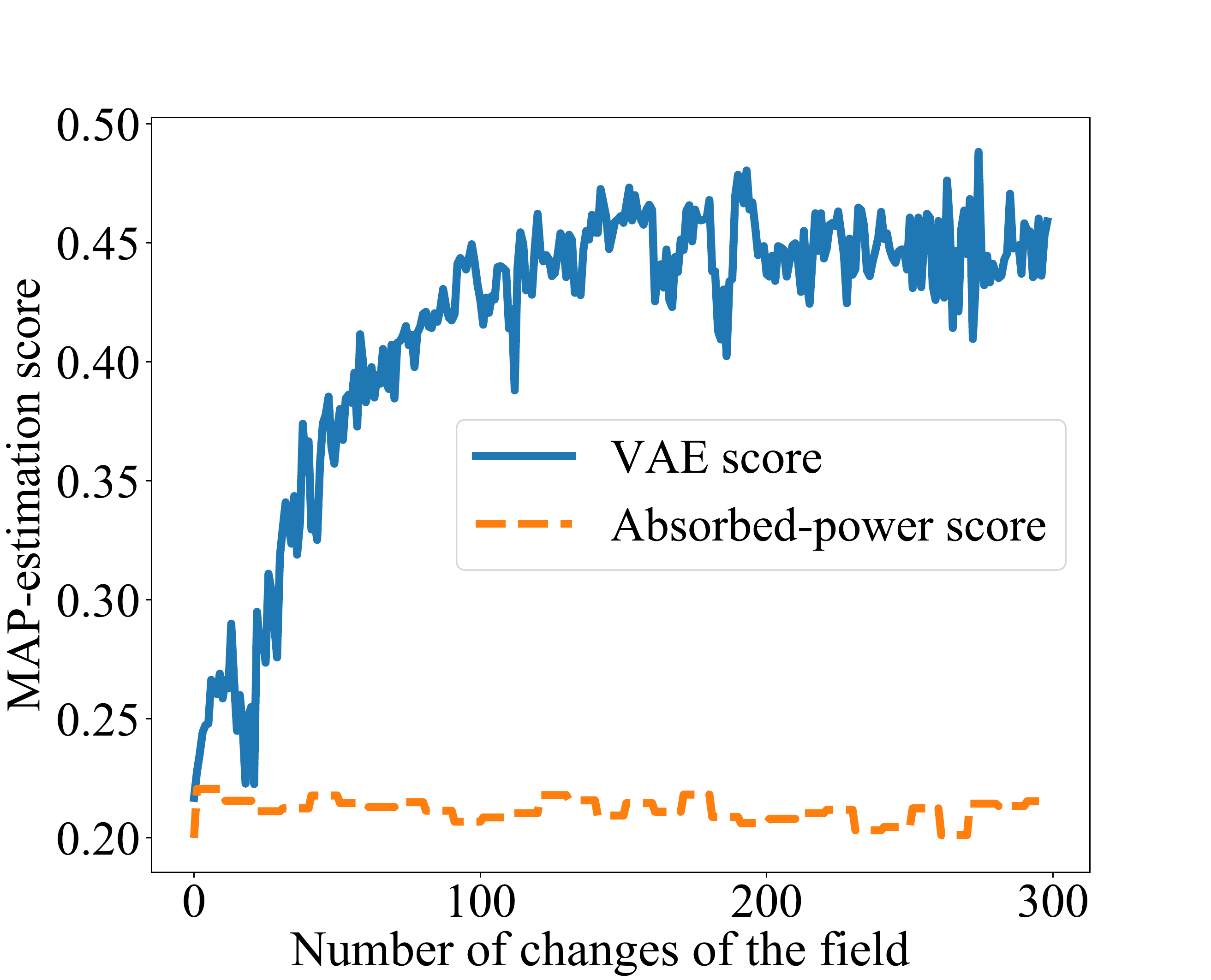}
\caption{\caphead{Quantification of memory capacity:}
A spin glass was trained on one of five drives in each of many trials.
Each drive was formed from 40 fields selected from 50 fields.
We quantified the spin glass's ability to classify the drives
with the score of maximum \emph{a priori} (MAP) estimation
performed with a variational autoencoder (upper, blue line).
The score of MAP estimation performed with absorbed power
is shown for comparison (lower, orange line).}
\label{fig_Capacity}
\end{figure}

The VAE's score is compared with the absorbed power's score,
calculated as follows.
For a fixed time $t$, we identified the power absorbed at $t$
in each VAE-testing trial.
We histogrammed the power absorbed
when $\mathcal{D}_j$ was applied at $t$,
for each $j = 1, 2, \ldots, 5$.
We then identified the power absorbed at $t$ in a test trial.
Comparing with the histograms, we inferred which drive
was most likely being applied.
We repeated this inference with each other test trial.
In which fraction of the trials did the absorbed power identify the drive correctly?
This number forms the absorbed power's score.
The score is plotted as the lower, orange curve in Fig.~\ref{fig_Capacity}.

The higher the score, the greater the memory capacity
attributed to the spin glass.
The absorbed power identifies the drive 
in approximately $20\%$ of the trials,
as would random guessing.
The score remains approximately constant,
because the number of fields exceeds the spin-glass capacity
measured by the absorbed power.
In contrast, the VAE's score grows over
$\approx 150$ changes of the field,
then plateaus at $\approx 0.450$.
The VAE points to the wrong drive most of the time
but succeeds significantly more often than the absorbed power.
Hence representation learning uncovers
more of the spin glass's memory capacity
than absorbed power measure does.

\textbf{A many-body system's memory capacity is quantified with
the greatest number of fields in any drive 
on which MAP estimation, based on a VAE's latent space,
scores better than random guessing.}

% Reference for the following text: Meeting notes --> Jacob, Weishun, Jeremy - 9/12/19 --> Weishun's plots --> the 60-field plots
% The MAP-estimate score quantifies an observation:
% The VAE's latent space consists of clusters.
% Each cluster consists of dots generated by 
% configurations generated the same field.
% The clusters do not overlap, if the latent space has enough dimensions.
% The number of dimensions is much less than the system size.
% For example, three dimensions separate 60 drives.
% Hence the latent space serves the VAE
% as a macrostate space serves thermodynamic problems,
% as a low-dimensional representation.
% But the VAE's ``macrostate space'' detects memory
% overlooked by the absorbed power.

%
%
%
\subsection{Discrimination: How new is this field?}
\label{sec_Discriminate}

% Reference: email chain "Linear-combination study" -- results sent by Weishun on 11/10/19
% Old reference: Meeting notes —> folder “England group - 8/1/19” —> notes of the same name —> p. 2-4, especially p. 4

Suppose that a many-body system learns fields $A$ and $B$,
then encounters a field that interpolates between them.
Can the system recognize that 
the new field contains familiar constituents?
Can the system discern how much $A$ contributes
and how much $B$ contributes?
The answers characterize the system's discrimination ability,
which we quantify with a MAP-estimation score (Sec.~\ref{sec_Classify}).
Estimates formed from a VAE's latent space
reflect more of the system's discriminatory ability
than do estimates formed from absorbed power.

We illustrate with the spin glass, forming a drive $\{A, B, C\}$.
In each of 300 time intervals,
a field was selected randomly from the drive and applied.
The spin glass was then tested with
a linear combination $D_w = w A + (1 - w) B$.
The weight $w$ varied from 0 to 1, in steps of $1/6$, across trials.

We measured the spin glass's discrimination using the VAE as follows.
The final configuration assumed by the spin glass 
in each test trial was collected.
The configurations were split into VAE-training data
and VAE-testing data.
On the configurations generated by $D_w$ in the VAE-training data,
the VAE was trained.
Then, the VAE received a configuration generated by $D_w$
in a VAE-testing trial.
The VAE mapped the configuration to a latent-space point.
We inferred which field most likely generated that point,
using MAP estimation on the latent space.
We tested the VAE many times, 
then calculated the fraction of MAP estimates that were correct,
the VAE's score.

Similarly, we measured the spin glass's discrimination
using the absorbed power.
For each trial in the VAE-training data,
we calculated the power $\mathcal{P}$ absorbed by the spin glass
after the application of $D_w$.
We histogrammed $\mathcal{P}$, inferring the probability that,
if shown $D_w$ for a given $w$,
the spin glass will absorb an amount $\mathcal{P}$ of power.
Then, we calculated the power absorbed
during a VAE-testing trial.
We inferred which field most likely generated that point,
using MAP estimation on the latent space.
Repeating MAP estimation with all the VAE-testing trials,
we calculated the absorbed power's score.

The VAE's score equals about double the absorbed power's score,
for latent spaces of dimensionality 2 to 20.
The VAE scores between 0.448 and 0.5009, 
whereas the absorbed power scores 0.2381.
Hence the representation-learning model picks up on 
more of the spin glass's discriminatory ability
than the absorbed power does.

\textbf{A many-body system's ability to discriminate
combinations of familiar fields
is quantified with the score of MAP estimates
formed from a VAE's latent space.}

\subsection{Novelty detection: Has this drive been encountered before?}
\label{sec_ROC}

% Reference 1: Meeting notes —> folder “England group - 4/22/19” —> notes of the same name —> p. 4: “Sarah’s novelty-detection study”
% Reference 2, about work-dissipation ROC: Meeting notes --> folder "England group - 4/29/19" --> notes of the same name --> p. 4
% Reference 3, about late iteration of the novelty-detection study: email chain "RE: Harvard CMTists" -- messages exchanged in late Oct. and early Nov. (Beware, this chain split into multiple chains.)

% Figure: ROC
% References, from earliest to latest:
% Reference 1: original figure from Sarah: email chain "Doing this during OH" -- messages sent by Sarah on 9/19/19 and afterward
% Data set used:  sim_4_10_4_40_expanded
% Reference 2: figure that contains curve drawn with VAE latent space: email chain "Harvard CMTists" -- message sent by Weishun on 11/2/19
% Reference 3: email chain "RE: Harvard CMTists" -- message sent by Weishun on 11/9/19
\begin{figure}[hbt]
\centering
\includegraphics[width=.6\textwidth, clip=true]{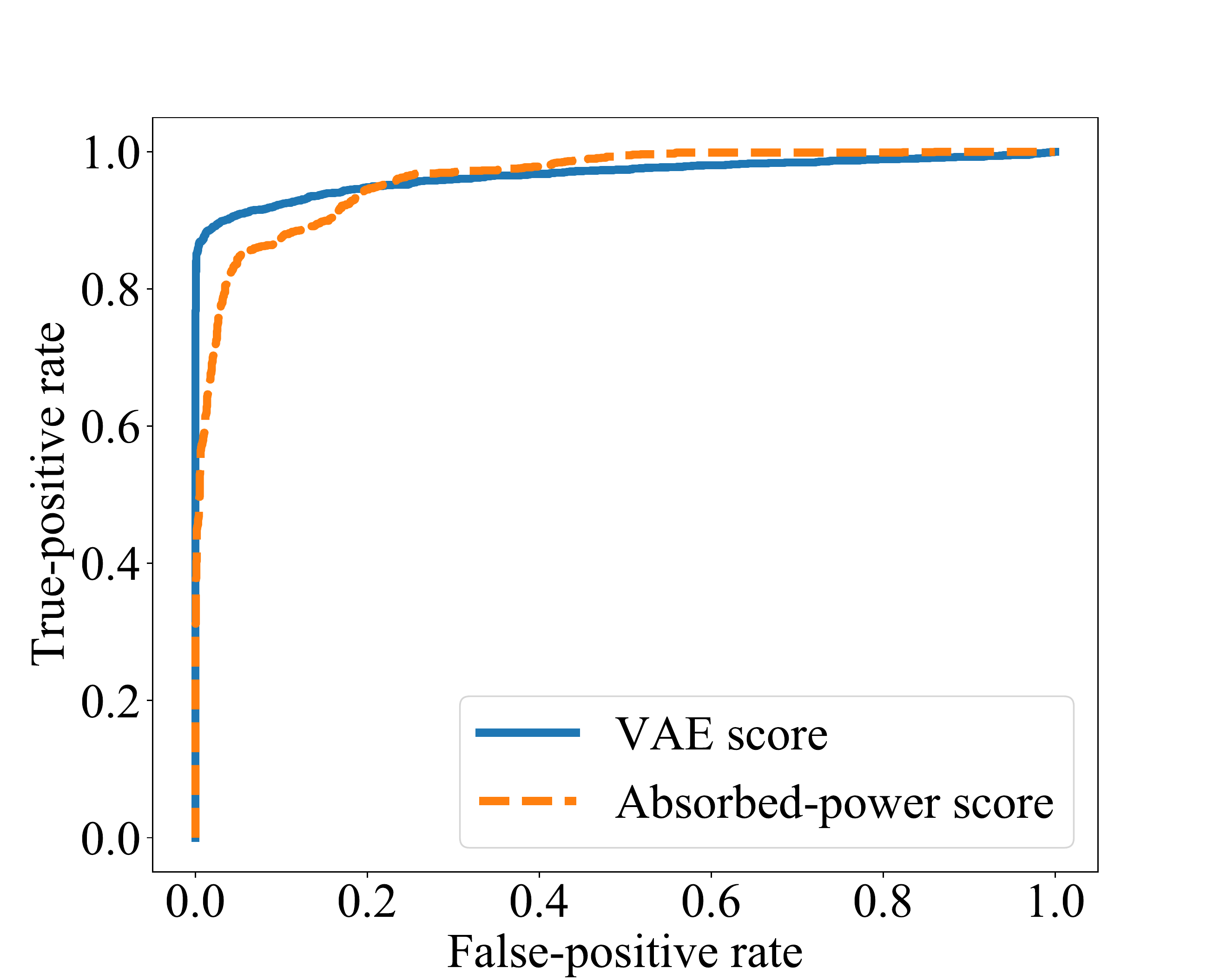}
\caption{\caphead{Receiver-operating-characteristic (ROC) curve:}
The spin glass was trained with three drives,
then tested with a familiar drive or with a novel drive.
From a response of the system's, an ROC curve can be defined.
The blue, solid curve is defined in terms of a variational autoencoder;
and the orange, dashed curve is defined in terms of absorbed power.
}
\label{fig_ROC}
\end{figure}

At the start of the introduction, we described
how absorbed power has been used to identify novelty detection.
A system detects novelty when labeling a stimulus as familiar or unfamiliar.
The stimulus produces a response that exceeds a threshold or lies below.
If the stimulus exceeds the threshold, 
an observer should guess that the stimulus is novel.
Otherwise, the observer should guess that the stimulus is familiar.

The observer can err in two ways:
One commits a \emph{false positive} by believing a familiar drive to be novel.
One commits a \emph{false negative} by believing a novel drive to be familiar.
The errors trade off:
Raising the threshold lowers the probability
$p( \text{pos.} | \text{neg.} )$,
suppressing false positives at the cost of false negatives.
Lowering the threshold lowers the probability 
$p( \text{neg.} | \text{pos.} )$,
suppressing false negatives at the cost of false positives.

The \emph{receiver-operating-characteristic} (ROC) curve
depicts the tradeoff's steepness 
(see~\cite{Brown_06_Receiver} and Fig.~\ref{fig_ROC}).
Each point on the curve corresponds to one threshold value.
The false-positive rate $p( \text{pos.} | \text{neg.} )$ 
runs along-the $x$-axis; and the true-positive rate, 
$p( \text{pos.} | \text{pos.} )$, along the $y$-axis.
% The curve rises sharply if
% the response sensitively reflects accurate novelty detection:
% Suppressing false positives barely suppresses true positives.
The greater the area under the ROC curve,
the more sensitively the response reflects accurate novelty detection.

We measure a many-body system's 
novelty-detection ability using an ROC curve.
Let us illustrate with the spin glass.
We constructed two random drives,
$\{A, B, C \}$ and $\{D, E, F \}$.
We trained the spin glass on $\{A, B, C \}$.
In each of 3,000 trials, we then tested the spin glass with
$A$, $B$, or $C$.
In each of 3,000 other trials, we tested with $D$, $E$, or $F$.
We defined one response in terms of a VAE, as detailed below;
measured the absorbed power, a thermodynamic response;
and, from each response, drew an ROC curve
(Fig.~\ref{fig_ROC}).
The curves show that representation learning offers 
greater sensitivity to the spin glass's novelty detection.

We defined the representation-learning response in terms of a VAE as follows.
We trained the VAE on the configurations 
assumed by the spin glass during its training.
The VAE populated latent space with three clumps of dots.
We modeled the clumps with a hard mixture 
$p_{ABC} (z_1,  z_2)$ of three Gaussians.\footnote{
% < f >
A mixture is hard if it models each point as
belonging to only one Gaussian.}
% < /f >
%%% About how the mixture of Gaussians is created:
% The Gaussians' parameters are estimated via maximum likelihood, with some coding package.
% Weishun doesn't tell the program the number of clusters (the number of Gaussians that the mixture should contain).
% Reference: Meeting notes --> Weishun - 11/27/19
We then fed the VAE the configuration that resulted from testing the spin glass.
The VAE mapped the configuration to a latent-space point 
$(z_1^{\rm test},  z_2^{\rm test})$.
We calculated the probability 
$p_{ABC} (z_1^{\rm test},  z_2^{\rm test}) \, dz_1 dz_2$ 
that the $ABC$ distribution produced the new point.
This protocol led to the blue, solid curve in Fig.~\ref{fig_ROC}.

We defined a thermodynamic ROC curve in terms of absorbed power.
Consider the trials in which the spin glass is tested with field $A$.
We histogrammed the power absorbed by the spin glass 
at the end of the testing.
We form another histogram from the $B$-test trials;
and a third histogram, from the $C$-test trials.
To these histograms was compared
the power $\mathcal{P}$ that the spin glass absorbed
during a test with an arbitrary field.
We inferred the likelihood that $\mathcal{P}$ resulted from
a familiar field.
The results form the orange, dashed curve in Fig.~\ref{fig_ROC}.

% Reference for the areas: email chain "Re: Harvard CMTists" -- message sent by Weishun on 11/9/19 -- Beware: This email chain split into multiple chains.
The two ROC curves enclose regions of approximately the same area:
The VAE curve encloses an area-0.9633 region;
and the thermodynamic curve, an area-0.9601 region.
On average across all thresholds, therefore,
the responses register novelty detection approximately equally.
Yet the responses excel in different regimes:
The VAE achieves greater true-positive rates at low false-positive rates,
and the absorbed power achieves greater true-positive rates 
at high false-positive rates.
This two-regime behavior persisted across batches of trials,
though the enclosed areas fluctuated a little.
Hence anyone paranoid about avoiding false positives
should measure a many-body system's novelty detection with a VAE,
while those more relaxed might prefer the absorbed power.

Why should the VAE excel at low false-positive rates?
Because of the VAE's skill at generalizing, we conjecture.
Upon training on cat pictures, a VAE generalizes from the instances.
Shown a new cat, the VAE recognizes its catness.
Perturbing the input a little perturbs the VAE's response little.
Hence changing the magnetic field a little,
which changes the spin-glass configuration little,
should change latent space little,
obscuring the many-body system's novelty detection.
This obscuring disappears when the magnetic field changes substantially.

\textbf{A many-body system's novelty-detection ability is quantified with
an ROC curve formed from a VAE's latent space
or a thermodynamic response, depending on the false-positive threshold.}

% %
% %
% %
% \begin{appendices}

% % To switch for two-column to one-column formatting here, using the command 
% \onecolumngrid

% % Number subsections in the appendices as in the main text,
% % except skip the capital Roman numerals.
% \renewcommand{\thesection}{\Alph{section}}
% \renewcommand{\thesubsection}{\Alph{section} \arabic{subsection}}
% \renewcommand{\thesubsubsection}{\Alph{section} \arabic{subsection} \roman{subsubsection}}

% % Label the equations in Section L as L1, L2, ...
% \makeatletter\@addtoreset{equation}{section}
% \def\theequation{\thesection\arabic{equation}}

%
%
%
%
%
\section{Details about the variational autoencoder}
\label{sec_NN_Details}

% References: 
% (1) Email chain "Specification of NNs" -- which now consists of several separate threads
% (2) Email chain "VAE app -- manuscript files" -- begun in late Nov. 2019

We briefly motivate and review VAEs,
then describe the VAE applied in the main text.
Further background about VAEs can be found in~\cite{Kingma_13_Auto,JR_14_Stochastic,Doersch_16_Tutorial}.
We denote vectors with boldface in this section.

Denote by $\mathbf{X}$ data that has a probability 
$p_{ \bm{\theta} }(\mathbf{x})$
of assuming the value $\mathbf{x}$.
$\bm{\theta}$ denotes a parameter,
and $p_{ \bm{\theta} } (\mathbf{x})$ is called the \emph{evidence}.
We do not know the form of $p_{ \bm{\theta} }(\mathbf{x})$,
when using representation learning.
We model $p_{ \bm{\theta} }(\mathbf{x})$ by identifying 
latent variables $\mathbf{Z}$ 
that assume the possible values $\mathbf{z}$.
Let $p_{ \bm{\theta} }(\mathbf{x} | \mathbf{z})$ 
denote the conditional probability that
$\mathbf{X} = \mathbf{x}$, given that $\mathbf{Z} = \mathbf{z}$.
We model the evidence, using the latent variables, with
\begin{align}
 p_{ \bm{\theta} }(\mathbf{x}) 
 = \int d\mathbf{z} \; 
 p_{ \bm{\theta} }( \mathbf{x} | \mathbf{z} ) 
 p( \mathbf{z} ) .
\end{align}

$p_{ \bm{\theta} }(\mathbf{x} | \mathbf{z} )$ 
can be related to the posterior distribution 
$p_{ \bm{\theta} } (\mathbf{z} | \mathbf{x})$.
The posterior is the probability that, if 
$\mathbf{X} = \mathbf{x}$, then $\mathbf{Z} = \mathbf{z}$. By Bayes' rule,
$p_{ \bm{\theta} }( \mathbf{z} | \mathbf{x} ) 
= p_{ \bm{\theta} }( \mathbf{x} | \mathbf{z} )
p( \mathbf{z} )  /  p_{ \bm{\theta} }( \mathbf{x} )$. 
Calculating the posterior is usually impractical,
as $p_{ \bm{\theta} }(\mathbf{x})$ 
is typically intractable (cannot be calculated analytically).
% <-- Explanation from Weishun (email chain "VAE app -- manuscript files," message sent on 11/30/19): "intractable here refers to the fact that pθ(x) is like a partition function, in most cases you cannot analytically calculate it because you have to integrate over all possible latent variables (just like summing over all possible configurations)."
Hence we approximate the posterior with a variational model 
$q_{ \bm{\phi} }( \mathbf{z} | \mathbf{x} )$.
The optimization parameter $\bm{ \phi }$ denotes 
the NN's weights and biases.
% Reference for the preceding line: email chain "VAE app -- manuscript files" -- message sent by Weishun on 2/12/20

The approximation introduces an inference error,
quantified with the Kullback-Leibler (KL) divergence.
Let $P(\mathbf{u})$ and $Q(\mathbf{u})$ denote distributions over
the possible values $\mathbf{u}$ of a variable.
The KL divergence quantifies the distance between the distributions:
\begin{align}
   D_\KL \LParen P(\mathbf{u}) || Q(\mathbf{u}) \RParen
   % % %
   & :=  \mathbb{E}_{ P(\mathbf{u}) }
        \left[  \ln P( \mathbf{u} ) \right]
        -  \mathbb{E}_{ P(\mathbf{u}) }
        \left[ \ln  Q(\mathbf{u})  \right] \\
   % % %
   \label{eq_D_KL_Nonneg}
   & \geq  0 .
\end{align}
We denote by $\mathbb{E}_{ P(\mathbf{u}) } [ f( \mathbf{u} ) ]$
the expectation value of a function $f(\mathbf{u})$.
Operationally, the KL divergence equals the maximal efficiency 
with which the distributions can be distinguished, 
on average, in a binary hypothesis test.
We quantify our inference error with the KL divergence between
the variational model and the posterior,
$D_\KL \LParen  
q_{ \bm{\phi} }( \mathbf{z} | \mathbf{x} ) || 
                         p_{ \bm{\theta} }( \mathbf{z} | \mathbf{x} )  \RParen$. 

Recall that we wish to estimate $p_{\bm \theta} ( \mathbf{x} )$:
An accurate estimate lets us predict $\mathbf{x}$ accurately.
We wish also to estimate the latent posterior distribution, 
$q_{\bm \phi} (\mathbf{z} | \mathbf{x} )$.
We therefore write out the KL divergence's form,
apply Bayes' rule to rewrite the $p_{ \bm{\theta} }( \mathbf{z} | \mathbf{x} )$,
rearrange terms,
and repackage terms into a new KL divergence:
% From Weishun (email chain "VAE app -- manuscript files," message sent on 11/30/19): "p(z) should not have subscript \theta, because this is prior distribution of latent variables (represents our belief in the distribution of latent variables before we see any data) and it is independent of how we choose our model parameters \theta."
\begin{align}
   \label{eq_Log_Like_1}
   \ln p_{ \bm{\theta} } ( \mathbf{x} )
   % % %
   =  D_\KL \LParen  q_{\bm \phi} ( \mathbf{z} | \mathbf{x} )
                                                     || p_{\bm \theta} ( \mathbf{z} | \mathbf{x} )
                   \RParen
   +  \mathbb{E}_{ q_{\bm \phi} ( \mathbf{z} | \mathbf{x} )}
   \left[ \ln p_{\bm \theta} ( \mathbf{x} | \mathbf{z} )  \right]
   -  D_\KL  \LParen  q_{\bm \phi} ( \mathbf{z} | \mathbf{x} )
                                                       || p ( \mathbf{z} )  \RParen .
\end{align}
The penultimate term encodes our first goal;
and the final term, our second goal.

Recall that the KL divergence is nonnegative.
The sum of the final two terms therefore lower-bounds the log-likelihood,
$\ln p_{\bm \theta}( \mathbf{x} )$.
$\mathbf{x}$ denotes the event observed,
$\bm{\theta}$ denotes a possible cause,
and $p_{\bm \theta}$ denotes the likelihood that 
$\bm{\theta}$ caused $\mathbf{x}$.
Maximizing each side of Eq.~\eqref{eq_Log_Like_1},
and invoking Ineq.~\eqref{eq_D_KL_Nonneg}, yields
\begin{align}
   \label{eq_Log_Like_2}
   \max_\theta  \left\{   \ln p_{ \bm{\theta} } ( \mathbf{x} )   \right\}
   % % %
   \geq \max_\theta  \left\{
   \mathbb{E}_{ q_{\bm \phi} ( \mathbf{z} | \mathbf{x} )}
   \left[ \ln p_{\bm \theta} ( \mathbf{x} | \mathbf{z} )  \right]
   -  D_\KL  \LParen  q_{\bm \phi} ( \mathbf{z} | \mathbf{x} )
                                                       || p ( \mathbf{z} )  \RParen 
   \right\} .
\end{align}
The RHS is called the \textit{evidence lower bound} (ELBO). 

A VAE is a neural network that implements the ELBO.
$q_{\bm \phi} ( \mathbf{z} | \mathbf{x} )$ encodes the input $\mathbf{X}$, 
and $p_{\bm \theta} ( \mathbf{x} | \mathbf{z} )$ decodes.
The VAE has the cost function
% From Weishun (email chain "VAE app -- manuscript files," message sent on 11/30/19): "p(z) should not have subscript \theta. And p(x) also should not have subscript \theta, either. Here the p(x) should really be p_{empirical} (x), meaning the data distribution as we see in the actual dataset. In contrast, p_{\theta}(x) is what we use to model p_{empirical}(x)."
\begin{equation}
   \label{eq_VAE_Cost}
   \mathcal{L}_{\text{VAE}} 
   := \mathbb{E}_{  p_{\rm emp}( \mathbf{x} )  }
   \left[ \mathbb{E}_{  q_{\bm \phi} ( \mathbf{z} | \mathbf{x} )  }
   \left[  \ln p_{\bm \theta} ( \mathbf{x} | \mathbf{z} )  \right] 
   -  D_\KL  \LParen  q_{\bm \phi} ( \mathbf{z} | \mathbf{x} ) 
                                                        \| p( \mathbf{z} )  \RParen \right].
\end{equation}
$p_{\rm emp}( \mathbf{x} )$ denotes the distribution inferred from the empirical dataset.
Given input values $\mathbf{x}$, the VAE generates a latent distribution 
$q_{\bm \phi} ( \mathbf{z} | \mathbf{x} ) 
= \mathcal{N}( \bm{\mu}_{ \mathbf{z} | \mathbf{x} }, 
                                           \bm{\Sigma}_{ \mathbf{z} | \mathbf{x} } )$.
We denote by $\mathcal{N} ( \bm{\mu}, \bm{\Sigma} )$
the standard multivariate normal distribution
whose vector of means is $\bm{\mu}$ 
and whose covariance matrix is $\bm{\Sigma}$.
Neural-network layers parameterize the VAE's
$\bm{\mu}_{ \mathbf{z} | \mathbf{x} }$ and 
$\bm{\Sigma}_{ \mathbf{z} | \mathbf{x} }$.
%%% What we mean by "the layers parameterize mu and Sigma"
% (Reference: Meeting notes --> WZ - 12/14/19 --> final page)
% mu and Sigma are functions of the input, x.
% f(x) = mu,  g(x) = Sigma
% f and g can be approximated with NNs, because NNs are universal functional approximators.
% A NN identifies good choices of parameters in functions that approximate f and g.
Latent vectors are sampled according to 
$q_{\bm \phi}( \mathbf{z} | \mathbf{x} )$,
then decoded into outputs distributed according to
$p_{\bm \theta} ( \mathbf{x} | \mathbf{z} ) 
= \mathcal{N}( \bm{\mu}_{ \mathbf{x} | \mathbf{z} }, 
                        \sigma^2_{ \mathbf{x} | \mathbf{z} }  \id  )$.
Neural-network layers parameterize the mean vector 
$\bm{\mu}_{ \mathbf{x} | \mathbf{z} }$.
The variance $\sigma^2_{ \mathbf{x} | \mathbf{z} }$ is a hyperparameter.

A VAE with the following architecture produced the results in the main text.
The style was borrowed from~\cite{Hafner_18_Building}.
Two fully connected $200$-neuron hidden layers
process the input data.
One fully connected two-neuron hidden layer parameterizes each of
$\bm{\mu}_{ \mathbf{z} | \mathbf{x} }$ and 
$\bm{\Sigma}_{ \mathbf{z} | \mathbf{x} }$.
Two fully connected $200$-neuron hidden layers process the latent variables.
An output layer reads off the outputs. 
% From Weishun (email chain "VAE app -- manuscript files," message sent on 11/30/19): "Read-off is a fancy way to say there is no activation function for this layer, that it is linear (a common practice for the output layer)."
We choose $\sigma^2_{ \mathbf{x} | \mathbf{z} }=1$ 
and use Rectified Linear Unit (ReLU) activations
% From Weishun (email chain "VAE app -- manuscript files," message sent on 11/30/19): "ReLU stands for “Rectified Linear Unit”, define to be ReLU(x) = max{0, x}, it is the most commonly used non-linear activation function in ML, so I don’t think we need to cite a reference for it."
for all hidden layers.

\section{Distinction between robust learning \\
and two superficially similar behaviors}
% The spin glass is neither enslaved to the drive nor frozen in the good-learning regime}
\label{sec_Not_Enslaved_Or_Frozen}

% References:
% (i) Meeting notes --> Weishun, Jacob - 10/4/19
% (ii) Meeting notes —> “Weishun, Jacob, Jeremy - 7/10/19”
% (iii) Meeting notes —> “England group - 8/1/19”
% (iv) Meeting notes —> folder “Jacob, Weishun, Jeremy - 8/8/19” —> notes of the same name —> p. 4: “Weishun’s report, part 1”
% (v) Meeting notes --> folder "Weishun, Jacob - 10/4/19" --> p. 7
% (vi) Email chain "fraction of positive field energy spins" -- begun on 9/11/19
% (vii) Meeting notes --> Jacob, Weishun, Jeremy - 9/12/19

Learning contrasts with two other behaviors
that the spin glass could exhibit,
entraining to the field and near-freezing.

\subsection{Entraining to the field}
\label{sec_Enslave}

% Reference for discussion about what distinguishes silly-putty behavior from nontrivial learning: Meeting notes —> Youngsters - 9/26/19 —> p. 2
Imagine that most spins align with any field $A$. 
The configuration reflects the field
as silly putty reflects the print of a thumb
pressing on the silly putty.
Smoothing the silly putty's surface wipes the thumbprint off.
Similarly, applying a field $B \neq A$ to the spin glass
wipes the signature of $A$ from the configuration.
From the perspective of the end of the application of $B$,
the spin glass has not learned $A$.
The spin glass lacks a long-term memory of the field;
the field is encoded in no robust, deep properties of the configuration.

We can distinguish learning from entraining
by calculating the percentage of the spins 
that align with the field at the end of training.
If the spins obeyed the field, 100\% would align.
If the spins ignored the field, $50\%$ would align, on average.
Hence the spin glass's entraining is quantified with
\begin{align}
   2 ( \text{Percentage of spins aligned with the field} )  - 100 .
\end{align}
(This measure does not apply to alignment percentages $< 50$,
which are unlikely to be realized.)

% Figure: Fraction of the spins that antialign with the field
\begin{figure}[hbt]
\centering
\includegraphics[width=.6\textwidth, clip=true]{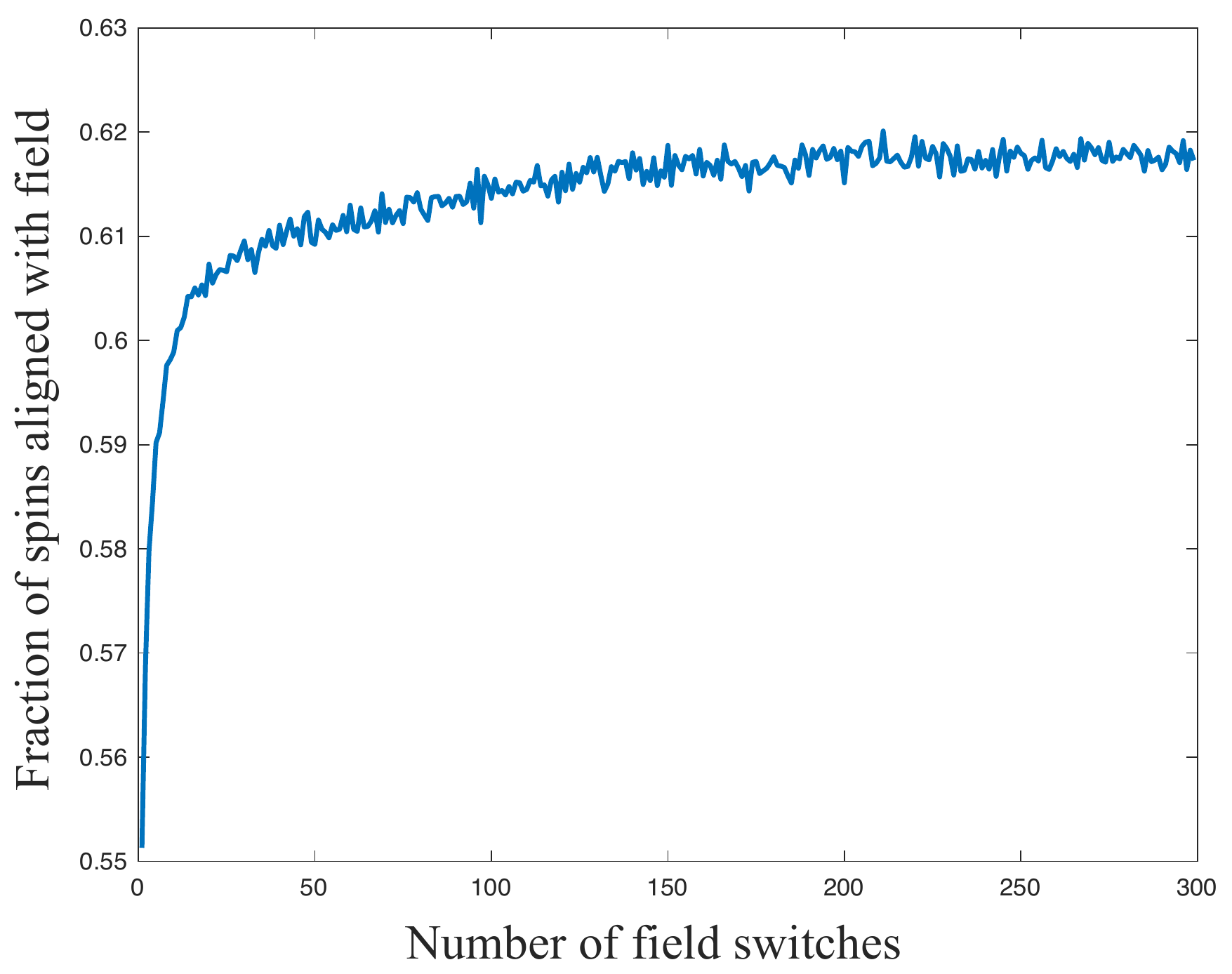}
\caption{\caphead{Fraction of the spins aligned with the field,
as a function of time:}
If a fraction $\approx 1$ of the spins align, the spin glass resembles silly putty,
which shallowly reflects the print of a thumb that presses on it.
Robust learning stores information deep in a system's structure.}
\label{fig_Frac_Aligned}
\end{figure}

Figure~\ref{fig_Frac_Aligned} shows data collected 
about the spin glass in the good-learning regime (Sec.~\ref{sec_Spin_Glass}).
%%% Drive strength: quantified with a standard deviation of 3, as in the setup
% Number of trials: probably 1,000
% <-- Reference: email chain "Are you ok?" -- message sent by Jacob on 12/18/19
The number of aligned spins is plotted against
the amount $t$ of time for which the spin glass has trained.
After the application of one field, 55\% of the spins align with the field.
At the end of training, 62\% align. 
Hence the spins' entraining grows from 10\% to 24\%.
Growth is expected, as the spin glass learns the training drive.
%%
% NYH: But we've said that the spin glass resembles silly putty early on, if the drive is strong enough, then loses its silly-putty behavior, upon learning the drive more deeply. So the entraining should shrink. ?
% JMG (email chain "Jacob - please" -- message sent by NYH on 1/9/20 -- contains message sent by Jacob): "I might try playing around with a similar quantity that's a little more informative, that accounts for the magnitude of the field that the spins are aligning with. I imagine lower magnitude fields might be more aligned early on than they are later, while high magnitude fields becomes more aligned."
%%
But $24\%$ is an order of magnitude less than $100\%$, 
so the spin glass is not entrained to the field.

\subsection{Near-freezing}
\label{sec_Near_Freeze}

Suppose that the spin glass is nearly frozen.
Most spins cannot flip, but a few jiggle under most fields.
The spin glass does not learn any field effectively,
being mostly immobile.
But the few flippable spins reflect the field.
A bottleneck NN could guess the field from those few spins.
The NN's low loss function would induce a false positive, 
leading us to believe that the spin glass had learned.

% References about the following analysis:
% (i) Meeting notes —> “Weishun, Jacob, Jeremy - 7/10/19”
% (ii) Meeting notes —> “England group - 8/1/19”
% (iii) Meeting notes —> folder “Jacob, Weishun, Jeremy - 8/8/19” —> notes of the same name —> p. 4: “Weishun’s report, part 1”

We can avoid false positives by measuring two properties.
First, we measure the percentage of the spins
that antialign with the field.
If the percentage consistently $\gg 0$,
many of the spins are not frozen.
Figure~\ref{fig_Frac_Aligned} confirms that many are not.

% Figure: Which spins the AE uses to infer the drive
% Source of figure: email chain "Decoding latent space" -- message sent by Weishun on 1/6/20
% Source of old version of figure: email chain "Standardizing plots" -- messages sent by Weishun on 11/28/19
\begin{figure}[hbt]
\centering
\includegraphics[width=.6\textwidth, clip=true]{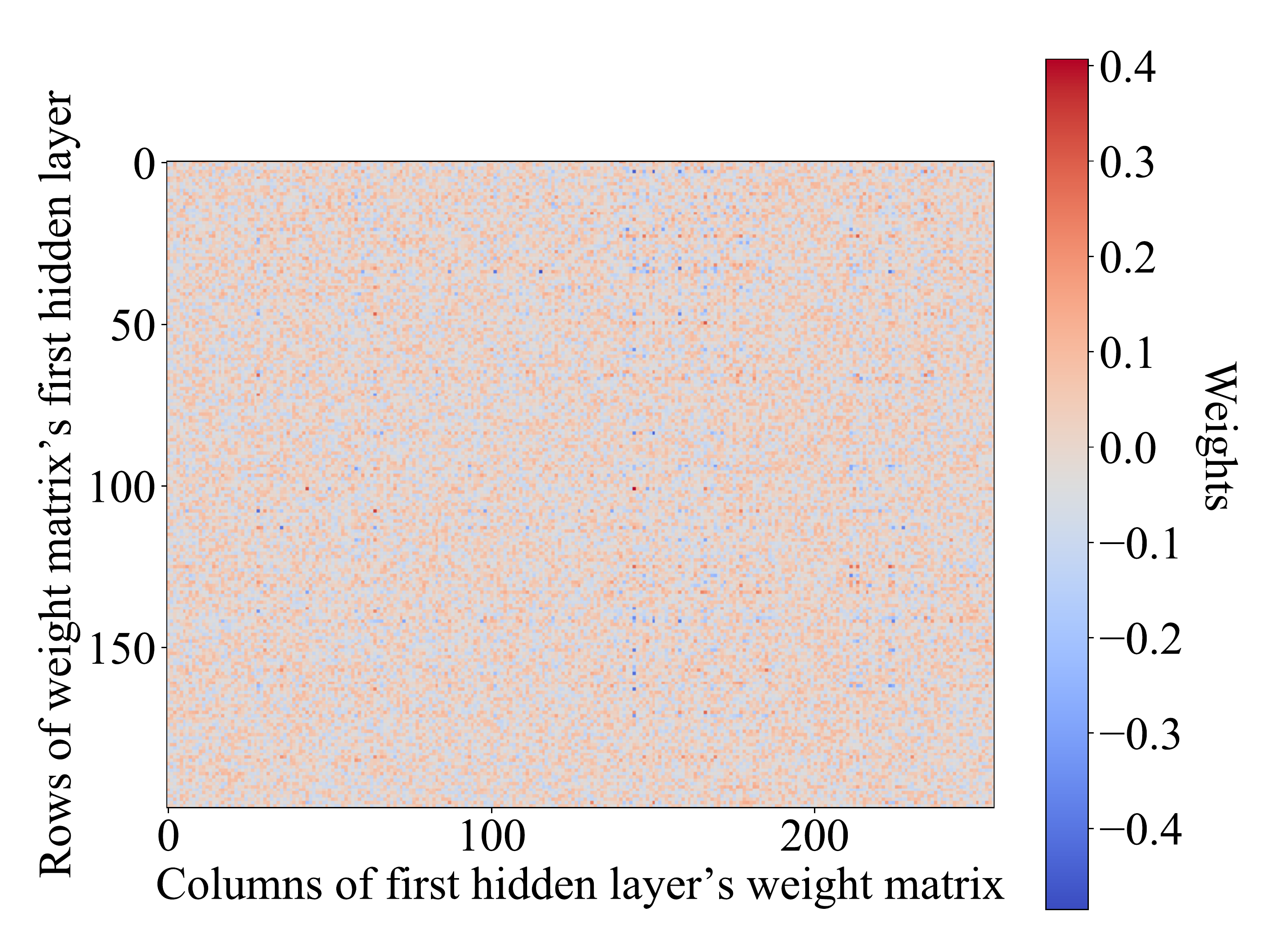}
\caption{\caphead{How much information about each spin 
the variational autoencoder compresses:}
This figure represents the first hidden layer's weight matrix.
The weight matrix transforms the input layer, which consists of 256 neurons,
into the first hidden layer, which consists of 200 neurons.
The matrix's elements are replaced with colors.
Each vertical line corresponds to one spin.
The farther leftward a stripe, the lesser the spin's field energy 
[Eq.~\eqref{eq_Hamiltonian_j}].
} 
\label{fig_Which_Spins_Used_By_AE_b}
\end{figure}

Second, we check that the VAE compresses information about
spins that have many different field energies $A_j(t) s_j$
[Eq.~\eqref{eq_Hamiltonian_j}].
We illustrate with the protocol used to generate Fig.~\ref{fig_Latent_Space_Visualize}:
We trained the spin glass on a drive $\{A, B, C\}$
in each of many trials.
On the end-of-trial configurations, the VAE was trained.

% Reference: Meeting notes --> WZ - 12/14/19 --> p. 4
A configuration is represented in the VAE's input layer, a column vector.
A weight matrix transforms the input layer
into the first hidden layer, another column vector.
The weight matrix is depicted in Fig.~\ref{fig_Which_Spins_Used_By_AE_b}.
The matrix's numerical entries have been replaced with colors.
Each vertical stripe corresponds to one spin.
The farther leftward a stripe, the lesser the spin's field energy.
The darker a stripe, the more information about the spin 
the VAE uses when forming $Z$.
% Blank stripes represent spins about which the NN encodes no information.
The plot is approximately invariant, at a coarse-grained level,
under translations along the horizontal.
(On the order of ten exceptions exist.
These vertical stripes contain several dark dots.
An example appears at $x \approx 150$.
But the number of exceptions is much less than the number of spins:
$\approx 10 \ll 256$.)
Hence the NN uses information about spins of many field energies.
The spins do not separate into
low-field-energy flippable spins
and high-field-energy frozen spins.

\section{Maximum \emph{a posteriori} estimation (MAP estimation)}
\label{app_MAP}

% Reference for first draft of Section: 
% Email chain “MAP estimates and score” — begun 9/2019

% Reference about the relationship between the MI and the MAP-estimate score:
% Email chain "Questions for Sarah: MAP-estimate score vs. MI" -->
% message sent by Sarah on 9/26/19

This Section details the MAP estimation
applied in Sections~\ref{sec_Classify}-\ref{sec_Discriminate}.
MAP estimates help answer the question
``How accurately can the drive be identified from the spin configuration?''
We return to the notation used in the introduction,
denoting the drive by $Y$ and the configuration by $X$.

In information theory, we answer this question using the conditional entropy,
\begin{align}
   \label{eq_Cond_Ent}
   H(Y | X) := - \sum_{x, y} p(x, y)  \log \frac{ p(x, y) }{ p(x) } .
\end{align}
$p(x, y)$ denotes a joint distribution; and $p(x)$, a marginal.
The conditional entropy quantifies
the uncertainty about the drive, given the configuration. 
Equation~\eqref{eq_Cond_Ent} does not refer to
any estimator of $Y$.
Rather, $H(Y | X)$ underlies a bound on the accuracy with which
any estimator can reconstruct the drive from the configuration,
by Fano's inequality.
Estimating $H(Y | X)$ proves difficult, due to undersampling:
An enormous amount of data is needed to estimate the distribution
$p(y | x)$ accurately enough to estimate $H(Y | X)$
(Sec.~\ref{sec_Opportunities}).

Undersampling plagues also the mutual information,
a sister of the conditional entropy:
$I(X; Y)  :=  H(Y) - H(Y | X)$.
The Shannon entropy, $H(Y) := - \sum_y p(y) \log p(y)$,
quantifies the randomness in the drive variable.
The mutual information quantifies 
the information about the drive in the configuration and vice versa.

$H(Y | X)$ and $I(X; Y)$ offer one answer to our question.
Another comes from using MAP estimation 
to predict drives from configurations,
then scoring the predictions.
MAP estimation proceeds as follows.
One approximates the conditional probability distribution $p(y|x)$ 
from the data via any possible strategy.
(We detail one strategy below.)
Let $\tilde{p} (y | x)$ denote the approximation.
Given a configuration $x$, one predicts that it resulted from the drive
\begin{equation}
   \label{eq_How_To_MAP}
    \hat{y} = \arg\max_{y}  \Set{  \tilde{p} ( y | x )  }
\end{equation}
that has the greatest conditional probability.
Equation~\eqref{eq_How_To_MAP} is the MAP estimator.
We use it to map all the configurations $x$
to drive predictions $\hat{y}$.
The frequency with which $\hat{y} = y$ is the estimator's score.

To use the MAP estimator~\eqref{eq_How_To_MAP},
we must approximate the conditional probability distribution $p(y | x)$.
Our approximation suffers from undersampling.
Hence we invoke the map $f(x) = z$ from configurations $x$
to the low-dimensional latent-space variable $z$.
Approximating $p \LParen y | f(x) \RParen$ proves easier than
approximating $p(y | x)$. By Bayes' rule, 
$p \LParen y | f(x) \RParen
=  \frac{  p \LParen f(x) | y \RParen  p(y)  }{  p \LParen f(x) \RParen  }$.
The approximation $\tilde{p} \LParen y | f(x) \RParen$ factors analogously.
We redefine our estimator as
\begin{align}
   \hat{y}
   % % %
   & =  \arg  \max_y  \Set{  \tilde{p} \LParen y | f(x) \RParen  } 
   % % %
   =  \arg \max_y  \Set{
   \frac{  \tilde{p} \LParen f(x) | y \RParen  \tilde{p}(y)  }{  
             \tilde{p} \LParen f(x) \RParen  }  }
   % % %
   =  \arg \max_y  \Set{   \tilde{p} \LParen f(x) | y \RParen  \tilde{p}(y)  } .
\end{align}
The final equality holds because the arg-max over $y$ cannot depend on 
the $y$-independent $\tilde{p} \LParen f(x) \RParen$.
The fields $y$ are chosen uniformly randomly from the drive.
Hence $p(y) \approx \tilde{p}(y)$ is constant, and
\begin{align}
   \hat{y}
   % % %
   & \approx  \arg  \max_y  \Set{  \tilde{p} \LParen y | f(x) \RParen  } .
\end{align}
This MAP estimate equals the maximum-likelihood estimate.
Generally, MAP estimation with a uniform prior
amounts to maximum-likelihood estimation.
We use only uniform priors.
Other applications of our toolkit, however,
can benefit from alternative priors, if extra information is available.
Hence we present the MAP generalization of maximum-likelihood estimation.
A Gaussian distribution approximates $p \LParen y | f(x) \RParen$ well,
so $\hat{y}$ can be approximated easily.

\section{Memory capacity attributed to the many-body system \\
by the absorbed power}
\label{app_Capacity_Work}

% References:
% (1) Meeting notes --> Youngsters - 9/5/19 --> subfolder "Jacob - 9/5/19"
% (2) Meeting notes --> Jacob, Weishun, Jeremy - 9/12/19 --> p. 6
% (3) Email chain “Meeting today @2:30 PM EDT/11:30 AM PDT” — messages sent on 9/5/19 and afterward
% (4) Meeting notes --> Group skype - 9/19/19

In Sec.~\ref{sec_Capacity}, we compared the memory capacity
registered by the VAE
to the capacity registered by the absorbed power.
The study involved MAP estimation on 
drives of 40 fields selected from 50 fields.
The choice of 50 is explained here:
Fifty fields exceed the spin-glass capacity
registered by the absorbed power.

Recall how memory has been detected thermodynamically~\cite{Gold_19_Self}:
Let a many-body system be trained with
a drive that includes a field $A$.
Consider testing the system, afterward, 
with an unfamiliar field $B$, and then with $A$.
Suppose that the absorbed power jumps substantially when $B$ is applied
and less when $A$ is reapplied.
The many-body system identifies $B$ as novel
and remembers $A$,
according to the absorbed power.

% Reference for the following paragraph: email chain “Re: Meeting today @2:30 PM EDT/11:30 AM PDT”
We sharpen this analysis.
First, we divide the trial into time windows.
% Reference for the following line: email chain “Re: Meeting today @2:30 PM EDT/11:30 AM PDT” -- message sent by Jacob on 10/3/19
During each time window, the field switches 10 times.
% Reference for the following parenthesized comment: email chain “Jacob - please" -- message sent by Jacob on 1/8/20
(The 10 eliminates artificial noise and is not critical.
Our qualitative results are robust with respect to changes in such details.)
We measure the absorbed power at the end of each time window
and at the start of the subsequent window.
We define ``the absorbed power jumps substantially'' as
``the absorbed power jumps, on average over trials, 
by much more than the noise
(by much more than the absorbed power fluctuates across a trial)'':
\begin{align}
   \label{eq_Capacity_Condn_Work}
   & \langle ( \text{Power absorbed at start of later time window} ) 
    \\ \nonumber & \quad
   - ( \text{Power absorbed at end of preceding time window} )
   \rangle_{ \text{trials} } 
   \\ \nonumber & 
   % % %
   \gg  \text{Standard deviation in } 
   [ ( \text{Power absorbed at start of later window} )
   \\ \nonumber & \quad
   % \\ \nonumber & \qquad   \qquad \qquad \qquad \qquad
   -  ( \text{Power absorbed at end of preceding window} ) ] .
\end{align}
Consider including only a few fields in the training drive, 
then growing the drive in later trials.
The drive will tax the spin glass's memory until exceeding the capacity.
The LHS of~\eqref{eq_Capacity_Condn_Work}
will come to about equal the RHS.

Figure~\ref{fig_Capacity_Work} illustrates with the spin glass.
On the $x$-axis is the number of fields in the training drive.
On the $y$-axis is the ratio of 
the left-hand side of Ineq.~\eqref{eq_Capacity_Condn_Work}
to the right-hand side (LHS/RHS).
Where LHS/RHS $\approx 1$, the spin glass reaches its capacity.
This spin glass can remember $\approx 15$ fields,
according to the absorbed power.

% Figure: Capacity according to absorbed power
% Reference: email chain "Capacity Study" -- message sent by Jacob on 11/19/19
% Reference for old version of plot (spin-flip rate vs. drive size): Meeting notes --> Group skype - 9/19/19 --> p. 14
% Jacob’s initial construction of a plot like this one: email chain “Re: Meeting today @2:30 PM EDT/11:30 AM PDT” — message sent by Jacob on 9/5/19
\begin{figure}[hbt]
\centering
\includegraphics[width=.6\textwidth, clip=true]{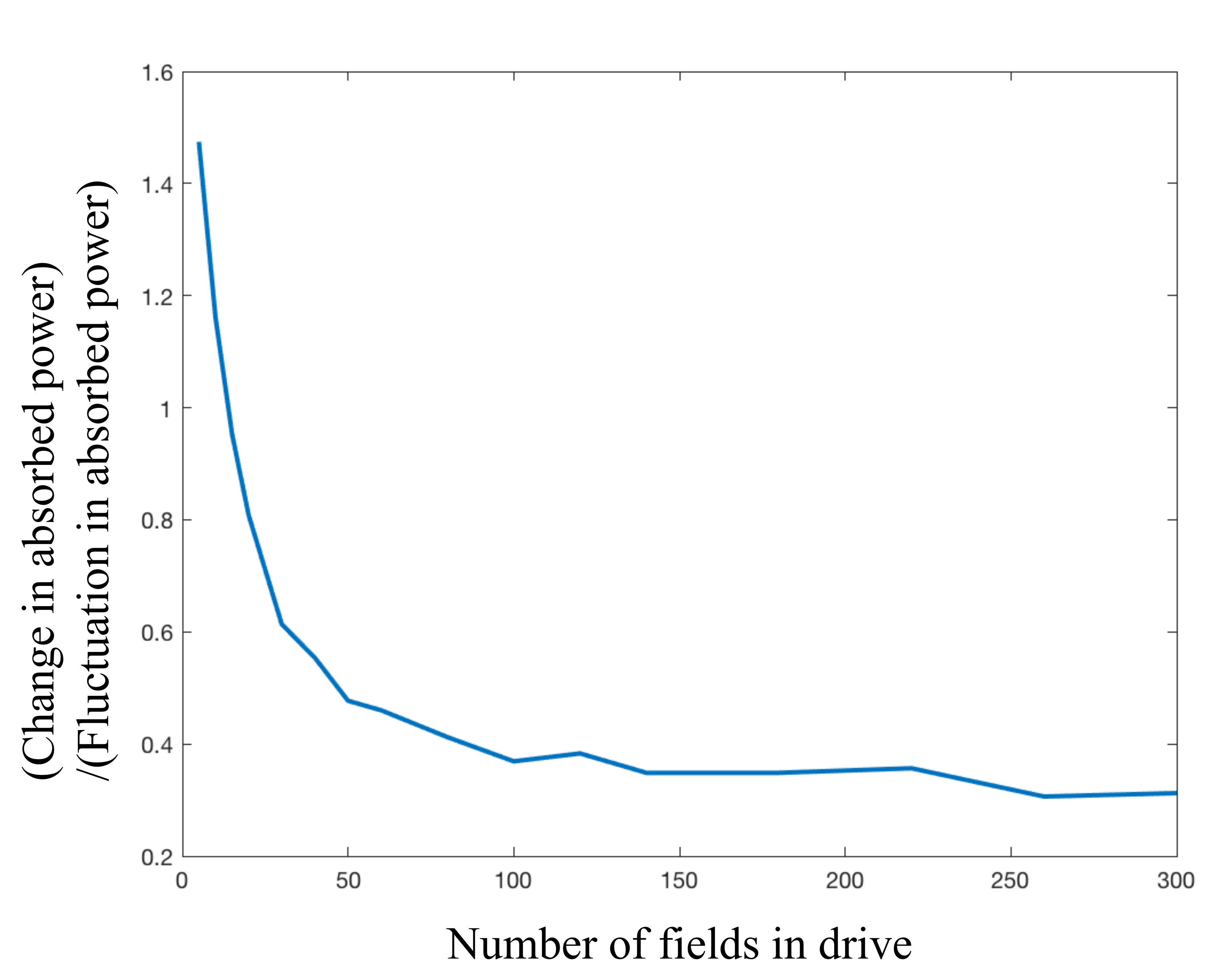}
\caption{\caphead{Estimate of memory capacity by absorbed power:}
A many-body system reaches its capacity, 
according to the absorbed power, when
[left-hand side of Ineq.~\eqref{eq_Capacity_Condn_Work}]
/ (right-hand side) $\approx 1$.
The curve $\approx 1$, and a 256-spin glass reaches its capacity,
when the training drive contains $\approx 15$ fields.}
\label{fig_Capacity_Work}
\end{figure}
\section{Justification of use of machine learning}
\label{sec_Justify_ML}

% References:
% Meeting notes —> Compare - lin'r model, clustering
% Meeting notes —> Why ML? - 11/16/19
% Meeting notes --> WZ - 11/27/19 - why ML, capacity

Deep learning is a powerful tool.
Is it necessary for recovering our results?
Could simpler algorithms detect and quantify many-body learning as sensitively?
Comparable simpler algorithms tend not to, we find.
Two competitors suggest themselves:
single-layer linear autoencoders,
related to PCA~\cite{Bourland_88_Auto}, 
% <-- Alternative reference: https://en.wikipedia.org/wiki/Autoencoder#Dimensionality_Reduction
and clustering algorithms.
Alternatives include generalized linear models~\cite{Bishop_06_Pattern}
and supervised linear autoencoders.
These models, however, perform supervised learning.
They receive more information than the VAE
and so enjoy an unfair advantage.
We analyze the two comparable competitors sequentially.

\subsection{Comparison with single-layer linear autoencoder}

The linear autoencoder is a single-layer NN.
The input, $X$, undergoes a linear transformation:
$Y = mX + b$.
% Reference about the AE's performance: Files from others —> Weishun - 11/17/19 - why ML —> compare various networks —> rightmost column
We compare, as follows, the linear autoencoder's 
detection of field classification 
with the VAE's detection:
% Reference for parts of the following: email chain "VAE app -- manuscript files" -- messag sent by Weishun on 1/2/20
We trained the spin glass on a drive in each of 
3,000-5,000 trials.
Ninety percent of the trials were designated as NN-training data;
and 10\%, as NN-testing data.
For each training trial, we identified the spin glass's final configuration.
On these configurations, each NN performed unsupervised learning.
Each NN then received the configuration with which
the spin glass ended a NN-testing trial.
We inferred the field most likely to have produced this configuration,
using MAP estimation.
The fraction of trials in which the NN points to the correct field 
constitutes the NN's score.
On a three-field drive, the linear autoencoder scored 0.771,
while the VAE scored 0.992.
On a five-field drive, the linear autoencoder scored 0.3934,
while the VAE scored 0.829.
Hence the VAE picks up on more of 
the spin glass's ability to classify fields.

\subsection{Comparison with clustering algorithm}

A popular, straightforward-to-apply algorithm is 
\emph{$k$-means clustering}~\cite{Bishop_06_Pattern}.
$k$ refers to a parameter inputted into the algorithm,
the number of clusters expected in the data.
We inputted the number of drives imposed on the spin glass,
in addition to inputting configurations.
The VAE receives just configurations and so less information.
We could level the playing field by automating the choice of $k$,
using the Bayesian information criterion (BIC)~\cite{Bishop_06_Pattern}.
But clustering with the BIC-chosen $k$ 
would perform no better than
clustering performed with the ideal $k$,
and the ideal clustering performs worse than the VAE.

The protocol run on the spin glass is described 
at the beginning of Sec.~\ref{sec_Capacity}.
Five thousand trials were performed.
The configuration occupied by the spin glass 
at the end of each trial was collected.
Splitting the data into testing and training data
did not alter results significantly.
Hence we fed all the configurations,
with the number $k = 5$ of drives,
to the clustering algorithm.
The algorithm partitioned the set of configurations into subsets.
Each subset contained configurations 
likely to have resulted from the same drive.

% Reference: Email chain "Were you able to get ahold of Jacob?" --> message sent by Weishun on 12/19/19
% Message from Weishun: For the clustering evaluation task, I have corresponded with Sarah, and she also doesn’t have a way to calculate MAP score for clustering. So I dug a little bit more into the Rand Index, and find that it can actually be interpreted as the percentage of correct decisions made by the algorithm in clustering (see the Rand Index section in https://en.wikipedia.org/wiki/Cluster_analysis#External_evaluation). Since MAP score also measures the percentage of correct classification, we can actually directly compare Rand Index and MAP score!
% Personal-favorite reference about Rand index: https://en.wikipedia.org/wiki/Rand_index
Clustering algorithms are assessed with the Rand index,
denoted by RI~\cite{Rand_71_Objective}.
The Rand index differs from the MAP-estimation score (Sec.~\ref{sec_Classify}).
How to compare the clustering algorithm
with the VAE, therefore, is ambiguous.
However, the Rand index quantifies 
the percentage of the algorithm's classifications that are correct.
Hence the Rand index and the MAP-estimation score
have similar interpretations, despite their different definitions.

The clustering algorithm's Rand index began at $\text{RI} = 0$, at $t = 0$.
RI rose during the first $\approx 200$ changes of the drive,
then oscillated around 0.125.
Figure~\ref{fig_Capacity} shows the VAE's performance.
The VAE's score rose during the first $\approx 150$ changes of the drive,
then oscillated around $0.450 > 0.125$.
Hence the VAE outperformed the clustering algorithm.

\section{Discussion}
\label{sec_Discussion}

We have detected and quantified a many-body system's learning of its drive,
using representation learning,
with greater sensitivity than absorbed power affords.
The scheme relies on a parallel that we identified
between statistical mechanical problems and VAEs.
Uniting statistical mechanical learning with machine learning,
the definition is conceptually satisfying.
The definition also has wide applicability,
not depending on whether 
the system exhibits magnetization or strain or another thermodynamic response.
Furthermore, our representation-learning toolkit signals many-body learning
more sensitively than does
the seemingly best-suited thermodynamic tool.

The rest of this section is organized as follows.
In Sec.~\ref{sec_Decode_Latent}, 
we decode latent space in terms of thermodynamic variables.
In Sec.~\ref{sec_Feasibility}, we argue for the feasibility of applying our toolkit.
In Sec.~\ref{sec_Opportunities}, 
we discuss problems that our toolkit can illuminate.
We also motivate the development of new representation-learning tools.

\subsection{Decoding latent space}
\label{sec_Decode_Latent}

Thermodynamicists parameterize macrostates with
volume, energy, magnetization, etc.
Thermodynamic macrostates parallel latent space,
as illustrated in Fig.~\ref{fig_VAE_SM_Parallel}.
What variables parameterize the VAE's latent space?
Latent space could suggest definitions of new thermodynamic variables,
or hidden relationships amongst known thermodynamic variables.
We begin decoding latent space in terms of thermodynamic quantities,
leaving the full decoding for future research.

We illustrate with part of the spin-glass protocol 
in Sec.~\ref{sec_Classify}:
Train the spin glass with a drive $\{ A, B, C \}$ in each of many trials.
On the end-of-trial configurations, train the VAE.

% Figure
% Source of figures: email chain "Decoding latent space" -- message sent by Weishun on 1/6/20
\begin{figure}[h]
\centering
\begin{subfigure}{\textwidth}
\centering
\includegraphics[width=0.45\textwidth]{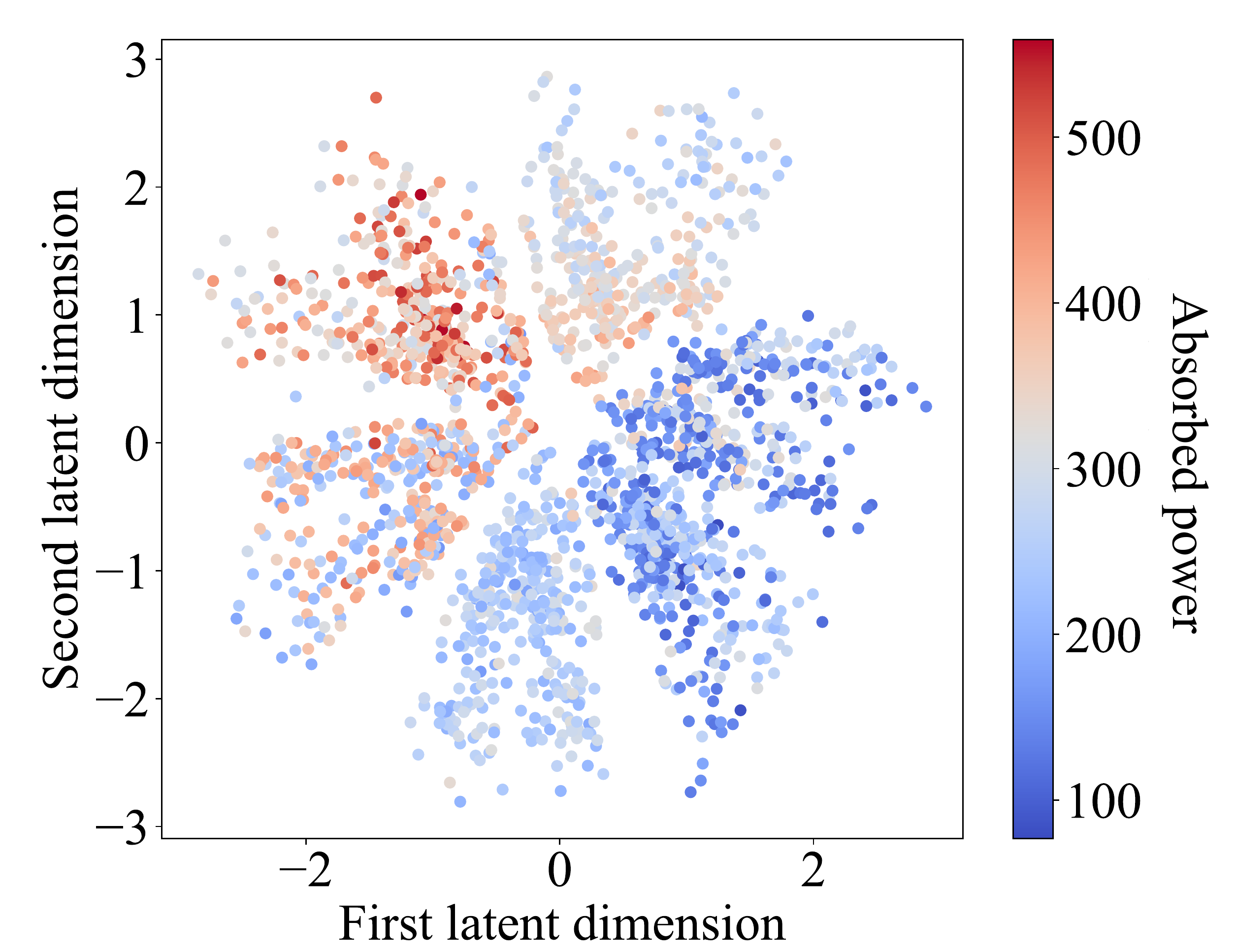}
\caption{\caphead{Correspondence of absorbed power to
the bottom-right-to-upper-left diagonal}}
\label{fig_Latent_Space_Dissipation}
\end{subfigure}
\begin{subfigure}{\textwidth}
\centering
\includegraphics[width=0.45\textwidth]{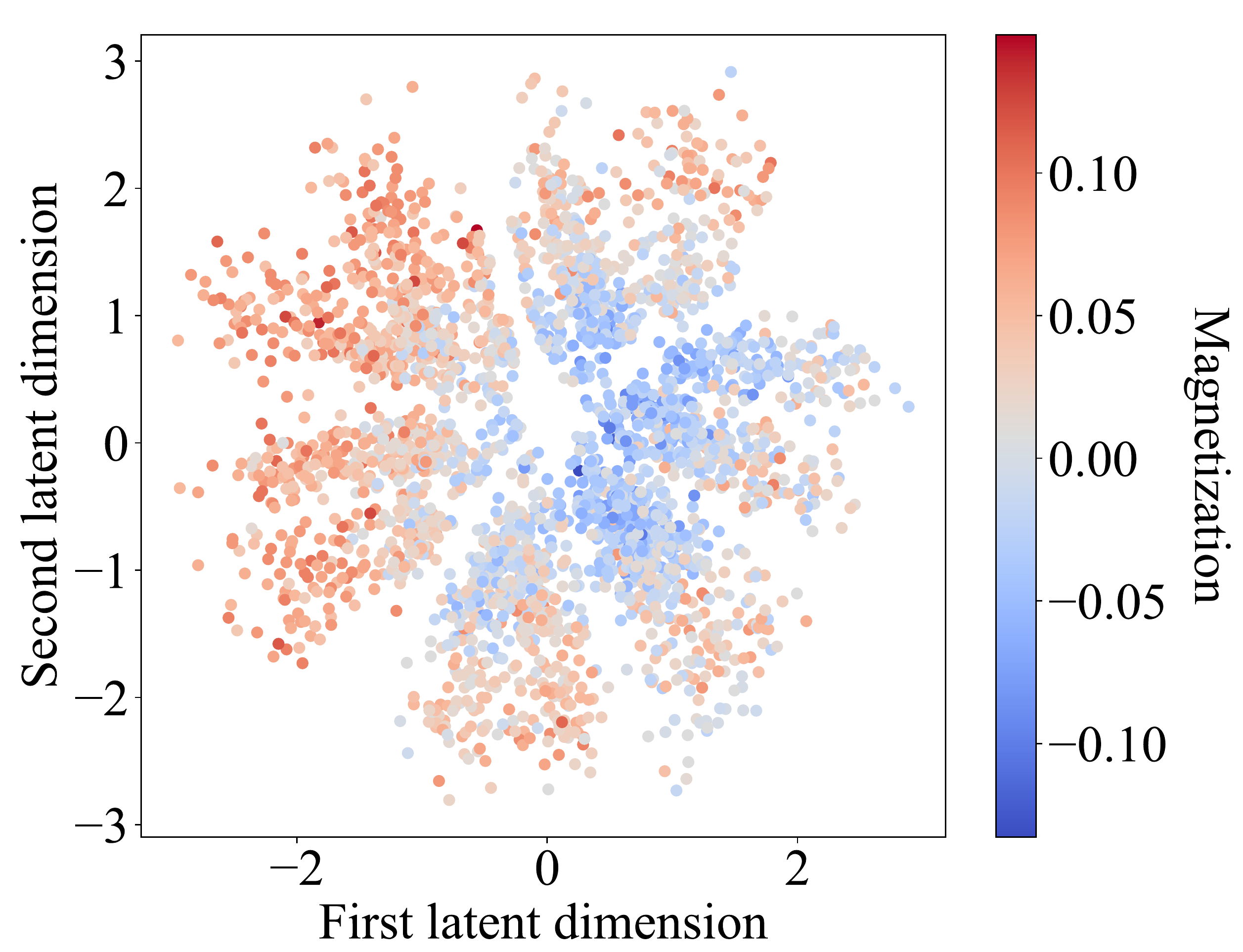}
\caption{\caphead{Correspondence of magnetization to the radial direction.}}
\label{fig_Latent_Space_Magnetization}
\end{subfigure}
\caption{\caphead{Correspondence of latent-space directions
to thermodynamic quantities:}
Each plot depicts the latent space constructed by
a variational autoencoder (VAE).
The VAE trained on the configurations assumed by a spin glass 
during its training with fields $A$, $B$, and $C$.
We have color-coded each plot to highlight
how a thermodynamic property changes 
along some direction.
According to Fig.~\ref{fig_Latent_Space_Dissipation}, 
the absorbed power grows
from the bottom righthand corner to the upper lefthand corner.
According to Fig.~\ref{fig_Latent_Space_Magnetization}, 
the magnetization grows along the radial direction.}
\label{fig_Latent_Space_Visualize}
\end{figure}

Figure~\ref{fig_Latent_Space_Visualize} shows 
two visualizations of the latent space.
Each visualization spotlights a correlation between 
a latent-space direction and a thermodynamic variable.
In Fig.~\ref{fig_Latent_Space_Dissipation},
blue dots represent configurations in which
the spin glass absorbs little work.
Red dots highlight high-absorbed-power configurations.
The dots change from blue to red along the diagonal
from the lower right-hand corner to the upper left-hand.
Hence a point's distance along the diagonal
correlates with the absorbed power.

In Fig.~\ref{fig_Latent_Space_Magnetization},
blue represents low magnetizations,
and red represents high.
Blue dots cluster near the latent space's center,
and red dots occupy the outskirts.
Hence magnetization correlates with a dot's radial coordinate.
Magnetization correlates, to some extent, also with 
distance along the bottom-right-to-upper-left diagonal.
After all, magnetization is related to the absorbed power.

In summary, the diagonal corresponds roughly to the absorbed power,
and the radial direction corresponds roughly to magnetization.
The directions are nonorthogonal, suggesting 
a nonlinear relationship between the thermodynamic variables.
We leave the parameterization of the relationship,
and the possible decoding of other latent-space directions into
new thermodynamic variables, for future work.

%
%
%
% Reference: Meeting notes —> Jacob, Weishun, Jeremy - 8/30/19 —> notes of the same name —> p. 2
\subsection{Feasibility}
\label{sec_Feasibility}

Applying our toolkit might appear impractical,
since microstates must be inputted into the NN.
Measuring a many-body system's microstate may daunt experimentalists.
Yet the use of microstates hinders our proposal little, for three reasons.

% References: 
% (1) Email exchange with Jeremy -- summer 2019
% (Reference for the following line: )2) Mukherji_19_Strength — bottom of p. 1 —> “High-speed imaging (Photron Fastcam SA4, Photron, United Kingdom) of the rafts during rheological measurements allowed simultaneous quantification of single-particle dynamics.”
% (3) Keim_19_Memory contains more examples of measurements of many-body microstates. I underlined at least some of these examples in my paper copy of the paper.
% Also, Jeremy mentioned multiplexing tools used in brain monitoring.
First, microstates can be calculated in numerical simulations,
which inform experiments.
Second, many key properties of many-body microstates 
have been measured experimentally.
High-speed imaging has been used to monitor soap bubbles' positions~\cite{Mukherji_19_Strength}
and colloidal suspensions~\cite{Cheng_11_Imaging}.
Similarly wielded tools, such as high magnification, have advanced
active-matter~\cite{Sanchez_12_Spontaneous}
and gene-expression~\cite{Lonsdale_13_Genotype} studies.

One might worry that the full microstate
cannot be measured accurately or precisely.
Soap bubbles' positions can be measured with finite precision,
and other microscopic properties might be inaccessible.
% Reference for the information below: email chain "Potential criticism: configuration reliance" -- message received from Weishun on 8/31/19
But, third, some bottleneck NNs denoise their inputs~\cite{Vincent_08_Extracting,Goodfellow_16_Deep}:
The NNs learn the distribution from which samples are generated ideally,
not systematic errors.
Denoising by VAEs is less established but is progressing~\cite{Im_15_Denoising}.

% Fourth, suppose that experimentalists could measure 
% only thermodynamic properties.
% These properties could be combined in the nonlinear functions
% suggested by a VAE (Sec.~\ref{sec_Decode_Latent}).
% Such a combination could enhance one's inference of the drive.

Furthermore, one might wonder whether 
our study requires deep learning.
Could simpler algorithms detect and measure many-body learning
as sensitively?
Section~\ref{sec_Justify_ML} addresses this question.
We compare the VAE with simpler competitors
that perform unsupervised learning:
a single-layer linear NN,
related to principal-component analysis (PCA)~\cite{Bourland_88_Auto},
% <-- Alternative reference: https://en.wikipedia.org/wiki/Autoencoder#Dimensionality_Reduction
and a clustering algorithm.
The VAE outperforms both competitors.

\subsection{Opportunities}
\label{sec_Opportunities}

Several opportunities emerge from 
this combination of statistical mechanical learning
and bottleneck NNs.
% References for text below: 
% (i) Email chain “papers” — messages exchanged on 7/31/19
% (ii) Miller_19_Raft (Physics Today article) -- my electronic copy -- green highlighting
First, our toolkit may resolve open problems in the field
of statistical mechanical learning.
One example concerns the soap-bubble raft in~\cite{Mukherji_19_Strength}.
Experimentalists trained a raft of soap bubbles with
an amplitude-$\gamma_{\rm t}$ strain.
The soap bubbles' positions were tracked,
and variances in positions were calculated.
No such measures distinguished trained rafts
from untrained rafts;
only stressing the raft and reading out the strain could~\cite{Mukherji_19_Strength,Miller_19_Raft}.
Bottleneck NNs may reveal what microscopic properties distinguish
trained from untrained rafts.

% Reference: email chain "Applications to open problems" --> message sent by Jeremy on 9/2/19
Similarly, representation learning might facilitate 
the detection of active matter.
Self-organization is detected now through 
simple, large-scale, easily visible signals~\cite{Heylighen_02_Science}.
Bottleneck NNs could identify patterns
invisible in thermodynamic measures.

Second, our framework calls for extensions to quantum systems.
Far-from-equilibrium many-body systems have been realized with 
many quantum platforms, including ultracold atoms~\cite{Langen_15_Ultracold},
% ~\cite{Kinoshita_06_Quantum,Schreiber_15_Observation,Erne_18_Universal,Eigen_18_Universal,Prufer_18_Observation}, 
trapped ions~\cite{Friis_18_Observation,Smith_16_Many}, 
and nitrogen vacancy centers~\cite{Kucsko_18_Critical}.
% A reference that contains the references above:
% the introduction of https://arxiv.org/pdf/1711.03528.pdf
Applications to memories have been proposed~\cite{Abanin_19_Colloquium,Turner_18_Weak}.
% of many-body localization~\cite{Abanin_19_Colloquium} and many-body scars~\cite{Heller_84_Bound} to memories have been proposed.
Yet quantum memories that remember 
\emph{particular coherent states} have been focused on.
The learning \emph{of strong drives} by quantum many-body systems
calls for exploration,
as the learning of strong drives by polymers, soap bubbles, etc.
has proved so productive in classical statistical mechanics.
Our framework can guide this exploration.

% Conversation about the extent to which representation learning is understood: Meeting notes --> Weishun - 11/27/19
Third, we identified a parallel between representation learning and
statistical mechanics. % in the presence of a drive.
The parallel enabled us to use representation learning
to gain insight into statistical mechanics.
Recent developments in information-theoretic far-from-equilibrium statistical mechanics 
(e.g.,~\cite{Still_12_Thermodynamics,Parrondo_15_Thermodynamics,Crutchfield_17_Origins,Kolchinsky_17_Dependence})
might, in turn, shed new light on representation learning.

Fourth, the mutual information between configuration and drive
can be calculated as a function of time.
Let $p(x, y)$ denote the joint probability that
the configuration $X = x$ and the drive $Y = y$.
Let $p(x)  :=  \sum_y  p(x, y)$ and 
$p(y)  :=  \sum_x  p(x, y)$ denote the marginal distributions.
The mutual information quantifies the information 
about the drive in the configuration and vice versa:
$I(X; Y)  
=  \sum_{x, y}  p(x, y)
\log  \left(  \frac{ p(x, y) }{ p(x)  p(y) }  \right)$.
The mutual information is expected to grow
as the many-body system learns.
% Reference for the following line: email chain “MI unclear” — message received from Sarah on 9/19/19
Estimating $I(X; Y)$ proved difficult due to undersampling;
hence our use of the MAP-estimate score (Sec.~\ref{sec_Classify}),
a cousin of the mutual information (App.~\ref{app_MAP}).
This work motivates the development of 
techniques for estimating $I(X; Y)$ from little data.

% References:
% (1) Meeting notes --> Weishun - 9/3/19 - paper draft --> 9/3/19 --> final page
% (2) Meeting notes --> Group skype - 9/19/19 --> p. 11
Such techniques could be complemented by
a sampling strategy based on our VAE, fifth.
The VAE populates latent space, analogous to the space of macrostates, 
as in Fig.~\ref{fig_Latent_Space}.
Consider choosing an unpopulated point,
analogous to an unfamiliar macrostate,
and having the VAE decompress the point.
The VAE will construct a configuration.
Such configurations could improve $p(x, y)$ estimates
and so $I(X; Y)$ estimates.
Rough initial studies suggest that the constructed configurations
resemble the true samples that they should mimic.
% Further verification is needed.

Sixth, given $I(X; Y)$, one can benchmark the many-body system
against the \emph{information curve}~\cite{Tishby_00_Information}.
The information curve quantifies the tradeoff in representation learning:
Recall the general bottleneck NN described in the introduction.
The NN compresses $X$ into $Z$,
then decompresses $Z$ into $Y$ [Fig.~\ref{fig_VAE_SM_Parallel}(a)].
The more the NN compresses $X$,
the less space $Z$ requires.
Hence shrinking $I(X; Y)$ is desirable.
Yet $Z$ must carry enough information about $X$
to generate an accurate $Y$ prediction $\hat{Y}$.
Hence $I(Z; \hat{Y})$ should be large.
One can tune the mutual informations' relative importance,
using a parameter $\beta$.
One chooses a $\beta \in [0, 1]$, then maximizes the objective function
$I(Z; \hat{Y} )  -  \beta I (Z; X)$.
This strategy is called the \emph{information bottleneck}~\cite{Tishby_99_Information}.
Consider varying $\beta$.
At each $\beta$ value, the optimal $I(X; Y)$ can be 
plotted against the optimal $I(X; Z)$.
The resulting \emph{information curve} represents an ideal:
Physical systems can reach the points inside the curve,
not points outside.
Consider plotting a many-body system's 
$\LParen  I(X; Z),  I(X; Y)  \RParen$ as a point.
The point's distance from the information curve
will quantify how close the many-body system approaches to the ideal.

Seventh, we partially decoded the VAE's latent space 
in terms of thermodynamic variables (Sec.~\ref{sec_Decode_Latent}).
Further analysis merits exploration.
Convention biases thermodynamicists toward measuring
volume, magnetization, heat, work, etc.
The VAE might identify new macroscopic variables
better-suited to far-from-equilibrium statistical mechanics,
or hidden nonlinear relationships amongst thermodynamic variables.
A bottleneck NN could uncover new theoretical physics,
as discussed in, e.g.,~\cite{Carleo_19_Machine,Wu_19_Toward,Iten_20_Discovering}.

%auto-ignore
\chapter{Generative modeling by disordered quantum spins}
\label{born}
\section{Introduction}
%%%%%%%%%%%%%%%%%%%%%%%%%%%%%%%%%%
%%%%%%%%%%%%%%%%%%%%%%%%%%%%%%%%%%

The computational power of quantum processors is the subject of considerable amount of recent research, in particular with regard to scaling and a potential quantum advantage \cite{IBM_eagle,arute2019quantum,IBM2019,Zhong2020,Gong2021,wu2021strong,ebadi2021quantum}. While the advent of a fully error corrected quantum computer requires yet another milestone, the immediate application of noisy quantum hardware with a clear advantage over classical computation becomes even more crucial. In this regard, the interface of quantum computing and machine learning has been increasingly brought into focus. 
For example, the rise of hybrid variational algorithms, such as variational quantum eigensolvers (VQE) \cite{VQE} and the quantum approximate optimization algorithm (QAOA) \cite{QAOA}, which use a parametrized quantum circuit as variational ans"atze and optimize the parameters classically, has been considered particularly promising as they aim to obtain heuristic and approximate solutions.

However, the exponential dimension of the Hilbert space and the random characteristics of parametrized quantum circuits makes their training very challenging  due to the existence of barren plateaus\cite{Jarrod_QNN}. More recently, yet another approach to quantum machine learning has emerged, which is known as  brain-inspired\cite{QNC_Grollier,QMemresistor,Gonzalez-Raya1,Gonzalez_Raya2,Torrontegui}. One interesting category consists of quantum reservoir computing (QRC) where a fixed reservoir geometry scrutinizing the unitary dynamics of an interacting quantum system allows versatile machine learning tasks \cite{QNC_Fuji1,QNC_Fuji2,QRNN_Ryd,xia2022reservoir}. While QRCs have shown many advantages, they are mainly appropriate for discriminative tasks such as classification or regression. 

The goal of generative models, however, is to learn an unknown data probability distribution $p_{\text{data}}$ in order to subsequently sample from $p_{\text{data}}$ and thus generate new and previously unseen data. Such tasks can, for example, be performed by the recently introduced Born machines\cite{Born_MPS,Born_TTN}. Born machines for many-body problems have early on shown to be successful in conjunction with tensor network state ans\"atze.  The elements of these matrix product states or tree tensor networks and their bond dimensions can be optimized during training to effectively approximate $p_{\text{data}}$\cite{Born_MPS,Born_TTN,BEBM_Abigail}. While Born machines have also been tested with parameterized quantum circuits \cite{Born_PQC}, we address here the question of whether there are other quantum many-body states that can be used as anasatz for Born machine to any advantage.  

Quantum many-body systems display many phases in the presence of disorder, in particular, the break-down of thermalization and thus localization of the wavefunction in the so-called \emph{many-body localized} (MBL) phase. Here, emergent integrals of motion can be utilized as quantum memories\cite{Huse2013}. The failure of such systems to anneal\cite{altshuler2010anderson} has inspired their use in QRC\cite{xia2022reservoir} for learning tasks, with particular enhancement close to the phase transition\cite{Martinez2021}.

Here, we extend quantum inspired generative models into the MBL phase, and introduce a hidden architecture to increase the representation power of our generative model. While recent work has also studied Born machines in the MBL phase\cite{tangpanitanon2020expressibility}, using a similar quenched approach, our work differs in the hidden architecture and the characterization of learnability and expressibility. 
%
% %%%%%%%%%%%%%%%%%%%%%%%%%%%%%%%%%%%%%
%\begin{widetext}
    %\begin{minipage}{\linewidth}
    \begin{figure*}[ht]
     \centering
            \includegraphics[scale=0.17]{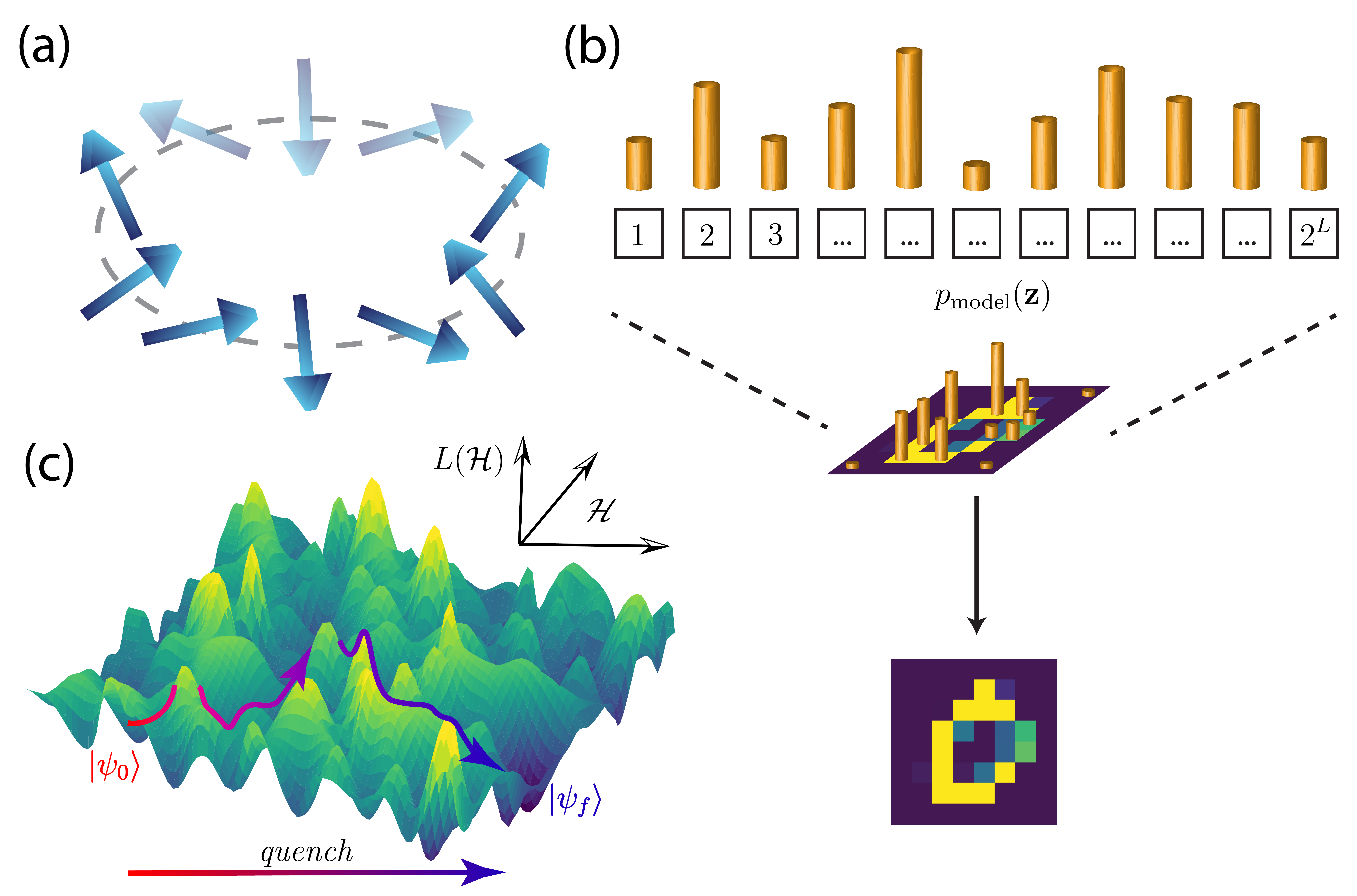}
            \caption{Illustration of the MBL hidden Born machine. (a) XXZ spin chain in 1D with periodic boundary condition. The faded color spins are the hidden units $h_i$, and the solid color spins are the visible units $v_i$. (b) The probability distribution of finding individual states in the z-basis represents the model distribution for the generative model, which are coded as normalized pixel values of an image. (c) An illustration of the loss landscape defined by our hidden MBL Born machine. The training is done by optimizing disorder configurations in the Hamiltonian
            during each quantum quench, which is then used to evolve the initial state $\ket{\psi_0}$ over successive layers of quenches toward a final state $\ket{\psi_f}$ which gives rise to the desired model distribution. }
            \label{fig:model schematics}
        \end{figure*}    
    %\end{minipage}
%\end{widetext}
% %%%%%%%%%%%%%%%%%%%%%%%%%%%%%%%%%%%%%
%
In this article, we first introduce the hidden Born machine in section \ref{Hidden-BM}, and prove that including hidden units into Born machine leads to expressive power advantage with respect to the basic architecture. Then, in section \ref{Hidden-BM-expressible}, we comment on the expressiblty of the MBL-Born machine, leveraging the fact that our model can be mapped into existing models with proven advantage in expressibility over classical models \cite{Gao_2017}. We describe our training algorithm in section \ref{sec:Training} and introduce the randomly driven Born machine in section \ref{sec:RDBM} and compare its performance with the hidden Born machine by learning patterns of MNIST hand written digits. Next, in section \ref{sec:different_phases} we investigate the learning power of the hidden Born machine both in the thermal phase and the MBL phase, and numerically show that the thermal phase fails to learn data obtained from quantum systems either in MBL or in thermal phase. Tracking von Neumann entanglement entropy and Hamming distance during training suggests that localization is crucial to learning. In section \ref{sec:parity} we further show that while the hidden Born machine trained in the MBL phase is able to capture the underlying structure of the parity data, a hidden Born machine trained in the Anderson localized phase fails to do so, shedding light on the fact that the interplay between interaction and disorder plays an important role in learning. Finally, we conclude and discuss possible direction for future works. 

%%%%%%%%%%%%%%%%%%%%%%%%%%%%%%%%%%
%%%%%%%%%%%%%%%%%%%%%%%%%%%%%%%%%%
\section{Hidden Born machines}
\label{Hidden-BM}
%%%%%%%%%%%%%%%%%%%%%%%%%%%%%%%%%%
%%%%%%%%%%%%%%%%%%%%%%%%%%%%%%%%%%
Born machine \cite{liu2018differentiable,cheng2018information,benedetti2019generative,benedetti2019parameterized,coyle2020born} is a generative model that parameterized the probability distribution of observing a given configuration $\z$ of the system according to the probabilistic interpretation of its associated quantum wavefunction $\psi(\z)$,
\begin{equation}
\label{eqn:born}
    p_{\text{Born}}(\z) = \frac{|\psi(\z)|^2}{\mathcal{N}},
\end{equation}
where $\mathcal{N}=\sum_{\z}|\psi(\z)|^2$ is the overall normalization of the wavefunction. Note that $\mathcal{N}$ is only required in tensor network ans\"atze but not in physical systems. Training of the Born machine is done by minimizing the discrepancy between the model distribution $p_{\text{Born}}(\z)$ and the data distribution $q_{\text{data}}(\z)$.

% \subsection{Hidden Born machine}

Born machine \cite{liu2018differentiable,cheng2018information,benedetti2019generative,benedetti2019parameterized,coyle2020born} is a generative model that parameterized the probability distribution of observing a given configuration $\z$ of the system according to the probabilistic interpretation of its associated quantum wavefunction $\psi(\z)$,
\begin{equation}
\label{eqn:born}
    p_{\text{Born}}(\z) = \frac{|\psi(\z)|^2}{\mathcal{N}},
\end{equation}
where $\mathcal{N}=\sum_{\z}|\psi(\z)|^2$ is the overall normalization of the wavefunction. Note that $\mathcal{N}$ is only required in tensor network ans\"atze but not in physical systems. Training of the Born machine is done by minimizing the discrepancy between the model distribution $p_{\text{Born}}(\z)$ and the data distribution $q_{\text{data}}(\z)$.

% \subsection{Hidden Born machine}

In the language of Boltzmann machine\cite{smolensky1986information,hinton2006reducing,ackley1985learning}, the units that are used for generating configurations are called `visible'. Meanwhile, adding `hidden' units prove to be a powerful architecture for the Boltzmann machine as it provides a way to decouple the complex interaction among the visible units at the expense of introducing interaction between the hidden and the visible units\cite{mehta2019high,gao2017efficient}.
In Eqn.\eqref{eqn:born}, all units of the system are used to generate configurations that are compared against data and therefore all units are visible. In a similar spirit, we introduce hidden units to the Born machine by defining the probability distribution of observing a given visible spin configuration $\z$ to be its expectation value in $z$-basis after tracing out the hidden units, 
\begin{equation}
\label{eqn:hidden}
    p_{\text{hidden}}(\z) = \Tr\rho_{\text{vis}} \Pi_Z,    
\end{equation}
where
\begin{equation}
\label{eqn:rho_vis}
     \rho_{\text{vis}} = \Tr_{h} \ket{\psi}\bra{\psi},
\end{equation}
is the reduced density matrix for the visible units, and $\Pi_Z = \ket{\z}\bra{\z}$ is the projection operator onto the $z-$basis of the visible part of the system (see Fig.\ref{fig:model schematics}(a) for an illustration of our model). Note that normalization is implicit in Eqn.\eqref{eqn:rho_vis} for $\rho_{\text{vis}}$ to be a density matrix. 

We argue that the hidden Born machine in Eqn.\eqref{eqn:hidden} offers expressive power advantage over the basic Born machine in Eqn.\eqref{eqn:born}. We demonstrate this by first proving a proposition about the hidden architecture, showing that adding hidden units generalizes the basic Born machine and therefore the achievable generalization error is at least as good as the original one. While our proof is independent of the particular choice of Hamiltonian, we support our claim with numerical evidence of a hidden Born machine realized with a XXZ spin chain in section \ref{sec:Training}.

%%%%%%%%%%%%%%%%%%%%%%%%%%%%%%%%%%
%\subsection{The hidden model}
\subsection{The hidden advantage}
%%%%%%%%%%%%%%%%%%%%%%%%%%%%%%%%%%
In this section, we prove that the hidden Born machine Eqn.\eqref{eqn:hidden} generalizes the basic Born machine (BM) defined by Eqn.\eqref{eqn:born}, in the sense that the class of probability distributions expressible by the basic Born machine is a subset of that of the hidden Born machine. In the following, we assume only that the visible and hidden part couple through an interaction term in the Hamiltonian.

Let's consider a basic Born machine consisting of only visible units $\vis=\{v_i\}$, with Hamiltonian $\hat{\mathcal{H}}_v$. Now consider adding hidden units $\h=\{h_j\}$ to the system with Hamiltonian $\hat{\mathcal{H}}_h$ and the hidden units couple with the visible ones through an interaction Hamiltonian $\hat{\mathcal{H}}_{\text{int}}$. The full Hamiltonian can be written as 
\begin{equation}
\label{eqn:fullH}
    \hat{\mathcal{H}}_{vh}[\vis,\h] = \hat{\mathcal{H}}_v[\vis] + \hat{\mathcal{H}}_h[\h] + \hat{\mathcal{H}}_{\text{int}}[\vis,\h],
\end{equation}
where all the $\hat{\mathcal{H}}$'s in general can be time-dependent. Let's assume that the basic Born machine model is described by just the visible part of Hamiltonian in Eqn.\eqref{eqn:fullH}, $\hat{\mathcal{H}}_{\text{BM}} = \hat{\mathcal{H}}_v(\bTheta^{\text{BM}})$, and the hidden Born machine is described by the full Hamiltonian, $\hat{\mathcal{H}}_{\text{hBM}} = \hat{\mathcal{H}}_{vh}(\bTheta^{\text{hBM}})$, where $\bTheta^{\text{BM}}$ and $\bTheta^{\text{hBM}}$ denotes the parameters in the Hamiltonian to be optimized during learning.

\begin{proposition}
\label{prop_1}
For the same set of visible spins $\vis$, let $p_{\text{BM}}(\z)$ denote the model distribution realized by the basic Born machine, and $p_{\text{hBM}}(\z)$ denote the model distribution realized by the hidden Born machine, then $\{p_{\text{BM}}(\z)\}\subseteq \{p_{\text{hBM}}(\z)\}$.

\end{proposition}

% \cmt{add clarification on the assumption, mention that $H_v + H_h$ scales like $L$ and $H_{int}$ scales like $\partial L$}.

\begin{proof}
Let's denote the initial state for the BM as $\ket{\psi_{0}^v}\in \mathcal{H}_{v}$. Let $\hat{\mathcal{U}}_{v} = \hat{\mathcal{T}}\exp \left(-i\int_0^T dt \hat{\mathcal{H}}_{v} \right)$. Then, the final state of BM is $\ket{\psi_f^v} = \hat{\mathcal{U}}_{v}\ket{\psi_0^v}$.
Choose an initial product state for the hBM, $\ket{\psi^{vh}_0} = \ket{\psi_0^v} \otimes \ket{\psi_0^h} \in \mathcal{H}_{v} \otimes \mathcal{H}_{h}$ for some $\ket{\psi_0^h} \in \mathcal{H}_h$. Choose $\bTheta^{\text{hBM}}$ to be such that $\hat{\mathcal{H}}^{\text{hBM}}_{v} = \hat{\mathcal{H}}^{\text{BM}}_{v}$, and $\lvert\lvert \hat{\mathcal{H}}^{\text{hBM}}_{v} \rvert\rvert \gg \lvert\lvert \hat{\mathcal{H}}^{\text{hBM}}_{\text{int}} \rvert\rvert$.

Then, we have 
\begin{align}
    \ket{\psi_f^{vh}} \approx \hat{\mathcal{U}}_{v} \ket{\psi_0^v} \otimes \hat{\mathcal{U}}_h \ket{\psi_0^h} = \ket{\psi_f^v} \otimes \ket{\psi_f^h}
\end{align}
where we have defined $\ket{\psi_f^h} \equiv \hat{\mathcal{U}}_h \ket{\psi_0^h}$. With this choice, now $\rho_{\text{vis}} = \Tr_h \ket{\psi_f^{vh}} = \ket{\psi_f^v}\bra{\psi_f^v}$, and $p_{\text{hBM}}(\z) = \Tr \rho_{\text{vis}}\Pi_Z = |\psi_f^v(\z)|^2 = p_{\text{BM}}(\z)$, where $p_{\text{BM}}$ is automatically normalized ($\mathcal{N}=1$) for physical systems as in our case. Therefore, the class of probability distributions described by BM is contained in hBM.

\end{proof}

\begin{corollary}
\label{cor2}
Comparing the minimum achievable loss $\mathcal{L}^*$ of the hidden Born machine and the basic Born machine on any given loss function, we have $\mathcal{L}^*_{\text{hBM}}\leq \mathcal{L}^*_{\text{BM}}$.
\end{corollary}

Prop.\ref{prop_1} suggests that the hidden Born machine is able to represent a larger class of probability distributions and thus generalizes the basic Born machine.  Cor.\ref{cor2} indicates that the achievable training loss for the hidden Born machine is less than or equal to that of the regular Born machine, a property that we will confirm numerically in section \ref{sec:Training}.

%%%%%%%%%%%%%%%%%%%%%%%%%%%%%%%%%%%%%%%%%%%%%%%%%%%%%%%%%%%%%%%%%%%%
%%%%%%%%%%%%%%%%%%%%%%%%%%%%%%%%%%%%%%%%%%%%%%%%%%%%%%%%%%%%%%%%%%%%
\section{Expressibilty of MBL-Born machine}
\label{Hidden-BM-expressible}
%%%%%%%%%%%%%%%%%%%%%%%%%%%%%%%%%%%%%%%%%%%%%%%%%%%%%%%%%%%%%%%%%%%%
%%%%%%%%%%%%%%%%%%%%%%%%%%%%%%%%%%%%%%%%%%%%%%%%%%%%%%%%%%%%%%%%%%%%
%\section{Many Body Localized ansatz}

Previously, different ans\"atze for $\ket{\psi}$ has been introduced for the Born machine, notably tensor networks states and states prepared by both digital quantum circuits and analog quantum many-body systems\cite{Born_MPS,Born_TTN,Born_PQC,tangpanitanon2020expressibility}. In this paper, we will be adopting the latter approach, and focus on a specific type of quantum many-body systems that admits a many-body localization (MBL) phase. 
In the following, we first discuss the simple model that give rise to the MBL phase. Then, leveraging on the fact that the XXZ model under appropriate choice of quench parameters can be mapped into a 2D Ising model that has quantum computational advantage\cite{Gao_2017}, we show that our MBL Born machine possesses more expressive power than classical models.

\subsection{Many-body localized ans\"atze}

It is generally believed that, thermalization in quantum system wipes out the microscopic information associated with the initial state. Even in the case of closed quantum system, the information of initial state quickly spreads throughout the entire system, implying that no local measurements can retrieve those information\cite{Deutsch,srednicki1994chaos}. However, it's known that strong disorder leads to localization, preventing the system to thermalize. Furthermore, the localization manifests itself in the form of memory associated with the lack of transport. While the localization in the presence of strong disorder was first introduced in non-interacting systems by Anderson\cite{Anderson}, more recently, it was shown that the localization and break down of thermalization can also happen in strongly interating systems, leading to new dynamical of phase of matter known as many-body localization (MBL)\cite{MBL1,MBL2}. 

In the MBL phase, eigenstates of the system do not satisfy Eigenstate Thermalization Hypothesis (ETH) and the wavefunctions become localized in the Hilbert space. Such ergodicity breaking renders the system to retain memory of its initial state, and offers advantage in controlling and preparing desired quantum many-body states and has been also realized experimentally\cite{MBL_exp}. The XXZ model of quantum spin chain is well-known to develop a MBL phase when the disorder strength exceeds the MBL mobility edge\cite{MBL_edge}. 

We perform numerical simulation with the XXZ-Hamiltonian
$\hat{\mathcal{H}}_{\text{XXZ}}$ defined as:
    \begin{align}
    \label{eqn:XXZ}
            &\mathcal{\hat{H}}_{\text{XXZ}} = \sum_i^{L-1} J_{xy}(\hat{S}^{x}_{i}\hat{S}^x_{i+1} + \hat{S}^y_{i}\hat{S}^y_{i+1}) + \sum_i ^{L-1} J_{zz} \hat{S}^z_i \hat{S}^z_{i+1},
    \end{align}
where $\hat{S}_i^{\alpha}(\alpha\in \{x,y,z\})$ are Pauli spin 1/2 operators acting on spins $i\in 1,..,L$, and $L=L_v+L_h$ consists of $L_v$ visible units and $L_h$ hidden units. $J_{xy}$ and $J_{zz}$ are couplings in the $xy$ plane and $z$ direction, respectively. Then, we consider a series of $M$ quenches $\hat{\mathcal{H}}_{\text{quench}}(\bTheta_m)$ in the $z$-direction:
    \begin{align}
    \label{eqn:H_total}
        \mathcal{\hat{H}}_{\text{total}}&=\mathcal{\hat{H}}_{\text{XXZ}}+\mathcal{\hat{H}}_{\text{quench}}(\bTheta_m),
    \end{align}
where $\mathcal{\hat{H}}_{\text{quench}}(\bTheta_m) = \sum_i h_i^m \hat{S}^z_i$ and we have denoted the tunable parameters in the system collectively as $\bTheta_m = \{h_i^m\}$. During each quench $m$, $h_i^m$ are drawn i.i.d. from the uniform distribution over the interval $[-h_{d},h_{d}]$, where $h_d$ is the disorder strength. Notice that when $J_{zz}=0$, the model reduces to non-interacting XY model with random transverse field exhibiting single particle localization. Once we turn on the $J_{zz}$ interaction, the spins couple via Heisenberg interaction and MBL phase emerges when $h_c\sim 3.5$ (for $J_{zz}=J_{xy}=1$) \cite{MBL_transition1,MBL_edge,MBL_transition2}. See more details in Section \ref{app:MBL_check}.

In section \ref{sec:Training}, we will explain the training algorithm under the time evolution implied by series of quenches in $\hat{\mathcal{H}}_{total}$, and learning through optimizing the values of disordered field $h_i^m$ at each site. 

%%%%%%%%%%%%%%%%%%%%%%%%%%%%%%%%%%%%%%%%%%%%%%%%%%%%%%%%%%%%%%%%%%%%
\subsection{Mapping XXZ chain into Ising model}
%%%%%%%%%%%%%%%%%%%%%%%%%%%%%%%%%%%%%%%%%%%%%%%%%%%%%%%%%%%%%%%%%%%%

There has been extensive studies on the expressive power of quantum models. In particular, quantum computational advantage for sampling problem has been proved (based on standard computational complexity assumptions) in a translation-invariant Ising model \cite{Gao_2017}. While our numeric are mostly restricted to the 1-dimension case as it can be studied by exact diagonalization, the XXZ model can be realized in any dimensions.

In this section, we show that the XXZ model in 2-dimension, with proper choice of disorder parameters, can be reduced to an Ising model that contains brickwork state that is classically intractable Ref.~\cite{Gao_2017}. This classically-hard instance implies that our model cannot be simulated in polynomial time by a classical computer and therefore offers an advantage in its expressive power. 

\begin{proposition}

The XXZ model in 2D subject to quench in z-direction can be reduced to an Ising model.
\end{proposition}

\begin{proof}

In 2D, 

\begin{equation}
    \mathcal{\hat{H}}_{\text{XXZ}} = \sum_{\langle i,j \rangle} J_{xy}(\hat{S}^{x}_{i}\hat{S}^x_j + \hat{S}^y_{i}\hat{S}^y_j) + \sum_{\langle i,j \rangle} J_{zz} \hat{S}^z_i \hat{S}^z_{j},
\end{equation}
where the interactions are between nearest neighbours. During a quench $\hat{\mathcal{H}}_{\text{quench}}$ of duration $t_m$, we can divide the disorder into a time-dependent and a time-independent part, 
\begin{equation}
    h_i^m(t) = J_i^m(t) + B_i^m.
\end{equation}

In the case of bipartite lattice, we can partition the vertices into two partitions, and denote the sites in one partition as $\mathcal{K} = \{k_1, k_2,..,k_{L/2}\}$ and another partition as $\mathcal{N} = \{n_1,n_2,...,n_{N/2}\}$ (assuming $N$ even). For example, in the case of a square lattice, $\mathcal{K}$ and $\mathcal{N}$ correspond to the black and white sites of the checkerboard coloring. For the set of $\mathcal{K}$ spins, we turn on a $\pi-$pulse in the middle of the quench ($k\in \mathcal{K}$), such that: 
\begin{equation}
J_{k}^m(t)=\begin{cases}
          0 \quad &\text{if} \, 0 \leq t<\frac{t_m}{2} \\
          \frac{\pi}{2\Delta t} \quad &\text{if} \, \frac{t_m}{2} \leq t<\frac{t_m}{2} + \Delta t \\
          0 \quad &\text{if} \, \frac{t_m}{2} + \Delta t \leq t<t_m \\
     \end{cases}
\end{equation}
where $\Delta t \ll t_m$ is a short duration of time. With this choice of disorder, the time evolution operator reduces to 
\begin{equation}
    \hat{\mathcal{U}} = e^{-i\hat{\mathcal{H}}_{\text{XXZ}}\Delta t}(\Pi_{k} i \hat{Z}_{k})e^{-i\sum_k B_k^m \hat{Z}_k}.
\end{equation}
Now the Pauli $Z$'s effectively flip the signs of the $\hat{S}^x\hat{S}^x$ and $\hat{S}^y\hat{S}^y$ terms in the XXZ-Hamiltonian, and upon integrating over the duration of a quench cancels out with the corresponding terms in first half of the quench. Therefore, after time evolution of a quench, the effective Hamiltonian is left with only Ising interactions, 
\begin{equation}
    \bar{\mathcal{H}}_{\text{eff}}^m = \sum_{\langle i,j \rangle} J_{zz} \hat{S}^z_i \hat{S}^z_j + \sum_i B_i^m \hat{S}^z_{i}
\end{equation}

\end{proof}

% \begin{corollary}
% MBL Born machine has expressibility quantum advantage.
% \end{corollary}

As a result, we have recovered the case in Ref.~\cite{Gao_2017} and showed that our model's expressibility has quantum advantage \footnote{Note that for the proof in Ref.~\cite{Gao_2017} to work, one also need to initialize the system in all $\ket{+}$ states and subsequently perform all measurements in the $x-$basis}.

\section{Training of hidden MBL Born machine}
\label{sec:Training}
%%%%%%%%%%%%%%%%%%%%%%%%%%%%%%%%%%%%%

%%%%%%%%%%%%%%%%%%%%%%%%%%%%%%%%%%%%%
\subsection{Learning algorithm}
\label{sec:algorithm}
%%%%%%%%%%%%%%%%%%%%%%%%%%%%%%%%%%%%%
%
The basic idea behind the training of Hidden Born machine is the following:
given target distribution $q_{\text{data}}$, and a loss function $\mathcal{L}(p_{\text{model}},q_{\text{data}})$ that measures the discrepancy between model distribution and data distribution, training of the MBL Born machine is achieved through time-evolving the system with the Hamiltonian in Eqn.\eqref{eqn:XXZ}, then optimizing $\bTheta_m$ over $N$ different disorder realizations for each quench $m$. After obtaining the final state at the $M$-th quench, we evaluate the model distribution of the MBL hidden Born machine from Eqn.\eqref{eqn:hidden} and use it as our generative model (see Fig.\ref{fig:model schematics}(c) for an illustration of the learning process). We use Maximum Mean Discrepancy (MMD) loss as our loss function:
\begin{align}
\label{eqn:MMD}
    \mathcal{L}_{MMD} &= \norm{\sum_x p(x)\phi(x)  - \sum_x q(x)\phi(x)}^2,
\end{align}
where $\phi(x)$ are kernel functions that one can choose (see more details in Section \ref{app:MMD}).
The learning algorithm is summarized by the pseudo-code in Alg.\ref{alg:born} and illustrated in \ref{fig:alg_schem}.
Given the reduced density matrix $\rho_M$ of the $L_v$ visible spins at the final layer of the quench $m=M$, we compute the model distribution from Eqn.\eqref{eqn:hidden}, which gives the probability $p_{\text{model}}(\z)$ of finding each of the $2^{L_v}$ basis states $\z$ in the visible part of the system. For learning image data, we then interpret the probabilities as pixel values (normalized to be within 0 and 1), and reshape it into an image of size $2^{L_v/2} \times 2^{L_v/2}$ (see Fig.\ref{fig:model schematics}(b)). For quantum data, we interpret these probabilities as measurement outcomes obtained from the quantum state. For more details, see Section \ref{app:MNIST}.
\begin{algorithm}
\caption{Training of MBL hidden Born machine}\label{alg:born}
\begin{algorithmic}[1]
\State Initialize the system in some initial state $\ket{\psi(\bm{\Theta}_{\text{m}=0})} \equiv \ket{\psi_0}$ and choose $\bm{\Theta}_0 = \mathbf{0}$;
\While{$m < M$}
\While{$n<N$}
\State Sample $\bm{\Theta}_{m}^{(n)}$ uniformly from the interval $[-h_{d},h_{d}]$;
\State Time-evolve the state $\ket{\psi_{m+1}^{(n)}} = \hat{\mathcal{U}}(\bm{\Theta}_m^{(n)})\ket{\psi_m}$ with \\  \qquad \qquad $\hat{\mathcal{U}} = \hat{\mathcal{T}}\exp \left(-i\int_0^T dt \hat{\mathcal{H}}_{\text{total}} \right)$;
\State Trace out the hidden units $\rho_{m+1}^{(n)} = \Tr_{h} \ket{\psi_{m+1}^{(n)}}\bra{\psi_{m+1}^{(n)}}$;
\State Compute $\mathcal{L}(\bm{\Theta}_m^{(n)})$ from $ p_{\text{hidden}}^{(n)}(\z) = \Tr\rho_{m+1}^{(n)}\Pi_Z$;
\State $n \gets n+1$
\EndWhile
\State $\bm{\Theta}_{m} = \argmax_{\bm{\Theta}_{m}^{(n)}} \mathcal{L}(\bm{\Theta}_m^{(n)})$; 
\State $\ket{\psi_{m+1}} = \hat{\mathcal{U}}(\bm{\Theta}_m)\ket{\psi_m}$
\State $m \gets m + 1$
\EndWhile
\State Denote the training outcome as $p_{\text{model}}(\z) = \Tr\rho_{M}\Pi_Z$;
\end{algorithmic}
\end{algorithm}
%
% %%%%%%%%%%%%%%%%%%%%%%%%%%%%%%%%%%%%%

\begin{figure}[h!]
 \centering
        \includegraphics[scale=0.2]{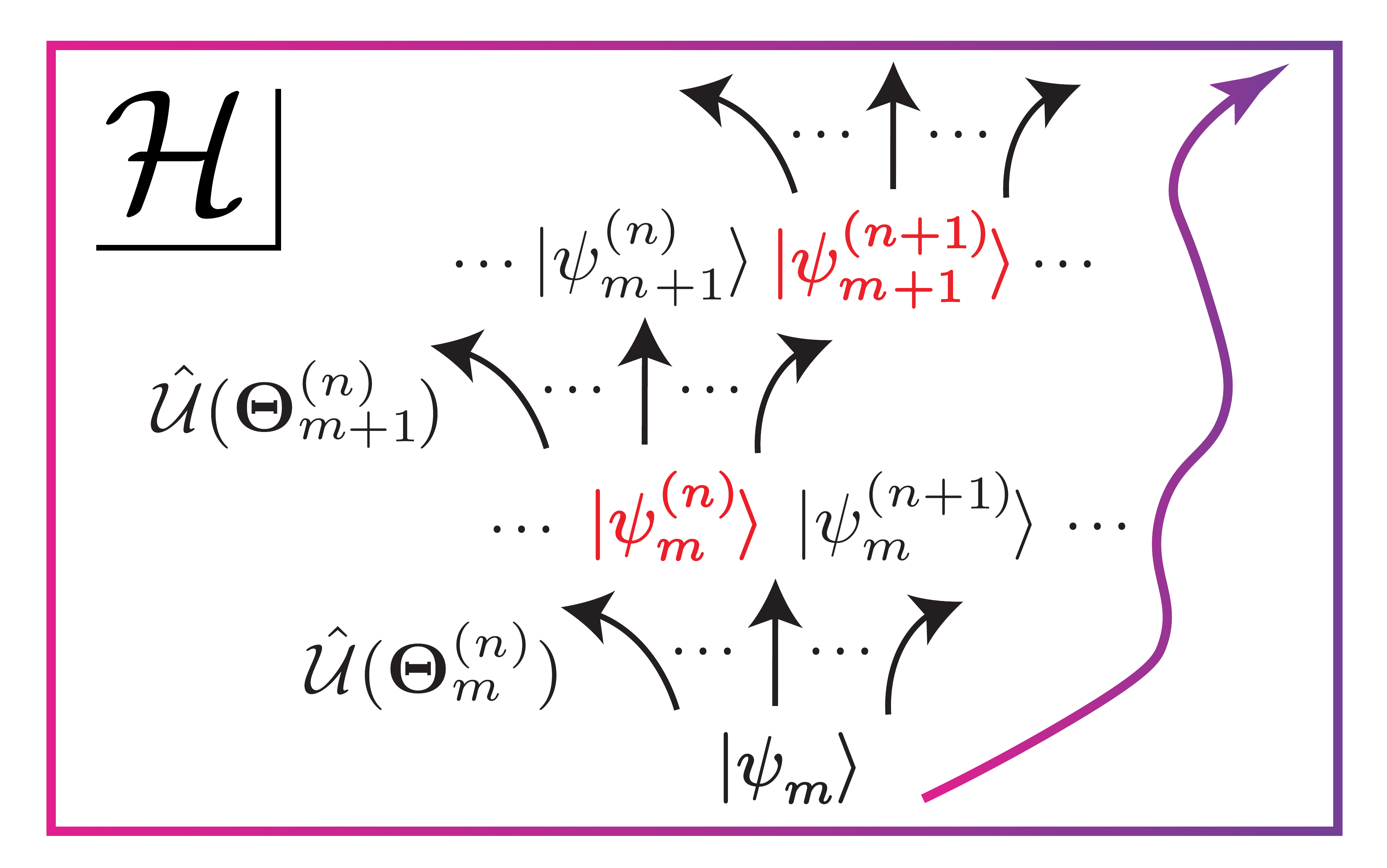}
        \caption{Schematics of the learning algorithm as in Alg.\ref{alg:born}. At the $m$-th quench, we independently evolve $N$ copies of the state $\ket{\psi_{m}}$ with different time-evolution operators $\bTheta_m^{(n)}$ sampled from the same distribution. At the $(m+1)$-th quench, we pick the $\ket{\psi_{m}^{(n)}}$ with the lowest loss value (based on the loss function Eqn.\eqref{eqn:MMD}) from the previous quench as our new starting point and evolve again. As we repeat this process, the learning resembles a directed random walk in the Hilbert space.}
        \label{fig:alg_schem}
\end{figure}    
% %%%%%%%%%%%%%%%%%%%%%%%%%%%%%%%%%%%%%

%%%%%%%%%%%%%%%%%%%%%%%%%%%%%%%%%%%%%
\subsection{Randomly driven MBL Born machine}
\label{sec:RDBM}
%%%%%%%%%%%%%%%%%%%%%%%%%%%%%%%%%%%%%
In classical machine learning, stochasticity is found to have the effect of smoothing out loss landscape and helps to avoid local minima \cite{bottou1991stochastic,goodfellow2016deep,mehta2019high}. When introducing the hidden Born machine in Eqn.\eqref{eqn:hidden}, the hidden units are traced out and effectively act as a heat bath for the remaining visible units and provide a source for stochasticity. In order to understand the extent to which stochasticity aids learning in the hidden Born machine, in this section, we construct a Born machine with random drive that mimics the heat bath. In Fig.\ref{fig:loss_MNIST}, we numerically demonstrate that the randomly driven Born machine (RDBM) outperforms the basic Born machine, defined in Eqn.\eqref{eqn:born} without random drives, and approaches the performance of (yet still underperforms) the hidden Born machine trained with Alg.\ref{alg:born}.\\
Let's consider the Hamiltonian Eqn.\eqref{eqn:XXZ} with applied external random drives $\hat{\mathcal{H}}_{\text{RD}}$ in the $x-$direction (we can also apply random drives in the $xy-$plane and the result will be similar),
\begin{equation}
\label{eqn:randomdrive}
    \hat{\mathcal{H}}_{\text{RD}}(t) = \sum_i d_i^m(t) \hat{S}_i^x.
\end{equation}
To model the heat bath, we would like $\{d_i^m(t)\}$ to be like a white noise,
\begin{equation}
\label{eqn:whitenoise}
    \langle d_i^m (t) d_i^m(0) \rangle = 2D \delta(t),
\end{equation}
where $D$ is the amplitude of the white noise and is proportional to the temperature of the bath. 
In the simulation, we split the driven interval $T$ into intervals of auto-correlation time $\tau$, and require that Eqn.\eqref{eqn:whitenoise} holds for $t>\tau$. Outside of this correlation time, $d_i^m (t)$ is drawn i.i.d. from $\mathcal{N}(0,\sqrt{2D})$.
%%%%%%%%%%%%%%%%%%%%%%%%%%%%%%%%%%%%%

% %%%%%%%%%%%%%%%%%%%%%%%%%%%%%%%%%%%%%

\begin{figure}[h!]
 \centering
        \includegraphics[scale=0.18]{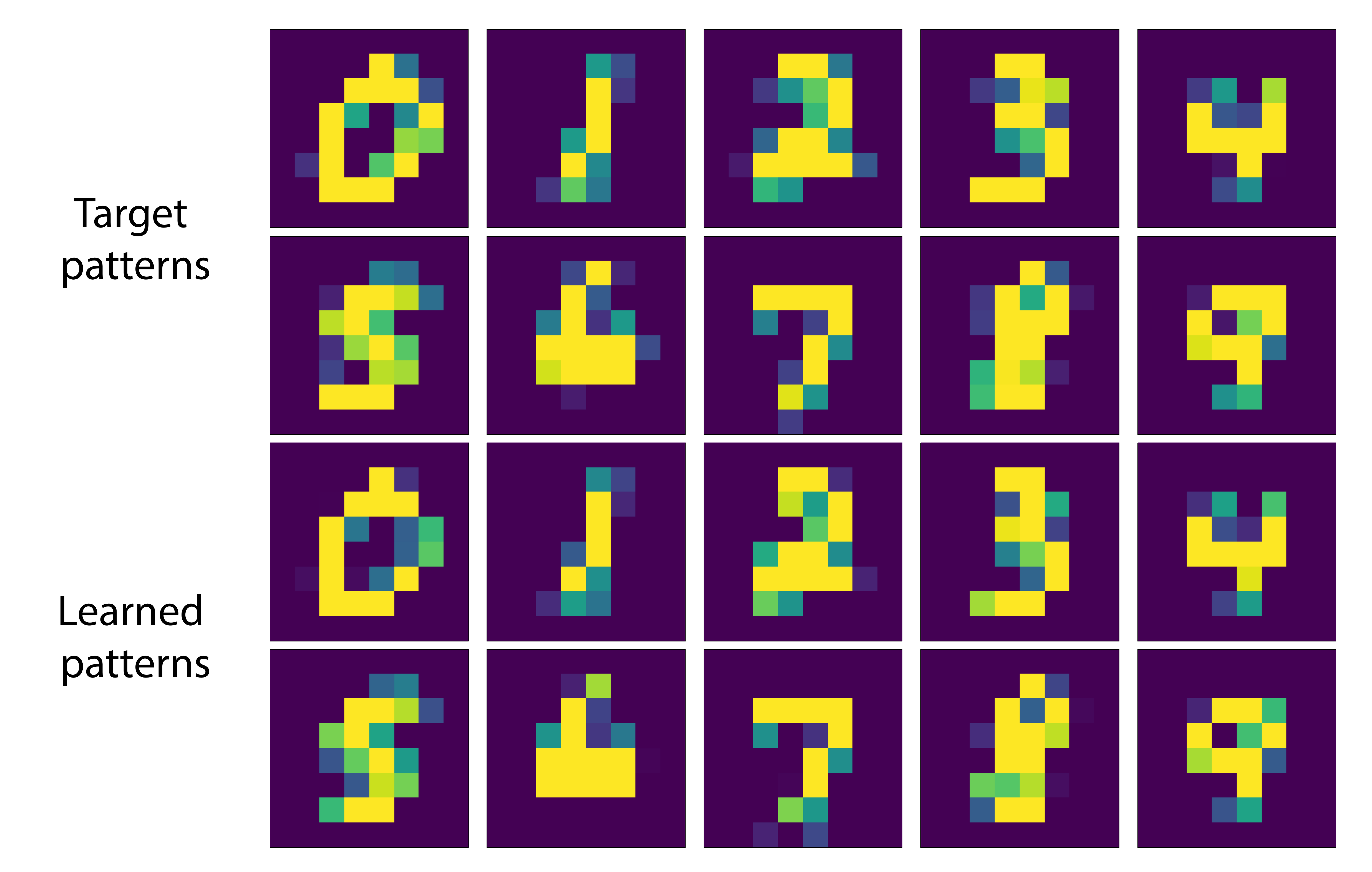}
        \caption{Learning toy MNIST digit patterns. The top two rows are different data instances $q_{\text{data}}$ in our toy MNIST digit patterns dataset. The bottom two rows are the corresponding learning outcome $p_{\text{model}}$ from our MBL hidden Born machine (each digit trained separately).}
        \label{fig:MNIST_0_9}
\end{figure}    
% %%%%%%%%%%%%%%%%%%%%%%%%%%%%%%%%%%%%%

%%%%%%%%%%%%%%%%%%%%%%%%%%%%%%%%%%%%%

To illuminate on the learning power of the hidden Born machine, here, we compare the three models: the basic Born machine, the Randomly Driven Born machine, and the hidden Born machine. We task all three models with a toy dataset constructed from the images of MNIST dataset\cite{deng2012mnist} (downsampled to $2^{L_v}$ pixels). Our toy dataset consists of mean pixel values across all different styles within a single type of MNIST digit, see `target patterns' in Fig.\ref{fig:MNIST_0_9} (also see Section \ref{app:MNIST}).

We perform the training of the hidden Born machine using the algorithm described in Alg.\ref{alg:born}, and show the corresponding learning outcomes in Fig.\ref{fig:MNIST_0_9}. Our results indicate our hidden model is able to learn different patterns of MNIST digits accurately (the result of basic BM and RDBM are omitted). 

%%%%%%%%%%%%%%%%%%%%%%%%%%%%%%%%%%%%%
\begin{figure}[h!]
\centering
\includegraphics[scale=0.3]{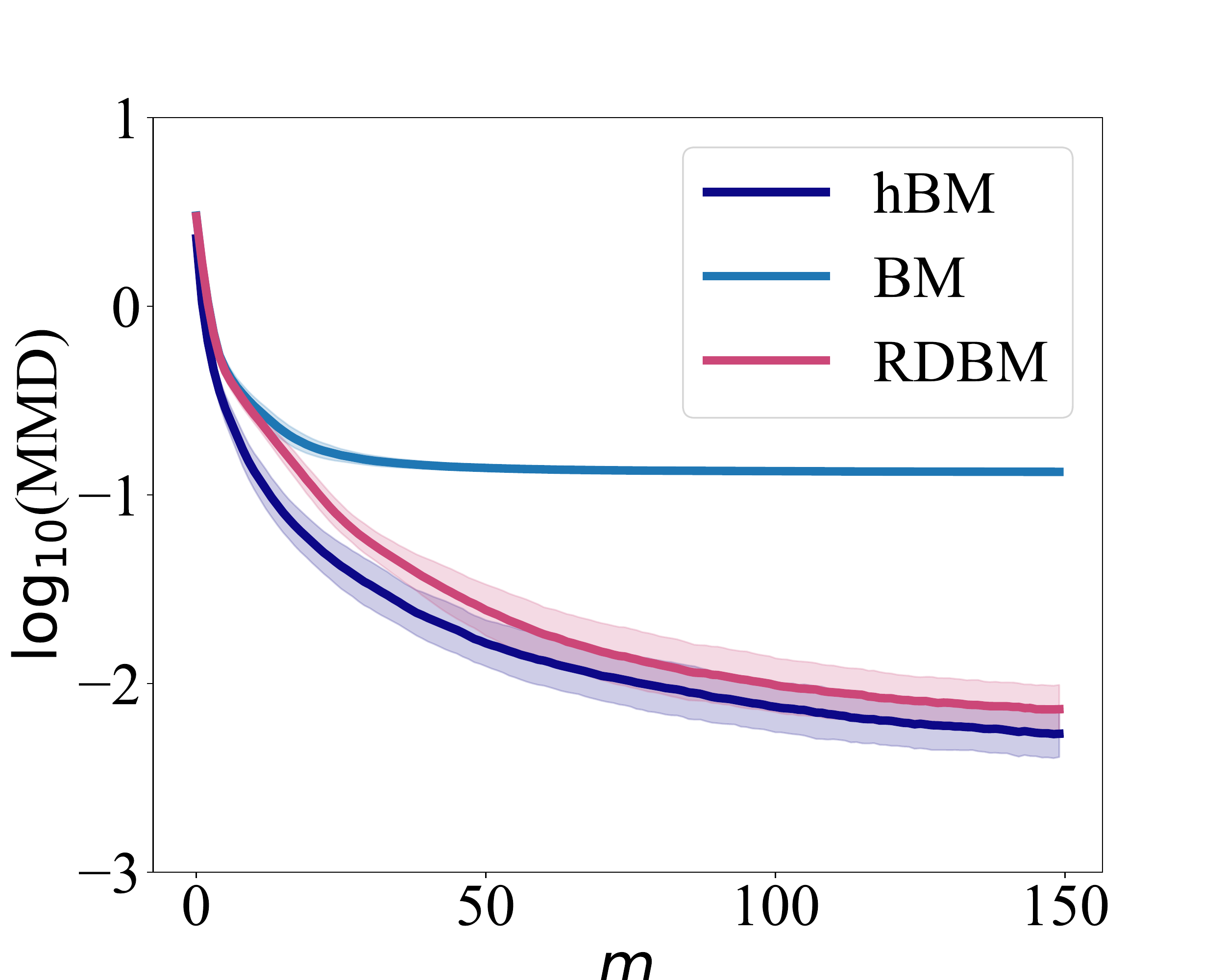}
\caption{Model comparisons. (a) Basic Born machine(BM), Randomly Driven Born machine(RDBM), and hidden Born machine(hBM). Log-MMD loss as a function of quench layer number $m$. The solid curves are averaged over 100 different realizations, with one standard deviation included as the shades. The hidden Born machine achieves the lowest MMD loss throughout and at the end of the training.}
\label{fig:loss_MNIST}
\end{figure}
%%%%%%%%%%%%%%%%%%%%%%%%%%%%%%%%%%%%%
%
We plot the loss as a function of quenches $m$ in Fig.\ref{fig:loss_MNIST}, and we can see that the hidden model performs best out of the three both in terms of final MMD loss on the dataset, with the hierarchy being hBM $\succeq$ RDBM $\succeq$BM. \\ 

%%%%%%%%%%%%%%%%%%%%%%%%%%%%%%%%%%%%%%%%%%%%%%%%%%%%%%%%%%%%%%%%%%%%%%%%%%
%%%%%%%%%%%%%%%%%%%%%%%%%%%%%%%%%%%%%%%%%%%%%%%%%%%%%%%%%%%%%%%%%%%%%%%%%%

% %%%%%%%%%%%%%%%%%%%%%%%%%%%%%%%%%%%%%
\begin{figure*}[ht]
    \centering
    \includegraphics[scale=0.23]{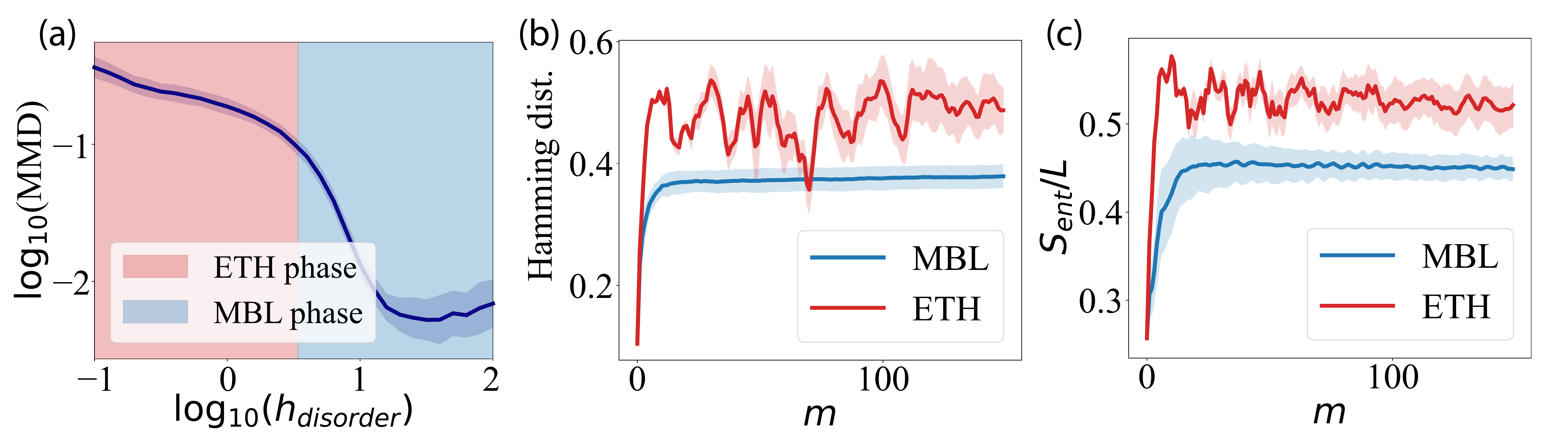} %width=0.9\linewidth
    \caption{Training hidden Born machine in thermal and MBL phases. (a) The terminal (at the final layer of quench) MMD loss of hidden Born machine on the toy MNIST task is plotted as a function of disorder strength $h_{d}$. The results are averaged over 100 realizations and one standard deviation is included as shade. (b) Hamming distance with respect to the initial state (normalized by $L$) as a function of quenches $m$. In the thermal phase, states change discontinuously over successive quenches, whereas in the MBL phase states change gradually toward the target state that gives rise to the desired distribution. (c) Entanglement entropy per site as a function of quenches $m$, confirming that our system evolves under dynamics distinctive in the thermal/MBL phases.}
    \label{fig:compare models}
\end{figure*}    
% %%%%%%%%%%%%%%%%%%%%%%%%%%%%%%%%%%%%%

\section{Learnability in different phases}
\label{sec:different_phases}
%%%%%%%%%%%%%%%%%%%%%%%%%%%%%%%%%%%%%%%%%%%%%%%%%%%%%%%%%%%%%%%%%%%%%%%%%%
%%%%%%%%%%%%%%%%%%%%%%%%%%%%%%%%%%%%%%%%%%%%%%%%%%%%%%%%%%%%%%%%%%%%%%%%%%
We have already seen that the hidden Born machine in the MBL phase can properly learn the toy MNIST dataset (Fig.\ref{fig:MNIST_0_9}). An important question arises that whether learning can happen in the thermal phase. In the thermal phase, information spreads throughout the system, which makes it difficult to extract. In the quenched approach as in Eqn.\eqref{eqn:H_total}, the state of the system in the thermal phase changes wildly between successive quenches and effectively only parameters in the last layer of the quench would be trained. In contrast, as the system become more disorderd and enters the MBL phase,  the breakdown of thermalization and emergence of local integrals of motion leads to local memory, which is useful for directing the state toward a target corner of the Hilber space (Fig.\ref{fig:model schematics}(c)). We aim to understand the effect of disorder in learning by comparing the learning ability of the hidden Born machine in the MBL and thermal phases.

In Fig.\ref{fig:compare models}(a), we show a log-log plot of the final quench layer MMD loss on the toy MNIST dataset as a function of disorder strength $h_d$. By varying the disorder strength, the system in Eqn.\eqref{eqn:XXZ} can exhibit both a thermal phase (denoted as ETH) and an MBL phase depending on whether the disorder strength exceeds the critical value $h_c\sim 3.5$ (for $J_{zz}=J_{xy}=1$). We observe that the loss value has a significant change at the transition from the thermal phase (corresponding to $h_d < 3.5$ indicated by pink shade) into the MBL phase (indicated by blue shade). The relatively high value of MMD loss in the thermal phase indicates that the hidden Born machine fails to learn. As we increase the disorder, the MMD loss deep in the MBL phase decreases significantly, indicating better learning power of the MBL phase. We can attribute the better learning power in the MBL phase to the quantum memory and the emergent local integral of motions. In contrary to the thermal phase, the thermalization mechanism wipes out all the information from the initial conditions, as observed similarly in the case of quantum reservoir computing in the MBL phase \cite{xia2022reservoir}.

To better quantify the learning mechanism in the MBL phase, we investigate the time evolution of quantities underlying MBL physics during the quenched steps. First, we investigate the Hamming distance (HD) defined as
\begin{equation}
\label{eqn:Hamming_distance}
\mathcal{D}(t)= \frac{1}{2}-\frac{1}{2L}\sum_i \langle \psi_0|\sigma_i^{z}(t)\sigma_i^{z}(0)|\psi_0\rangle,
\end{equation}
which gives a measure of number of spin flips with respect to the initial state $\psi_0$ normalized by the length of chain $L$. It's expected that in the long time the HD approaches 0.5 in the thermal phase and decreases as one increases the disorder\cite{Hauke_2015}. In Fig.\ref{fig:compare models}(b), we show the trajectory of HD at the end of each quench $\mathcal{D}^m(t=T)$. We observe that, evolving in the thermal phase the HD fluctuates around the value of 0.5 as expected, while in the MBL phase the HD reaches a lower value about 0.33. The more significant fluctuations in the thermal phase indicate that the system retains little information about the most recent quench, and therefore is difficult to be manipulated toward a target state that gives desired probability distribution. In contrast, the relatively small fluctuations in the MBL phase suggest that system changes gradually between successive quenches and is more amenable to directed evolution by quenches.

One hallmark of MBL phase is the logarithmically slow growth of von Neumann entanglement entropy ($S_{\text{ent}}^m=-\rm Tr\rho_m \ln \rho_m$) due to the presence of strong interaction. Notice that $\rho_m$ is the reduced density matrix at quench $m$, which can be obtained by tracing over the complementary part of system with respect to the subsystem of interest. This can be considered as slow dephasing mechanism implying that not all information of initial state survives \cite{MBL_ent1,MBL_ent2,MBL_ent3}. In order to confirm that our system indeed evolves under MBL/thermal dynamics when trained in these two phases, in Fig.\ref{fig:compare models}(c), we track the value of $S_{\text{ent}}^m$ over different quenches. In the MBL phase, $S_{\text{ent}}^m$ shows a quick saturation, while in the thermal phase the entanglement entropy changes significantly from successive quenched steps, a behavior expected from the thermal phase.

\subsection{Pattern recognition}
\label{sec:pattern recognition}

Pattern recognition has been implemented in a variety of analog classical systems ranging from molecular self-assembly to elastic networks \cite{winfree1998algorithmic,qian2010efficient,murugan2015multifarious,zhong2017associative,o2019temporal,stern2020continual,zhong2021machine,stern2022learning}. It is interesting to ask whether quantum systems possesses similar power. In this section, we demonstrate the pattern recognition ability of the MBL hidden Born machine. Here, we take the same toy dataset of MNIST digit patterns as in Fig.\ref{fig:MNIST_0_9}. Each pattern $\xi^{\mu} \in [0,1]^{2^{L_v}}$ is a (normalized) vector in the pixel space, where $L_v$ is the length of the visible units, and $\mu=1,2,..,P$ denotes the pattern index. We encode the patterns into the hidden Born machine by setting $p_{\text{data}} = \sum_{\mu}\xi^{\mu}$ \footnote{While there exists other more sophisticated encoding schemes, here we choose the simplest one for illustration.}. 
Again, we perform the training of the hidden Born machine using the algorithm in Alg.\ref{alg:born} (see first column of Fig.\ref{fig:learning MNIST} for the learned patterns from $p_{\text{model}}$). After training, we obtain the target final state $\ket{\psi_M}$, along with a series of unitaries $\{ \hat{\mathcal{U}}(\bTheta_m) \equiv \hat{\mathcal{U}}_m \}_{m=0}^{M}$ that defines the entire history of intermediate states during successive quenches, $\ket{\psi_m} = \prod_{i=0}^m \hat{\mathcal{U}}_{m-i} \ket{\psi_0}$, which upon tracing out hidden units becomes intermediate model distributions, $p_m = \Tr \Tr_{h} \ket{\psi_{m}}\bra{\psi_{m}}\Pi_Z$. Now given a partially corrupted pattern $\Tilde{\xi}^{\mu}$ and the state $\ket{\Tilde{\psi}^{\mu}}$ that gives rise to this corrupted pattern, $|\Tilde{\psi}^{\mu}(\boldsymbol{z})|^2/\mathcal{N} = \Tilde{\xi}^{\mu}$ (see second column of Fig.\ref{fig:learning MNIST} for examples of corrupted patterns), we can identify the `closest' intermediate state $\ket{\psi_{m^*}}$ where $m^*=\argmax_m \text{MMD}(\Tilde{\xi}^{\mu},p_{m})$. Then we apply unitary time-evolution to the corrupted state $\ket{\Tilde{\psi}^{\mu}}$ using the series of learned unitaries starting from $m^*$ and obtain the `retrieved' state $\ket{\hat{\psi}^{\mu}} \equiv \prod_{i=0}^{M-m^*}\hat{\mathcal{U}}_{M-i} \ket{\Tilde{\psi}^{\mu}}$. We can then compute the corresponding retrieved pattern as $\hat{\xi}^{\mu} \equiv \Tr \Tr_{h} \ket{\hat{\psi}^{\mu}}\bra{\hat{\psi}^{\mu}}\Pi_z$ (see last column of Fig.\ref{fig:learning MNIST} for the retrieved patterns).

\label{sec:MNIST}
%%%%%%%%%%%%%%%%%%%%%%%%%%%%%%%%%%%%%
\begin{figure}[h!]
\centering
\includegraphics[scale=0.3]{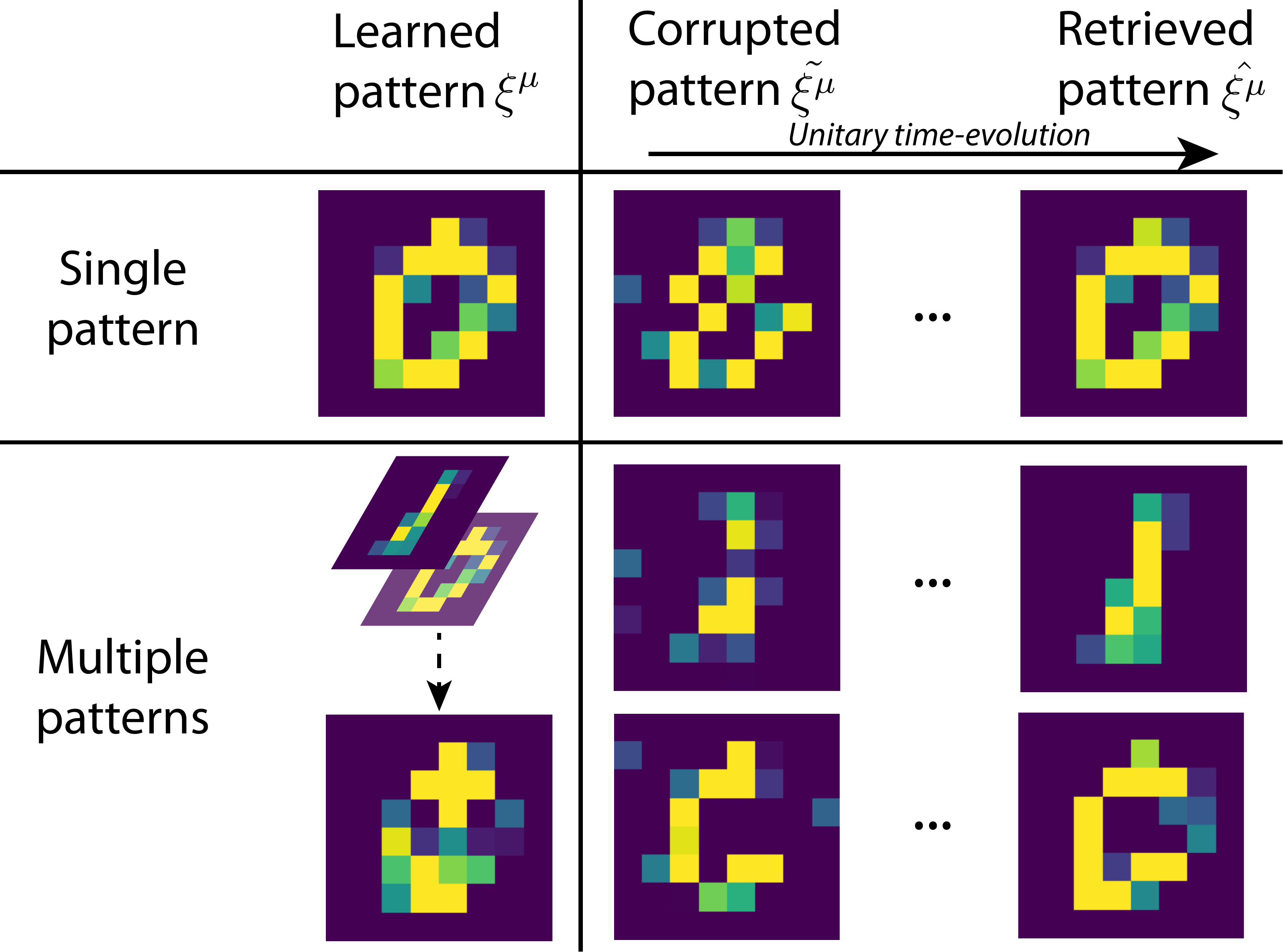}
\caption{Pattern recognition task by the MBL hidden Born machine. Given a corrupted pattern $\Tilde{\xi}^{\mu}$ and its corresponding corrupted state $\ket{\Tilde{\psi}^{\mu}}$, we find the quench layer number $m^*$ such that the intermediate model distribution $p_m^*$ resembles the corrupted pattern most. Then we time-evolve $\ket{\Tilde{\psi}^{\mu}}$ with the series of learned unitaries $\hat{\mathcal{U}}_i$ starting from $m^*$ to obtain the retrieved state $\ket{\hat{\psi}^{\mu}}$, from which we can then obtain the retrieved pattern $\hat{\xi}^{\mu}$. Top row: after learning a single pattern (digit `$0$'), a complete `$0$' can be retrieved from a partially corrupted `$0$'. Bottom row: after learning multiple patterns (superposition of digit `$0$' and `$1$'), complete `$0$' or `$1$' can be selectively retrieved from partially corrupted `$0$' and `$1$', respectively.}
\label{fig:learning MNIST}
\end{figure}
%%%%%%%%%%%%%%%%%%%%%%%%%%%%%%%%%%%%%

As shown in the top row of Fig.\ref{fig:learning MNIST}, in the case of a single pattern (a digit `0'), the MBL hidden Born machine is able to retrieve a complete pattern from a corrupted pattern (a partially corrupted digit `0'). As shown in the bottom row of Fig.\ref{fig:learning MNIST}, in the case of multiple patterns (a superposition of `0' and `1'), the MBL hidden Born machine is able to selectively retrieve complete patterns (`0' or `1') based on the input corrupted pattern \footnote{However, one should note that just like in classical pattern recognition \cite{hopfield1982neural}, if the input pattern gets too corrupted and does not resemble any of the encoded patterns, this procedure will fail.}.

%%%%%%%%%%%%%%%%%%%%%%%%%%%%%%%%%%%%%
\subsection{Learning quantum dataset}
\label{sec:quantum_dataset}
%%%%%%%%%%%%%%%%%%%%%%%%%%%%%%%%%%%%%

%%%%%%%%%%%%%%%%%%%%%%%%%%%%%%%%%%%%%

\begin{figure*}[ht]
    \centering
    \includegraphics[scale=0.18]{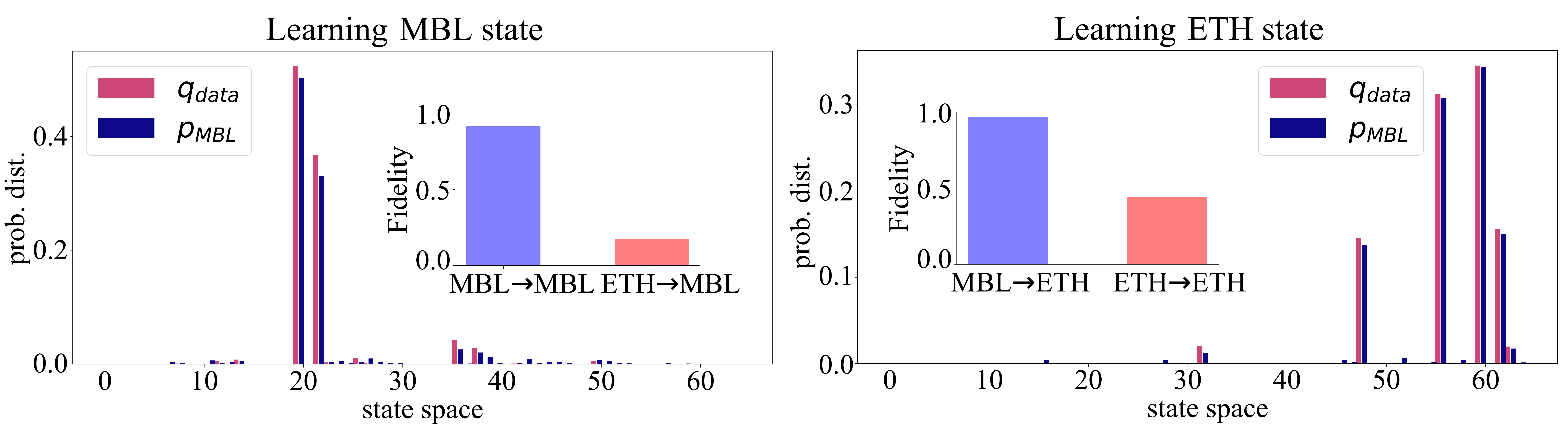} %width=0.9\linewidth
    \caption{Learning quantum dataset. Left/right: MBL hidden Born machine trained in MBL phase learns the probability distribution corresponding to an MBL/thermal (denoted as ETH) target state. Insets: classical fidelities between the model and the data distributions. Model trained in the MBL phase has better learning capability than model trained in the thermal phase. }
    \label{fig:learning_quantum_states}
\end{figure*}    

%%%%%%%%%%%%%%%%%%%%%%%%%%%%%%%%%%%%%

We have demonstrated the power of MBL Born machine in learning classical data of the toy MNIST digit patterns, now we explore the ability of the MBL Born machine in learning data obtained from measurements of quantum states. While quantum state tomography is the standard method for state reconstruction, it becomes a daunting task as the system size increases. In this respect, quantum machine learning has shown great success in learning quantum states from limited amount of data\cite{carrasquilla_reconstructing_2019,torlai_neural_network_2018,tomography_Wang_2020,huang_provably_2021,huang_predicting_2020,randomized_toolbox,Abigail_Enhanced_Born}. In this section, we use the hidden Born machine to learn data obtained from quantum many-body states prepared by Eqn.\eqref{eqn:H_total} subject to a single layer of quench, but with disorder strengths $h_d$ different from the phases that the hidden Born machine is trained in.

%%%%%%%%%%%%%%%%%%%%%%%%%%%%%%%%%%%%%

In Fig.\ref{fig:learning_quantum_states}, we demonstrate the learning ability of Born machine in the thermal and MBL phase. In Fig.\ref{fig:learning_quantum_states} left/right, we compare the measurement outcome sampled from the exact simulation $q_{\text{data}}$ in MBL/thermal phase (denoted as ETH)  (shown in purple), with those learned via hidden Born machine trained in MBL phase (shown in blue). In the insets we show the classical fidelity between the model distribution $p$ and data distribution $q$, $F(p,q)=\left( \sum_i \sqrt{p_i q_i} \right)^2$. We see that the hidden Born machine trained in MBL phase is able to capture the underlying probability distribution obtained from both the MBL and thermal phases with high fidelity ($\sim 0.98$), while the hidden Born machine trained in thermal phase fails to learn either. Notice that in order to learn the quantum state, one needs to perform measurement in the informationally-complete basis as reported in Ref.\cite{Abigail_Enhanced_Born}.

%%%%%%%%%%%%%%%%%%%%%%%%%%%%%%%%%%%%%
\subsection{Learning parity dataset}
\label{sec:parity}
%%%%%%%%%%%%%%%%%%%%%%%%%%%%%%%%%%%%%

%%%%%%%%%%%%%%%%%%%%%%%%%%%%%%%%%%%%%
\begin{figure}[h!]
\centering
\includegraphics[scale=0.33]{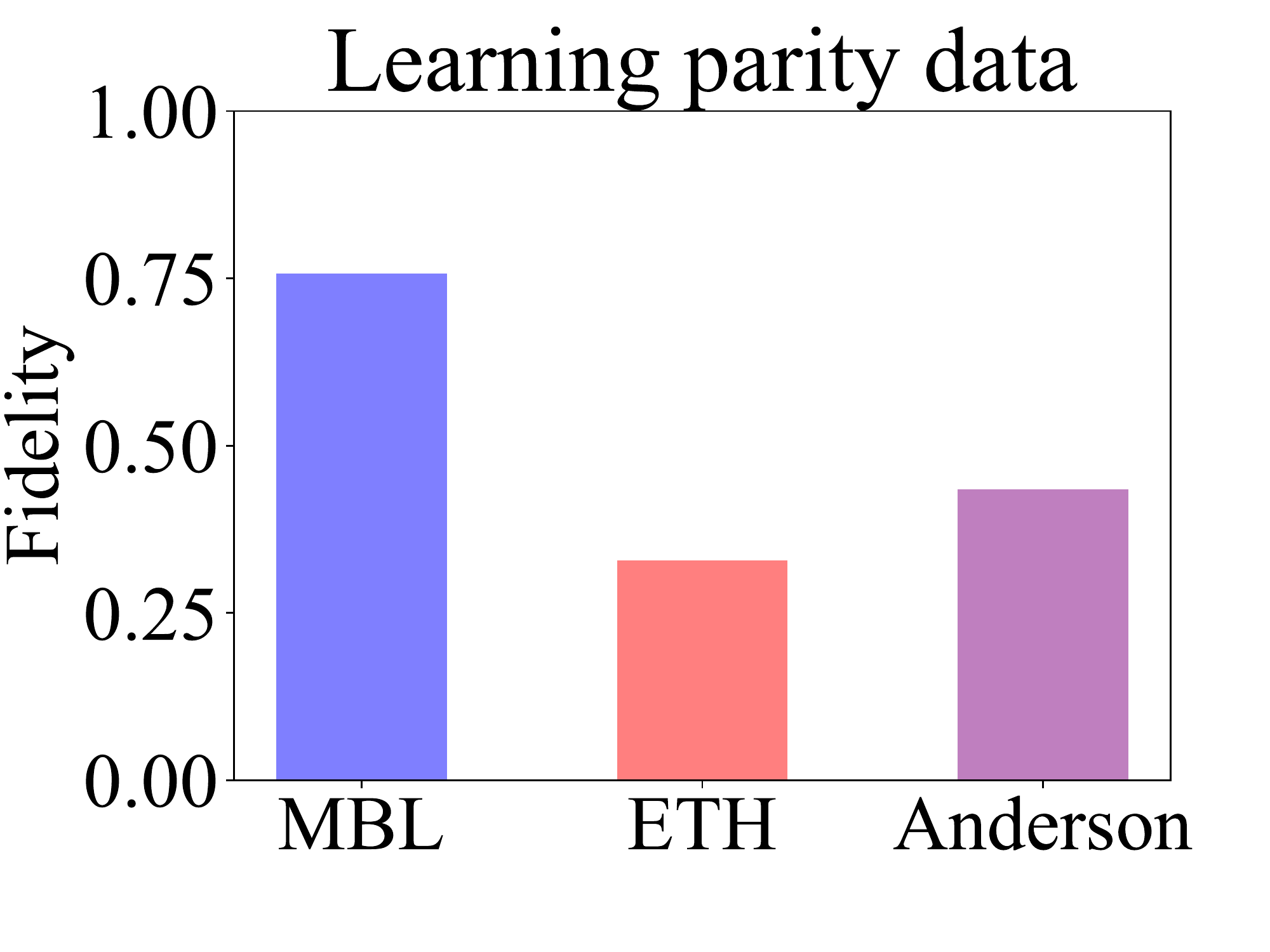}
\caption{Learning parity dataset. Different bars in the horizontal-axis correspond to model trained in the MBL, thermal, and Anderson localized phases, respectively. Vertical-axis shows the classical fidelity of the model. Model trained in the MBL phase exhibits the highest fidelity despite the dataset is highly nonlocal. Comparing model performances in three phases suggest that both disorder and interaction are important for learning.}
\label{fig:learning parity}
\end{figure}
%%%%%%%%%%%%%%%%%%%%%%%%%%%%%%%%%%%%%

In the previous sections, we have discuss the role of localization and emergent memory in learning various datasets, however, the role played by interaction in the many-body localized phase remains unclear. To shed light into the role of interaction and its interplay with disorder, here, we investigate the power of MBL phase in learning parity dataset and compare it with both thermal and Anderson localized phase which can be obtained by setting $J_{zz}=0$ in Eqn.\eqref{eqn:XXZ}. Here, we consider the even parity dataset, which is defined as set of bit-string $(b_1,b_2,..,b_L)$ of length $L$ with $b_i\in\{0,1\}$, such that the parity function $\Pi(b_1,b_2,..,b_L):=\sum_{i=1}^{N} b_i \mod 2 $ is equal to $0$. While this is a classical bitstring, it appears as measurement outcome of particular quantum observables in certain basis such as measurement outcome of GHZ state in the $x-$basis.

Previous studies has indicate challenging learning on this dataset, in particular training the Born machine based on MPS with gradient descent optimization schemes has encountered failures \cite{Najafi_GHZ_nonlocal}, while quantum inspired optimization schemes such as density matrix learning has shown great success with the caveat in their scaling\cite{Bradley_2020}. Here, we investigate the power of our hidden Born machine across various phases in learning the parity dataset. Our numerical results (Fig.\ref{fig:learning parity}) demonstrates the interesting fact that both the MBL phase and Anderson localized phase show better performance compare to the thermal phase. The better learning performance in these two phases suggest that the emergence of integral of motion and memory plays an important role in learning. We further notice that the MBL phase has a better performance even though the Anderson localized phase is known to have better memory. In the latter the strong localization prevents the transport of information across the system, leading to a lesser learning power. While the value of fidelity around $F_{\text{MBL}}=0.75$ is not too high, reflecting the hardness of learning the parity dataset, our MBL hidden Born machine still shows a better performance compare to MPS Born machine which was reported a fidelity of $F_{\text{MPS}}=0,48$\cite{Najafi_GHZ_nonlocal}. Our numerical results indicate the importance of the presence of both disorder and interaction in the MBL hidden Born machine, and suggests that successive quenches defined by the learning cuts through a path in the Hilbert space that harnesses both local memory and interaction in order to arrive the target state.

\section{MBL phase transition}
\label{app:MBL_check}

In this section, we present the details of the numerical simulation of XXZ model (Eqn.\eqref{eqn:XXZ}) and confirm the thermal to MBL phase transition. We simulate the XXZ model using exact diagonalization methods provided by the QuSpin package \cite{weinberg2017quspin,weinberg2019quspin}. Throughout the paper we use parameters $J_{xy}=J_{zz}=1$. 

One hallmark of the MBL phase is the Poission distribution of level spacings in the eigenspectrum of the Hamiltonian. \cite{nandkishore2014many,alet2018many,Abanin_2019}. The level statistics $Pr(r_{\alpha})$ is defined as the normalized distribution of
\begin{equation}
    r_{\alpha} = \frac{\min (\Delta_{\alpha+1},\Delta_{\alpha})}{\max (\Delta_{\alpha+1},\Delta_{\alpha})},
\end{equation}
where $\Delta_{\alpha} = E_{\alpha+1}-E_{\alpha}$ are the level spacings in the eigenspectrum. In Fig.\ref{fig:level_statistics_L=16}, we show the level statistics in a simulation of $L=16$ spins described by Eqn.\eqref{eqn:H_total} subject to a single quench $M=1$, at two different disorder strengths: $h_d=0.1$ and $h_d=3.9$ (the critical disorder strength is $h_c \sim 3.5$ for $J_{zz}=J_{xy}=1$). We see that indeed the level statistics in the thermal phase obeys Wigner-Dyson statistics, and in the MBL phase obeys Poisson statistics, confirming the existence of thermal-MBL phase transition. 

%%%%%%%%%%%%%%%%%%%%%%%%%%%%%%%%%%%%%
\begin{figure}[h!]
\centering
\includegraphics[scale=0.26]{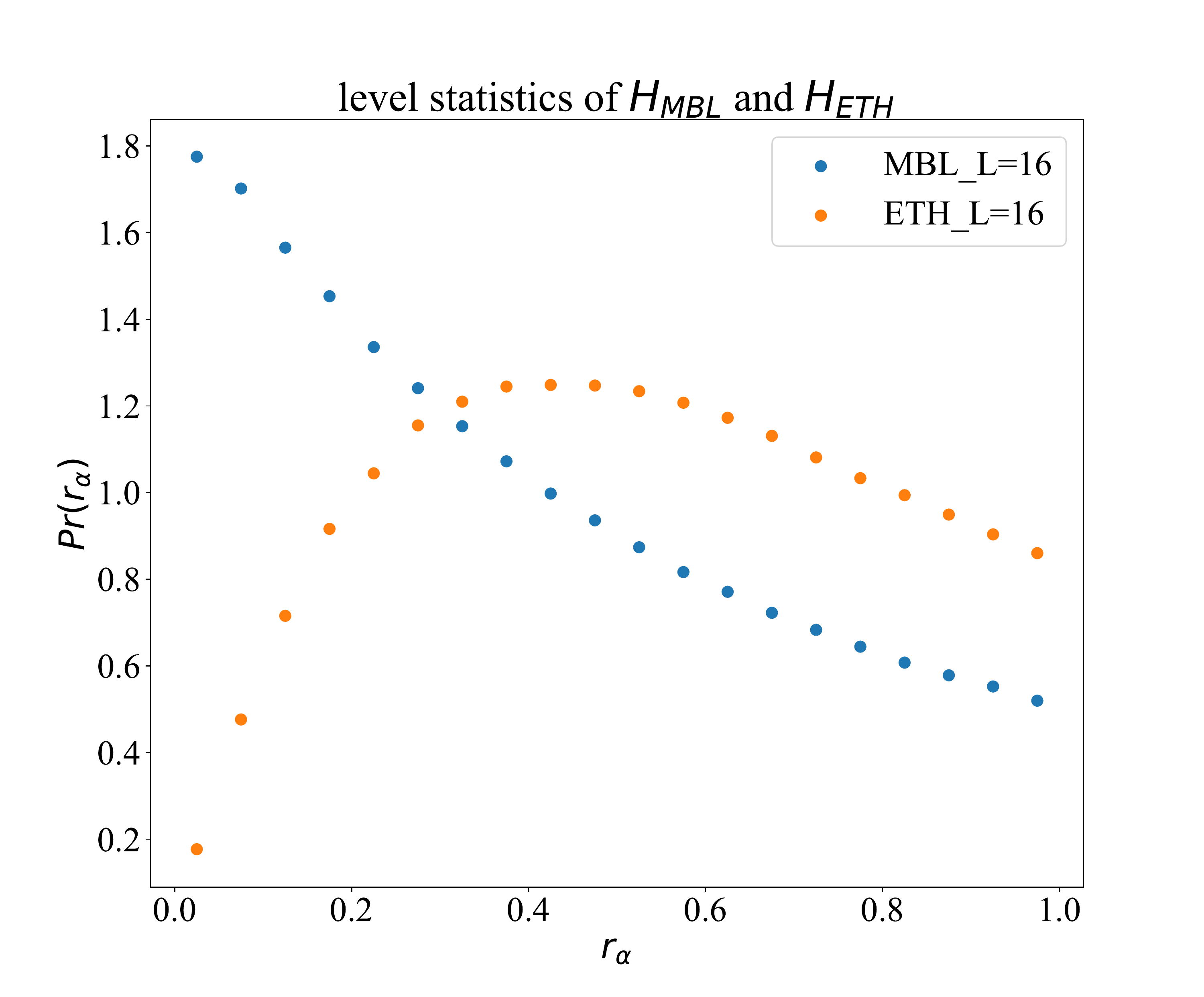}
\caption{Level statistics of $L=16$ XXZ model subject to quenches in the $z$-direction. The thermal phase (denoted as ETH) is simulated with $h_d=0.1$ and the MBL phase is simulated with $h_d=3.9$. Results are averaged over 1000 different realizations.}
\label{fig:level_statistics_L=16}
\end{figure}
%%%%%%%%%%%%%%%%%%%%%%%%%%%%%%%%%%%%%

Another hallmark of MBL phase is the area law scaling of von Neumann entanglement entropy ($S_{\text{ent}}=-\rm Tr\rho \ln \rho$), compared to the volume law scaling in the thermal phase. We numerically calculate the half-system entanglement entropy in the middle of the spectrum for the Hamiltonian in Eqn.\eqref{eqn:H_total}, and perform a scaling analysis for different $L$ and different disorder strengths $h$ (see Fig.\ref{fig:ent_entropy_scaling}). Our numerical results agrees with those reported in \cite{MBL_edge}.

%%%%%%%%%%%%%%%%%%%%%%%%%%%%%%%%%%%%%
\begin{figure}[h!]
\centering
\includegraphics[scale=0.26]{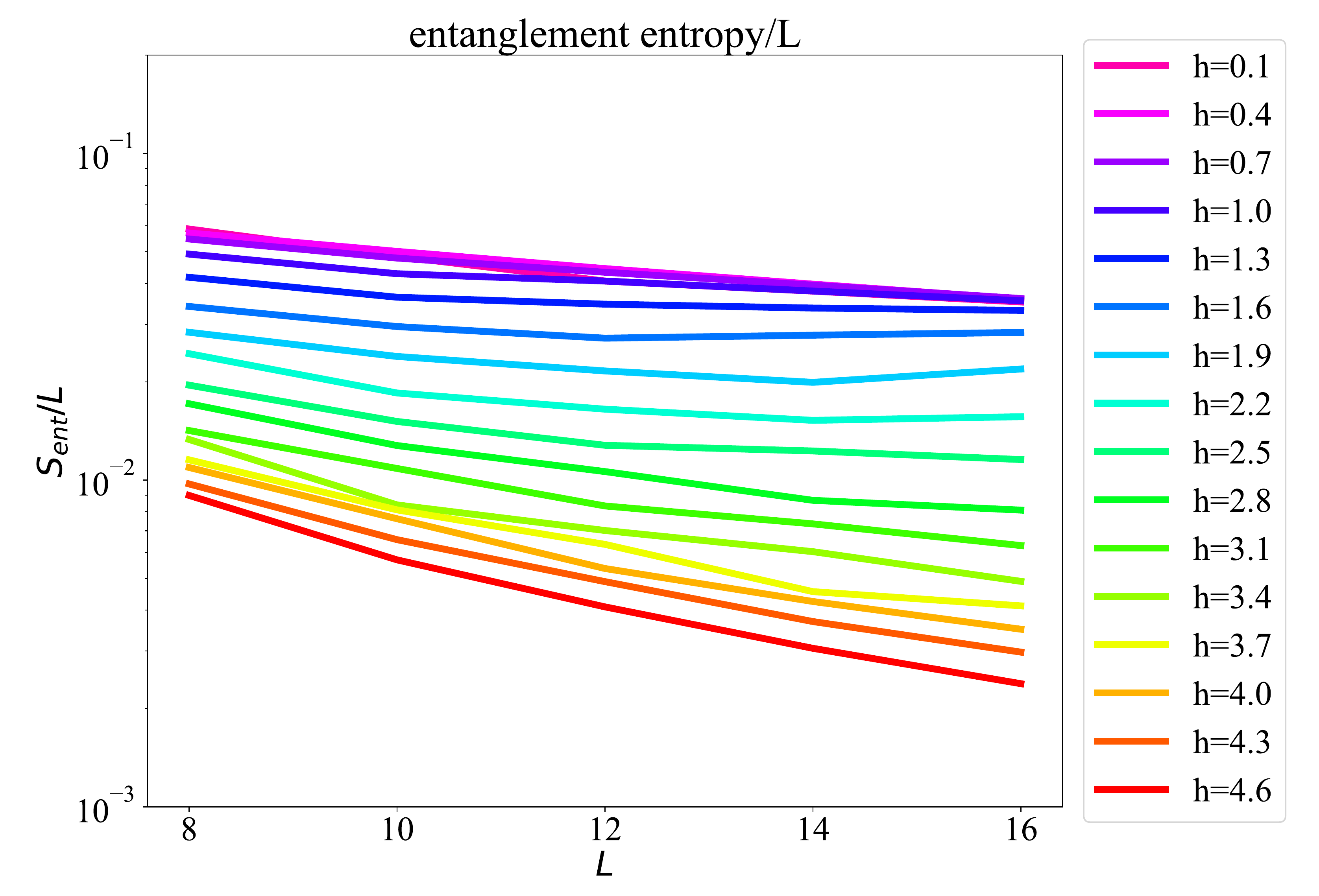}
\caption{Scaling analysis of entanglement entropy. We plot the entanglement entropy per site $S_{ent}/L$ as a function of system size $L$ for different disorder strengths $h$. Volume law scaling in the thermal phase (small $h$) leads to constant $S_{ent}/L$, while area law scaling in the MBL phase (for large $h$) leads to decreasing $S_{ent}/L$. }
\label{fig:ent_entropy_scaling}
\end{figure}
%%%%%%%%%%%%%%%%%%%%%%%%%%%%%%%%%%%%%

\section{Training MBL hidden Born machine}
\label{app:MMD}
Previously, KL-divergence has been suggested for training MBL Born machine as a generative model \cite{tangpanitanon2020expressibility}. However, KL-divergence does not capture correlations within data, and suffers from infinities outside the support of data distribution. To remedy these situations, the Maximum Mean Discrepancy (MMD) loss has been proposed for training Born machines \cite{liu2018differentiable}. The MMD loss measures the distance between model distribution $p$ and target distribution $q$, by comparing their mean embeddings in the feature space. The (squared) MMD loss can be written as
\begin{align}
    \mathcal{L}_{MMD} &= \norm{\sum_x p(x)\phi(x)  - \sum_x q(x)\phi(x)}^2 \nonumber\\
    &= \mathbb{E}_{x,x'\sim p} k(x,x') + \mathbb{E}_{y,y'\sim q} k(y,y') \\
    &- 2\mathbb{E}_{x\sim p, y\sim q} k(x,y) \nonumber,
\end{align}
where we have employed the kernel trick and write $k(x,y) = \phi(x)^T \phi(x)$. In our model, we use a Gaussian mixture kernel $k(x,y) = \frac{1}{c}\sum_{i=1}^c \exp\left(-\frac{1}{2\sigma_i^2}|x-y|^2\right)$ of four channels $c=4$, with corresponding bandwidths $\sigma_i^2 = [0.1,0.25,4,10]$. The bandwidths are chosen such that our Gaussian kernels are able to capture both the local features and the global features in the target distribution.

In the training of our MBL hidden Born machine, we use $N=6+2$ ($6$ visible spins and $2$ hidden spins), and $M=100$ quenches and search over $N=500$ different disorder realizations.

\section{Data encoding}
\label{app:MNIST}

Here, we describe the detailed data encoding scheme and our toy dataset of MNIST digit patterns in this section. Given a reduced density matrix $\rho_{\text{vis}}$ of $L$ visible spins, we compute the distribution of finding each of the $2^L$ basis states in our computational basis, and interpret the result as pixel values. We then reshape this probability vector into an image of size $2^{L/2} \times 2^{L/2}$.

On the other hand, given an image $\vec{x}_{\mu} \in \mathbb{R}^{n\times n}$, where $n\times n$ is the number of original pixels in the image, we first downsample it to $2^{L/2} \times 2^{L/2}$ pixels, then normalized the pixel values to be within $0$ and $1$. 

Our toy dataset of MNIST digit patterns are constructed as follows: we take all the training images $\vec{x}^{\mu}$ from a digit class, downsample to $2^{L/2} \times 2^{L/2}$ pixels, and compute each pixel as the average value $\bar{x}_i = 1/P \sum_{\mu=1}^P x^{\mu}_i$ across different styles within this digit class, where $i=1,...,2^L$. We then normalized the pixels to $\bar{x}_i \to \bar{x}_i/\sum_i \bar{x}_i$ and interpret the result as $q_{\text{data}}$. We take caution that this is different from learning the MNIST distribution in generative models. The latter refers to learning the joint probability distribution over all pixels in the image, and our toy data set corresponds to taking the mean-field limit of this joint probability distribution, which ignores the complicated correlations among pixels. This is akin to learning a single pattern (the averaged MNIST images shown in Fig.\ref{fig:MNIST_0_9}), and the reason for taking the average pixel value is such that we will be able to perform pattern recognition with imperfect initial states.

%%%%%%%%%%%%%%%%%%%%%%%%%%%%%%%%%%%%%%%%%%%%%%%%%%%%%%%%%%%%%%%%%%%%%%%%%%
%%%%%%%%%%%%%%%%%%%%%%%%%%%%%%%%%%%%%%%%%%%%%%%%%%%%%%%%%%%%%%%%%%%%%%%%%%
\section{Conclusion and outlook}
%%%%%%%%%%%%%%%%%%%%%%%%%%%%%%%%%%%%%%%%%%%%%%%%%%%%%%%%%%%%%%%%%%%%%%%%%%
%%%%%%%%%%%%%%%%%%%%%%%%%%%%%%%%%%%%%%%%%%%%%%%%%%%%%%%%%%%%%%%%%%%%%%%%%%
In this work, we have introduced the hidden MBL Born machine as a powerful quantum inspired generative model. Although parameterized quantum circuit has become one of the focal point in the realm of quantum machine learning, their training scheme poses many challenges as one requires to search in an exponential Hilbert space, which resembles finding a needle in haystack\cite{Jarrod_QNN}. While other variational algorithms such as QAOA offer a different scheme of finding solution in Hilbert space which is akin to adiabatic computing, here, by utilizing unique properties of MBL phase such as localization and memory, we develop a Born machine evolving under MBL dynamics such that by optimizing over values of disorder at each site we can reach a desired target state in the Hilbert space. 

Despite the localized nature of the MBL phase, we show the expressibilty of the MBL Born machine by mapping the 2D XXZ chain into 2D Ising model with proven expressibility advantage. Furthermore, we show that by including hidden units, we obtain expressive power advantage over the basic Born machine. We numerically demonstrate this advantage in learning both classical and quantum data. In this work, we aimed to answer two key questions, namely, whether MBL phase can be used as resource for learning, and what is the underlying mechanism of learning. By performing various numerical experiments in the thermal phase, non-interacting Anderson localized phase, and the MBL phase, we show that successful learning relies on both interaction and localization during training.

Our work opens up a new horizon in utilizing exotic quantum phases of matter as quantum inspired generative models. While we have explored the role of disorder in the MBL phase, an immediate question that follows is whether other disordered quantum phase would be capable of learning, which is left for future work. Furthermore, our quenched Born machine resembles specific adiabatic schedule, and whether we can utilize our model as quantum variational algorithm awaits further investigation. Although we have quantified the learning mechanism during the training by tracking both local and non-local quantities such as Hamming distance and entanglement entropy, more quantitative studies such as the existence of Barren Plateau and over-paramtrization in the context of quantum kernel learning remains an important question for future study\cite{Jarrod_QNN,Taylor_barren,liu2022analytic}.
%auto-ignore
\chapter{Future directions}
\label{future}

\section{Statistical mechanics of architecture-constrained neural networks}
Realistic machine learning tasks often operate with a number of parameters that far exceeds the number of data samples. In this over-parameterized regime, conventional computer science approaches to estimating generalization error bounds, such as Rademacher complexity and VC dimension, are no longer effective, and alternative methods that characterize average-case performance are required. Statistical mechanics of disordered many-body systems have been found to provide such a toolkit \cite{bahri2020statistical,engel2001statistical}. Furthermore, recent advances in different neural architectures call for a theory that incorporates such structural information. Building on Chapters \ref{disco}, \ref{vae}, and \ref{attractor}, a potential next step would be to establish a correspondence between architecture-constrained neural networks and spin glass models, viewing data as quenched disorder in the energy landscape, network parameters as dynamical degrees of freedom, and network architecture as constraints in the partition function. Using this correspondence, we could calculate the average-case generalization performance in the over-parameterized regime, by applying the replica method. One potential application is to theoretically elucidate the computational consequences of empirical neural architecture designs (such as dropout and skip connections). We could also use this formalism to search for new architecture designs, such as distribution-constrained layers (Chapter \ref{disco}), that are theoretically-principled and functionally-interpretable.

\section{Structure-function relation in neural and biological computation}
Biological functions are often determined by the underlying physical structures, and the physical structures themselves are in turn shaped by their functional purposes across different timescales. The elucidation of such structure-function relations has historically led to many important discoveries in biology from molecular \cite{honey2007network} to systems level \cite{will2011spliceosome}. Recently, breakthroughs in experimental technologies have led to a wealth of structural data such as high-resolution imaging of cells and brain connectomes. These emerging high-quality datasets present a remarkable opportunity for theoretical modeling. Extending Chapters \ref{disco} and \ref{attractor}, we could employ a combination of theoretical and data-driven approaches including but not limited to statistical mechanics, information theory, and machine learning, to investigate the relationship between physical structures and their corresponding information-processing capabilities. In particular, we could focus on the learning and memory aspects in neural systems such as cortical circuits in connectomics and navigation systems in the hippocampus. Our goal is to uncover hitherto unknown functional purposes of the observed structures in these data. Moreover, it would also be interesting to study how connectivity and interaction affect emergent computation capabilities in other biological systems, such as biochemical networks and molecular self-assembly.

\section{Modeling biological data with machine learning and statistical mechanics}
Traditional physical theories aim to use a small number of variables to capture the essence of complex phenomena. Such high-level abstraction offers conceptual simplicity but often has limited predictive power in realistic settings. This is especially the case in the era of high-throughput biological experiments, in which massive amounts of high-dimensional data defy simple description. High-dimensional statistics methods like machine learning can perform effective dimensionality reduction while preserving necessary details of the data, and have demonstrated great potential in computational modeling of various biological systems \cite{eraslan2019deep,richards2019deep}. Moreover, statistical mechanical models of many-body systems provide yet another effective approach to extracting global features of these systems \cite{lezon2006using,schneidman2006weak}.

It would be worthwhile to transfer the methods used in statistical mechanics of many-body systems (Chapters \ref{spinglass}, \ref{born}) and machine learning (Chapter \ref{vae}) to build effective models of complex biological phenomena such as cognitive function. Possible directions include modeling of sensory-motor neural activity, protein pathways in neural circuits, and gene regulatory networks. Our goal is that through scrutinizing effective computational models distilled from data, we will be able to gain insights into the underlying organization principles that govern the structure and function in these systems.

\section{Generative models with many-body dynamics}
\subsection{Classical}
State-of-the-art classical generative models, such as diffusion models, are inspired by non-equilibrium thermodynamics \cite{sohl2015deep}. As next steps, we would like to design new generative models that use many-body dynamics for learning, such as nucleation and self-assembly dynamics \cite{murugan2015multifarious,zhong2017associative}. Other than computer vision tasks, these non-equilibrium dynamical processes can be used to describe complicated distributions that defy simple equilibrium descriptions. This is particularly the case in data that arise in physical sciences, such as chemical reactions, protein interactions, and turbulent flows.

An interesting task would be to design domain-specific generative models to tackle these challenging datasets, with the goal of generating realistic data samples that could accelerate scientific simulation and experimental design. By leveraging many-body dynamics and non-equilibrium processes, these generative models could capture the inherent complexity of the data and provide more accurate representations of the underlying phenomena. Ultimately, this could lead to improved understanding of complex systems and facilitate the development of new technologies and methodologies in a variety of scientific domains.

\subsection{Quantum}
Building on generative models in Chapter \ref{vae}, \ref{born}, and many-body learning in Chapter \ref{spinglass}, we could extend our work into designing hybrid generative models that use natural many-body dynamics for learning. In particular, it would be interesting to chart out an atlas for many-body Born machines \cite{coyle2020born} based on two axes: (i) different Hamiltonians, such as Ising, XXZ, and spin glasses; (ii) different phases, ranging from thermal to many-body localized, and paramagnetic to glassy.

A potential goal is to search for universality in learning that could lead to the discovery of architecture-independent learning principles. Using those principles, one could potentially classify and understand different physical models based on the symmetries in their learned representations. This would provide a comprehensive framework for analyzing various many-body systems and their learning capabilities, leading to a deeper understanding of the relationship between the underlying physics and the learning performance of these systems.

Moreover, by exploring the space of many-body Born machines and identifying the common learning principles across different Hamiltonians and phases, we could potentially develop novel generative models that are better suited for specific tasks or datasets. This, in turn, could help advance the field of generative modeling and improve the efficiency of scientific simulations. 

%\appendix
%\include{appa}
%\include{appb}
%auto-ignore
%% This defines the bibliography file (main.bib) and the bibliography style.
%% If you want to create a bibliography file by hand, change the contents of
%% this file to a `thebibliography' environment.  For more information 
%% see section 4.3 of the LaTeX manual.
\begin{singlespace}
\bibliography{main}
\bibliographystyle{plain}
\end{singlespace}

\end{document}